\documentclass[11pt,reqno]{amsart}

\usepackage[margin=1.0in]{geometry}
\usepackage{amsmath}
\usepackage{amsthm, amsfonts,amssymb,euscript}
\usepackage{mathtools}
\usepackage{enumitem}
\usepackage{bm}
\usepackage{amssymb}
\usepackage{graphicx}
\usepackage{caption}
\usepackage{subcaption}
\usepackage{amsfonts}
\usepackage{float}
\usepackage{color}
\usepackage{slashed}
\usepackage[affil-it]{authblk}
\usepackage{mathrsfs}
\usepackage{upgreek}
\usepackage{pifont}
\usepackage{bbm}
\usepackage[colorlinks=true]{hyperref}

\usepackage[usenames,dvipsnames,svgnames,table]{xcolor}
\hypersetup{linkcolor=BrickRed, urlcolor=green, citecolor=blue, linktoc=page}

\numberwithin{equation}{section}

\newtheorem*{proposition*}{Proposition}
\newtheorem*{theorem*}{Theorem}
\newtheorem*{conjecture*}{Conjecture}
\newtheorem*{claim*}{Claim}
\newtheorem*{lemma*}{Lemma}
\newtheorem*{corollary*}{Corollary}
\newtheorem{theorem}{Theorem}[section]
\newtheorem{proposition}[theorem]{Proposition}

\newtheorem{lemma}[theorem]{Lemma}
\newtheorem{corollary}[theorem]{Corollary}

\newtheorem*{definition*}{Definition}
\newtheorem{definition}{Definition}[section]
\newtheorem*{assumption*}{Assumption}

\newtheorem*{remark*}{Remark}
\newtheorem{remark}{Remark}[section]

\newcommand{\la}{\langle}
\newcommand{\ra}{\rangle}
\newcommand{\R}{\mathbb{R}}
\newcommand{\s}{\mathbb{S}}
\newcommand{\C}{\mathbb{C}}

\newcommand{\Z}{\mathbb{Z}}
\newcommand{\N}{\mathbb{N}}

\newcommand{\lf}{\mathfrak{l}}

\DeclareMathOperator{\supp}{\textnormal{supp}}
\DeclareMathOperator{\re}{\textnormal{Re}}
\DeclareMathOperator{\im}{\textnormal{Im}}
\newcommand{\snabla}{\slashed{\nabla}}

\newcommand{\hpsi}{\widehat{\psi}}
\newcommand{\hphi}{\widehat{\phi}}

\newcommand{\ssubset}{\subset\joinrel\subset}

\allowdisplaybreaks

\usepackage[foot]{amsaddr}

\begin{document}

\title{Quasinormal modes on Kerr spacetimes}
\author{Dejan Gajic$^1$}
\author{Claude M. Warnick$^2$}
\address{$^1$\small Institut f\"ur Theoretische Physik, Universit\"at Leipzig, Br\"uderstrasse 16, 04103 Leipzig, Deutschland}
\address{$^2$\small Department of Applied Mathematics and Theoretical Physics, University of Cambridge, Wilberforce Road, Cambridge CB3 0WA, United Kingdom}
\email{$^1$dejan.gajic@uni-leipzig.de, $^2$cmw50@cam.ac.uk}
\date{}
\maketitle

\begin{abstract}
We introduce a rigorous framework for defining quasinormal modes on stationary, asymptotically flat spacetimes as isolated eigenvalues of the infinitesimal generator of time translations. We consider time functions corresponding to a foliation of asymptotically hyperboloidal hypersurfaces and restrict to suitable Hilbert spaces of functions. These functions have finite Sobolev regularity in bounded regions, but need to be Gevrey-regular at null infinity. This framework is developed in the context of sub-extremal Kerr spacetimes, but also gives uniform-in-$\Lambda$ resolvent estimates on Kerr--de Sitter spacetimes with a small cosmological constant $\Lambda$. As a corollary, we also construct the meromorphic continuation (in a sector of the complex plane) of the cut-off resolvent in Kerr that is associated to the standard Boyer--Lindquist time function. The framework introduced in this paper bridges different notions of quasinormal modes found in the literature.  As further applications of our methods, we prove stability of quasinormal frequencies in a sector of the complex plane, with respect to suitably small perturbations and establish convergence properties for Kerr--de Sitter quasinormal frequencies when the cosmological constant approaches zero.
\end{abstract}

\tableofcontents

\section{Introduction}
\label{sec:intro}
The quantitative behaviour of gravitational radiation emitted at late times, as black hole solutions to the Einstein equations settle down to a final, stationary state, is intimately tied to the geometry of the stationary spacetime that emerges as the final state. In the limit as time $t\to\infty$, the dynamics are expected to be governed by solutions to wave equations, which decay inverse-polynomially in time, with powers that, in many cases, only depend on the asymptotic flatness of the spacetime and therefore are  not characteristic to the particular stationary spacetime under consideration; see \cite{Price1972,paper2,hintzprice, aagkerr, aagprice, lukoh24} and references therein.\footnote{Extremal black holes form a notable exception. In the case of extremal Kerr black holes, the decay rates can differ significantly from the sub-extremal case, due to the dominant role played by the degenerate event horizon. In this case, the inverse-polynomial decay rates of gravitational radiation could actually be used as a signature of the extremality of the stationary spacetime \cite{gaj22b}. See also \cite{aag18,bgs23, aags23} for more work on observational signatures of extremality. We refer \S \ref{sec:future} for a further discussion on quasinormal modes on extremal Kerr.}

In late, \emph{but finite}, time intervals, however, an important role is instead expected to be played by \emph{quasinormal modes}. These are oscillating and exponentially damped solutions to wave equations that govern the spacetime dynamics, with frequencies that form a discrete set and are in fact characteristic to the particular geometry of the stationary spacetime in question. They form a dispersive analogue to the purely oscillating, vibrational \emph{normal modes} of a vibrating string or drumhead.

In contrast with the decay rates of inverse-polynomial tails, the complex quasinormal frequencies, which determine the temporal behaviour of quasinormal modes, may serve as characteristic signatures of the geometry of the stationary spacetime that can be read off from gravitational radiation emitted to future null infinity. This type of inverse problem is known as ``black hole spectroscopy''. For this reason, quasinormal modes are of particular interest to gravitational wave experiments, where they have been used to estimate the masses and angular momenta of Kerr black hole spacetimes that are expected to arise asymptotically as the final states of various observed black hole mergers; see for example \cite{ligo21}.

While they have played a prominent role in the theoretical and experimental physics literature on black hole dynamics, the status of quasinormal modes on asymptotically flat spacetimes as well-defined mathematical objects is far less developed. 

In this paper, we present a mathematical construction of quasinormal modes in the context of the geometric wave equation on fixed sub-extremal Kerr black hole exteriors \cite{kerr63}:
\begin{equation}
\label{eq:waveeq}
\square_{g_{M,a}}\phi=0,
\end{equation}
with $M$ the mass of the Kerr spacetime and $a$ its specific angular momentum. In fact, the results in the present paper apply uniformly (in $\Lambda$) to Kerr(--de Sitter) spacetimes with cosmological constant $\Lambda\geq 0$. Indeed, we consider also the more general, conformally covariant Klein-Gordon equation:
\begin{equation}
\label{eq:waveeqds}
\square_{g_{M,a,\Lambda}}\phi=\frac{2}{3}\Lambda\phi,
\end{equation}
with $g_{M,a,\Lambda}$ a Kerr--de Sitter metric with $\Lambda\geq 0$, $M^2\Lambda\ll 1$, such that the event horizon is non-degenerate\footnote{This means that the surface gravity associated to the event horizon is strictly positive (but the surface gravity of the cosmological horizon goes to zero as $\Lambda\downarrow 0$).}. We summarize the main results below in two theorems.

The first theorem demonstrates that the spectrum of $\mathcal{A}$, the infinitesimal generator of time translations associated to \eqref{eq:waveeqds} with $\Lambda\geq 0$, consists of isolated eigenvalues, when $\mathcal{A}$ is restricted to appropriate Hilbert space domains, and it moreover includes the poles of meromorphic continuations of cut-off resolvent operators as a subset. This therefore motivates a new definition of quasinormal frequencies as eigenvalues of restrictions of $\mathcal{A}$, extending a definition of quasinormal frequencies as \emph{scattering resonances}, i.e.\  poles of meromorphic continuations of cut-off resolvent operators.
\begin{theorem}[Rough version]
\label{thm:rough1}
Let $\mathcal{A}$ be the infinitesimal generator of time translations corresponding to \eqref{eq:waveeqds} on a Kerr(--de Sitter) metric with a non-degenerate event horizon, with initial data on an asymptotically hyperboloidal, horizon-intersecting hypersurface. Then:
\begin{enumerate}[label=\emph{(\roman*)}]
\item There exists a family of subsets $\Omega_{\alpha}\subset \C$ and a family of restrictions $\mathcal{A}_{\alpha}$ of $\mathcal{A}$ to suitable initial data Hilbert spaces $\mathbf{H}_{\alpha}$, such that $\mathcal{A}_{\alpha}$ has a pure point spectrum of isolated eigenvalues in $\Omega_{\alpha}$ (``\textbf{regularity quasinormal frequencies}'') with corresponding finite-dimensional eigenspaces (spanned by ``\textbf{regularity quasinormal modes}''). The regularity quasinormal frequencies are independent of the precise choice of initial data hypersurface and of the index $\alpha$.
\item The union of all $\Omega_{\alpha}$ satisfies
\begin{equation*}
\Omega:=\bigcup_{\alpha}\Omega_{\alpha}=\left\{{|\rm arg} z|<\frac{2\pi}{3}\right\},
\end{equation*}
and the corresponding union of all eigenvalues, denoted $\mathscr{Q}_{\rm reg}$, has no accumulation points in $\Omega$.
\item The standard cut-off resolvent associated to \eqref{eq:waveeq} with respect to the Boyer--Lindquist time function $t$ can be meromorphically continued to $\Omega$ and its poles are contained in $i\mathscr{Q}_{\rm reg}$.
\end{enumerate}
\end{theorem}

In the second theorem, we establish the stability of the quasinormal spectrum under suitably small perturbations and the existence of convergent sequences of quasinormal frequencies in the limit $\Lambda\downarrow 0$.

\begin{theorem}[Rough version]
\label{thm:rough1b}
The following additional perturbative properties hold for regularity quasinormal frequencies associated to \eqref{eq:waveeqds} with $\Lambda\geq 0$:
\begin{enumerate}[label=\emph{(\roman*)}]
\item When $\Lambda=0$, the quasinormal frequencies in $\mathscr{Q}_{\rm reg}$ are stable under perturbations of the wave operator that are suitably small with respect to the norms on $\mathbf{H}_{\alpha}$.
\item Let $\{\Lambda_n\}$ be a sequence of cosmological constants with $\Lambda_n\downarrow 0$ as $n\to \infty$. For each quasinormal frequency $s_*\in \mathscr{Q}_{\rm reg}$ with respect to $\Lambda=0$, there exists a sequence $\{s_n\}$ of quasinormal frequencies in $\mathscr{Q}_{\rm reg}$ with respect to $\Lambda=\Lambda_n$, with limit $s_*$. Conversely, for each sequence $\{s_n\}$ of quasinormal frequencies in $\mathscr{Q}_{\rm reg}$ with respect to $\Lambda=\Lambda_n$ and with limit $s_*\in \Omega$, $s_*$ must be a quasinormal frequency in $\mathscr{Q}_{\rm reg}$ with respect to $\Lambda=0$. 
\end{enumerate}
\end{theorem}

Precise versions of the above statements (including definitions of the Hilbert spaces $\mathbf{H}_{\alpha}$) are given in \S \ref{sec:mainthmprecise}, following the introduction of the necessary notation and geometric and functional analytic concepts in \S\S \ref{sec:geometry}--\ref{sec:infigen}. The main ideas and techniques in the proof of Theorem \ref{thm:rough1} are already outlined in \S \ref{sec:ideastechniques}.

By Theorem \ref{thm:rough1}(iii), the set $i\mathscr{Q}_{\rm reg}$, includes the set of scattering resonances in an appropriate sector of the complex plane, which we will take to be the poles of meromorphically continued cut-off resolvent operators that are defined with respect to ``standard'' time functions with asymptotically flat level sets that intersect the bifurcation sphere; see also \S \ref{sec:scatres}.\footnote{The set $i\mathscr{Q}_{\rm reg}$ also includes poles of meromorphically continued cut-off resolvents with respect to time functions that intersect the future event horizon to the future of the bifurcation sphere. }  In this sense, regularity quasinormal frequencies are more general objects than scattering resonances. In contrast with scattering resonances, regularity quasinormal frequencies are natural objects to consider in the setting of initial data whose support includes the event horizon of the Kerr black hole, or which extend to future null infinity in conformal sense.

The definition of regularity quasinormal modes in the present paper does \emph{not} rely on a decomposition of the corresponding solution to the wave equation in terms of oblate spheroidal harmonics via Carter separation \cite{car68}. Indeed, for frequencies $\omega\in i\mathscr{Q}_{\rm reg}\subset \C\setminus \R$, such a decomposition is \emph{not guaranteed}; see Remark \ref{rmk:radialode} for more details. 

The tools developed in the present paper can be applied to more general wave equations on Kerr, like the Teukolsky equations, and to wave equations on more general stationary, asymptotically flat spacetimes; see \S \ref{sec:future} for a further discussion.

In the final stages of writing this paper, we learned about upcoming related, but independent, work by Stucker \cite{stuc24} on scattering resonances for \eqref{eq:waveeq} that complement the results in the present paper. In \cite{stuc24}, meromorphic continuations are constructed for cut-off resolvents corresponding to time functions with future-horizon intersecting, asymptotically flat level sets by combining the constructions introduced by Vasy in \cite{vasy1} in a spacetime region of bounded radius, together with a general complex scaling argument in the remaining region, analogous to the complex scaling argument introduced in \cite{sjzw91}.

An advantage of a complex scaling-based approach is that it allows for a construction of scattering resonances in the whole region $\C\setminus i(-\infty,0]$, in contrast with Theorem \ref{thm:rough1}(iii), which only provides such a construction in a sector of the complex plane. 

A disadvantage of such an approach is that, in contrast with the methods introduced in the present paper, it relies fundamentally on the conjugation of resolvent operators with cut-off functions, which limits the applicability of resolvent estimates to the setting of compactly supported initial data for \eqref{eq:waveeq} and does not allow for the consideration of the radiation field $r\phi|_{\mathcal{I}^+}$ at future null infinity. 

The method of complex scaling moreover does not directly allow for a uniform analysis of Kerr--de Sitter spacetimes with $\Lambda>0$ (where instead the methods of \cite{vasy1, warn15, peva21} apply) as $\Lambda\downarrow 0$, whereas the methods introduced in the present paper enable a uniform analysis of resolvent operators for $\Lambda\geq 0$, with $M^2\Lambda\ll 1$. This enables the study of the convergence properties of Kerr--de Sitter quasinormal frequencies as $\Lambda\downarrow 0$ in Theorem \ref{thm:rough1b}(ii).

\section{Acknowledgements}
The authors thank Peter Hintz and Thomas Stucker for several insightful conversations. DG acknowledges funding through the ERC Starting Grant 101115568.

\section{Related works on quasinormal modes}
In this section, we give a overview of previous literature on quasinormal modes and provide some context to the results in Theorems \ref{thm:rough1} and \ref{thm:rough1b}.

\subsection{Early numerical work}
The first numerical observation of the dominance of characteristic oscillations (independent of the precise initial data under consideration) in the time evolution of linear wave equations on asymptotically flat black holes was made by Vishveshwara \cite{vish70}. It was observed that at late, but finite time intervals, angular momentum $\ell=2$ solutions to the Regge--Wheeler wave equations on fixed Schwarzschild background spacetimes are dominated by a damped oscillation at a single-frequency, depending only on the Schwarzschild mass $M$.

Further numerical investigations for higher angular momenta $\ell$ were carried out by Press \cite{press71}, who coined the term ``quasinormal mode'' to describe these characteristic oscillations; ``normal mode'' because of the occurrence of characteristic oscillations, analogous to the normal mode solutions of wave equations on compact spatial domains and ``quasi'' because, in contrast with normal modes, these oscillations seemed to be damped and do not fully determine the global temporal behaviour of solutions arising from generic initial data.

An identification of quasinormal modes as a particular class of solutions to the wave equation with exact time dependence $e^{-i\omega t}$, where $\omega\in \C$, was made in a numerical analysis by Chandrasekhar--Detweiler \cite{chadet75}; see also related earlier work by Zerilli \cite{zer70}. Solutions to \eqref{eq:waveeq} on Schwarzschild of the form $\phi(t,r_*,\theta,\varphi)=e^{-i\omega t}\frac{u(r_*)}{r}Y_{\ell m}(\theta,\varphi)$, with $Y_{\ell m}$ spherical harmonics, satisfy the following radial ODE:
\begin{equation}
\label{eq:radialODEschw}
u''+(\omega^2-V_{\ell})u=0,	
\end{equation}
where
\begin{equation*}
	V_{\ell}(r_*)\sim \begin{cases}
		e^{-\frac{r_*}{4M}}\quad r_*\to-\infty,\\
		\frac{\ell(\ell+1)}{r_*^2}\quad r_*\to +\infty.
	\end{cases}
\end{equation*}
By standard asymptotic ODE analysis, it can be seen that $u\sim e^{\pm i \omega r_*}$ as $r_*\to \pm \infty$. If $V$ were compactly supported in $r_*$, then the behaviour $e^{\pm i \omega r_*}$ would in fact be exact, rather than merely asymptotic. In \cite{chadet75}, quasinormal modes were characterized as solutions to \eqref{eq:waveeq} with a radial part $u$ that satisfies \eqref{eq:radialODEschw}, such that $u\sim e^{- i \omega r_*}$ as $r_*\to -\infty$ (``ingoing at the horizon'') and $u\sim e^{+ i \omega r_*}$ as $r_*\to \infty$  (``outgoing at infinity'').``Quasinormal frequencies" would then correspond to the values of $\omega$ which allow for non-trivial solutions $u$ satisfying these boundary conditions to exist.

When $\im \omega<0$, this characterization amounts to demanding exponential growth in $|r_*|$ as $r_*\to \pm \infty$. However, as the relation $u\sim e^{\pm i \omega r_*}$ is merely an asymptotic one, this characterization is ill-defined:\footnote{Note that if $V_{\ell}$ were compactly supported in $r_*$, then the behaviour $e^{\pm i \omega r_*}$ would be exact, so there would be no ambiguity in selecting the desired boundary conditions.} it is unclear how to filter out solutions with exponential decay in $|r_*|$, as they are always dominated by exponential growth. This ambiguity was already observed in \cite{chadet75}.

\subsection{Quasinormal modes and converging power series}
In the numerical scheme of \cite{chadet75}, ingoing/outgoing boundary conditions are implemented by requiring $u$ to take the form of convergent power series towards $r_*=-\infty$ and $r_*=\infty$:
\begin{equation}
\label{eq:chadetQNM}
	u(r_*)= \begin{cases}
 	e^{ i \omega r_*}\sum_{j=0}^{\infty}a_jr^{-j}\quad \textnormal{$r$ large}\\
 	e^{ -i \omega r_*}\sum_{j=0}^{\infty}b_j(r-2M)^j\quad \textnormal{$r-2M$ small}
 \end{cases}
\end{equation}
and then estimating $\omega$ for which these power series (up to some large number of terms) match at an intermediate value of $r$, say $r=3M$.

Consistent with this approach is a characterization of quasinormal modes due to Leaver \cite{leav85}, who defined quasinormal modes in Schwarzschild as solutions to \eqref{eq:waveeq} of the form $\phi(t,r_*,\theta,\varphi)=e^{-i\omega t}\frac{u(r_*)}{r}Y_{\ell m}(\theta,\varphi)$ with a single power series expression for $u$ converging on $[2M,\infty)$:
\begin{equation}
\label{eq:leaverQNM}
u(r)=\left(1-\frac{2M}{r}\right)^{-i2M\omega}r^{i2M\omega}e^{i\omega(r-2M)}\sum_{j=0}^{\infty}c_j\left(1-\frac{2M}{r}\right)^j.
\end{equation}
The requirement that $u$ can indeed be written as the above power series imposes restrictions on the class of solutions $u$ to \eqref{eq:radialODEschw} and implements therefore a notion of ``ingoing at the horizon'' and ``outgoing at infinity. Indeed, by rewriting \eqref{eq:leaverQNM} in terms of the tortoise radial coordinate $r_*$, it is straightforward to show that \eqref{eq:chadetQNM} follows from \eqref{eq:leaverQNM}.

Note that the convergence of the power series at $r=2M$ is consistent with the fact that $r=2M$ is a regular singular point of the ODE \eqref{eq:radialODEschw}, so it admits solutions that are analytic at $r=2M$. The boundary condition ``ingoing at the horizon'' can be interpreted as the requirement that $e^{+i \omega r_*}u$ be \emph{analytic} in $r$ at $r=2M$. In fact, for solutions $u$ with the asymptotic behaviour $e^{+i\omega r_*}$, the expression $e^{+ i \omega r_*}u$ has only finite regularity in $r$ at the event horizon (with the order of regularity depending on $\im \omega$), so the ingoing boundary condition can already be implemented by restricting to $u$ with \emph{smooth} $e^{i \omega r_*}u$ (or $C^k$, with $k$ sufficiently large, depending on $\im \omega$).

The situation at $r=\infty$ is more complicated, because it corresponds to an \emph{irregular} singular point of the ODE \eqref{eq:radialODEschw}. In this case, the power series appearing in the expression for general $u$ is only asymptotic and the rescaled $e^{ -i \omega r_*}u$ will \emph{always} be smooth in $x$ at $x=0$, with $x=\frac{1}{r}$. Demanding the convergence of \eqref{eq:leaverQNM} therefore implies that the following power series around $x=\frac{1}{2M}$ converges in $[0,\frac{1}{2M}]$:
\begin{equation*}
\left(1-2M x\right)^{i4M\omega}(e^{ -i \omega r_*}u)(x)=\sum_{j=0}^{\infty}c_j(-2M)^j\left(x-\frac{1}{2M}\right)^j.
\end{equation*}
This can be thought of as a possible characterization of the ``outgoing at infinity'' boundary condition. Note however that this is a strictly weaker condition than analyticity in $[0,\frac{1}{2M}]$, since that would also require convergence for (suitably small) $x<0$ and not just for $x\geq 0$. In contrast with the boundary condition at the horizon, it is not clear \emph{a priori} if there exist \emph{any} solutions with the desired convergence property at $x=0$.

The ingoing and outgoing boundary conditions of \cite{chadet75,leav85} have a cleaner interpretation when expressed in terms of a time function $\tau$ whose level sets are appropriate asymptotically hyperboloidal or null hypersurfaces intersecting the future event horizon and future null infinity. In that case, a quasinormal mode in the sense of Leaver would correspond to a solution to \eqref{eq:waveeq} on Schwarzschild of the form $\phi(\tau,x=\frac{1}{r},\theta,\varphi)=e^{-i\omega \tau}x\hat{\psi}_{\ell m}(x)Y_{\ell m}(\theta,\varphi)$, such that $\hat{\psi}_{\ell m}$ is smooth (or analytic) at the horizon $x=\frac{1}{2M}$ and such that $\hat{\psi}_{\ell m}$ can be expanded as follows for $x\geq 0$:
\begin{equation}
\label{eq:hypfoliationmode}
\hat{\psi}_{\ell m}(x)=p(x;M,\omega)\sum_{j=0}^{\infty}c_j(-2M)^j\left(x-\frac{1}{2M}\right)^j,
\end{equation}
with $p$ analytic at $x=0$ and smooth at $x=\frac{1}{2M}$, depending on the precise choice of asymptotically hyperboloidal or null foliation.

The ingoing and outgoing boundary conditions can therefore be thought of as \emph{regularity conditions} at the future event horizon and future null infinity imposed on solutions to \eqref{eq:waveeq}. See also \cite{jamash21} for recent numerical methods exploiting the advantages of the asymptotically hyperboloidal setting.

Using the characterization \eqref{eq:leaverQNM}, Leaver developed a numerical scheme for determining quasinormal frequencies in \cite{leav85} (i.e.\ values of $\omega$ for which \eqref{eq:leaverQNM} holds) based around a continued fraction method introduced by Jaff\'e \cite{jaf34} and generalized by Baber--Hass\'e \cite{baha35} in the context of the Schr\"odinger equation for hydrogen-molecule-like ions. This method was also applied to the Kerr setting, where quasinormal modes were taken to be solutions of the form 
\begin{equation}
\label{eq:kerrmode}
\phi(t,r_*,\theta,\varphi)=e^{-i\omega t}\frac{u(r_*)}{r}S_{\ell m}(\theta;a\omega)e^{im\varphi},
\end{equation}
 with $S_{\ell m}(\theta;a\omega)e^{im\varphi}$ oblate spheroidal harmonics and $u$ satisfying an equation analogous to \eqref{eq:leaverQNM}, but with a dependence on the angular eigenvalues associated to $S_{\ell m}(\theta;a\omega)e^{im\varphi}$; see Remark \ref{rmk:radialode} for more details regarding subtleties connected with the assumption \eqref{eq:kerrmode}.

For a comprehensive overview of additional numerical results on quasinormal modes, we refer the reader to the review papers \cite{kosc99,bcs09}.

\subsection{Scattering resonances}
\label{sec:scatres}
An alternative way to characterize quasinormal modes is motivated by the study of scattering resonances. When $\im \omega>0$, we can umambiguously fix ingoing boundary conditions at the horizon and outgoing boundary conditions at null infinity, since in this case, this amounts to selecting solutions that \emph{decay} exponentially in $|r_*|$, rather than those that grow exponentially, as would be the case for $\im \omega<0$. 

Let $\im\omega>0$ and let $u_+(r_*)$ and $u_{\infty}(r_*)$ denote solutions to \eqref{eq:radialODEschw} such that $u_+(r_*)\sim e^{-i\omega r_*}$ as $r_*\to-\infty$ and $u_{\infty}(r_*)\sim e^{+i\omega r_*}$ for $r_*\to \infty$. Let $W$ denote the Wronskian corresponding to the solutions $u_+$ and $u_{\infty}$, which depends only on $\omega$.

By an energy conservation argument, it follows that $W$ does not have any zeroes when $\im \omega>0$ and the resolvent operator;
\begin{equation*}
	R_{\ell}(\omega)=[\partial_{r_*}^2+(\omega^2-V_{\ell})]^{-1}: L^2(\R)\to L^2(\R)
\end{equation*}
can be defined via Green's formula:
\begin{equation}
\label{eq:resolventschwell}
	R_{\ell}(\omega)(f)=W^{-1}\left(u_{\infty}(r_*)\int_{-\infty}^{r_*} u_{+}(y)f(y)\,dy+u_{+}(r_*)\int_{r_*}^{\infty}u_{\infty}(y)f(y)\,dy\right).
\end{equation}

If we assume that we can make sense of the analytic continuation of $W$ to $\C\setminus i(-\infty,0]$, the meromorphic continuation of the \emph{cut-off resolvents} $\chi R_{\ell}(\omega)\chi: L^2(\R)\to L^2(\R)$ to $\im \omega\leq 0$, where $\chi$ are arbitrary smooth compactly supported function, is well-defined. The poles of this meromorphic continuation correspond to zeroes of the analytic continuation of $W$ and are known as \emph{scattering resonances}.

In light of the above, Detweiler suggested an alternative characterization of ``quasinormal frequencies'' in \cite{det77}, by identifying them with scattering resonances, where quasinormal modes are the corresponding resonant states. In contrast with Leaver's characterization, which is based around \eqref{eq:leaverQNM}, the relevance of scattering resonances to the temporal behaviour of solutions to \eqref{eq:waveeq} is clearer. Indeed, by denoting $s=-i\omega$, one can express $\phi_{\ell m}(t,r)$ via the \emph{Bromwich integral}: let $s_0\in (0,\infty)$, then
\begin{equation}
\label{eq:bromwich}
	(\chi\cdot \phi_{\ell m})(t,r)=\frac{1}{2\pi i}\lim_{S\to \infty}\int_{s_0-iS}^{s_0+iS}[(\chi R_{\ell}(is)\chi)(f_{\ell m})](r)\,ds,
\end{equation}
where $f_{\ell m}$ is determined by initial data for \eqref{eq:waveeq} at $t=0$: $f_{\ell m}=-s\phi_{\ell m}|_{t=0}-\partial_t\phi_{\ell m}|_{t=0}$.

The above contour of integration in $\{\re s>0\}$ (i.e.\ $\im \omega >0$) can be deformed to obtain a contour intersecting the region $\{\re s<0\}$ (i.e.\ $\im \omega <0$) that encloses potential poles of $\chi R_{\ell}(\omega)\chi$. By the residue theorem, the poles will therefore contribute to $\chi\phi_{\ell m}(t,r)$ as terms with time dependence $e^{-i\omega t}$, multiplied by a solution to \eqref{eq:radialODEschw} (the resonant state).

Note that from their definitions, it is not immediately clear whether scattering resonances actually correspond to the same values of $\omega\in \C$ as those that give solutions to \eqref{eq:radialODEschw} satisfying Leaver's power series expansion (though there is ample numerical evidence in favour of such a correspondence). See \S \ref{sec:regmodesasympfl} for a further discussion on this matter.

In \cite{ba}, Bachelot--Motet-Bachelot proved that the cut-off resolvents $\chi R_{\ell}(\omega)\chi$ corresponding to the resolvent in \eqref{eq:resolventschwell} can indeed be meromorphically continued to $\C\setminus i(-\infty,0]$ for all fixed angular momenta $\ell$, so the notion of quasinormal frequencies as scattering resonances is mathematically well-defined. Furthermore, the results in \cite{ba} imply that for fixed $\ell$, there are no accumulation points away from $i(-\infty,0]$, so the set of quasinormal frequencies in any sector of the complex plane that excludes $i(-\infty,0]$ is a discrete subset, provided one restricts to fixed $\ell$.

The proof of the meromorphic continuation of the resolvent in \cite{ba} uses the method of \emph{complex scaling}, developed in \cite{agco71,baco71,si73}, which relies on the analyticity of $V$ in $r_*$.

Scattering resonances have a long history in mathematical physics and quantum mechanics. A general theory for defining and studying scattering resonances for the wave equation on Minkowski with compactly supported potentials was introduced by Lax--Phillips \cite{laxphillips}. We refer the interested reader to \cite{dyzw19} for an introduction to scattering resonances as well as compendium of relevant references.

\subsection{Regularity quasinormal modes in asymptotically de Sitter and anti de Sitter spacetimes}
A drawback of the Leaver characterization of quasinormal modes is that it does not offer a clear connection to the temporal evolution of solutions to \eqref{eq:waveeq} arising from prescribed initial data. Furthermore, it is not immediate how to place this method on a solid mathematical footing and derive properties about the quasinormal spectrum like its discreteness or the asymptotic behaviour of quasinormal frequencies with large frequencies $|\re  \omega| \gg M^{-1}$. 

The characterization of quasinormal modes as resonant states of cut-off resolvents, on the other hand, is more amenable to a mathematical study. It also provides a clear link with the time evolution of solutions to \eqref{eq:waveeq}, but only for \emph{compactly supported initial data} that are moreover supported away from the event horizon. Furthermore, due to appearance of the cut-off function $\chi$ on the left-hand side of \eqref{eq:bromwich}, it is not clear how resonant states contribute to the time evolution along future null infinity, the region that is most relevant for modelling gravitational wave observations.

\emph{Regularity quasinormal modes} provide a way to bridge these two different characterizations of quasinormal modes and resolve the above mentioned issues. They were first introduced in \cite{vasy1} in the context of asymptotically hyperbolic manifolds and on a subfamily of Kerr--de Sitter spacetimes. A related characterization was provided by the second author in \cite{warn15} in the context of asymptotically anti de Sitter spacetimes.

The definition of regularity quasinormal modes in asymptotically (anti) de Sitter spacetimes can be motivated by the following two observations:
\begin{itemize}
\item With respect to a time function $\tau$ whose level sets intersect the future event horizon and the future cosmological horizon (in the asymptotically de Sitter case) away from the bifurcation sphere(s), the ingoing (at the event horizon) and outgoing (at the cosmological horizon) boundary conditions can be implemented by requiring the analogues of the functions $\hpsi_{\ell m}$ from \eqref{eq:hypfoliationmode} to be elements of the Sobolev space of functions $H^k$, where $k$ is determined by the subset of the complex plane under consideration and is related to the surface gravities $\kappa_+$ and $\kappa_c$ of the event horizon and cosmological horizon, respectively.
\item It is possible to construct a resolvent operator $R(s):  H^{k}\to H^{k+1}$ with respect to the time function $\tau$ for $\re s>-\frac{\min\{\kappa_+,\kappa_c\}}{2}(1+2k)$, rather than constructing the resolvent operator in $\{\re s>0\}$ and then meromorphically continuing it after composition with cut-off functions to a sector in $\{\re s<0\}$.

In fact, the poles of this resolvent operator correspond to the eigenvalues of the infinitesimal generator of times translations $\mathcal{A}$, with $S(\tau')=e^{s \tau' \mathcal{A}}$ the map that sends initial data $(\phi,\partial_{\tau}\phi)$ at $\tau=0$ to $(\phi,\partial_{\tau}\phi)$ at $\tau=\tau'$.
\end{itemize}
Regularity quasinormal modes in this setting can therefore be \emph{defined} as the eigenfunctions of the infinitesimal generator of time translations $\mathcal{A}$, restricted to $H^{k+1}\times H^k$ with $k\in \N_0$ suitably large. The corresponding eigenvalues $s$ are the regularity quasinormal frequencies. 

It can easily be shown that the set of regularity quasinormal frequencies \emph{includes} the set
\begin{equation*}
	\{-i\omega\:|\: \textnormal{$\omega$ scattering resonance}\};
\end{equation*}
see for example \cite{warn15}. This inclusion can, however, be \emph{strict}. This, for example, is the case in the model problem studied in \S6 of \cite{warn15} and for the wave equation on the de Sitter spacetime \cite{hixi21, joy22}.

Let us also mention the earlier work of S\'a Barreto--Zworski \cite{bazw97}, who constructed and studied the properties of scattering resonances on Schwarzschild--de Sitter at large values of $|\re \omega|$ and Bony--H\"afner \cite{bonyh}, who derived resonance expansions on Schwarzschild--de Sitter. See also the recent work on the distribution of Schwarzschild(--de Sitter) scattering resonances in \cite{hizw24}. The results of \cite{bazw97, bonyh} were generalized to the slowly-rotating Kerr--de Sitter setting by Dyatlov \cite{semyon1, dya12}. 

The construction of regularity quasinormal modes from \cite{vasy1}  has recently also been extended to the full sub-extremal Kerr--de Sitter family of spacetimes in \cite{peva21}. In the asymptotically anti de Sitter setting, it was shown by Gannot \cite{gan14} that scattering resonances on Schwarzschild-anti de Sitter approach the real axis as $|\re \omega|\to\infty$, which is intimately connected to the presence of stable trapping of null geodesics and the sharpness of energy decay estimates with slow, logarithmic decay in Kerr-anti de Sitter; see \cite{gusmu2}.

Recently, the existence of slowly-damped (``zero-damped'') modes in a neighbourhood of $0\in \C$ was proved in \cite{hixi22} in the setting of Schwarzschild--de Sitter with $M^2\Lambda \ll 1 $, and similarly in near-extremal Reissner--Nordstr\"om--de Sitter spacetimes (where the difference between the event and inner horizon radius is very small) in \cite{joy22}. See also \cite{hin24} for analogous results in sub-extremal Kerr--de Sitter spacetimes. These constructions make use of \emph{quasinormal co-modes} or \emph{dual resonant} states; see also Theorem \ref{thm:mainthmA}(iv) for a characterization of such objects in the asymptotically flat setting.

\subsection{Regularity quasinormal modes in asymptotically flat spacetimes}
\label{sec:regmodesasympfl}
In the asymptotically flat setting, a definition of regularity quasinormal modes is more complicated, due to the fact that the outgoing boundary condition at future null infinity cannot be implemented by considering Sobolev Hilbert spaces. Indeed, the convergence of the power series in \eqref{eq:hypfoliationmode} cannot even be guaranteed by restricting to the space of smooth functions.

This difficulty was first overcome in the setting of the geometric wave equation on extremal Reissner--Nordstr\"om spacetimes by the authors in \cite{gajwar19a} and in a one-dimensional toy model setting in \cite{gajwar19b}. The key insight in these papers is the use of $L^2$-based spaces of functions which are $(2,\sigma)$-Gevrey-regular at infinity, with $\sigma>0$ depending on the region of the complex plane under consideration (i.e.\ Gevrey-regular in $x=\frac{1}{r}$ at $x=0$). A function $f:[0,1]\to \R$ can be said to be $(\alpha,\sigma)$-Gevrey-regular if there exists a constant $C_{\sigma}>0$, such that for all $k\in \N_0$:
\begin{equation*}
	\|f^{(k)}\|_{L^{2}([0,1])}\leq C_{\sigma} \sigma^{-k}k!^{\alpha}.
\end{equation*}
If the above holds for all $\sigma\in \R$, the function $f$ is said to be $\alpha$-Gevrey-regular.

The space of $(\alpha,\sigma)$-Gevrey-regular functions with $\alpha> 1$ therefore may be thought of as sitting between the space of analytic functions ($\alpha=1$) and the space of smooth functions in terms of regularity.

In \cite{gajwar19a,gajwar19b} it was moreover shown that the set of regularity quasinormal frequencies includes the set $\{-i\omega\:,\: \textnormal{$\omega$ scattering resonance}\}$ in a sector of the complex plane excluding the half-line $(-\infty,0]$.

Furthermore, \cite{gajwar19b} demonstrated that the convergence of the power series in \eqref{eq:hypfoliationmode} implies Gevrey regularity of $\hpsi_{\ell m}$ and therefore the Leaver quasinormal frequencies are also included when considering regularity quasinormal frequencies.

As a \emph{corollary} of the methods developed \cite{gajwar19a,gajwar19b}, it was moreover shown that, in the one-dimensional toy model setting of \cite{gajwar19b}, the construction of the meromorphic continuation of the standard cut-off resolvent applies even when the potential under consideration is $(2,\sigma)$-Gevrey-regular, rather than analytic. A straightforward implementation of the complex scaling method fails in this case, but it was observed in \cite{gazw21} that $(2,\sigma)$-Gevrey-regular potentials can be split into an analytic part and an exponentially decaying part, the latter which can treated as a small perturbation when studying cut-off resolvents. With this splitting, classical implementations of the complex scaling method therefore also apply to this $(2,\sigma)$-Gevrey-regular toy model, at least in the context of cut-off resolvents.

In comparison with the results in \cite{gajwar19a,gajwar19b}, Theorem \ref{thm:rough1} of the present paper identifies $L^2$-based spaces of functions that are applicable in a much more general setting of asymptotically flat spacetimes, including the sub-extremal Kerr spacetimes, where a suitable notion of Gevrey-like regularity must also be imposed at infinity.

Besides providing a definition of quasinormal modes that is more general and widely applicable than the definition in terms of scattering resonances, Theorem \ref{thm:rough1}(iii) provides the first construction of the meromorphic continuation of the cut-off resolvent on sub-extremal Kerr spacetimes with respect to the standard Boyer--Lindquist time function. Theorem \ref{thm:rough1} moreover provides an extension of the fixed-$\ell$ result in Schwarzschild \cite{ba} in a sector of the complex plane, showing that scattering resonances do not feature any accumulation points in the limit $\ell\to\infty$ (in the above mentioned sector). See also \cite{stuc24} and the discussion at the end of \S \ref{sec:intro}.

As Theorem \ref{thm:rough1} and Theorem \ref{thm:rough1b}(ii) demonstrate, in contrast with methods involving cut offs and complex scaling, regularity quasinormal frequencies (and therefore also scattering resonances) in Kerr and Kerr--de Sitter with suitably small $M^2\Lambda$ can be studied simultaneously and uniformly in $\Lambda$ with the methods of the present paper.

\subsection{Stability of the quasinormal mode spectrum}
A natural question in the context of the eigenvalue problem of $\mathcal{A}$ is the stability of the point spectrum under suitably small perturbations of the corresponding wave operator. Since $\mathcal{A}$ is not self-adjoint, one cannot simply appeal to classical stability results in spectral theory.

Stability properties of quasinormal spectra were investigated numerically by Nollert and Nollert--Price in \cite{nol96, npr99} in the context of (toy models for) perturbed wave operators of the form $\square_{g_{M,0}}+\delta V$, with $V$ a discontinuous potential function. More recently, the quasinormal spectrum (in)stability question was revisited for \eqref{eq:waveeq} on Schwarzschild in \cite{dgm20, jamash21, cheetal22} for different choices of $V$.

Results pertaining to ``instability'' in the above works are in fact all consistent with the \emph{stability} statement in Theorem \ref{thm:rough1b}(i), because \textbf{in all of the above cases, the potentials under consideration do \underline{not} correspond to small perturbations for $0<\delta\ll 1$, as measured by the norms on $\mathbf{H}_{\alpha}$}. Roughly, speaking $\mathbf{H}_{\alpha}$ consists of functions with finite regularity in a bounded $r$ region $r\leq R$, but with $(2,\sigma)$-Gevrey regularity in $x=\frac{1}{r}$ as $r\to \infty$ with $\sigma$ determined by the index $\alpha$. In particular, the rescaled potential function $r^2V$ therefore needs to be small and more regular than smooth in $x$ near $x=0$ for Theorem \ref{thm:rough1b}(i) to apply, but away from $x=0$, far less regularity is required. 

Let us emphasize however, that it remains open how to even \emph{define} quasinormal modes (either in the sense of Leaver, scattering resonances or regularity quasinormal modes), let alone study their stability properties mathematically, when $r^2V$ is not at least $(2,\sigma)$-Gevrey at $x=0$ for some $\sigma> 0$.

Recently, the question of (in)stability and the characterization of ``smallness' have also been addressed in the context of the conformally covariant wave operator in de Sitter by the second author \cite{war24}.

\section{Main ideas and techniques}
\label{sec:ideastechniques}
In this section, we give an outline of the proofs of Theorems \ref{thm:rough1} and \ref{thm:rough1b} and highlight the main new ideas and tools that are introduced in the paper.\\
\paragraph{\textbf{Infinitesimal generators and Laplace transformed operators}}
Let $\phi$ be a solution to the wave equation \eqref{eq:waveeqds} and let $S(\tau)$ denote the following time translation map:
\begin{equation*}
S(\tau'): (\varrho\cdot \phi, \varrho \cdot \partial_{\tau}\phi)|_{\tau=0}\mapsto (\varrho\cdot \phi, \varrho\cdot \partial_{\tau}\phi)|_{\tau=\tau'},
\end{equation*}
where $\varrho$ is an appropriate radial coordinate (see the later paragraphs of this section for a discussion on the choice of radial coordinate).

By standard local-in-time energy estimates, $S(\tau)$ is well-defined for all $\tau\in [0,\infty)$ with initial data that are smooth in a conformal sense, i.e.\ with respect to the differentiable structure $(x=\frac{1}{\varrho}, \vartheta,\varphi_*)$, with $(\vartheta,\varphi_*)$ appropriate angular coordinates. The corresponding \emph{infinitesimal generator of time translations} is defined as the time derivative $\mathcal{A}:=\frac{d}{d\tau}|_{\tau=0} S(\tau)$ and, formally, takes the form:
\begin{equation*}
\mathcal{A}\left( (\varrho\cdot \phi, \varrho \cdot \partial_{\tau}\phi)|_{\tau=0}\right)=(\varrho \cdot \partial_{\tau}\phi, \varrho\cdot \partial_{\tau}^2\phi)|_{\tau=0}.
\end{equation*}

The time function $\tau$ is chosen so that: its level sets $\Sigma_{\tau}$ are spacelike hypersurfaces that intersect $\mathcal{H}^+$ strictly to the future of the bifurcation sphere and: 1) are asymptotically hyperboloidal in the $\Lambda=0$ case, or 2) intersect the cosmological horizon to the future of the bifurcation sphere in the $\Lambda>0$ case. 

The level sets $\Sigma_{\tau}$ capture the radiation of energy through $\mathcal{H}^+$ and future null infinity $\mathcal{I}^+$ or the cosmological horizon $\mathcal{C}^+$, in contrast with the standard Boyer--Linquist time function $t$, whose level sets $\Sigma'_{t}$ are spacelike hypersurfaces that intersect the bifurcation sphere and approach spacelike infinity $i^0$. See Figure \ref{fig:kerrfoli} for an illustration in the $\Lambda=0$ case.

\begin{figure}[H]
	\begin{center}
\includegraphics[scale=0.45]{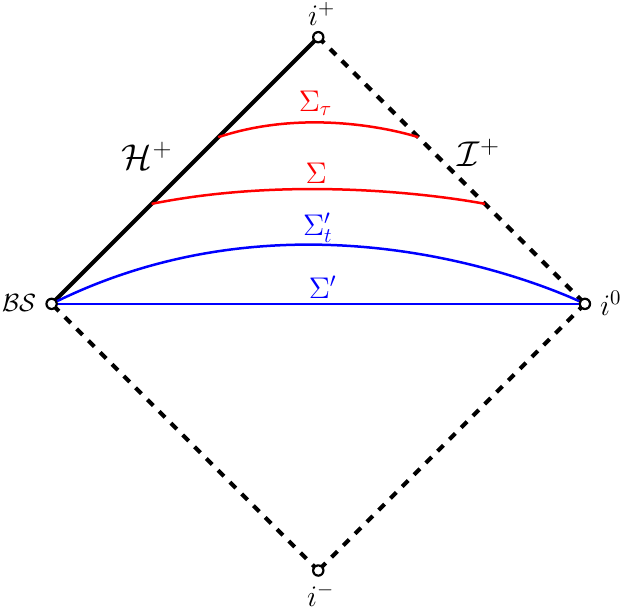}
\end{center}
\caption{A Penrose diagrammatic representation of the Kerr black hole exterior (including the future event horizon $\mathcal{H}^+$), with $\Sigma'=\{t=0\}$ and $\Sigma'_t$ the $t$-level sets, and with $\Sigma=\{\tau=0\}$ and $\Sigma_{\tau}$ the $\tau$-level sets. The point $\mathcal{BS}$ represents the bifurcation sphere and $i^+$, $i^0$, $i^-$ represent future timelike/spacelike/past timelike infinity.}
	\label{fig:kerrfoli}
\end{figure}

In \S \ref{sec:infigen}, we show that $\mathcal{A}$ can be related to the operators $L_{s,\Lambda}$, which is defined as the Laplace transform of the rescaled wave operator $\varrho^3\circ \left(\square_{g_{M,a,\Lambda}}-\frac{2}{3}\Lambda\right)\circ \varrho^{-1}$, i.e.\
\begin{equation*}
L_{s,\Lambda}\hat{\psi}=e^{-s\tau}\varrho^{3}\left(\square_{g_{M,a,\Lambda}}-\frac{2}{3}\Lambda\right)(e^{s\tau}\varrho^{-1} \hat{\psi}),
\end{equation*}
with $\hat{\psi}$ a function on $\Sigma$. We will show that:
\begin{equation*}
\mathcal{A}-s\mathbf{1}=\begin{pmatrix}
0 & \mathbf{1}\\
\mathbf{1} & P-s\mathbf{1}
\end{pmatrix}\begin{pmatrix}
\frac{1}{b} L_{s,\Lambda}& 0\\
0&\mathbf{1}
\end{pmatrix}\begin{pmatrix}
\mathbf{1}& 0\\
-s \mathbf{1} & \mathbf{1}
\end{pmatrix},
\end{equation*}
with $b$ a bounded function that is moreover bounded away from 0 and with $P$ an appropriate first order differential operator.

Formally, the corresponding resolvent operator $(\mathcal{A}-s\mathbf{1})^{-1}$ then takes the form:
\begin{equation*}
(\mathcal{A}-s\mathbf{1})^{-1}=\begin{pmatrix}
\mathbf{1} & 0\\
s \mathbf{1}& \mathbf{1}
\end{pmatrix}\begin{pmatrix}
  L_{s,\Lambda}^{-1}\circ b \mathbf{1}& 0\\
0&\mathbf{1}
\end{pmatrix}\begin{pmatrix}
-(P-s \mathbf{1})& \mathbf{1}\\
\mathbf{1} & 0
\end{pmatrix}.
\end{equation*}
Hence, we can determine the spectrum ${\rm Spect}(\mathcal{A})$ of $\mathcal{A}$ from a suitable characterization of the invertibility properties of the operator $L_{s,\Lambda}$.  In particular, eigenvectors of $\mathcal{A}$ are elements of the form $(\hat{\psi}_s, s\cdot \hat{\psi}_s)$ with $\hat{\psi}_s\in \ker L_{s,\Lambda}$. Theorem \ref{thm:rough1} can therefore be thought of as a \emph{corollary} of the following theorem:

\begin{theorem}[Rough version]
\label{thm:rough2}
Let $\Lambda\geq 0$ and let $L_{s,\Lambda}$ with $s\in \C$ be the Laplace transform of the operator $\varrho^3\circ (\square_{g_{M,a,\Lambda}}-\frac{2}{3}\Lambda)\circ \varrho^{-1}$ with respect to the time function $\tau$ associated to an asymptotically hyperboloidal and future-horizon-intersecting foliation of hypersurfaces. Let $N_+\in \N_0$, $\sigma\in \R$ and $R_{\infty}>r_+$. Denote by ${H}_{\sigma,R_{\infty}}^{N_+}$ and a ${H}_{\sigma,R_{\infty},2}^{N_++1}$ corresponding family of Hilbert spaces of functions on $\Sigma$ (see \S \ref{sec:hilbertspaces} precise definitions). Then, for suitably large $M^{-1}R_{\infty}$:
\begin{enumerate}[label=\emph{(\roman*)}]
\item For $M^2\Lambda$ suitably small, there exists a family of subsets $\Omega_{\sigma,N_+}\subset \C$ with
\begin{equation*}
\Omega:=\bigcup_{\sigma\in \R, N_+\in \N_0}\Omega_{\sigma,N_+}=\left\{{|\rm arg} z|<\frac{2\pi}{3}\right\},
\end{equation*}
such that $L_{s,\Lambda}$ is defined as follows: $L_{s,\Lambda}:  {D}_{\sigma,R_{\infty}}^{N_+}(L_s)\to {H}_{\sigma,R_{\infty}}^{N_+}$, with ${H}_{\sigma,R_{\infty},2}^{N_++1}\subseteq {D}_{\sigma,R_{\infty}}^{N_+}(L_s)\subseteq {H}_{\sigma,R_{\infty}}^{N_+}$ and ${D}_{\sigma,R_{\infty}}^{N_+}(L_s)$ a suitable closed subspace,  admits an inverse $L_{s,\Lambda}^{-1}$ for $s$ in the complement of a discrete subset of $\Omega_{\sigma,N_+}$, such that $s\mapsto L_{s,\Lambda}^{-1}: H_{\sigma,R_{\infty}}^{N_+}\to H_{\sigma,R_{\infty},2}^{N_++1}$ is meromorphic in $\Omega_{\sigma,N_+}$. The set of poles of $s\mapsto L_{s,\Lambda}^{-1}$ does not depend on the precise choice of time function $\tau$ and indices $\sigma$ and $N_+$.
\item The union $\mathscr{Q}_{\rm reg}$ over $\sigma\in \R$ and $N_+\in \N_0$ of all poles of $s\mapsto L_{s,\Lambda}^{-1}$ has no accumulation points in $\Omega$.
\item The cut-off resolvent of \eqref{eq:waveeqds} corresponding to the Boyer--Lindquist time function  $t$ can be meromorphically continued to $\Omega$ and its poles are contained in $i \mathscr{Q}_{\rm reg}$ provided $M^2\Lambda$ is suitably small.
\end{enumerate}
\end{theorem}

Similarly, Theorem \ref{thm:rough1b} can be thought of as a corollary of the following theorem:
\begin{theorem}[Rough version]
\label{thm:rough2b}
\hspace{1pt}
\begin{enumerate}[label=\emph{(\roman*)}]
\item The poles of $L_{s,0}^{-1}$ are stable under small perturbations $L_s+\delta \widetilde{Q}$, with $ \widetilde{Q}$ a bounded operator with respect to appropriate Hilbert spaces and $\delta>0$ a suitably small constant.
\item Let $\{\Lambda_n\}$ be a sequence of cosmological constants with $\Lambda_n\downarrow 0$ as $n\to \infty$. For each pole $s_*$ of $s\mapsto L_{s,0}^{-1}$, there exists a sequence $\{s_n\}$ of poles of $s\mapsto L_{s,\Lambda_n}^{-1}$ with limit $s_*$. Conversely, for each sequence $\{s_n\}$ of poles of $s\mapsto L_{s,\Lambda_n}^{-1}$ with limit $s_*\in \Omega$, $s_*$ is a pole of $s\mapsto L_{s,0}^{-1}$.
\end{enumerate}
\end{theorem}
We will now sketch the proofs of Theorems \ref{thm:rough2} and \ref{thm:rough2b}. References to detailed versions of the steps below can be found in \S \ref{sec:guide}.
\\
\paragraph{\textbf{Sketch of the proof of Theorem \ref{thm:rough2}(i)}}
The proof of (i) constitutes the main part of Theorem \ref{thm:rough2}. The goal is to show that the following operator is meromorphic in $s$:
\begin{equation*}
L_{s,\Lambda}^{-1}: {H}_{\sigma,R_{\infty}}^{N_+}\to {H}_{\sigma,R_{\infty},2}^{N_++1}\subset H_{\sigma,R_{\infty}}^{N_+},
\end{equation*}
with ${H}_{\sigma,R_{\infty},2}^{N_++1}$ a Hilbert space consisting of functions in ${H}_{\sigma,R_{\infty}}^{N_+}$ whose derivatives up to first order are also contained in ${H}_{\sigma,R_{\infty},2}^{N_++1}$ and whose second-order derivatives, restricted outside a suitable bounded set, with potential additional degenerate factors, are also contained in ${H}_{\sigma,R_{\infty}}^{N_+}$;  see \S \ref{sec:hilbertspaces} for a precise definition. Key to the proof is that ${H}_{\sigma,R_{\infty},2}^{N_++1}$ is \emph{compactly} embedded in ${H}_{\sigma,R_{\infty}}^{N_+}$.

In the discussion below, we will first restrict to the $\Lambda=0$ setting. We denote $L_s=L_{s,0}$. As we will, see a simplified version of $L_{s,\Lambda}$ with $\Lambda>0$ will play an important role, even when investigating $L_s$.

We first show that we can decompose $L_s$ as follows:
\begin{equation}
\label{eq:introsplittingLs}
L_s=\mathfrak{L}_{s,\lambda}+\epsilon_0\cdot B_s+K_{s,\lambda},
\end{equation}
with $\epsilon_0>0$ arbitrarily small and with $\lambda=\mathfrak{l}(\mathfrak{l}+1)$, where $\mathfrak{l}\in \N_0$ is a large parameter that arises from an \emph{angular frequency shift} of lower angular frequencies. Furthermore:
\begin{align*}
\mathfrak{L}_{s,\lambda}^{-1}:&\: {H}_{\sigma,R_{\infty}}^{N_++1}\to {H}_{\sigma,R_{\infty},2}^{N_++1}\ssubset{H}_{\sigma,R_{\infty}}^{N_++1} \quad \textnormal{is bounded for all $s\in \Omega_{\sigma, N_+}$, for suitably large $\mathfrak{l}\in \N_0$},\\
 B_s:&\: {H}_{\sigma,R_{\infty},2}^{N_++1}\to {H}_{\sigma,R_{\infty}}^{N_+}\quad \textnormal{is bounded for all $s\in \Omega_{\sigma, N_+}$},\\
 K_{s,\lambda}:&\: {H}_{\sigma,R_{\infty},2}^{N_++1}\to {H}_{\sigma,R_{\infty},2}^{N_++1}\ssubset{H}_{\sigma,R_{\infty}}^{N_++1}\quad \textnormal{is bounded for all $s\in \Omega_{\sigma, N_+}$ and $\mathfrak{l}\in \N_0$}.
\end{align*}
For $\epsilon_0>0$ suitably small, it follows that $\|\epsilon_0 B_s\circ \mathfrak{L}_{s,\lambda}^{-1}\|<1$, so
\begin{equation*}
(\mathfrak{L}_{s,\lambda}+\epsilon_0\cdot B_s)^{-1}=\mathfrak{L}_{s,\lambda}^{-1}(\mathbf{1}+B_s\mathfrak{L}_{s,\lambda}^{-1})^{-1}
\end{equation*}
is a well-defined and bounded linear operator from ${H}_{\sigma,R_{\infty}}^{N_+}$ to ${H}_{\sigma,R_{\infty},2}^{N_++1}$. Meromorphicity of $s\mapsto L_s^{-1}$ then follows from a direct application of the analytic Fredholm theorem, using that $K_{s,\lambda}\circ \mathfrak{L}_{s,\lambda}^{-1}$ is a compact operator from  ${H}_{\sigma,R_{\infty}}^{N_+}$ to ${H}_{\sigma,R_{\infty}}^{N_+}$.\\

\paragraph{\textbf{Coordinate choices}}
In order to guarantee splitting of $L_s$ in \eqref{eq:introsplittingLs}, we do not work in ingoing Kerr coordinates $(v,r,\theta,\varphi_*)$; see \S \ref{sec:geometry} for more details on these coordinates. Instead, we modify $r$ and $\theta$ in the region $r\geq R_{\infty}$, with $R_{\infty}\gg M$ to obtain a new radial function $\varrho$ and a new angular function $\vartheta$. This change is motivated by the fact that in the case $M=0$, the Kerr spacetime in Boyer--Lindquist coordinates $(t,r,\theta,\varphi)$ is isometric to the Minkowski spacetime in non-standard coordinates, where the surfaces of constant $t$ and $r$ are ellipsoids with eccentricity parameter $a^2$. The functions $\varrho$ and $\vartheta$ are chosen so that in the $M=0$ case the surfaces of constant $t$ and $\varrho$ are round spheres. In other words, we choose coordinates so that as our radial coordinate goes to infinity, the Kerr metric approaches the Minkowski metric in spherical coordinates in a suitably fast way.\footnote{The coordinate transformation $(r,\theta)\mapsto (\varrho,\vartheta)$ in the $M=0$ setting also plays an important role in the strategy for constructing double null coordinates on Kerr spacetime by Pretorius--Israel in \cite{preis98}. In \cite{preis98}, however, the choice of $(\rho,\vartheta)$ is $M$-dependent, so that the spheres of constant $t$ and $\rho$ form the intersection of the null hypersurfaces of a double null foliation. In the present paper, it suffices to take $(\rho,\vartheta)$ to agree with the $M=0$ expressions for suitably large values of $r$.}

In practice, it is often more convenient to work with the inverted coordinate $x=\frac{1}{\varrho}$, which has a finite range on the spacetime.

In the $\Lambda>0$ case, we consider a related change of coordinates. In the $M=0$ case, these correspond to standard coordinates on a static patch of de Sitter. We refer to \S \ref{sec:KdS} for the precise coordinate transformation.\\
\paragraph{\textbf{Inverting $\mathfrak{L}_{s,\lambda}$}}
Property (i) therefore follows from invertibility of $\mathfrak{L}_{s,\lambda}$ for $s\in \Omega_{\sigma,N_+}$. Rather than establishing this directly, we first \emph{regularize} the operator at infinity by considering
\begin{equation*}
\mathfrak{L}_{s,\lambda,\kappa}:=\kappa \partial_{x}(\chi_{\varrho\geq R_{\infty}} x \partial_{x} (\cdot))+\mathfrak{L}_{s,\lambda},
\end{equation*}
with $\kappa>0$ and $\chi_{\varrho\geq R_{\infty}}$ a smooth cut-off function that is equal to 1 in $\varrho\geq R_{\infty}$. The inclusion of the $\kappa$ term can be interpreted in two ways:
\begin{enumerate}
\item (A geometric perspective) Consider the Kerr-de Sitter metric $g_{M,a,\Lambda}$, which features an additional parameter $\Lambda>0$, such that $g_{M,a}=g_{M,a,0}$, and consider the conformally invariant wave operator $\square_{g_{M,a}}-\frac{2\Lambda}{3}\mathbf{1}$.  The level set $\{x=0\}$ then corresponds to the future cosmological horizon $\mathcal{C}^+$. The consideration of $\Lambda>0$ leads to the following perturbation of $L_s$:
\begin{equation*}
2\kappa_c \partial_{x}(\chi_{\varrho\geq R_{\infty}} x \partial_{x} (\cdot))+\ldots
\end{equation*}
with $\kappa_c>0$ the surface gravity of $\mathcal{C}^+$ with $\ldots$ denoting terms that can be treated as small perturbations. Hence, the $\kappa$-term in $\mathfrak{L}_{s,\lambda,\kappa}$ \emph{effectively} adds a non-zero surface gravity to future null infinity $\mathcal{I}^+$ (interpreted as a null hypersurface in a conformal extension of the Kerr spacetime). See \S \ref{sec:KdS} for details.
\item (An ODE perspective) In the region $\varrho \geq R_{\infty}$, $\mathfrak{L}_{s,\lambda,\kappa }$ takes the following form when restricted to \emph{spherical harmonic modes} $\hpsi_{\ell}$, which are projections of $\hpsi$ to the spherical harmonics $Y_{\ell m}$ with angular frequency $\ell$:
\begin{equation*}
(\mathfrak{L}_{s,\lambda,\kappa}\hpsi)_{\ell}=\partial_{x}(x(\kappa+x)\partial_{x}\hpsi_{\ell})+2 s\partial_{x}\hpsi_{\ell}-\mathfrak{l}(\mathfrak{l}+1)\hpsi_{\ell},
\end{equation*}
where $\mathfrak{l}=\ell+l$ for $\ell\leq l$ and $\mathfrak{l}=\ell$ for $\ell>l$. The ODE $(\mathfrak{L}_{s,\lambda,\kappa}\hpsi)_{\ell}=f_{\ell}$ has a regular singular point at $x=0$ if $\re s<0$ and $\kappa>0$, but the singular point becomes irregular as $\kappa=0$. Hence, $\kappa$ has a regularizing effect on the ODE. 

In contrast, $\varrho=r_+$ (the future event horizon $\mathcal{H}^+$) is always a regular singular point, since the surface gravity of the event horizon $\kappa_+$ is always positive on sub-extremal Kerr spacetimes.
\end{enumerate}
The advantage of considering $\mathfrak{L}_{s,\lambda,\kappa}$ is that it is possible to obtain estimates for $\mathfrak{L}_{s,\lambda,\kappa}^{-1}$ involving function spaces with only a finite number of derivatives. Using the existence of $\mathfrak{L}_{s,\lambda,\kappa}^{-1}$ as an operator on appropriate Sobolev spaces, restricted to finite-$m$ azimuthal modes, as derived in \cite{warn15}, it is possible to show that for $\re s>- \frac{\kappa}{2}(1+2N_{\kappa})$,
\begin{equation*}
\mathfrak{L}_{s,\lambda,\kappa}^{-1}: {H}^{N_{\kappa}}(\widehat{\Sigma}) \to {H}^{N_{\kappa}+1}(\widehat{\Sigma}),
\end{equation*}
is well-defined and bounded, where $N_{\kappa}>0$ is suitably large depending on $\kappa$. In particular, $\|\mathfrak{L}_{s,\lambda,\kappa}^{-1}\|_{{H}^{N_{\kappa}}(\widehat{\Sigma}) \to {H}^{N_{\kappa}+1}(\widehat{\Sigma})}$ is finite and \underline{depends on $\kappa$.}

The core of part (i) is to show that for $s\in \Omega_{\sigma, N_+}$, we can in fact estimate \underline{uniformly in $\kappa$}:
\begin{equation}
\label{eq:intromainest}
\|\mathfrak{L}_{s,\lambda,\kappa}^{-1}f\|_{{H}_{\sigma,R_{\infty},2}^{N_++1}}\leq C_{\lambda,\sigma,R_{\infty}} \|f\|_{{H}_{\sigma,R_{\infty}}^{N_+}}.
\end{equation}
Then $\mathfrak{L}_{s,\lambda}^{-1}$ can be constructed as the following limit:
\begin{equation*}
\mathfrak{L}_{s,\lambda}^{-1}:=\lim_{\kappa \downarrow 0}\mathfrak{L}_{s,\lambda,\kappa}^{-1}
\end{equation*}
where the limit is defined with respect to the operator norm.

Most of the analysis in the paper concerns the uniform-in-$\kappa$ estimate \eqref{eq:intromainest}. \\

\paragraph{\textbf{Nature of the relevant norms}}
The ${H}_{\sigma,R_{\infty}}^{N_+}$ norm can roughly be decomposed into four parts, which each have a different character: for $R_{\infty}\gg M$, with $ \slashed{\nabla}_{\s^2}$ denoting the covariant derivative on the unit round sphere $\s^2$ with respect to the angular coordinates $(\vartheta,\varphi_*)$ and $\slashed{\Delta}_{\s^2}$ the corresponding angular Laplacian we have:
\begin{enumerate}[label=(\alph*)]
\item \textbf{Sobolev norms in bounded regions}: $\|\cdot\|_{H^{N_+}(\Sigma \cap\{\varrho\leq 3R_{\infty}\})}+\|\slashed{\nabla}_{\s^2}(\cdot)\|_{H^{N_+}(\Sigma \cap\{\varrho\leq \frac{R_{\infty}}{3}\})}$, 
\item Sobolev norms in far-away, but bounded regions that are \textbf{microlocal in the angular directions}: $\|(\mathbf{1}-\slashed{\Delta}_{\s^2})^{-\frac{w_p}{4}}(\cdot)\|_{H^{N_+}(\Sigma \cap\{ 3R_{\infty}\leq \varrho\leq 4R_{\infty}\})}$, with $w_p$ a smooth cut-off function such that $w_p(\varrho)=0$ for $\varrho\leq 3R_{\infty}$ and $w_p(\rho)=p$ for $\varrho\geq 4R_{\infty}$, where $p>0$ is a suitably large constant.
\item \textbf{1-Gevrey-type boundary norms} at $\varrho=2R_{\infty}$ with infinitely many derivatives: 
\begin{equation*}
\|\cdot\|^2_{B_{R_{\infty}}}=\sum_{\ell,n\in \N_0}\frac{c^n}{\max\{n+1,\ell+1\}^{2n}}\|(\varrho^2\partial_{\varrho})^n(\cdot)|_{\varrho=2R_{\infty}}\|^2_{L^2(\s^2)},
\end{equation*}
for some constant $c>0$.
\item \textbf{$L^2$-based 2-Gevrey-type norms}: for lower radial derivatives
\begin{equation*}
\sum_{\ell\in \N_0}\sum_{n=0}^{\max\{\ell-1,0\}}\frac{(4\sigma)^{2n}(\ell-(n+1))!^2}{(\ell+n+1)!^2} \int_{2R_{\infty}}^{\infty}\|(\varrho^2\partial_{\varrho})^n(\cdot)\|^2_{L^2(\s^2)}(x)\, \varrho^2 d\varrho
\end{equation*}
and for higher radial derivatives
\begin{equation*}
	\sum_{\ell\in \N_0}(\ell+1)^3\sum_{n=\ell}^{\infty}\frac{\sigma^{2n}}{n!^2(n+1)!^2}  \int_{2R_{\infty}}^{\infty}\|(\varrho^2\partial_{\varrho})^n(\cdot)\|^2_{L^2(\s^2)}\, \varrho^2 d\varrho.
	\end{equation*}
\end{enumerate}

Note that (c) and (d) involve so-called ``Gevrey-type norms'' in $0\leq x\leq x_0$, with $x_0=\frac{1}{2R_{\infty}}$. 

For a function $f: [0,x_0]\times \s^2\to \C$, we define the $L^2$-based $(\alpha,b)$-Gevrey norm as follows:
\begin{multline*}
\sqrt{\sum_{n\in\N_0}\sum_{n_1+n_2=n}\frac{b^{2n}}{n!^{2\alpha}}\int_0^{x_0} \|\snabla_{\s^2}^{n_1}\partial_x^{n_2}f\|^2_{L^2(\s^2)}\,dx}\\
\sim \sqrt{\sum_{\ell \in\N_0}\sum_{n\in\N_0}\sum_{m\in \N_0}\frac{(e\cdot b)^{2(n+m)} (\ell(\ell+1))^{m}}{(n+m)^{2\alpha (n+m+1)}}\int_0^{x_0} \| \partial_x^{n}f_{\ell}\|^2_{L^2(\s^2)}\,dx},
\end{multline*}
where the equivalence of norms follows from the orthogonality of the spherical harmonic modes $f_{\ell}$ of $f$ and from an application of Stirling's approximation formula: $n!\sim n^{n+1}e^{-n}$. 

Similarly, we consider $L^{\infty}$-based $(\alpha,b)$-Gevrey norms:
\begin{multline*}
\sqrt{\sum_{n\in\N_0}\sum_{n_1+n_2=n}\frac{b^{2n}}{n!^{2\alpha}} \|\snabla_{\s^2}^{n_1}\partial_x^{n_2}f\|^2_{L^{\infty}([0,x_0]\times \s^2)}}\\
\sim \sqrt{\sum_{\ell \in\N_0}\sum_{n\in\N_0}\sum_{m\in \N_0}\frac{(e\cdot b)^{2(n+m)} (\ell(\ell+1))^{m}}{(n+m)^{2\alpha (n+m+1)}}\| \partial_x^{n}f_{\ell}\|^2_{L^{\infty}([0,x_0]\times \s^2)}}
\end{multline*}

Note that the space of $(1,b)$-Gevrey norms corresponds to the space of analytic functions with analyticity radius depending on the exponent $b$.

It is straightforward to show that the boundary norm in (iii) can be bounded above by $L^{\infty}$-based $(1,b)$-Gevrey norms restricted to an arbitrarily small neighbourhood of $x=x_0$ for a suitable exponent $b$.

In (d), the $n\geq \ell$ part of the norm can immediately be bounded above by an $L^2$-based $(2,b)$-Gevrey norm with suitable exponent $b$. The same holds for the $n\leq \ell-1$ part of the norm, but this is less straightforward. Suppose $n\leq (1-\epsilon)(\ell-1)$, for some $0<\epsilon<1$. Without loss of generality, we can take $\ell\geq 1$. Then, by Stirling's approximation formula:
\begin{equation*}
\frac{(\ell-(n+1))!^2}{(\ell+(n+1))!^2}\lesssim e^{4(n+1)}\epsilon^{2(\ell-(n+1)+1)}\frac{\ell^{2(\ell-(n+1)+1)}}{(2 \ell+1)^{2(\ell+(n+1)+1)}}\leq e^{4(n+1)}\ell^{-4(n+1)}\leq e^{4(n+1)}(n+1)^{-4(n+1)}.
\end{equation*}
Hence, the $n\leq (1-\epsilon)(\ell-1)$ part of the form can also be bounded by an $L^2$-based $(2,b)$-Gevrey norm with suitable exponent $b$. Finally, consider the part where $(1-\epsilon)(\ell-1)\leq n\leq \ell-1$, where we may assume without loss of generality that $n\geq 1$. Then we apply Stirling's approximation formula again to estimate:
\begin{equation*}
\frac{(\ell-(n+1))!^2}{(\ell+(n+1))!^2}\lesssim \frac{1}{(2\ell)!^2}\lesssim e^{-2\ell}(2\ell)^{-4\ell-2}\leq e^{-2n}2^{-4n-2}n^{-4n-2}.
\end{equation*}

\paragraph{\textbf{Deriving \eqref{eq:intromainest}}}
The different norms in (a)--(d) arise from different types of fixed-frequency analogues of energy estimates and elliptic estimates in different ranges of the radial coordinate $\varrho$, which are indicated in Figure \ref{fig:rvalues}. 
\begin{figure}[h!]
	\begin{center}
\includegraphics[scale=0.75]{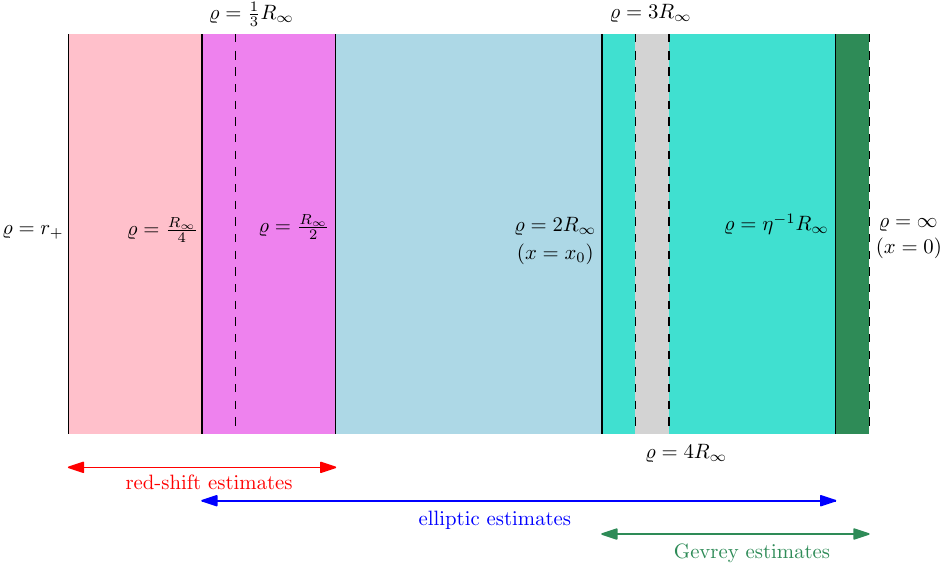}
\end{center}
\caption{The spacetime regions labelled by the radial function $\varrho$ where the different types of estimates are applied.}
	\label{fig:rvalues}
\end{figure}

These are
\begin{itemize}
\item \textbf{Red-shift estimates}: We consider the integral of:
\begin{equation*}
\re(\varrho^{\delta}\overline{\partial_{\varrho}(\varrho^{-1}\hpsi})\cdot\mathfrak{L}_{s,\kappa}\hpsi)
\end{equation*}
for $\delta>0$ suitably small and integrate by parts. These estimates rely crucially on the positivity of $\kappa_+$, the surface gravity of the Kerr event horizon, and they may be thought of as zero-frequency analogues of the physical space red-shift estimates of \cite{redshift}. In contrast with their physical space analogues, it is important that these estimates are valid for large values of $\varrho$ and not just in a neighbourhood of $\varrho=r_+$, with $r_+$ the radius of the event horizon. The \emph{ergoregion} of the black hole, where the vector field $\partial_{\tau}$ fails to be timelike, is entirely included in the region where we apply red-shift estimates. With these estimates, we gain control of \underline{one} extra derivative of $\hpsi$ compared to $\mathfrak{L}_{s,\kappa}\hpsi$.
\item \textbf{Elliptic estimates}: These estimates concern a far-away but bounded region, excluding the ergoregion, where $\partial_{\tau}$ is timelike and the operator $\mathfrak{L}_{s,\lambda,\kappa}$ is therefore elliptic. As a result, we can consider elliptic estimates obtained by splitting
\begin{equation*}
\mathfrak{L}_{s,\lambda,\kappa}\hpsi=\textnormal{ second order derivatives}+\mathfrak{R}_{s,\lambda}
\end{equation*}
and integrating the weighted square
\begin{equation*}
x^{1-\delta}(\lf+1)^{-w_p}|\mathfrak{L}_{s,\lambda,\kappa}\hpsi-\mathfrak{R}_{s,\lambda}\hpsi|^2.
\end{equation*}
For suitably large $\lambda$, we can arrive at an elliptic estimate by integrating by parts, with boundary terms on the right-hand side. With these estimates, we gain control of \underline{two} extra derivatives of $\hpsi$ compared to $\mathfrak{L}_{s,\kappa}\hpsi$. The red-shift and elliptic estimates combined provide control of the norms  in (a) and (b). 
\item \textbf{Higher-order gevrey estimates}: These estimates take place in the region $\varrho\geq 2R_{\infty}$ where $\mathfrak{L}_{s,\lambda,\kappa}$ takes the simple form:
\begin{equation*}
(\mathfrak{L}_{s,\lambda,\kappa}\hpsi)_{\ell}+\mathfrak{l}(\mathfrak{l}+1)\hpsi_{\ell}=\partial_{x}(x(\kappa+x)\partial_{x}\hpsi_{\ell})+2 s\partial_{x}\hpsi_{\ell}.
\end{equation*}
Ignoring the $\kappa$ term, this is simply the Laplace transform of $\varrho^{3}\square (\varrho^{-1} \cdot)$ with $x=\frac{1}{\varrho}$ and $\square$ the standard wave operator with respect to the Minkowski spacetime and standard spherical coordinates $(\varrho,\vartheta,\varphi_*)$, where $\tau$-level sets are outgoing null hypersurfaces.

The above equation can easily be commuted with $\partial_x=-\varrho^2\partial_{\varrho}$ and we integrate:
\begin{equation}
\label{eq:introhoeqL}
\left|[\lf(\lf+1)-n(n+1)]\partial_x^n\hpsi+\partial_x^n (\mathfrak{L}_{s,\lambda,\kappa}\hpsi)_{\ell}\right|^2
\end{equation}
to obtain weighted $L^2$-estimates for derivatives (up to second order) of $\hphi$. In order for these type of estimates to close, we sum them up with the following weights:
\begin{equation*}
	\sum_{n=\lf}^{N_{\infty}}\frac{|2s|^2\tilde{\sigma}^{2n}}{n!^2(n+1)!^2},
\end{equation*}
with $\tilde{\sigma}^2\in (0,\infty)$ restricted to an appropriate range.

 Here it is important to note that for $n=\lf$, the following equation holds:
\begin{equation*}
\partial_{x}(x(\kappa+x)\partial_{x}^{\lf+1}\hpsi_{\ell})+2 \left[s+(n+1) \left(x+\frac{\kappa}{2}\right)\right]\partial_{x}^{\lf+1}\hpsi_{\ell}=(\mathfrak{L}_{s,\lambda,\kappa}\hpsi)_{\ell}.
\end{equation*}
When $\kappa=0$, the physical space analogue of the above equation is:
\begin{equation}
\label{eq:consvlaw}
\partial_{x}(x^2\partial_{x}^{\lf+1}\psi_{\ell})+2 \left[\partial_{\tau}+(n+1) x\right]\partial_{x}^{\lf+1}\psi_{\ell}=0
\end{equation}
with $\psi=\varrho\phi$. Evaluating the above expression at $x=0$ gives $\partial_{\tau}\partial_{x}^{\lf+1}\psi_{\ell}|_{\mathcal{I}^+}=0$, which is a conservation law that was first pointed out by Newman and Penrose \cite{np2}. \emph{The existence of this conservation law in Minkowski is directly related to the fact that the Gevrey estimates close when summing over $n\geq \lf$.} With these estimates, we gain control of \underline{two} extra derivatives of $\hpsi$ compared to $\mathfrak{L}_{s,\kappa}\hpsi$. These estimates are the origin of the restriction to the sector $\{\arg z<\frac{2\pi}{3}\}$.
\item \textbf{Lower-order gevrey estimates}: It remains to control an appropriately weighted sum terms with $n< \lf$. By summing elliptic-type estimates in the region $\varrho\geq 2R_{\infty}$, we obtain control of lower-order derivatives in terms of the higher-order derivatives (which we already control independently!) and boundary terms at $\varrho=2R_{\infty}$. To control the boundary terms, we simply use the equation $\eqref{eq:introhoeqL}$ to express an appropriately weighted sum of higher-order derivatives $\partial_{\varrho}^{k+2}\hpsi$ in terms of $\hpsi$ and $\partial_{\varrho}\hpsi$. In this step, we obtain, in particular, control over the boundary norm in (c) above and we can use this to control the remaining lower-order (in $n$) Gevrey norms in (d). With these estimates, we gain control of \underline{two} extra derivatives of $\hpsi$ compared to $\mathfrak{L}_{s,\kappa}\hpsi$. 
\end{itemize}
The estimates above that take place in different ranges of $\varrho$ are all \emph{coupled}, apart from the higher-order Gevrey estimates, which can be closed independently of all the other estimates. To enable a coupling of the estimates, we need to take $\lambda$ suitably large compared to the other parameters and we need to choose the value $p$ in the weight function $w$ in the elliptic estimates appropriately.

Now consider the more general operator $L_{s,\Lambda}$, with $\Lambda\geq 0$. Let $\kappa_c\geq 0$ be the surface gravity associated to the cosmological horizon $\mathcal{C}^+$. Note that $\kappa_c\lesssim \sqrt{\Lambda}$. Then we show that we can decompose $L_{s,\Lambda}$ as follows:
\begin{equation*}
L_{s,\Lambda}=\mathfrak{L}_{s,\lambda,2\kappa_c}+\epsilon_0\cdot B_{s,\Lambda}+K_{s,\lambda,\Lambda}+ \kappa_c^2 B_{s,\Lambda}^{\mathcal{C}^+}.
\end{equation*}
The operator $\mathfrak{L}_{s,\lambda,2\kappa_c}$ closely resembles the $\kappa$-regularized operator $\mathfrak{L}_{s,\lambda,\kappa}$, with $\kappa=2\kappa_c$, which plays a key role in the proof of Theorem \ref{thm:rough2}(i) for $\Lambda=0$, with only minor differences in the region $r\leq \frac{R_{\infty}}{2}$ that do not affect the relevant estimates. Boundedness of the operator $\kappa_c B_{s,\Lambda}^{\mathcal{C}^+}\circ \mathfrak{L}_{s,\lambda,2\kappa_c}^{-1}: {H}_{\sigma,R_{\infty}}^{N_+}\to {H}_{\sigma,R_{\infty}}^{N_+}$ follows by applying the following $\kappa_c$-dependent generalization of \eqref{eq:intromainest}:
\begin{equation}
\label{eq:intromainest2}
\|\mathfrak{L}_{s,\lambda,\kappa_c}^{-1}f\|_{{H}_{\sigma,R_{\infty},2}^{N_++1}}+\kappa_c \|x\partial_x^2\mathfrak{L}_{s,\lambda,2\kappa_c}^{-1}(f)\|_{G_{\sigma,R_{\infty}}}+\kappa_c \|\slashed{\Delta}_{\s^2 }\mathfrak{L}_{s,\lambda,2\kappa_c}^{-1}(f)\|_{G_{\sigma,R_{\infty}}}\leq C_{\lambda,\sigma,R_{\infty}} \|f\|_{{H}_{\sigma,R_{\infty}}^{N_+}}
\end{equation}
and exploiting the control over additional second-order derivatives on the left-hand side.

Hence, for $M\kappa_c\ll 1$, $\kappa_c^2 B_{s,\Lambda}^{\mathcal{C}^+}\circ \mathfrak{L}_{s,\lambda,2\kappa_c}^{-1}$ can be interpreted as a small bounded operator, just like $\epsilon_0\cdot B_{s,\Lambda}\circ \mathfrak{L}_{s,\lambda,2\kappa_c}^{-1}$. The operator $K_{s,\lambda,\Lambda}\circ \mathfrak{L}_{s,\lambda,2\kappa_c}^{-1}$, on the other hand, is shown to be compact, as in the $\Lambda=0$ case. \\

\paragraph{\textbf{Sketch of the proof of Theorem \ref{thm:rough2}(ii)}}
We would like to take the union of all $s$ for which $\ker L_s\neq \{0\}$ with $L_s$ restricted to a subspace of $H^{N_+}_{\sigma,R_{\infty}}$. A priori, it is not clear how the non-triviality of the kernel of $L_s$ depends on the parameters $\sigma$, $N_+$ and $R_{\infty}$. We show that the kernel of restrictions of $L_s$ to $H^{N_+}_{\sigma,R_{\infty}}$ is \emph{independent} of the choice of parameters. The underlying reason is that we can use the estimates for $\mathfrak{L}_{s,\lambda,\kappa}^{-1}$ to further improve the regularity of elements of $\ker L_s$. In particular, we conclude in this step that elements of the $\ker L_s$ are smooth on $\Sigma$.

Here, we moreover use that any element of $\ker L_s$ can be approximated by elements of the kernel of the regularized operator $L_{s,\kappa}=L_s+\partial_{x}(x(\kappa+x)\chi_{\varrho\geq R_{\infty}}\partial_{x}\hpsi)$ with $\kappa>0$, by an application of Hurwitz's theorem of complex analysis. This is very closely related to the quasinormal frequency convergence property in Theorem \ref{thm:rough2b}(ii).\\
\\
\paragraph{\textbf{Sketch of the proof of Theorem \ref{thm:rough2}(iii)}}
Note that of the norms in (a)--(d) in the sketch of the proof of (i), only (a) is relevant when restricting to functions supported in $\{\varrho\leq \frac{1}{2}R_{\infty}\}$ and $R_{\infty}$ can be chosen arbitrarily large.
In particular, the following operator can easily shown to be well-defined for $s\in \Omega\setminus \mathscr{Q}_{\rm reg}$:
\begin{equation*}
L_s^{-1}: C_c^{\infty}(\Sigma)\to C^{\infty}(\Sigma).
\end{equation*}

Consider $\hphi: \Sigma'\to \C$. Then we define:
\begin{equation*}
A(\omega)\hat{\phi}=e^{i\omega t}\square_{g_{M,a}}(e^{-i\omega t} \hat{\phi}).
\end{equation*}
The inverse $R(\omega)=A(\omega)^{-1}$ on appropriate function spaces is called the \emph{resolvent operator} associated the the wave operator $\square_{g_{M,a}}$ and the Boyer--Lindquist time function $t$.

Let $Q_s$ be the map that takes the trivial extension of a function on $\Sigma'$ to the rest of the spacetime and restricts it to the hypersurface $\Sigma=\{\tau=0\}$. Then we can relate $A(\omega)$ to $L_s$ in the following way:
\begin{equation*}
A(\omega)=Q_s^{-1}\circ (\varrho^{-3}L_s(\varrho \cdot))\circ Q_s
\end{equation*}
with $s=-i\omega$, so the cut-off resolvent, defined as follows: $\chi R(\omega) \chi$ for arbitrary $\chi\in C^{\infty}_c(\Sigma')$, satisfies
\begin{equation*}
\chi R(\omega) \chi=(\chi Q_s^{-1})\circ (\varrho^{-1}L_s(\varrho^{3} \cdot))^{-1}\circ (Q_s\chi).
\end{equation*}
Since $\chi Q_s^{-1}$ and $Q_s\chi$ define a bounded linear operator on Sobolev spaces of arbitrary order, it follows that the right-hand side of the equation above is a well-defined map from $ C^{\infty}(\Sigma)$ to $C^{\infty}(\Sigma)$ that is meromorphic in $s$ for $s\in \Omega$ and hence the cut-off resolvent can be meromorphically continued to the subset $\Omega\subset \C$ with poles contained in $\mathscr{Q}_{\rm reg}$.\\
\\

\paragraph{\textbf{Sketch of the proof of Theorem \ref{thm:rough2b}}}
Stability of the quasinormal spectrum follows from the following two observations: 1) Invertibility of $L_s$ implies invertibility of $L_s+\delta \widetilde{Q}$ for $\delta \|\widetilde{Q}L_s^{-1}\|<1$. Hence, small perturbations of $L_s$ are invertible for $s$ contained in compact subsets of $\Omega_{\sigma,N_+}$ that exclude arbitrarily small disks around the poles of $s\mapsto L_s^{-1}$. The smaller the disks, the larger $\|\widetilde{Q}L_s^{-1}\|$ can be, so $\delta$ needs to be taken appropriately small. 2) $\ker L_s$ can be approximated by elements of modification $\ker (L_s+\delta \widetilde{Q})$ (by an application of Hurwitz's theorem), so arbitrarily small disks around the poles of $s\mapsto L_s^{-1}$ contain poles of $s\mapsto (L_s+\delta \widetilde{Q})^{-1}$ for suitably small $\delta>0$.

Theorem \ref{thm:rough2b}(ii) follows from a similar application of Hurwitz's theorem, with $\kappa_c$ taking the role of the smallness parameter $\delta$ and the factor that the estimates for $L_{s,\Lambda}^{-1}$ are uniform in $\Lambda$.\\
\\

\section{Extensions and further applications}
\label{sec:future}
We describe below future directions of research related to the results and methods derived in the present paper.
\subsection{General asymptotically flat spacetimes and wave equations}
The methods developed in the present paper can be applied to more general settings of wave operators on asymptotically flat spacetimes. They rely on the following key assumptions, which we will only sketch: Let $\Sigma_{\tau}\cong \R \times \s^2$ or $\Sigma_{\tau}\cong \R ^3$ be isometric asymptotically hyperboloidal hypersurfaces that foliate a stationary, asymptotically flat spacetime. Let $\mathcal{L}_s$ denote the Laplace transform of a wave operator with all time derivatives replaced by multiplications with the frequency $s\in \C$. Let $r$ be a radial-like coordinate.
\begin{enumerate}
\item There exists coordinates on $\Sigma=\Sigma_0$ with respect to which the coefficients of $\mathcal{L}_s$ approach the coefficients of the Laplace transform of the wave operator on Minkowski, with respect to an asymptotically hyperboloidal time function. Furthermore, these coefficients are suitably regular (in a Gevrey sense, after suitable rescaling) at infinity.
\item The zero-frequency operator $\mathcal{L}_0$ satisfies the following estimate: there exists $\beta\in \R$, such that for all $k\in \N_0$
\begin{equation*}
	\frac{1}{C_k}\|r^{-1+\delta}\hphi\|^2_{H^{k+1}(\Sigma)}+(k-\beta)\|r^{-1+\delta}\partial_r^{k+1}\hphi\|^2_{L^2(\Sigma)}\leq \|r^{-1+\delta}\mathcal{L}_0\hphi\|^2_{H^{k}(\Sigma)}+\|r^{-1+\delta}\hphi\|^2_{H^{k}(\Sigma)},
\end{equation*}
where $C_k>0$ are uniform constants. The above estimate may be thought of as a ``zero-frequency'' resolvent estimate, incorporating moreover the presence of an ``enhanced red shift effect'' for higher radial derivatives due to the growing factor $k$ on the left-hand side.
\end{enumerate}
One could state the above assumptions precisely and incorporate them into the arguments of the present paper as a ``black box'' to extend the validity of the results to a much more general setting. We do not carry out such a procedure in the present paper.

Analogues of (2) also play an important role in the derivation of sharp decay estimates and late-time (inverse-polynomial) asymptotics for wave equations on asymptotically flat spacetimes; see for example \cite{tataru3, paper2,hintzprice, aagkerr, lukoh24}.

In the case of sub-extremal Reissner--Nordstr\"om spacetimes (which include Schwarzschild spacetimes), (2) can be derived by exploiting the (degenerate) ellipticity of $\mathcal{L}_0$ and combining it with red-shift estimates in a neighbourhood of the event horizon. As already mentioned in \S \ref{sec:ideastechniques}, this argument does not work on sub-extremal Kerr, due to the presence of an ergoregion. Nevertheless, as shown in the present paper, the estimate in (2) does hold. Note also that the necessary convergence to Minkowski in (1) does not hold with respect to Boyer--Lindquist coordinates $(r,\theta,\varphi)$, but does hold with respect to  modified coordinates $(\varrho,\vartheta,\varphi)$.

The Teukolsky wave equations (of arbitrary spin) on sub-extremal Kerr can be shown to satisfy (2). Property (1) can similarly be shown to hold by considering as the main variables appropriate $\varrho$-rescalings and $\varrho^2\partial_{\varrho}$-derivatives of the Teukolsky variables, which are closely related to the conserved Newman--Penrose charges \cite{np2} for the Teukolsky equations on Minkowski; see for example the expressions in \cite{gake22}[\S 3].

\subsection{Higher spacetime dimensions}
The methods in the present paper extend also to wave equations on stationary asymptotically flat spacetimes with spacetime dimensions $n+1$, where $n\geq 3$. Important here is that the wave equation on the $n+1$-dimensional Minkowski spacetime is equivalent to the following equation with respect to standard spherical coordinates: let $\psi=r^{\frac{n-1}{2}}\phi$, then:
	\begin{equation*}
	0=-\partial_t^2\psi+ \partial_r^2\psi-\frac{1}{4r^2}(n-1)(n-3)\psi+r^{-2}\slashed{\Delta}_{\s^{n-1}}\psi.
	\end{equation*}
Using that higher-dimensional spherical harmonics $Y_{\ell m}$ satisfy $\slashed{\Delta}_{\s^{n-1}}Y_{\ell m}=-\ell(\ell+n-2)Y_{\ell m}$, this results in the following equation for $\psi_{\ell}$, the projection of $\psi$ onto the spherical harmonics with angular frequency $\ell$:
\begin{equation}
\label{eq:highdimwefixedl}
	0=-\partial_t^2\psi_{\ell}+ \partial_r^2\psi_{\ell}-\frac{1}{4r^2}[4\ell(\ell+n-2)+(n-1)(n-3)]\psi_{\ell}.
	\end{equation}
	When $n=2k+3$, $k\in \N_0$, \eqref{eq:highdimwefixedl} is equivalent to:
	\begin{equation*}
	0=-\partial_t^2\psi_{\ell}+ \partial_r^2\psi_{\ell}-\frac{1}{r^2}(\ell+k)(\ell+k+1)\psi_{\ell}.
	\end{equation*}
	each $\psi_{\ell}$ satisfies a conservation law after taking appropriately many outgoing null derivatives, analogous to \eqref{eq:consvlaw}. This is directly connected to the validity of the strong Huygens principle. The higher-order Gevrey estimates, which rely on this conservation law, therefore proceed analogously to the $n=3$ case.\footnote{When $n=2k+2$, $k\in \N_0$, the conservation law breaks down (and the strong Huygens principle does not hold). However, in that case, using that the relevant wave operator in the far-away region $\{\varrho \geq 2R_{\infty}\}$ is just the Minkowski operator, one can apply the method of descent, by considering as an auxiliary operator, the Minkowski wave operator in $n+2$ spacetime dimensions, restricted to functions that do not depend on the Cartesian coordinates $x_{n+1}$. Then estimates can be carried out in the far-away region for the rescaled functions $\tilde{r}^{\frac{n}{2}}\phi$, with $\tilde{r}^2=r^2+x_{n+1}^2$ by considering the Laplace transform of the $(n+2)$-dimensional wave operator, with respect to the time function $\tilde{\tau}=t-\tilde{r}$. This can be used to prove the existence of $L_s^{-1}$ on appropriate Hilbert spaces, away from a subset of isolated $s\in \C$.}
		
\subsection{Extending the sector of the complex plane} 
\label{sec:extendingsector}
The construction of regularity quasinormal mode frequencies in the present paper applies only to the sector
\begin{equation*}
	\left \{|{\rm arg} z|<\frac{2\pi}{3}\right\}\subset \C.
\end{equation*}
In contrast, scattering resonances, $i s\in \C$, \underline{are} well-defined in the case of fixed spherical harmonics on Schwarzschild for $s\notin (-\infty,0]$ by \cite{ba}; see also the upcoming work \cite{stuc24} for an analogous statement on sub-extremal Kerr.

The condition $|{\rm arg} z|<\frac{2\pi}{3}$ is not expected to be sharp. Indeed, in the context of a model problem, which includes the wave equation restricted to fixed spherical harmonic modes on Minkowski, a slightly larger sector was obtained in \cite{gajwar19b}. 

It remains an open problem whether the excluded sector around $ \R_{\leq 0}$ can be made \emph{arbitrarily} small in the setting of regularity quasinormal modes, even in the context of the standard wave equation on Minkowski (even though in this case, the \emph{cut-off resolvent} is in fact holomorphic in the whole complex plane).\footnote{In fact, since the condition $|{\rm arg} z|<\frac{2\pi}{3}$ arises purely at the level of estimates for the exact Minkowski wave operator in the far-away region $r\geq 2R_{\infty}$, an improved condition on Minkowski would suffice to obtain an improved condition in Kerr (and more general asymptotically flat spacetimes).}

\subsection{Extremal black hole spacetimes}
In the case of extremal black hole spacetimes, (2) does \underline{not} hold, due to the absence of the horizon red-shift effect. However, modifications exist with additional degeneracies at the event horizon; see \cite{paper4, gajwar19a, gaj22b}. For this reason, regularity quasinormal modes can still be defined on extremal Reissner--Nordstr\"om. We expect that the methods of the present paper can also be extended to extremal Kerr, provided the near-horizon region is treated in a similar way to the far-away region in the sub-extremal case. An important difference with the sub-extremal setting is that the frequencies $\omega=m\upomega_+$,  with $\upomega_+$ the angular velocity of the black hole and $m\in \Z$ the azimuthal number, play the role of $\omega=0$ in the far-away region. The region of the complex plane for which regularity quasinormal frequencies can be defined with the methods of the present paper would therefore have to exclude in addition sectors around the lines $s=(-\infty,0]+ im\upomega_+$ for all $m\in \Z$.

\subsection{Zero-damped frequencies and high frequencies}
With the results in the present paper and the methods developed in \cite{hixi22, joy22, hin24}, the properties of zero-damped quasinormal frequencies, which are expected to accumulate at $\omega=m\upomega_+$ with $m\neq 0$ in the extremal limit of sub-extremal Kerr can in principle be studied.

The behaviour of regularity quasinormal frequencies in sub-extremal Kerr for bounded $|\re s|$ (or bounded $|\im \omega|$) and very large $|\im s|$ (or very large $|\re  \omega|$), which is closely related to the structure of the trapping of null geodesics in the spacetime, would also be an interesting topic for further research in a uniform-in-$\Lambda$ setting. See \cite{bazw97,pra99, dya12,hizw24} for related results in de asymptotically de Sitter and Schwarzschild (fixed spherical harmonic, scattering resonance) setting.
\section{Preliminaries: geometry}
\label{sec:geometry}
We start by defining the main geometric objects related to Kerr spacetimes. We will briefly discuss relevant geometric aspects of Kerr--de Sitter spacetimes in \S \ref{sec:KdS}.
\subsection{Coordinates and vector fields}
Let $(\mathcal{M}_{M,a},g_{M,a})$ denote the the black hole exterior of a Kerr spacetime, including the future event horizon, with mass $M$ and angular momentum/mass $a$. With respect to \emph{ingoing Kerr coordinates} $(v,r,\theta,\varphi_*)$, we can express
\begin{equation*}
\mathcal{M}_{M,a}= \R_v\times [r_+,\infty)_r\times \s^2_{\theta,\varphi},
\end{equation*}
where
\begin{equation}
\label{eq:kerrmetric}
\begin{split}
g_{M,a}=&-\rho^{-2}\left(\Delta-a^2\sin^2\theta\right)dv^2+2dvdr-4Ma r \rho^{-2}  \sin^2\theta dv d\varphi_*-2a\sin^2\theta drd\varphi_*+\rho^2d\theta^2\\
&+\rho^{-2}((r^2+a^2)^2-a^2\Delta \sin^2\theta)\sin^2\theta d\varphi_*^2,
\end{split}
\end{equation}
with
\begin{align*}
\Delta=&\:r^2-2Mr+a^2=(r_+-r)(r-r_-),\\
r_{+,-}=&\:M\pm \sqrt{M^2-a^2},\\
\rho^2=&\: r^2+a^2\cos^2\theta.
\end{align*}

We obtain the \emph{Boyer--Lindquist coordinates} $(t,r,\theta,\varphi_*)$ via the following relations:
\begin{align*}
t=&\:v-r_*,\\
\varphi=&\:\varphi_*+\int_{r}^{\infty} \frac{a}{\Delta}\,dr' \mod 2\pi,\\
\end{align*}
where $r_*: (r_+,\infty)\to \R$ is a solution to
\begin{equation*}
\frac{dr_*}{dr}=\frac{r^2+a^2}{\Delta},
\end{equation*}
defined uniquely up to a constant. Note that the Boyer--Lindquist coordinates are only well defined on $\mathcal{M}_{M,a}\setminus \{r=r_+\}$.

It will be convenient to state moreover the inverse metric $g^{-1}_{M,a}$:
\begin{equation}
\label{eq:invmetric}
\begin{split}
g^{-1}_{M,a}=&\:a^2\rho^{-2}\sin^2\theta \partial_v\otimes \partial_v+\rho^{-2}(r^2+a^2)[ \partial_v\otimes \partial_r+ \partial_r\otimes \partial_v]+\Delta \rho^{-2}\partial_r\otimes \partial_r\\
&+a\rho^{-2} [(\partial_v+\partial_r)\otimes \partial_{\varphi_*}+\partial_{\varphi_*}\otimes (\partial_v+\partial_r)]+\rho^{-2}[\partial_{\theta}\otimes \partial_{\theta}+\sin^{-2}\theta \partial_{\varphi_*}\otimes\partial_{\varphi_*}].
\end{split}
\end{equation}

The \emph{surface gravity $\kappa_+$} associated to the future event horizon $\mathcal{H}^+=\{r=r_+\}$ takes the following expression:
\begin{equation*}
\kappa_+=\frac{r_+-M}{2Mr_+}.
\end{equation*}

We introduce the following vector fields, defined with respect to $(v,r,\theta,\varphi_*)$ coordinates:
\begin{align*}
Y=&\:\partial_r,\\
T=&\:\partial_v,\\
\Phi=&\:\partial_{\varphi_*}.
\end{align*}

We will now introduce the coordinate chart $(\tilde{v},\varrho,\vartheta,\varphi_*)$ which we motivate by first considering the $M=0$ case where the Kerr spacetime is locally isometric to the Minkowski spacetime and the coordinates $(v,r,\theta,\varphi_*)$ correspond to a foliation of the spacetime by axisymmetric ellipsoids with eccentricity parameter $a^2$. 

In the $M=0$ case, the transformation from $(v,r,\theta)$ to $(\tilde{v},\varrho,\vartheta)$ corresponds to passing from ellipsoids with a non-zero eccentricity parameter to round spheres. More precisely, $(r,\theta)$ have the following relation to Cartesian coordinates $(\tilde{x},\tilde{y},\tilde{z})$ on Minkowski:
\begin{align*}
\tilde{z}=&\:r\cos \theta,\\
\tilde{x}^2+\tilde{y}^2=&\:(r^2+a^2)\sin^2\theta.
\end{align*}
Let $(\varrho,\vartheta,\varphi)$ denote standard spherical coordinates on Minkowski, then we also have that
\begin{align*}
\tilde{z}=&\:\varrho\cos \vartheta,\\
\tilde{x}^2+\tilde{y}^2=&\:\varrho^2\sin^2\vartheta.
\end{align*}
From the above, it is straightforward to obtain the following relations in Minkowski:
\begin{align}
\label{eq:coordtrafo1}
\varrho^2=&\:r^2+a^2\sin^2\theta,\\
\label{eq:coordtrafo2}
\sin^2\vartheta=&\:\frac{r^2+a^2}{r^2+a^2\sin^2\theta}\sin^2\theta,\\
\label{eq:coordtrafo3}
r^2=&\:\frac{1}{2}(\varrho^2-a^2)\left(1+\sqrt{1+\frac{4a^2\varrho^2}{(\varrho^2-a^2)^2}\cos^2\vartheta}\right),\\
\label{eq:coordtrafo4}
\sin^2\theta=&\:\frac{1}{2a^2}\left(\varrho^2+a^2-\sqrt{(\varrho^2+a^2)^2-4a^2\varrho^2\sin^2\vartheta}\right)
\end{align}

Consider now the $M\neq 0$ case. In this case, we define $\varrho: [r_+,\infty)\times (0,\pi)\to [r_+,\infty)$ and $\vartheta: [r_+,\infty)\times (0,\pi)\to  (0,\pi)$ as smooth functions satisfying:
\begin{align*}
\varrho(r,\theta)=&\:\sqrt{r^2+\check{a}^2\sin^2\theta},\\
(\sin\vartheta)(r,\theta)=&\:\sqrt{\frac{r^2+\check{a}^2}{r^2+\check{a}^2\sin^2\theta}}\sin \theta,
\end{align*}
with $\check{a}: [r_+,\infty)\to \R$ a smooth function with $|\check{a}|\leq |a|$, such that $\check{a}(r)=a$ for $r\geq \frac{1}{\sqrt{2}}R_{\infty}$ and $\check{a}(r)=0$ for $r\leq \frac{1}{2}R_{\infty}$. Then $r\leq \varrho(r,\theta)\leq r\sqrt{1+a^2r^{-2}}\leq \sqrt{2}r$ and, for $r\leq \frac{1}{2} R_{\infty}$ we have that $\varrho=r$. Hence, in particular, we have that $\check{a}=a$ in the region $\{\varrho\leq \frac{1}{2}R_{\infty}\}$ and $\check{a}=0$ in the region $\{\varrho\geq R_{\infty}\}$.

Note that for $|a|R_{\infty}^{-1}$ suitably small,
$\begin{pmatrix} \partial_r \varrho & \partial_r \vartheta,\\
\partial_{\theta} \varrho & \partial_{\theta} \vartheta
\end{pmatrix}$
is clearly invertible everywhere.

In particular, for $r\geq \frac{1}{\sqrt{2}}R_{\infty}$:
\begin{align*}
\partial_r\vartheta=&\:\frac{a^2}{r(r^2+a^2)}\sin^2\vartheta\cos^2\vartheta,\\
\partial_r\varrho=&\: \frac{r}{\varrho}.
\end{align*}

Furthermore, it follows that we can write for all $r\geq r_+$:
\begin{align*}
r(\varrho,\vartheta)=&\:\varrho-f_r(\varrho,\vartheta),\\
\theta(\varrho,\vartheta)=&\:\vartheta-f_{\theta}(\varrho,\vartheta),
\end{align*}
where $f_r$ and $f_{\theta}$ are smooth everywhere, analytic in $r\geq \frac{1}{\sqrt{2}}R_{\infty}$ with respect to the variable $\frac{1}{r}$ and with:
\begin{align*}
f_r(\varrho,\vartheta)=&\:a^2O_{\infty}\left(\frac{1}{\varrho}\right),\\
f_{\theta}(\varrho,\vartheta)=&\:a^2O_{\infty}\left(\frac{1}{\varrho^2}\right).
\end{align*}

In particular,
\begin{align}
\label{eq:drvarrho}
\partial_r\varrho=&\:\frac{1}{2\varrho}\partial_r(\varrho^2)=\varrho^{-1}(r+\check{a}\partial_r\check{a}\sin^2\theta)=1+a^2\sin^2\vartheta O_{\infty}(\varrho^{-2}),\\
\label{eq:drvartheta}
\partial_r\vartheta=&\:\frac{\partial_r(\sin^2\vartheta)}{2 \sin \vartheta \cos\vartheta}=\frac{\partial_r(\check{a}^2(r^2+\check{a}^2\sin^2\theta)^{-1})\sin^2\theta \cos^2\theta}{2 \sin \vartheta \cos\vartheta}=a^2O_{\infty}(\varrho^{-3})\sin\vartheta \cos\vartheta
\end{align}

We define $\tilde{v}=t+r_*(\varrho)$ so that $\tilde{v}=v+r_*(\varrho)-r_*(r)$. Then $\tilde{v}=v$ in $r\leq \frac{R_{\infty}}{2}$.
 
\subsection{Asymptotically hyperboloidal foliations}
\label{sec:asymphypfol}
We will define a time function $\tau$, such that the level sets of $\tau$, denoted $\Sigma_{\tau}$, are appropriate asymptotically hyperboloidal hypersurfaces. Let
\begin{equation*}
\tau=\tilde{v}-\int_{r_+}^{\varrho} \mathbbm{h}(\varrho')\,d\varrho'-v_0,
\end{equation*}
where $v_0\in \R_{>0}$ and $\mathbbm{h}: [r_+,\infty)\to \R$ is a smooth, non-negative function, such that
\begin{align}
\label{eq:condh1}
2-\mathbbm{h}(1-2M \varrho^{-1})=&\:O(\varrho^{-2}) ,\\
\label{eq:condh2}
\rho^{-2}\left[-\mathbbm{h}(2(r^2+a^2)-\mathbbm{h}\Delta)+a^2\sin^2\theta\right]=g^{-1}_{M,a}(d\tau,d\tau)<&\:0.
\end{align}
By construction, $T(\tau)=1$.

Let $x=\frac{1}{\varrho}$ and $h:[0,\frac{1}{r_+}]_x\to \R$, with 
\begin{equation*}
M^2 h(x)=2x^{-2}-\mathbbm{h}(x^{-1})(x^{-2}-2Mx^{-1}). 
\end{equation*}
By assumption on $\mathbbm{h}$, $h(x)=O(1)$. \textbf{We will assume that $h$ is analytic in an arbitrarily small neighbourhood of $x=0$ and that $h(0)\neq 0$.} By $g^{-1}_{M,a}(d\tau,d\tau)<0$, we then necessarily have that $h(0)>0$.

Note that we can obtain $\tau=t$, by defining instead $\mathbbm{h}: (r_+,\infty)\to \R$ by $\mathbbm{h}(r)=\frac{r^2+a^2}{\Delta}$, which cannot be smoothly extended to $r=r_+$ and does not satisfy \eqref{eq:condh1}.

If $M=0$, then we can simply take $\mathbbm{h}(\varrho)=2$ and $\Sigma_{\tau}$ are outgoing null hypersurfaces.

For the sake of later convenience, we moreover assume that $1\leq \mathbbm{h}\leq \frac{r^2+a^2}{\Delta} $ in the region $\varrho \leq R_{\infty}$, where $R_{\infty}\gg M$ will be defined later.

We can extend the manifolds-with-boundary $\Sigma_{\tau'}=\{\tau'\}\times (0,r_+^{-1}]_x\times \s^2_{\vartheta,\varphi_*}$ as follows:
\begin{equation*}
\widehat{\Sigma}_{\tau'}=\{\tau'\}\times [0,r_+^{-1}]_x\times \s^2_{\vartheta,\varphi_*}.
\end{equation*}
We denote moreover $\mathcal{I}^+=\R_{\tau}\times \{0\}_x\times \s^2_{\vartheta,\varphi_*}$ and also $\Sigma:=\Sigma_0$ and $\widehat{\Sigma}:=\widehat{\Sigma}_0$.

The extended manifold $[0,\infty)_{\tau}\times [0,r_+^{-1}]_x\times \s^2_{\vartheta,\varphi_*}$ can be equipped with an extension of the conformally rescaled metric $r^{-2}g_{M,a}$ and forms a conformal compactification of the domain of dependence $D^+(\Sigma)$. The subset $\mathcal{I}^+$ is called future null infinity. See Figure \ref{fig:kerrfoli} for an illustration of the above geometric notions.

The tuple $(\tau,\varrho, \vartheta, \varphi_*)$ forms a well-defined coordinate chart on the spacetime. With respect to these coordinates, we can express:
\begin{align*}
T=&\:\partial_{\tau} \quad \textnormal{and we define}\\
\Theta:=&\: \partial_{\vartheta},\\
 X:=&\:\partial_{\varrho}-\mathbbm{h}\partial_{\tau}.
 \end{align*}
 Alternatively, we can consider the coordinates $(\tau,x, \vartheta, \varphi_*)$, to obtain the following expression for $X$:
 \begin{equation*}
 -X=x^{2}\partial_x+\frac{2-M^2 x^2h(x)}{1-2M x}\partial_{\tau}.
 \end{equation*}
 By applying \eqref{eq:drvarrho} and \eqref{eq:drvartheta}, we can further express $Y$ in terms of $X$ and $\Theta$:
\begin{equation*}
	Y=(\partial_r\varrho) X+(\partial_r \vartheta) \Theta=(1+a^2\sin^2\vartheta O_{\infty}(\varrho^{-2}))X+a^2O_{\infty}(\varrho^{-3})\sin\vartheta \cos\vartheta \Theta.
\end{equation*}

\section{Preliminaries: differential operators}
In this section, we introduce the main differential operators that are relevant for the analysis in the paper. We will first express the Kerr wave operator $\square_{g_{M,a}}$ with respect to suitable spacetime coordinates, introduced in \S \ref{sec:geometry}, and we will then write down the corresponding rescaled Laplace-transformed operator $L_s$, as well as $\mathfrak{L}_{s,\lambda,\kappa}$, which is a modification of $L_s$ that plays a leading role in the main estimates. In \S \ref{sec:KdS}, we will briefly discuss analogous objects in the context of Kerr--de Sitter spacetimes.
\subsection{Wave operator}
We will first express the wave operator $\square_{g_{M,a}}$ with respect to different choices of vector fields, adapted to the spacetime coordinates under consideration. With respect to the vector fields $(T,Y,\partial_{\theta},\Phi)$, introduced in \S \ref{sec:geometry}, we have that:
\begin{multline}
\label{eq:waveoporigcoord}
\square_{g_{M,a}}\phi=\rho^{-2}Y(\Delta Y\phi)+2a\rho^{-2}Y\Phi \phi+\rho^{-2}\slashed{\Delta}_{\s^2,\theta}\phi+2\rho^{-2}(r^2+a^2) YT\phi+2\rho^{-2}r T\phi\\
+a^2\sin^2\theta \rho^{-2}T^2\phi+2a\rho^{-2} T\Phi\phi.
\end{multline}

We will now express $\square_{g_{M,a}}$ in terms of the vector fields $(T,X,\Theta,\Phi)$, where we recall that in the region $\{\varrho\leq \frac{R_{\infty}}{2}\}$, $X=Y-\mathbbm{h}T$ and $\Theta=\partial_{\theta}$. We also introduce the following spherical differential operators:
\begin{align*}
\slashed{\Delta}_{\s^2,\theta}(\cdot )=&\:\frac{1}{\sin \theta}\partial_{\theta}(\sin\theta \partial_{\theta}(\cdot ))+\frac{1}{\sin^2\theta} \Phi^2,\\
\slashed{\Delta}_{\s^2,\vartheta}(\cdot )=&\:\frac{1}{\sin \vartheta}\partial_{\vartheta}(\sin\vartheta \partial_{\vartheta}(\cdot ))+\frac{1}{\sin^2\vartheta} \Phi^2.
\end{align*} 

We can decompose $f\in L^2(\s^2)$ into azimuthal modes:
\begin{equation*}
	f=\sum_{m\in \Z}f_m,
\end{equation*}
with $\Phi f_m=i mf_m$.

We can also decompose $f\in L^2(\s^2)$ into spherical harmonic modes:
\begin{equation*}
	f(\vartheta,\varphi_*)=\sum_{\ell\in \N_0} f_{\ell}(\vartheta,\varphi_*)=\sum_{m\in \Z}\sum_{\ell\geq |m|}f_{m\ell}Y_{\ell m}(\vartheta,\varphi_*),
\end{equation*}
with $\slashed{\Delta}_{\s^2,\vartheta}f_{\ell}=-\ell(\ell+1)f_{\ell}$.

We will reserve the notation $\slashed{\Delta}_{\s^2}$ for the Laplacian $\slashed{\Delta}_{\s^2,\vartheta}$ and $\snabla_{\s^2}$ for the covariant derivative $\snabla_{\s^2,{\vartheta}}$, both associated to spheres of constant $\tau$ and $\varrho$.

 First, we restrict to the region $\{\varrho\geq R_{\infty}\}$ and use the fact that in this region $\check{a}(r)=0$. In the $M=0$ case, the operator $\square_{g_{M,a}}$ with respect to $(T,X,\partial_{\vartheta},\Phi)$ corresponds simply to the wave operator on Minkowski in outgoing Eddington--Finkelstein coordinates $(\tilde{u},\varrho,\vartheta,\varphi_*)$, with $T=\partial_{\tilde{u}}$, $X=\partial_{\varrho}$ and $\Phi=\partial_{\varphi_*}$.

We can therefore express in $\{\varrho\geq R_{\infty}\}$:
\begin{multline*}
\rho^{-2}Y(r^2 Y\phi)+2a\rho^{-2}Y\Phi \phi+\rho^{-2}\slashed{\Delta}_{\s^2,\theta}\phi+2\rho^{-2}(r^2+a^2) YT\phi+2r T\phi+a^2\sin^2\theta \rho^{-2}T^2\phi+2a\rho^{-2} T\Phi\phi\\
=\square_{g_{0,a}}\phi=\varrho^{-2}X(\varrho^2 X \phi)+\varrho^{-2}\slashed{\Delta}_{\s^2,\vartheta}\phi+2\varrho^{-1}X(\varrho T\phi).
\end{multline*}
Hence, in $\{\varrho\geq R_{\infty}\}$, we have that
\begin{equation}
\label{eq:waveoprgeqRinfty}
\square_{g_{M,a}}\phi=\varrho^{-2}X(\varrho^2X \phi)+\varrho^{-2}\slashed{\Delta}_{\s^2,\vartheta}\phi+2\varrho^{-1}X(\varrho T\phi)-2M\varrho^{-2}Y(r Y\phi) .
\end{equation}

We now turn to the full region $\{r\geq r_+\}$. 
\begin{lemma}
\label{lm:operatorcoeffb}
Define $\psi:=\varrho \cdot \phi$. With respect to the coordinates $(\tau, x=\varrho^{-1},\vartheta,\varphi_*)$, we can write:
\begin{equation}
	\label{eq:waveeqxcoord}
	\varrho^3 \square_{g_{M,a}}\phi=\partial_{x}(x^2\partial_{x}\psi)+2 \partial_{x}\partial_{\tau}\psi+\slashed{\Delta}_{\s^2,\vartheta}\psi+Mx \cdot \mathcal{B}^{\infty}\psi+\mathcal{K}^{\infty}\psi,
\end{equation}
with
\begin{align*}
\mathcal{B}^{\infty}f:=&\:b_{xx}\cdot x^2\partial_x^2f+b_{x\tau }\cdot \partial_x \partial_{\tau} f+b_{\tau \tau}\partial_{\tau}^2f+a^2( xb_{x \vartheta}\cdot x^2\partial_x+x b_{\tau \vartheta}\partial_{\tau})  \sin\vartheta \partial_{\vartheta}f \\
+&\:a^4x^{4}\cdot b_{\vartheta \vartheta}(\sin\vartheta\partial_{\vartheta})^2f+ b_{x} \cdot x \partial_xf+ a^2x^2b_{\vartheta}\cdot \sin \vartheta\partial_{\vartheta} f+b_{\tau}\partial_{\tau}f+ b\cdot  f,\\
\mathcal{K}^{\infty}f:=&\:-2M^2h(0) \partial_{\tau}^2f,
\end{align*}
where the functions $b_{\star \star}, b_{\star},b$, with $\star\in \{x,\tau,\vartheta\}$, are smooth and bounded functions of $(\varrho,\vartheta)$, by \eqref{eq:waveoprgeqRinfty}. We moreover have that $b_{\star \star}, b_{\star},b$ are analytic in $[0,x_0]_x\times \s^2_{\theta,\varphi}$ and there exist constants $C,B_1,B_2>0$, such that
\begin{equation*}
\frac{1}{(n+m)!}(\|\snabla_{\s^2}^m \partial_x^n b\|_{L^{\infty}([0,x_0]\times \s^2)}+\|\snabla_{\s^2}^m \partial_x^n b_{\star}\|_{L^{\infty}([0,x_0]\times \s^2)}+\|\snabla_{\s^2}^m \partial_x^n b_{\star\star}\|_{L^{\infty}([0,x_0]\times \s^2)})\leq C B_1^n (a^2 x_0^2 B_2)^m.
\end{equation*}
In $\{\varrho\leq \frac{R_{\infty}}{2}\}$ we can express:
\begin{equation}
\label{eq:negttcoeff}
(M x b_{\tau\tau}-2M^2h(0)) =-\mathbbm{h}(2(r^2 +a^2)-\Delta \mathbbm{h})+a^2\sin^2\theta,
\end{equation}
which is strictly negative by \eqref{eq:condh2}.
\end{lemma}
\begin{proof}
Note that the form \eqref{eq:waveeqxcoord} and the smoothness property of $b_{\star \star}, b_{\star},b$ hold trivially in $\{\varrho \leq R_{\infty}\}$. Using that in $\{\varrho\leq \frac{R_{\infty}}{2}\}$, we have that $\varrho=r$ and $\vartheta=\theta$, so $Y=X+\mathbbm{h}T$ and the property \eqref{eq:negttcoeff} also follows immediately after expanding the $Y$ derivatives in \eqref{eq:waveoporigcoord} and considering the factor appearing in front of $T^2\phi$.

It remains to verify the stated asymptotic properties of $b_{\star \star}, b_{\star},b$ as $x\to 0$, for which we can restrict to the region $\{\varrho \geq R_{\infty}\}$. We first recall that:
 \begin{equation*}
 -X=x^{2}\partial_x+\frac{2-M^2 x^2h(x)}{1-2M x}\partial_{\tau},
 \end{equation*}
 so
\begin{multline}
\label{eq:waveeqfromXtox}
\varrho X(\varrho^2X \phi)+\varrho \slashed{\Delta}_{\s^2,\vartheta}\phi-2\varrho^2 X(\varrho T\phi)=\varrho^2 X^2(\varrho \phi)+2 \varrho^2 XT(\varrho \phi)+\varrho \slashed{\Delta}_{\s^2,\vartheta}\phi\\
=\left(\partial_{x}+x^{-2}\frac{2-M^2 x^2h(x)}{1-2M x}\partial_{\tau}\right)\left(x^2\partial_{x}\psi+\frac{2-M^2 x^2h(x)}{1-2M x}\partial_{\tau}\psi\right)-2 \left(\partial_{x}+x^{-2}\frac{2-M^2 x^2h(x)}{1-2M x}\partial_{\tau}\right)\partial_{\tau}\psi\\
+\slashed{\Delta}_{\s^2,\vartheta}\psi\\
=\partial_{x}(x^2\partial_{x}\psi)+2 \partial_{x}\partial_{\tau}\psi+\slashed{\Delta}_{\s^2,\vartheta}\psi-2M^2\frac{h(x)}{1-2M x}\partial_{x}\partial_{\tau}\psi\\
+\frac{4Mx^{-1}-M^2h}{(1-2M x)^2}(2-M^2 x^2h(x))\partial_{\tau}^2\psi+\partial_x\left(\frac{2-M^2 x^2h(x)}{1-2M x}\right)\partial_{\tau}\psi.
\end{multline}
Using that:
\begin{equation*}
Y=(1+a^2\sin^2\vartheta O_{\infty}(\varrho^{-2}))X+a^2O_{\infty}(\varrho^{-3})\sin\vartheta \cos\vartheta \partial_{\vartheta},
\end{equation*}
we moreover obtain:
\begin{multline}
\label{eq:doubleYder}
-2M\varrho Y(r Y\phi)\\
=2M x^{-1} \left[(1+a^2\sin^2\vartheta O_{\infty}(x^2))\left(x^{2}\partial_x+\frac{2-M^2 x^2h(x)}{1-2M x}\partial_{\tau}\right)+a^2O_{\infty}(x^3)\sin\vartheta \cos\vartheta\partial_{\vartheta}\right](x\psi)\\
-2M   x^{-1} r\left((1+a^2\sin^2\vartheta O_{\infty}(x^2))\left(x^{2}\partial_x+\frac{2-M^2 x^2h(x)}{1-2M x}\partial_{\tau}\right)+a^2O_{\infty}(x^3)\sin\vartheta \cos\vartheta \partial_{\vartheta}\right)^2(x\psi).
\end{multline}
Using that $r= x^{-1}+ a^2O(x)$, terms involving $\partial_{\tau}^2\psi$ appearing on the very right-hand side combine to obtain:
\begin{equation*}
-2M x^{-1} \frac{2+M^2O_{\infty}(x^2)+a^2O_{\infty}(x^2)}{(1-2M x)^2}(2-M^2 x^2h(x))\partial_{\tau}^2\psi,
\end{equation*}
which will cancel to leading-order in $x$ with the $\partial_{\tau}^2\psi$ term on the very right-hand side of \eqref{eq:waveeqfromXtox}, to obtain a total leading-order term:
\begin{equation*}
-2M^2h(0) \partial_{\tau}^2\psi
\end{equation*}
when we consider, via \eqref{eq:waveoprgeqRinfty}, 
\begin{equation*}
\varrho^3 \square_{g_{M,a}}\phi=\varrho X(\varrho^2X \phi)+\varrho \slashed{\Delta}_{\s^2,\vartheta}\phi-2\varrho^2 X(\varrho T\phi)-2M\varrho Y(r Y\phi)
\end{equation*}
in the region $\{\varrho\geq R_{\infty}\}$.

Similarly, the leading-order term in the factor multiplying $\partial_{\tau}\psi$ in \eqref{eq:doubleYder} is:
\begin{equation*}
-4M\partial_{\tau}\psi,
\end{equation*}
which cancels with the leading-order contribution in the $\partial_{\tau}\psi$ term on the very right-hand side of \eqref{eq:waveeqfromXtox}. 

For the remaining coefficients, the asymptotics in $x$ follow straightforwardly.

The analyticity properties of $b,b_{\star},b_{\star\star}$ in $\{r\geq R_{\infty}\}$ follow from the analyticity of $r(\varrho,\vartheta)$ by \eqref{eq:coordtrafo3} and the $a^2x_0^2$ factor in the $L^{\infty}$ estimate follows from the factor $\frac{4a^2\varrho^2}{(\varrho^2-a^2)^2}$ in front of the $\vartheta$-dependent term in \eqref{eq:coordtrafo3}, which plays the role of a scale factor when taking $\vartheta$-derivatives of $r$
\end{proof}

\subsection{Laplace-transformed operator}
\label{sec:laplacetransfop}
We define the differential operator $L_s$ to be the Laplace transform of the operator
\begin{equation*}
\varrho^3\circ \square_g\circ \varrho^{-1}
\end{equation*}
with respect to the time function $\tau$. More explicitly, for $\hpsi\in C^{\infty}(\widehat{\Sigma})$:
\begin{equation*}
L_s\hpsi=e^{-s\tau}\varrho^3\square_g(e^{s\tau}\varrho^{-1}\hpsi).
\end{equation*}

In view of \eqref{eq:waveeqxcoord}, we can then write:
\begin{equation}
\label{eq:generalexprLs}
	L_s=\partial_{x}(x^2\partial_{x}(\cdot))+2 s\partial_{x}(\cdot)+\slashed{\Delta}_{\s^2,\vartheta}+M x\cdot{B}^{\infty}_s+K^{\infty}_s,
\end{equation}
with $B_s^{\infty}$ and $K_s^{\infty}$ the Laplace transforms of $\mathcal{B}^{\infty}$ and $\mathcal{K}^{\infty}$, defined as follows:
\begin{align*}
B^{\infty}_sf:=&\:b_{xx}\cdot x^2\partial_x^2f+s b_{x\tau }\cdot \partial_x  f+s^2b_{\tau \tau}f+a^2( xb_{x \vartheta}\cdot x^2\partial_x+s x b_{\tau \vartheta})  \sin\vartheta \partial_{\vartheta}f \\
+&\:a^4x^{4}\cdot b_{\vartheta \vartheta}(\sin\vartheta\partial_{\vartheta})^2f+ b_{x} \cdot x \partial_xf+ a^2x^2b_{\vartheta}\cdot \sin \vartheta\partial_{\vartheta} f+sb_{\tau}f+ b\cdot  f,\\
K^{\infty}_sf:=&\:-2M^2h(0) s^2 f.
\end{align*}

\subsection{The Laplace-transformed operator as a perturbation}
The main estimates in this paper will not directly involve $L_s$. Instead we will treat $L_s$ as a \emph{perturbation} of a suitable operator that we will denote $\mathfrak{L}_{s,\lambda}$. Here, $\lambda$ will be a large parameter that can be interpreted as an angular frequency shift. In \S \ref{sec:invtLs}, we will show that $L_s$ can be split up in the following way:
\begin{equation}
\label{eq:Lsaspert}
L_s=\mathfrak{L}_{s,\lambda}+\epsilon_0 \cdot B_{s}+K_{s,\lambda},
\end{equation}
with $B_{s,\lambda}: H^{N_++2}_{\sigma, R_{\infty},2}\to H^{N_+}_{\sigma,R_{\infty}}$ and $K_{s,\lambda}:  H^{N_++1}_{\sigma, R_{\infty},2}\to H^{N_++1}_{\sigma, R_{\infty},2}$ bounded operators, with  $H^{N_++1}_{\sigma, R_{\infty},2}\subset H^{N_+}_{\sigma,R_{\infty}}$ Hilbert spaces that will be introduced in \S \ref{sec:hilbertspaces}. 

We will also show that that the composition $K_{s,\lambda}\circ R$, with $R: H^{N_+}_{\sigma,\delta}\to  H^{N_++1}_{\sigma, R_{\infty},2}$ a suitable bounded operator, is a compact operator and that we can arrange $\epsilon_0$ to be suitably small to obtain $\|\epsilon_0 B_{s} \circ R\|<1$. This splitting of $L_s$ will allow us to obtain \emph{meromorphicity} of the map $s\mapsto L_s^{-1}$ from the invertibility of $\mathfrak{L}_{s,\lambda}$ for $s$ restricted to suitable domains, via the analytic Fredholm theorem.

We will in fact introduce a family of operators $\mathfrak{L}_{s,\lambda, \kappa}$, where $\kappa\in [0,\infty)$, with $\mathfrak{L}_{s,\lambda}:=\mathfrak{L}_{s,\lambda, 0}$. The parameter $\kappa>0$ plays a regularizing role, which can be interpreted as making ``surface gravity of future-null infinity $\mathcal{I}^+$'' positive. These operators can also be thought of as simplifications of the operators that are relevant on Kerr--de Sitter spacetimes.

Using that $\varrho=r$ and $\vartheta=\theta$ in $\{r\leq \frac{1}{2}R_{\infty}\}$ and applying \eqref{eq:waveoporigcoord}, we also have that for $\varrho\leq \frac{1}{2}R_{\infty}$:
\begin{multline*}
r^{-2}\rho^{2}L_s=r\partial_{r}(\Delta \partial_r(r^{-1} \cdot))+2a r \partial_r\partial_{\varphi_*} (r^{-1}\cdot )+\slashed{\Delta}_{\s^2,\theta}(\cdot)-2r[\mathbbm{h}\Delta-(r^2+a^2)] s\partial_r(r^{-1}\cdot)+\left[2r-\frac{d}{dr}(\Delta \mathbbm{h})\right]s (\cdot)\\
+[\mathbbm{h}^2\Delta-2\mathbbm{h}(r^2+a^2)+a^2\sin^2\theta] s^2(\cdot)-2a(\mathbbm{h}-1) s\partial_{\varphi_*}(\cdot).
\end{multline*}

Define
\begin{equation*}
r^{-2}\rho^{2}{K}_{s}^+f:=\left[2r-\frac{d}{dr}(\Delta \mathbbm{h})\right]s (r^{-1}f)+[\mathbbm{h}^2\Delta-2\mathbbm{h}(r^2+a^2)+a^2\sin^2\theta] s^2 f-2a(\mathbbm{h}-1) s\partial_{\varphi_*} f.
\end{equation*}

Then we can express in $\{\varrho \leq \frac{1}{2}R_{\infty}\}$:
\begin{equation*}
	L_s=r^{2}\rho^{-2}\left[r\partial_r(\Delta \partial_r(r^{-1}\cdot))+2ar\partial_r\partial_{\varphi_*}(r^{-1}\cdot)+\slashed{\Delta}_{\s^2}(\cdot)+(4Mr^2-2(\mathbbm{h}-1)r\Delta )s\partial_r(r^{-1}\cdot)\right]+{K}_{s}^+.
\end{equation*}

Let $\lambda\in [0,\infty)$ and define
\begin{equation*}
S_{\lambda}f=2 l \sum_{\ell\in \N_0}\ell f_{\ell},
\end{equation*}
 with $l\in \R_{\geq 0}$ the non-negative solution to $l(l+1)=\lambda$. Then, in $\{r\geq R_{\infty}\}$, we can write
\begin{equation*}
(\slashed{\Delta}_{\s^2}-(\lambda \mathbf{1}+S_{\lambda}))f_{\ell}=-(\ell+l)(\ell+l+1)f_{\ell}.
\end{equation*}
The operator $-(\lambda \mathbf{1}+S_{\lambda})$ therefore has the effect of shifting the angular frequency $\ell$ by $l$. Later, we will need to restrict to $\lambda$ such that $l\in \N_0$ (and we need to take $l$ suitably large). 

Denote with $\pi_{\leq l}$ the projection of $f$ to spherical harmonics with angular momentum $\leq l$. Then
\begin{equation*}
(\slashed{\Delta}_{\s^2}-(\lambda \mathbf{1}+S_{\lambda})\pi_{\leq l})f_{\ell}=-\lf(\lf+1)f_{\ell},
\end{equation*}
with $\lf=\ell+l$ if $\ell\leq l$ and $\lf=\ell$ if $\ell>l$. We therefore have that $\lf\geq l$ for all $\ell\in \N_0$.

Let $\chi_{\varrho\geq R_{\infty}}$ be a smooth cut-off function that is 1 in $\varrho\geq R_{\infty}$ and 0 in $\varrho\leq \frac{1}{2}R_{\infty}$. Let $\chi_{\varrho\leq \frac{1}{2}R_{\infty}}$ be a smooth cut-off function that is 1 in $\varrho\leq \frac{1}{2}R_{\infty}$ and 0 in $\varrho\geq R_{\infty}$

We now define the main differential operator with respect to which we will derive energy estimates and elliptic estimates:
\begin{multline}
	\label{eq:defLcheck}
	\mathfrak{L}_{s,\lambda,\kappa}:=L_s+\kappa \partial_{x}(\chi_{\varrho\geq R_{\infty}} x \partial_{x} (\cdot))-M\chi_{\varrho\geq R_{\infty}} {B}^{\infty}_s\\
	-(\lambda\pi_{\leq l}+\chi_{\varrho\geq R_{\infty}}S_{\lambda} \pi_{\leq l}+{K}^{\infty}_s+\chi_{\varrho\leq \frac{1}{2}R_{\infty}}{K}_{s}^+).
	\end{multline}
In particular, $L_s$ takes the form \eqref{eq:Lsaspert} for $\kappa=0$, if we define:
\begin{align*}
B_s:=&\:M\epsilon_0^{-1} x\chi_{\varrho\geq R_{\infty}} {B}^{\infty}_s,\\
K_{s,\lambda}:=&\: \lambda\pi_{\leq l}+\chi_{\varrho\geq R_{\infty}}S_{\lambda} \pi_{\leq l}+{K}^{\infty}_s+\chi_{\varrho\leq \frac{1}{2}R_{\infty}}{K}_{s}^+.
\end{align*}

Note that $\mathfrak{L}_{s,\lambda}$ takes on a rather simple form when $\varrho\in [r_+,\infty)\setminus (\frac{R_{\infty}}{2},R_{\infty})$:
\begin{align}
\label{eq:operatorbounded}
	r^{-2}\rho^{2}\mathfrak{L}_{s,\lambda,\kappa}=&\:r\partial_r(\Delta \partial_r(r^{-1}\cdot))+2ar\partial_r\partial_{\varphi_*}(r^{-1}\cdot)+\slashed{\Delta}_{\s^2}(\cdot)\\ \nonumber
	&\:+sr(4 Mr-2(1-\mathbbm{h})\Delta)\partial_r(r^{-1}\cdot)-\lambda\pi_{\leq l}\quad \textnormal{for}\quad  r=\varrho\leq \frac{R_{\infty}}{2},\\
	\label{eq:operatorfaraway}
	(\mathfrak{L}_{s,\lambda,\kappa}f)_{\ell}=&\:\partial_{x}(x(\kappa+x)\partial_{x}f_{\ell})+2 s\partial_{x}f_{\ell}-\mathfrak{l}(\mathfrak{l}+1)f_{\ell}\quad \textnormal{for}\quad \frac{1}{x}=\varrho\geq R_{\infty}.
	\end{align}
	In general, we can express:
	\begin{multline}
			\label{eq:operatorgeneral}
	(\mathfrak{L}_{s,\lambda,\kappa}f)=\partial_x(x(x+\kappa \chi_{\varrho\geq R_{\infty}})\partial_x f)+2s\partial_x f+\slashed{\Delta}_{\s^2}f-\lambda \pi_{\leq l}f-\chi_{\varrho\geq R_{\infty}}S_{\lambda}\pi_{\leq l}f+\chi_{\varrho<R_{\infty}}({K}_s^{\infty}+M x{B}_s^{\infty})\\ 
	-\chi_{\varrho\leq \frac{1}{2}R_{\infty}}{K}_s^{+}\quad \textnormal{for all}\quad \varrho\geq r_+.
	\end{multline}
Finally, we define the $\kappa$-regularized operator $L_{s,\kappa}$ as follows:
\begin{equation}
\label{eq:defLskappa}
L_{s,\kappa}:=L_s+\kappa \partial_x(x \chi_{\varrho\geq R_{\infty}})\partial_x f).
\end{equation}
\section{Preliminaries: Hilbert spaces}
\label{sec:hilbertspaces}
In this section, we will define all the Hilbert spaces that are relevant for the present paper.

Let $S\subseteq \Sigma$ be closed. We define $H^k(S)$ as the Hilbert space obtained from the completion of $C^{\infty}(S)$ with respect to the Sobolev inner product:
\begin{equation*}
\la f,g\ra_{H^k(S)}=\sum_{k_1+k_1\leq k}\int_{S}\snabla_{\s^2}^{k_1}\partial_{\varrho}^{k_2}f\cdot  \snabla_{\s^2}^{k_1}\partial_{\varrho}^{k_2}\overline{g}\,\varrho^2d\upsigma d\varrho.
\end{equation*}
Here, $d\upsigma$ denotes the following standard volume form on the unit round sphere: $d\upsigma:=\sin \vartheta d\vartheta d\varphi_*$.

We denote moreover 
\begin{equation*}
	\la f,g\ra_{H^k_{\snabla}(S)}=\la f,g\ra_{H^k(S)}+\la \snabla_{\s^2}f,\snabla_{\s^2}g\ra_{H^k(S)},
\end{equation*}
and we let $H^k_{\snabla}(S)$ denote the completion of $C^{\infty}(S)$ with respect to this inner product.

Consider $\widehat{\Sigma}\cong [0,r_+^{-1}]_x\times \s^2_{\vartheta,\varphi_*}$. Let $\sigma\in \R$ and $\delta>0$. We introduce the following bilinear forms on $C^{\infty}([0,x_0]_x\times \s^2_{\vartheta,\varphi_*})$ where $x_0>0$, $d_0>0$ and $\sigma\in \R$: 
\begin{align*}
	\la f,g\ra_{B_{R_{\infty}}}:=&\:\sum_{\ell\in \N_0}\sum_{n=0}^{\infty}\frac{(d_0x_0)^{2n}}{(\max\{\ell+1,n+1\})^{2n}}\la \partial_x^{n}f_{\ell},\partial_x^{n}g_{\ell}\ra_{L^2(\s^2)}(x_0) ,\\
		\la f,g\ra_{B_{R_{\infty},2}}:=&\:\sum_{\ell\in \N_0}\sum_{n=0}^{\infty}\frac{(d_0x_0)^{2n}}{(\max\{\ell+1,n+1\})^{2n}} \Bigg[(\ell+1)^4\la \partial_x^{n}f_{\ell},\partial_x^{n}g_{\ell}\ra_{L^2(\s^2)}(x_0)\\
		&+(\ell+1)^2\la \partial_x^{n+1}f_{\ell},\partial_x^{n+1}g_{\ell}\ra_{L^2(\s^2)}(x_0)+\la \partial_x^{n+2}f_{\ell},\partial_x^{n+2}g_{\ell}\ra_{L^2(\s^2)}(x_0)\Bigg],\\
	\la f,g\ra_{G_{\sigma,R_{\infty}}}:=&\:\sum_{\ell\in \N_0}\sum_{n=0}^{\max\{\ell-1,0\}}\frac{(4\sigma)^{2n}(\ell-(n+1))!^2}{(\ell+n+1)!^2} \int_0^{x_0}\la \partial_x^{n}f_{\ell}, \partial_x^{n}g_{\ell}\ra_{L^2(\s^2)}(x)\, dx\\
	&+(\ell+1)^3\sum_{n=\ell}^{\infty}\frac{\sigma^{2n}}{n!^2(n+1)!^2} \int_0^{x_0}\la \partial_x^{n}f_{\ell}, \partial_x^{n}g_{\ell}\ra_{L^2(\s^2)}(x)\, dx\\
	\la f,g\ra_{G_{\sigma,R_{\infty},2}}:=&\:\sum_{\ell\in \N_0}\sum_{n=0}^{\max\{\ell-1,0\}}\frac{(4\sigma)^{2n}(\ell-(n+1))!^2}{(\ell+n+1)!^2}\\
	&\:\times \int_0^{x_0}(\ell-n)^2(\ell+n+1)^2\la \partial_x^{n}f_{\ell}, \partial_x^{n}g_{\ell}\ra_{L^2(\s^2)}(x)\\
	&+(1+x^2(\ell+1)^2)\la \partial_x^{n+1}f_{\ell}, \partial_x^{n+1}g_{\ell}\ra_{L^2(\s^2)}(x)+x^4\la \partial_x^{n+2}f_{\ell}, \partial_x^{n+2}g_{\ell}\ra_{L^2(\s^2)}(x)\,dx\\
	&+(\ell+1)^3\sum_{n=\ell}^{\infty}\frac{\sigma^{2n}}{n!^2(n+1)!^2} \int_0^{x_0} (n+1)^4\la \partial_x^{n}f_{\ell}, \partial_x^{n}g_{\ell}\ra_{L^2(\s^2)}(x)\\
	&+(1+(n+1)^2x^2)\la \partial_x^{n+1}f_{\ell}, \partial_x^{n+1}g_{\ell}\ra_{L^2(\s^2)}(x)+x^4\la \partial_x^{n+2}f_{\ell}, \partial_x^{n+2}g_{\ell}\ra_{L^2(\s^2)}(x)\, dx.
\end{align*}
We denote moreover
\begin{align*}
\|f\|_{B_{R_{\infty}}}:=&\:\sqrt{\la f,f\ra_{B_{R_{\infty}}}}\quad \textnormal{and}\quad \|f\|_{B_{R_{\infty},2}}:=\sqrt{\la f,f\ra_{B_{R_{\infty},2}}}\\
\|f\|_{G_{\sigma, R_{\infty}}}:=&\:\sqrt{\la f,f\ra_{G_{\sigma, R_{\infty}}}}\quad \textnormal{and}\quad\|f\|_{G_{\sigma,R_{\infty},2}}:=\sqrt{\la f,f\ra_{G_{\sigma,R_{\infty},2}}}.
\end{align*}
Let $w_p: [0,r_+^{-1}]\to \R$ be a smooth cut-off function that is defined as follows: let $p\in [0,\infty)$, then
\begin{equation}
\label{eq:defw}
w_p(x)=\begin{cases}
 p &x\leq  \frac{1 }{ 4 R_{\infty}},\\
0  & x\geq \frac{1}{3R_{\infty}}.
\end{cases}
\end{equation}

Let $N_+\in \N_0$ and $\sigma\in \R$. We fix $p=4(N_++1)+2$. We define the following inner products on $C^{\infty}(\widehat{\Sigma})$: let $x_0=\frac{1}{2R_{\infty}}$, then
\begin{align*}
	\la f,g\ra_{H^{N_+}_{\sigma,R_{\infty}}}:=&\:\la f,g\ra_{H^{N_+}(\Sigma\cap\{\varrho \leq 3R_{\infty}\})}+\la \snabla_{\s^2}f,\snabla_{\s^2}g\ra_{H^{N_+}(\Sigma\cap\{\varrho\leq \frac{1}{3}R_{\infty}\})}\\
&+\la (\mathbf{1}-\slashed{\Delta}_{\s^2})^{-\frac{w_p}{4}}f,(\mathbf{1}-\slashed{\Delta}_{\s^2})^{-\frac{w_p}{4}}g\ra_{H^{N_+}(\Sigma\cap\{3R_{\infty}\leq \varrho\leq 4R_{\infty}\})}+\la f,g\ra_{B_{R_{\infty}}}+\la f,g\ra_{G_{\sigma,R_{\infty}}},\\
\la f,g\ra_{H^{N_++1}_{\sigma, R_{\infty},2}}:=&\:\sum_{k=0}^1\la \snabla_{\s^2}^{k }f,\snabla_{\s^2}^kg\ra_{H^{N_++1}(\Sigma\cap \{\varrho\leq \frac{1}{3}R_{\infty}\})}+\la f,g\ra_{H^{N_++2}(\Sigma\cap \{\frac{1}{3}R_{\infty}\leq \varrho\leq 3 R_{\infty}\})}\\
&+\la (\mathbf{1}-\slashed{\Delta}_{\s^2})^{-\frac{w_p}{4}}f,(\mathbf{1}-\slashed{\Delta}_{\s^2})^{-\frac{w_p}{4}}g\ra_{H^{N_++2}(\Sigma\cap\{3R_{\infty}\leq \varrho\leq 4R_{\infty}\})}+\la f,g\ra_{B_{R_{\infty},2}}+\la f,g\ra_{G_{\sigma,R_{\infty},2}}.
\end{align*}
We define the Hilbert spaces $H^{N_+}_{\sigma,R_{\infty}}$ and $H^{N_++1}_{\sigma,R_{\infty},2}$ as the completions of $C^{\infty}(\widehat{\Sigma})$ with respect to the inner products $\la f,g\ra_{H^{N_+}_{\sigma,R_{\infty}}}$ and $\la f,g\ra_{H^{N_++1}_{\sigma,R_{\infty},2}}$, respectively, and we denote the corresponding norms by $\|\cdot\|_{H^{N_+}_{\sigma,R_{\infty}}}$ and $\|\cdot\|_{H^{N_++1}_{\sigma,R_{\infty},2}}$, respectively.

We also introduce the subspace $H_{\sigma,R_{\infty},1}^{N_++1}\subset  H_{\sigma,R_{\infty}}^{N_+}$, which is defined as follows:
\begin{equation*}
H_{\sigma,R_{\infty},1}^{N_++1}:=\{f\in H_{\sigma,R_{\infty}}^{N_+}\,|\, \partial_xf,\snabla_{\s^2}f\in  H_{\sigma,R_{\infty}}^{N_+}\},
\end{equation*}
and we denote $\|f\|_{H_{\sigma, R_{\infty},1}^{N_++1}}^2:=\sum_{1\leq k_1+k_2\leq 1}\|\snabla_{\s^2}^{k_1}\partial_x^{k_2}f\|_{H_{\sigma, R_{\infty}}^{N_+}}$.

Let $P: H^1(\widehat{\Sigma})\to L^2(\widehat{\Sigma})$ be the following linear operator:
\begin{equation*}
Pf= \frac{1}{2M^2h(0)-b_{\tau \tau}}\left( (2+M x b_{x\tau})\partial_x+Ma^2 \sin \vartheta x^2 b_{\tau \vartheta}\partial_{\vartheta}+M xb_{\tau}\mathbf{1}\right).
\end{equation*}
The role of the operator $P$ is explained in Theorem \ref{prop:relALs}. We define the following space:
\begin{equation*}
\mathbf{H}_{\sigma,R_{\infty}}^{N_+}= {H}_{\sigma,R_{\infty}}^{N_+}(P)\oplus {H}_{\sigma,R_{\infty}}^{N_+}.
\end{equation*}
where ${H}_{\sigma,R_{\infty}}^{N_+}(P)=\{f\in {H}_{\sigma,R_{\infty}}^{N_+}\,|\, Pf\in {H}_{\sigma,R_{\infty}}^{N_+}\}$ and we equip it with the inner product:
\begin{equation*}
\la (f,g),(\tilde{f},\tilde{g})\ra_{\mathbf{H}_{\sigma,R_{\infty}}^{N_+}}=\la f,\tilde{f}\ra_{H_{\sigma,R_{\infty}}^{N_+}}+\la Pf, P\tilde{f}\ra_{H_{\sigma,R_{\infty}}^{N_+}}+\la g,\tilde{g}\ra_{H_{\sigma,R_{\infty}}^{N_+}}
\end{equation*}
to obtain a Hilbert space of pairs of functions on $\widehat{\Sigma}$. We note the corresponding norm by $\|\cdot\|_{\mathbf{H}_{\sigma,R_{\infty}}^{N_+}}$.
\section{Preliminaries: the infinitesimal generator of time translations}
\label{sec:infigen}
We define the time evolution operators $\mathcal{S}(\tau)$ associated to the wave equation $\square_{g_{M,a}}\phi=0$ and the foliation $\{\widehat{\Sigma}_{\tau}\}_{\tau\in \R}$ as follows: let $N\geq 0$ and $\tau\in [0,\infty)$, then

\begin{align*}
\mathcal{S}(\tau):H^{N+1}(\widehat{\Sigma})\times H^{N}(\widehat{\Sigma}) &\: \to  H^{N+1}(\widehat{\Sigma})\times H^{N}(\widehat{\Sigma}),\\
(\Psi,\Psi')&\: \mapsto ( \psi(\tau,\cdot), \partial_{\tau}\psi(\tau,\cdot)),
\end{align*}
where $\psi=\varrho\cdot  \phi$, with $\phi: \mathcal{R}\to \C$ the solution to $\square_{g_{M,a}}\phi=0$ corresponding to initial data $(\psi|_{\Sigma},\partial_{\tau}\psi|_{\Sigma})=(\Psi,\Psi')$. The map $\mathcal{S}(\tau)$ is well-defined, by local-in-time energy estimates on $\widehat{\Sigma}_{\tau}$; see for example \cite{aagkerr}[Proposition 3.4]. The results in this section also hold with very minimal modifications on Kerr--de Sitter spacetimes for the equation $(\square_{g_{M,a,\Lambda}}-\frac{2}{3}\Lambda)\phi=0$.

From standard, local-in-time energy estimates, it follows moreover that $\{\mathcal{S}(\tau)\}_{\tau\in \R}$ form a $C_0$ (strongly continuous) semigroup, i.e.\
\begin{align*}
\mathcal{S}(0)=&\:{\rm id},\\
\mathcal{S}(\tau_1)\mathcal{S}(\tau_2)=&\:\mathcal{S}_{\tau_1+\tau_2}\quad \textnormal{for all $\tau_1,\tau_2\in \R$,}\\
\lim_{\tau\to 0}\|\mathcal{S}(\tau)(\Psi,\Psi')-(\Psi,\Psi')\|_{H^{N+1}(\widehat{\Sigma})\times H^{N}(\widehat{\Sigma})}=&\: 0\quad \textnormal{for all $(\Psi,\Psi')\in H^{N+1}(\widehat{\Sigma})\times H^{N}(\widehat{\Sigma}) $}.
\end{align*}

We define the \emph{infinitesimal generator} corresponding to the semigroup $\{\mathcal{S}(\tau)\}_{\tau\in \R}$ as follows:
\begin{align*}
\mathcal{A}:&\: C^{\infty}(\widehat{\Sigma})\times C^{\infty}(\widehat{\Sigma})\to H^{N_++1}(\widehat{\Sigma})\times H^{N_+}(\widehat{\Sigma}),\\
\mathcal{A}(\Psi,\Psi')=&\:\lim_{\tau\to 0} \frac{1}{\tau}\left(\mathcal{S}(\tau)(\Psi,\Psi')-(\Psi,\Psi')\right) \textnormal{for all $(\Psi,\Psi')\in C^{\infty}(\widehat{\Sigma})\times C^{\infty}(\widehat{\Sigma}) $,}
\end{align*}
where the limit is taken with respect to the $H^{N_++1}(\widehat{\Sigma})\times H^{N_+}(\widehat{\Sigma})$ norm.

Note that for $(\Psi,\Psi')\in C^{\infty}(\widehat{\Sigma})\times C^{\infty}(\widehat{\Sigma})$, the function $\psi=\varrho\phi$ with $\phi: D^+(\Sigma)\to \C$ the solution to $\square_{g_{M,a}}\phi=0$ corresponding to initial data $(\psi|_{\Sigma},\partial_{\tau}\psi|_{\Sigma})=(\Psi,\Psi')$ satisfies
\begin{equation*}
\mathcal{A}(\Psi,\Psi')=(\partial_{\tau}\psi|_{\Sigma},\partial_{\tau}^{2}\psi|_{\Sigma})\in C^{\infty}(\widehat{\Sigma})\times C^{\infty}(\widehat{\Sigma})=:( C^{\infty}(\widehat{\Sigma}))^2.
\end{equation*}

In the lemma below, we will relate $\mathcal{A}$ to the operator $L_s$ introduced in \S \ref{sec:laplacetransfop}. For this, we introduce the following domain for $L_s$: 
\begin{equation*}
 {\rm dom}\, (L_s):=\{f\in C^{\infty}(\widehat{\Sigma})\cap H^{N_+}_{\sigma,R_{\infty}}\,|\, L_sf\in H^{N_+}_{\sigma,R_{\infty}}\}.
\end{equation*}
We define the Hilbert space $D^{N_+}_{\sigma, R_{\infty}}(L_s)$ as the closure of ${\rm dom}\, (L_s)$ with respect to the graph norm $\|\cdot\|_{H^{N_+}_{\sigma,R_{\infty}}}+\|L_s(\cdot)\|_{H^{N_+}_{\sigma,R_{\infty}}}$. By construction, $L_s$ can be extended as a bounded linear operator to obtain:
\begin{equation*}
L_s: D^{N_+}_{\sigma, R_{\infty}}(L_s)\to H^{N_+}_{\sigma,R_{\infty}}.
\end{equation*}

\begin{proposition}
\label{prop:relALs}
The infinitesimal generator of time translations $\mathcal{A}$ (corresponding to Kerr spacetimes) satisfies the following properties:
\begin{enumerate}[label=\emph{(\roman*)}]
\item 
Let $s\in \C$. Then we can express:
\begin{equation}
\label{eq:relAandLs}
\mathcal{A}-s\mathbf{1}=\begin{pmatrix}
0 & \mathbf{1}\\
\mathbf{1} & P-s\mathbf{1}
\end{pmatrix}\begin{pmatrix}
\frac{1}{2M^2h(0)-Mxb_{\tau \tau}} L_s& 0\\
0&\mathbf{1}
\end{pmatrix}\begin{pmatrix}
\mathbf{1}& 0\\
-s \mathbf{1} & \mathbf{1}
\end{pmatrix}.
\end{equation}
\item Assume that $D^{N_+}_{\sigma, R_{\infty}}(L_s)\subseteq H_{\sigma, R_{\infty},2}^{N_++1}$ and that the inverse of $L_s$, $L_{s}^{-1}: H_{\sigma,R_{\infty}}^{N_+}\to D^{N_+}_{\sigma, R_{\infty}}(L_s)\subseteq H_{\sigma, R_{\infty},2}^{N_++1}$, is a bounded linear operator with respect to the operator norm $\|\cdot\|_{H_{\sigma,R_{\infty}}^{N_+}\to H_{\sigma, R_{\infty},2}^{N_++1}}$. Then the operator
\begin{align*}
\mathcal{R}(s): \mathbf{H}_{\sigma,R_{\infty}}^{N_+} \to &\:\mathbf{H}_{\sigma,R_{\infty}}^{N_+},\\
\mathcal{R}(s)=&\:\begin{pmatrix}
\mathbf{1} & 0\\
s \mathbf{1}& \mathbf{1}
\end{pmatrix}\begin{pmatrix}
 L_s^{-1}\circ (2M^2h(0)-M xb_{\tau \tau})\mathbf{1}& 0\\
0&\mathbf{1}
\end{pmatrix}\begin{pmatrix}
-(P-s \mathbf{1})& \mathbf{1}\\
\mathbf{1} & 0
\end{pmatrix}
\end{align*}
is a bounded linear operator on $\mathbf{H}_{\sigma,R_{\infty}}^{N_+}$, which satisfies $(\mathcal{A}-s\mathbf{1})\mathcal{R}(s)(f,g)=(f,g)$ for all $(f,g)\in  \mathbf{H}_{\sigma,R_{\infty}}^{N_+}\cap ( C^{\infty}(\widehat{\Sigma}))^2$ and $\mathcal{R}(s)(\mathcal{A}-s\mathbf{1})(\Psi,\Psi')=(\Psi,\Psi')$ for all $(\Psi,\Psi')\in  \textnormal{ran}\, \mathcal{R}(s)\cap ( C^{\infty}(\widehat{\Sigma}))^2$.

\item Assume that there exists a $s_0\in \C$, such that $H^{N_+}_{\sigma, R_{\infty}}(L_{s_0})\subseteq H_{\sigma, R_{\infty},2}^{N_++1}$ and $L_{s_0}^{-1}: H_{\sigma,R_{\infty}}^{N_+}\to H^{N_+}_{\sigma, \delta}(L_{s_0})\subseteq H_{\sigma, R_{\infty},2}^{N_++1}$ is a bounded linear operator with respect to the operator norm $$\|\cdot\|_{H_{\sigma,R_{\infty}}^{N_+}\to H_{\sigma, R_{\infty},2}^{N_++1}}.$$ Then the operator $\mathcal{A}$ can be extended as a closed linear operator on the following domain:
\begin{equation*}
\mathcal{A}:  \mathbf{H}_{\sigma,R_{\infty}}^{N_+}\supseteq \textnormal{ran}\, \mathcal{R}(s_0)\to \mathbf{H}_{\sigma,R_{\infty}}^{N_+}
\end{equation*}
and
\begin{equation*}
D^{N_+}_{\sigma,R_{\infty}}(L_s)\oplus {H}_{\sigma,R_{\infty}}^{N_+}(P)=\textnormal{ran}\, \mathcal{R}(s)=\textnormal{ran}\, \mathcal{R}(s_0)
\end{equation*}
for all $s\in \C$ such that $L_{s}^{-1}$ is well-defined.
\end{enumerate}
The above results generalize to or Kerr--de Sitter spacetimes with suitably small $M^2\Lambda$ by replacing the factor $\frac{1}{2M^2h(0)-Mxb_{\tau \tau}}$ with an analogous factor, depending on the choice of future-horizon-intersecting, spacelike foliation. 
\end{proposition}
\begin{proof}
We will restrict to Kerr in the proof. The Kerr--de Sitter case follows with minimal modifications. We compare the expression for $\varrho^3\square_{g_{M,a}}\phi$ in \eqref{eq:waveeqxcoord} to the expression for $L_s$ in \eqref{eq:generalexprLs} in order to obtain:
\begin{multline*}
\varrho^3\square_{g_{M,a}}\phi= L_s \psi+2\partial_x\partial_{\tau}\psi+M x b_{x\tau }\cdot \partial_x \partial_{\tau} \psi+(Mx b_{\tau \tau}-2M^2h(0))\partial_{\tau}^2\psi+Ma^2x^2\sin\vartheta  b_{\tau \vartheta}\partial_{\tau}\partial_{\vartheta}\psi+Mx b_{\tau}\partial_{\tau}\psi\\
-2s\partial_x\psi-sM x b_{x\tau }\cdot \partial_x  \psi-s^2(Mx b_{\tau \tau}-2M^2h(0))\psi-Ma^2 s x^2\sin\vartheta  b_{\tau \vartheta}\partial_{\vartheta}\psi-M s x b_{\tau}\psi\\
= L_s \psi+(Mx b_{\tau \tau}-2M^2h(0))\partial_{\tau}^2\psi+(2M^2h(0)-Mx b_{\tau \tau})P\partial_{\tau}\psi -s(2M^2h(0)-Mx b_{\tau \tau})P\psi+s^2(2M^2h(0)-Mx b_{\tau \tau})\psi.
\end{multline*}
Using that $\square_{g_{M,a}}\phi=0$, we can therefore express:
\begin{equation*}
\partial_{\tau}^2\psi=(2M^2h(0)-Mx b_{\tau \tau})^{-1}L_s \psi+P\partial_{\tau}\psi-sP\psi+s^2\mathbf{1}
\end{equation*}
and therefore, $\mathcal{A}$ takes the following form
\begin{equation*}
\mathcal{A}(\Psi,\Psi')=\begin{pmatrix} \Psi'\\
(2M^2h(0)-Mx b_{\tau \tau})^{-1}L_s \Psi+(s^2\mathbf{1}-sP)\Psi+P\Psi'
\end{pmatrix}.
\end{equation*}
The expression for $\mathcal{A}-s\mathbf{1}$ in (i) can then be easily verified.

We consider now (ii). We have that
\begin{equation}
\label{eq:rangeexpr}
\mathcal{R}(s)(f,g)=\begin{pmatrix}
-L_s^{-1}\left((2M^2h(0)-M xb_{\tau \tau})((P-s\mathbf{1})f-g)\right),\\
-sL_s^{-1}\left((2M^2h(0)-M xb_{\tau \tau})((P-s\mathbf{1})f-g)\right)+f
\end{pmatrix}.
\end{equation}
We can use that $L_{s}^{-1}: H_{\sigma,R_{\infty}}^{N_+}\to D^{N_+}_{\sigma, R_{\infty}}(L_s)\subseteq H_{\sigma, R_{\infty},2}^{N_++1}$ is a bounded linear operator and $f,Pf,g\in H_{\sigma,R_{\infty}}^{N_+}$ for $(f,g)\in \mathbf{H}_{\sigma,R_{\infty}}^{N_+}$ to ensure that $\mathcal{R}(s)$ is in fact well-defined and bounded on the domain $ \mathbf{H}_{\sigma,R_{\infty}}^{N_+}$ and $\mathcal{R}(s_0)$.

Note moreover that $(f,g)\in  \mathbf{H}_{\sigma,R_{\infty}}^{N_+}\cap ( C^{\infty}(\widehat{\Sigma}))^2$ implies that $\mathcal{R}(f,g)\in ( C^{\infty}(\widehat{\Sigma}))^2$. The property $(\mathcal{A}-s\mathbf{1})\mathcal{R}(s)=\mathcal{R}(s)(\mathcal{A}-s\mathbf{1})$ can then easily be verified by using the block diagonal form of $\mathcal{A}-s\mathbf{1}$ from (i) to formally invert the operator.

We will now prove (iii). First observe that we can extend the domain of $\mathcal{A}$ to $\textnormal{ran}\, \mathcal{R}(s_0)$ as follows. Let $(\Psi,\Psi')\in \textnormal{ran}\, \mathcal{R}(s_0)$. Then there exists an element $(f,g)\in \mathbf{H}_{\sigma,R_{\infty}}^{N_+}$ such that $(\Psi,\Psi')=\mathcal{R}(s_0)(f,g)$. We then define:
\begin{equation*}
	(\mathcal{A}-s\mathbf{1})(\Psi,\Psi'):=(f,g)
\end{equation*}	
By (ii), this extends the operator $\mathcal{A}$ defined on $\textnormal{ran}\, \mathcal{R}(s_0)\cap ( C^{\infty}(\widehat{\Sigma}))^2$ to $\textnormal{ran}\, \mathcal{R}(s_0)$. Closedness of $\mathcal{A}-s\mathbf{1}$ (and hence $\mathcal{A}$) then follows immediately from the boundedness of its inverse $\mathcal{R}(s_0)$.

We will now prove the remaining statement in (iii). Let $(\Psi,\Psi')\in   \textnormal{ran}\, \mathcal{R}(s)$. Then there exist $(f,g)\in \mathbf{H}_{\sigma,R_{\infty}}^{N_+}$ satisfying \eqref{eq:rangeexpr}.
From (i), we have that formally
\begin{equation*}
(\mathcal{A}-s_0)(\Psi,\Psi')=(\mathcal{A}-s)(\Psi,\Psi')+(s-s_0)\begin{pmatrix}\Psi\\ \Psi'\end{pmatrix}=\begin{pmatrix}f+(s-s_0)\Psi\\
g+(s-s_0)\Psi'
\end{pmatrix}.
\end{equation*}
Motivated by the above, using that $(\Psi,\Psi')=\mathcal{R}(s)(f,g)$, we define
\begin{align*}
\tilde{f}:=&\:f+(s-s_0)\Psi,\\
\tilde{g}:=&\:g+(s-s_0)\Psi'.
\end{align*}
It then follows that $(\tilde{f},\tilde{g})\in  \mathbf{H}_{\sigma,R_{\infty}}^{N_+}$. Furthermore, $\mathcal{R}(s_0)(\tilde{f},\tilde{g})=(\Psi,\Psi')$, so $(\Psi,\Psi')\in  \textnormal{ran}\, \mathcal{R}(s_0)$. We conclude that $\textnormal{ran}\, \mathcal{R}(s)\subseteq \textnormal{ran}\, \mathcal{R}(s_0)$. Reversing the roles of $s_0$ and $s$ in the above argument allows us to conclude that $\textnormal{ran}\, \mathcal{R}(s)= \textnormal{ran}\, \mathcal{R}(s_0)$.
\end{proof}
The operator $ \mathcal{R}(s)$ from Proposition \ref{prop:relALs} satisfies the defining properties of the \emph{resolvent operator} corresponding to $\mathcal{A}: \textnormal{ran}\, \mathcal{R}(s_0)\to \mathbf{H}_{\sigma,R_{\infty}}^{N_+}$, where $L_{s_0}^{-1}$ is well-defined. We can denote
\begin{equation*}
(\mathcal{A}-s\mathbf{1})^{-1}:=\mathcal{R}(s).
\end{equation*}

\begin{corollary}
\label{cor:relLA}
Assume that for $s\in O$, with $O\subset \C$ a connected domain, \underline{either} $L_{s}^{-1}: H_{\sigma,R_{\infty}}^{N_+}\to D^{N_+}_{\sigma,R_{\infty}}(L_{s_0})\subseteq H_{\sigma,R_{\infty},2}^{N_++1}$ is bounded with respect to the operator norm $\|\cdot\|_{H_{\sigma,R_{\infty}}^{N_+}\to H_{\sigma, R_{\infty},2}^{N_++1}}$, \underline{or} $\ker L_s\neq \{0\}$. Assume also that there exists a $s_0\in O$ such that $\ker L_{s_0}=\{0\}$. 

Consider $\mathcal{A}:  \textnormal{ran}\, \mathcal{R}(s_0)\to \mathbf{H}_{\sigma,R_{\infty}}^{N_+}$. Then the spectrum ${\rm Spect}\,(\mathcal{A})$ satisfies:
\begin{equation*}
{\rm Spect}\,(\mathcal{A})\cap O={\rm Spect}_{\rm point}\,(\mathcal{A})\cap O=\{s\in O \,|\, \ker L_s\neq \{0\}\}.
\end{equation*}
Furthermore, $(\Psi,\Psi')\in \textnormal{ran}\, \mathcal{R}(s_0)$ is an eigenvector of $\mathcal{A}$ if and only if $(\Psi,\Psi')=(\hpsi_s,s\hpsi_s)$, with $\hpsi_s\in \ker L_s$.
\end{corollary}
\begin{proof}
In view of the form of $\mathcal{R}(s)$ from (ii) of Proposition \ref{prop:relALs}, boundedness of $\mathcal{R}(s)$ is equivalent to boundedness of $L_s^{-1}$. Hence, $s\in {\rm Spect}\,(\mathcal{A})\cap O$ if and only if $L_s^{-1}$ is not bounded. By assumption, the latter implies that $\ker L_s\neq \{0\}$. Let $\hpsi_s\in \ker L_s$. Then
\begin{equation*}
\mathcal{A}(\hpsi_s,s\hpsi_s)=s\cdot (\hpsi_s,s\hpsi_s)\in \mathbf{H}_{\sigma,R_{\infty}}^{N_+}.
\end{equation*}
and $(f,g)=(\mathcal{A}-s_0\mathbf{1})(\hpsi_s,s\hpsi_s)=(s-s_0)(\hpsi_s, s\hpsi_s)\in \mathbf{H}_{\sigma,R_{\infty}}^{N_+}$, so $\mathcal{R}(s_0)(f,g)=\begin{pmatrix}\hpsi_s\\s\hpsi_s\end{pmatrix}$ and therefore $(\hpsi,s\hpsi_s)\in \textnormal{ran}\, \mathcal{R}(s_0)$, so $s\in {\rm Spect}_{\rm point}\,(\mathcal{A})\cap O$, with $\mathcal{A}$ defined on the domain $\textnormal{ran}\, \mathcal{R}(s_0)$.
\end{proof}

The \emph{Hermitian adjoint} of $L_s:D^{N_+}_{\sigma, R_{\infty}}(L_s)\supset {\rm dom}\,(L_s)\to H_{\sigma,R_{\infty}}^{N_+}$, with respect to the $H_{\sigma,R_{\infty}}^{N_+}$ inner product is defined as follows: let
\begin{equation*}
D_{\sigma,R_{\infty}}^{N_+}(L_s^*):=\left\{f\in H_{\sigma,R_{\infty}}^{N_+}\,,\, \la f, L_s(\cdot)\ra_{H_{\sigma,R_{\infty}}^{N_+}}: D_{\sigma,R_{\infty}}^{N_+}(L_s)\supset {\rm dom}\,(L_s)\to  \C \: \textnormal{is bounded}\right\}.
\end{equation*}
By the Riesz representation theorem, we have that for all $f\in D_{\sigma,R_{\infty}}^{N_+}(L_s^*)$, there exists a unique $g \in D^{N_+}_{\sigma, R_{\infty}}(L_s)$, such that $\la g, \hpsi\ra_{H_{\sigma,R_{\infty}}^{N_+}}=\la f, L_s \hpsi\ra_{H_{\sigma,R_{\infty}}^{N_+}}$ for all $\hpsi\in {\rm dom}\,(L_s)$. Then we define the Hermitian adjoint of $L_s$ by $L_s^*: D_{\sigma,R_{\infty}}^{N_+}(L_s^*)\to  H_{\sigma,R_{\infty}}^{N_+}$, with $L_s^*f=g$, with $f,g$ as above.

Similarly, we define the Hermitian adjoint of $\mathcal{A}$, with respect to the $\mathbf{H}_{\sigma,R_{\infty}}^{N_+}$ inner product, as follows: let $s_0\in \C$ be such that $L_{s_0}^{-1}$ is well-defined, then
\begin{equation*}
\mathbf{D}_{\sigma,R_{\infty}}^{N_+}(\mathcal{A}^*):=\left\{(f,g)\in \mathbf{H}_{\sigma,R_{\infty}}^{N_+}\,,\, \la f,\mathcal{A}(\cdot)\ra_{\mathbf{H}_{\sigma,R_{\infty}}^{N_+}}:{\rm ran}\,\mathcal{R}(s_0)\to \C\: \textnormal{is bounded}\right\}.
\end{equation*}
By the Riesz representation theorem, we have that for all $(f,g)\in \mathbf{D}_{\sigma,R_{\infty}}^{N_+}(\mathcal{A}^*)$, there exists a unique $(\Phi,\Phi') \in {\rm ran}\,\mathcal{R}(s_0)$, such that $\la (\Phi,\Phi'),(\Psi,\Psi')\ra_{\mathbf{H}_{\sigma,R_{\infty}}^{N_+}}=\la (f,g), \mathcal{A}(\Psi,\Psi')\ra_{\mathbf{H}_{\sigma,R_{\infty}}^{N_+}}$ for all $(\Psi,\Psi')\in  (C^{\infty}(\hat{\Sigma}))^2\cap \mathcal{R}(s_0)$. Then we define the Hermitian adjoint of $\mathcal{A}$ by $\mathcal{A}^*: \mathbf{D}_{\sigma,R_{\infty}}^{N_+}(\mathcal{A}^*)\to \mathbf{H}_{\sigma,R_{\infty}}^{N_+}$, with $\mathcal{A}^*(f,g)=(\Phi,\Phi')$, with $(f,g)$ and $(\Phi,\Phi')$ as above.

Taking the formal adjoint of the right-hand side of \eqref{eq:relAandLs}, we obtain
\begin{equation*}
\mathcal{A}^{\dag}-\overline{s}\mathbf{1}:=\begin{pmatrix}
\mathbf{1}& -\overline{s} \mathbf{1}\\
0 & \mathbf{1}
\end{pmatrix}\begin{pmatrix}
\frac{1}{2M^2h(0)-Mxb_{\tau \tau}} L_s^*& 0\\
0&\mathbf{1}
\end{pmatrix}\begin{pmatrix}
0 & \mathbf{1}\\
\mathbf{1} & P^*-\overline{s}\mathbf{1}
\end{pmatrix}.
\end{equation*}
From the above expression, it follows that $\mathcal{A}^{\dag}(f,g)$ is well-defined and agrees with $\mathcal{A}^{*}(f,g)$ if $g\in D_{\sigma,R_{\infty}}^{N_+}(L_s^*)\cap D_{\sigma,R_{\infty}}^{N_+}(P^*)$ and $f\in H_{\sigma,R_{\infty}}^{N_+}$, with $D_{\sigma,R_{\infty}}^{N_+}(P^*)$ the domain of the Hermitian adjoint $P^*$, with respect to the inner product on $H_{\sigma,R_{\infty}}^{N_+}$.

A direct corollary of the standard orthogonality property between the kernel of a linear operator and the range of its Hermitian adjoint results in the following:
\begin{corollary}
\label{cor:orthogonality}
 Assume additionally that there exists a $s_0\in \C$ such that $\ker L_{s_0}=\{0\}$. We have that
\begin{align*}
\textnormal{ran}(\mathcal{A}^*-\overline{s}\mathbf{1})\perp&\: \ker (\mathcal{A}-s\mathbf{1}),\\
\ker(\mathcal{A}^*-\overline{\tilde{s}}\mathbf{1})\perp&\: \ker (\mathcal{A}-s\mathbf{1}) \quad \textnormal{if $s\neq \tilde{s}$}.
\end{align*}
Let $\hpsi_s\in \ker L_s$ and denote $(\Psi_s,\Psi'_s)=(\hpsi_s,s\hpsi_s)\in \ker (\mathcal{A}-s\mathbf{1})$. Let $g_{\tilde{s}}\in \ker L_{\tilde{s}}^*$ and denote $(\tilde{\Psi}_{\tilde{s}},\tilde{\Psi}'_{\tilde{s}})=(-(P^*-\overline{\tilde{s}})g_{\tilde{s}},g_{\tilde{s}})$. Then $(\tilde{\Psi}_{\tilde{s}},\tilde{\Psi}'_{\tilde{s}})\in \ker (\mathcal{A}^{\dag}-\overline{\tilde{s}}\mathbf{1})$ and
\begin{equation*}
\la (\Psi_s,\Psi'_s), (\tilde{\Psi}_{\tilde{s}},\tilde{\Psi}'_{\tilde{s}})\ra_{\mathbf{H}^{N_+}_{\sigma,R_{\infty}}}=0
\end{equation*}
if $s\neq \tilde{s}$.
\end{corollary}

\section{Precise statements of the main theorems}
\label{sec:mainthmprecise}
In this section, we give precise statements of the main theorems of the paper, making use of the preliminaries introduced in \S\S \ref{sec:geometry}--\ref{sec:infigen}, and we give precise definitions of the notions of \emph{regularity quasinormal modes} and \emph{regularity quasinormal frequencies}.
\subsection{Regularity quasinormal frequencies}
The following theorem provides a precise formulation of Theorem \ref{thm:rough1}(i)--(ii) and includes also an additional statement on orthogonality properties of regularity quasinormal modes.
\begin{theorem}
\label{thm:mainthmA}
Let $\sigma\in \R$ and $N_+\in \N_0$. Let 
\begin{align*}
	\Omega_{\sigma}:=&\:\left\{x+iy\in \C\,|\, x\in \R_-,\,y\in \R,\, x^2+y^2<\sigma^2,\,3y^2-5x^2>\sigma^2\right\}\\
	&\cup \left\{	x+iy\in \C\,|\, x\in[0,\infty),\, \frac{\sigma^2}{4}<x^2+y^2<\sigma^2\right\},\\
	\Omega_{\sigma, N_+}:=&\: \Omega_{\sigma}\cap \left\{\re z> -\left(\frac{1}{2}+N_+\right)\kappa_+\right \}
	\end{align*}
	\begin{enumerate}[label=\emph{(\roman*)}]
	\item
Then there exists a suitably large $R_{\infty}>r_+$ such that the following restrictions of the infinitesimal generator of time translations corresponding to $\mathcal{S}(\tau)$ are closed linear operators:
\begin{equation*}
\mathcal{A}: \mathbf{H}^{N_+}_{\sigma,R_{\infty}}\supset H^{N_+}_{\sigma,R_{\infty}}(L_s)\oplus {H}_{\sigma,R_{\infty}}^{N_+}(P)\to \mathbf{H}^{N_+}_{\sigma,R_{\infty}}
\end{equation*}
and 
\begin{equation*}
{\rm Spect}\,(\mathcal{A})\cap \Omega_{\sigma, N_+}={\rm Spect}_{\rm point}\,(\mathcal{A})\cap \Omega_{\sigma, N_+}
\end{equation*}
consists of isolated eigenvalues, which are independent of the precise choice of $\tau$. We define \textbf{regularity quasinormal frequencies} to be the union of eigenvalues:
\begin{equation*}
\mathscr{Q}_{\rm reg}:=\bigcup_{N_+\in \N_0, \sigma\in \R}{\rm Spect}_{\rm point}\,(\mathcal{A})\cap \Omega_{\sigma,N_+}
\end{equation*}
and \textbf{regularity quasinormal modes} to be the corresponding eigenvectors. The sets ${\rm Spect}_{\rm point}\,(\mathcal{A})\cap \Omega_{\sigma,N_+}$ are independent of the precise choice of foliation-defining function $\mathbbm{h}$.
\item The union $\Omega:=\bigcup_{N_+\in \N_0, \sigma\in \R}\Omega_{\sigma, N_+}$ satisfies
\begin{equation*}
\Omega=\left\{z\in \C\setminus\{0\}\,,\, |{\rm arg\, z}|< \frac{2}{3}\pi\right\}.
\end{equation*}
The set of quasinormal frequencies $\mathscr{Q}_{\rm reg}$ has no accumulation points in $\Omega$ and is independent of the choice of $R_{\infty}$.
\item The eigenspaces of $\mathcal{A}: \mathbf{H}^{N_+}_{\sigma,R_{\infty}}\supset H^{N_+}_{\sigma,R_{\infty}}(L_s)\oplus {H}_{\sigma,R_{\infty}}^{N_+}(P)\to  \mathbf{H}^{N_+}_{\sigma,R_{\infty}}$ are finite dimensional.
\item We define \textbf{regularity quasinormal co-modes} to be eigenvectors of $\mathcal{A}^{\dag}$ in $((D_{\sigma,R_{\infty}}^{N_+})(L_s^*)\cap D_{\sigma,R_{\infty}}^{N_+}(P^*))\oplus H^{N_+}_{\sigma,R_{\infty}}$. The following orthogonality relations hold with respect to the inner product on $\mathbf{H}^{N_+}_{\sigma,R_{\infty}}$:
\begin{align*}
\textnormal{ran}(\mathcal{A}^{\dag}-\overline{s}\mathbf{1})\perp&\: \ker (\mathcal{A}-s\mathbf{1}),\\
\ker(\mathcal{A}^{\dag}-\overline{\tilde{s}}\mathbf{1})\perp&\: \ker (\mathcal{A}-s\mathbf{1}) \quad \textnormal{if $s\neq \tilde{s}$}.
\end{align*}
\end{enumerate}
\end{theorem}
\begin{proof}
Theorem \ref{thm:mainthmA} follows directly from Theorem \ref{thm:mainthmLs} below, in combination with Corollaries \ref{cor:relLA} and \ref{cor:orthogonality}.
\end{proof}

\begin{figure}[H]
	\begin{center}
\includegraphics[scale=0.64]{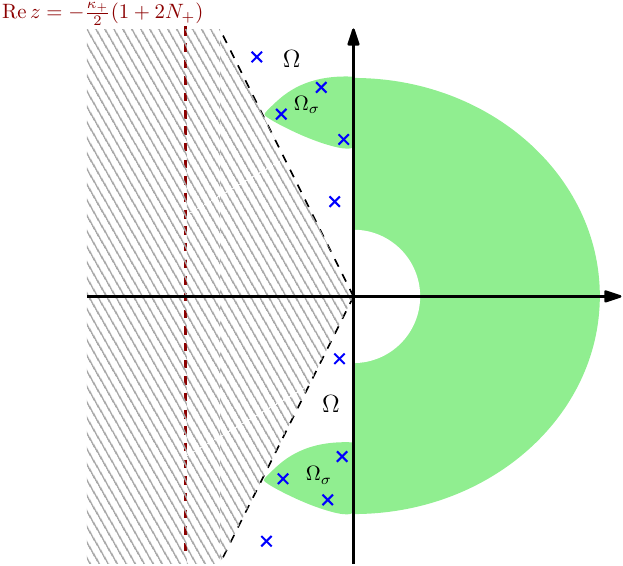}
\end{center}
\caption{A (schematic) pictorial representation of the following subsets of $\C$: $\Omega$ (the sector of the complex plane that is not dashed), $\Omega_{\sigma}$ (the shaded region), and possible regularity quasinormal frequencies represented by crosses.}
	\label{fig:complexplane}
\end{figure}

The proof of Theorem \ref{thm:mainthmA} uses crucially the meromorphicity of the inverse $L_s^{-1}$ of the Laplace transform of the wave operator in Kerr and $L_{s,\Lambda}^{-1}$ in Kerr--de Sitter (see \S \ref{sec:KdS} for a definition of $L_{s,\Lambda}$). Meromorphicity in this context is defined as follows:
\begin{definition}
\label{def:meromfamily}
Let $O\subset \C$ be a connected and open subset, $X,Y$ Banach spaces and $B(X,Y)$ the space of bounded linear operators from $X$ to $Y$. Let $A: O\to B(X,Y)$ be a family of bounded operators. Then $A$ or $\{A(s)\}$ is \emph{meromorphic} if for any $s_0\in O$ there exists a small neighbourhood $U$ of $s_0$, such that for all $s\in U$, we can write $A(s)=A_0(s)+\sum_{k=1}^m (s-s_0)^{-k}A_k$ in $U$, where $A_0(s)$ is holomorphic at $s_0$, i.e.\ $\lim_{s\to s_0}\frac{A_0(s)-A(s_0)}{s-s_0}$ is well-defined with respect to the operator topology on $B(X,Y)$. If $m>0$ for some $s_0\in O$, we refer to $s_0$ as a \emph{pole} of $A$.
\end{definition}

The orthogonality properties of regularity quasinormal modes are related to the cokernel of $L_s$, which is defined as follows:
\begin{definition}
\label{eq:comode}
Let $A: X\to Y$ be a linear operator between Banach spaces $X,Y$. Then the \emph{cokernel of $A$} is defined as the quotient space $\textnormal{coker}\, A=Y/\textnormal{ran}\,A$.
\end{definition}

The following theorem is a precise version of Theorem \ref{thm:rough2}(i)--(ii) and includes an additional statement on the orthogonality relations between the kernel and cokernel of $L_s$. \emph{This theorem constitutes the core of the present paper.}
\begin{theorem}
\label{thm:mainthmLs}
Let $\sigma\in \R$ and $N_+\in \N_0$. 
\begin{enumerate}[label=\emph{(\roman*)}]
\item For suitably large $R_{\infty}>r_+$ and suitably small $M^2\Lambda\geq 0$, the family of operators $\{L_{s,\Lambda}^{-1}: H_{\sigma,R_{\infty}}^{N_+}\to H_{\sigma, R_{\infty},2}^{N_++1}\}$ with $s\in \Omega_{\sigma,N_+}$ is meromorphic. We denote with $\mathscr{Q}_{\sigma,R_{\infty},N_+}$ the set of poles of $\{L_{s,\Lambda}^{-1}\}$, which does not depend on the precise choice of foliation-defining function $\mathbbm{h}$. Furthermore, there exists a $C_0\geq 0$ such that the family of operators is holomorphic when restricted to $s\in \{\re z> C_0\}$. 
\item Let $\widetilde{R}_{\infty}\geq R_{\infty}$, $\widetilde{N}_+\in \N_0$ and $\widetilde{\sigma}\in\R$. Then 
\begin{equation*}
\mathscr{Q}_{\sigma,\widetilde{R}_{\infty},\widetilde{N}_+}\cap \Omega_{\sigma,N_+}\cap \Omega_{\widetilde{\sigma},\widetilde{N}_+}=\mathscr{Q}_{\sigma, R_{\infty},N_+}\cap \Omega_{\sigma,N_+}\cap \Omega_{\widetilde{\sigma},\widetilde{N}_+}
\end{equation*}
 and $\bigcup_{N_+\in \N_0, \sigma\in \R}\mathscr{Q}_{\sigma,R_{\infty},N_+}$ has no accumulation points in $\Omega$.
\item Let $s\in \Omega_{\sigma,N_+}$. Then the space $ \ker L_{s,\Lambda}$ is finite dimensional with $\dim \ker L_{s,\Lambda}=\dim \textnormal{coker}\,  L_{s,\Lambda}$ and $\textnormal{ran}\,L_{s,\Lambda}^*=(\ker L_{s,\Lambda})^{\perp}\subset H_{\sigma, R_{\infty},2}^{N_++1}$, with $L_{s,\Lambda}^*$ the Hermitian adjoint and $\textnormal{coker}\, L_{s,\Lambda}^*$ isomorphic to $\ker L_{s,\Lambda}$. 
	\end{enumerate}
\end{theorem}
\begin{proof}
By standard, local energy estimates for $\square_{g_{M,a}}\phi=0$ and a Gr\"onwall inequality, one can obtain uniform constants $C,C_0>0$ such that:
\begin{equation}
\label{eq:gronw}
\|(\phi,\partial_{\tau}\phi)\|_{H^{1}(\Sigma)\times L^2(\Sigma)}(\tau)\leq C e^{C_0\tau}\|(\phi,\partial_{\tau}\phi)\|_{H^{1}(\Sigma)\times L^2(\Sigma)}(0).
\end{equation}
Suppose $\{L_{s,\Lambda}^{-1}\}$ had a pole in $s_0\in \{\re z> C_0\}$. Then the solution $\phi$ corresponding to initial data $(\phi_{s_0},s\phi_{s_0})$ with $\psi_{s_0}=\varrho\phi_{s_0}\in \ker L_{s_0,\Lambda}$ takes the form $(\phi,\partial_{\tau}\phi)(\tau)=(e^{s_0\tau }\phi_{s_0},s\phi_{s_0})$ and hence
\begin{equation*}
\|(\phi,\partial_{\tau}\phi)\|_{H^{1}(\Sigma)\times L^2(\Sigma)}(\tau)=e^{(\re s_0) \tau}\|(\phi,\partial_{\tau}\phi)\|_{H^{1}(\Sigma)\times L^2(\Sigma)}(0)
\end{equation*}
for all $\tau\geq 0$, which is in contradiction with \eqref{eq:gronw}. Hence, $\{L_{s,\Lambda}^{-1}\}$ is holomorphic when restricted to $\{\re z> C_0\}$. In the $\Lambda=0$ case, the remaining statement in (i) as well as (iii) follow from Proposition \ref{prop:mainprop} below. Part (ii) follows from Proposition \ref{prop:accumulation} below. This is then generalized to $\Lambda\geq 0$ in Corollary \ref{cor:Kdsmainresult}.
\end{proof}

\begin{remark}
\label{rm:modestab}
Uniform boundedness (and decay) of solutions to the wave equation $\square_{g_{M,a}}\phi=0$, established by Dafermos--Rodnianski--Shlapentokh-Rothman in \cite{part3} implies that in the $\Lambda=0$ case, the set $\mathscr{Q}_{\rm reg}$ is in fact entirely contained in the left-half plane $\{\re z<0\}$ ($C_0$ in the proof of Theorem \ref{thm:mainthmLs}(i) can be set to $0$). The validity of Theorem \ref{thm:mainthmLs}, however, is independent of a boundedness statement and applies to settings where such a boundedness statement is false, like the class of wave equations with potentials on Kerr studied in \cite{mos17}. 

Closely related to uniform boundedness for $\square_{g_{M,a}}\phi=0$, is the statement of ``mode stability'' \cite{whiting,sr15}, which implies the absence of regularity quasinormal modes of the form \eqref{eq:separabilityqnm} below. Since regularity quasinormal modes need not satisfy \eqref{eq:separabilityqnm}, mode stability alone is not enough to ensure that $\mathscr{Q}_{\rm reg}$ is contained in the left-half plane when $\Lambda=0$; see Remark \ref{rmk:radialode} below.

For $\Lambda>0$ with $M^2\Lambda\ll 1$, the results of \cite{hin24} imply that $\mathscr{Q}_{\rm reg}$ is contained in $\{\re z<0\}$. Note that in the general $\Lambda>0$ case however, not even mode stability is known; see \cite{cate22} for partial results in this direction.
\end{remark}

\begin{remark}
\label{rmk:radialode}
Consider the following angular ODE:
\begin{equation}
\label{eq:angode}
-\frac{1}{\sin \theta} \partial_{\theta}(\sin \theta \partial_{\theta}S_{m\ell})+\left[\frac{1}{\sin^2\theta}m^2+a^2\omega^2\sin^2\theta\right]S_{m\ell}=\Lambda_{m\ell}(a\omega)S_{m\ell},
\end{equation}
with $S_{m\ell}(\theta;a\omega)$ in $L^2(\sin \theta d\theta)$ the \emph{oblate spheroidal harmonics}.  When $a=0$, \eqref{eq:angode} is simply the Laplace equation on $\s^2$ restricted to functions with azimuthal angular dependence $e^{im\varphi}$ and the solutions are spherical harmonics.

When $\omega\in \R$, \eqref{eq:angode} admits solutions for a discrete set of $\Lambda_{m\ell}(a\omega)\in \R$, labelled by $\ell\in \N_{\geq |m|}$ such that $\Lambda_{m\ell}(a\omega)\geq \Lambda_{m\ell'}(a\omega)$ if $\ell\geq \ell'$. Furthermore, when $\omega\in \R$, the functions $\{ S_{m\ell}(\theta;a\omega)\}_{\substack{m\in \Z, l\geq |m|}}$ form an orthonormal basis of $L^2(\sin \theta d\theta)$, the self-adjointness of the differential operator on the left-hand side of \eqref{eq:angode} and the compactness of its inverse.

The functions $\omega\mapsto \Lambda_{m\ell}(a\omega)$ can be analytically continued to the rest of the complex plane as complex functions, with a finite number of branch points, for which the corresponding \eqref{eq:angode} admits solutions in $L^2(\sin \theta d\theta)$; see for example \cite{ms54}[Proposition 5 of Chapter 1.5, Proposition 2 of Chapter 3.22] for a proof in the context of general ODE and \cite{costa20}[Proposition 2.1] for a discussion on this matter in the context of \eqref{eq:angode}. However, when $\omega\notin \R$, $\{ S_{m\ell}(\theta;a\omega)\}_{\substack{m\in \Z, l\geq |m|}}$ no longer need to form an orthonormal basis.

In much of the physics literature, see for example \cite{leav85}, quasinormal mode solutions $\phi$ to the wave equation \eqref{eq:waveeq} are required to take the separated form
\begin{equation}
\label{eq:separabilityqnm}
\phi(t,r_*,\theta,\varphi)=e^{-i\omega t}\frac{u_{\ell m}(r_*;\omega)}{\sqrt{r^2+a^2}}S_{m\ell}(\theta;a\omega)e^{im\varphi},
\end{equation}
with $S_{m\ell}(\theta;a\omega)$ a solution to \eqref{eq:angode} and $u_{\ell m}(r;\omega)$ a solution to the ODE:
\begin{equation}
\label{eq:radialode}
u_{\ell m}''+(\omega^2-V_{m \ell}^{(a\omega)})u_{\ell m}=0,
\end{equation}
where $V_{m \ell}^{(a\omega)}(r_*)$ is a potential function that depends on $\Lambda_{m\ell}(a\omega)$.

In the present paper, we do \underline{not} impose the restriction \eqref{eq:separabilityqnm} and we instead allow for more general solutions of the form:
\begin{equation*}
\phi(t,r,\theta,\varphi)=e^{-i\omega t}\sum_{m\in \Z}\hat{\phi}_m(r,\theta)e^{im\varphi}.
\end{equation*}

In view of the finite dimensionality of the space of regularity quasinormal modes that follows from Theorem \ref{thm:mainthmA}(iii), we can infer that quasinormal modes corresponding to a single quasinormal frequency $s=-i\omega$ must be supported on a finite number of $m$. If $\{ S_{m\ell}(\theta;a\omega)\}_{\substack{m\in \Z, l\geq |m|}}$ formed an orthonormal basis for $L^2(\s^2)$, then we would also be able to conclude that any regularity quasinormal mode can be written as a finite linear combination of regularity quasinormal modes satisfying \eqref{eq:separabilityqnm}. This is the case when $a=0$. 

As mentioned above, we are not, however, guaranteed an orthonormal basis of functions $S_{m\ell}$ when $a\neq 0$, so we cannot draw such a conclusion. The set of regularity quasinormal frequencies could therefore in principle be \underline{strictly larger} than the set of frequencies corresponding to quasinormal modes that can be separated as in \eqref{eq:separabilityqnm}.
\end{remark}

\subsection{Application 1: stability of the quasinormal spectrum and convergence of Kerr--de Sitter quasinormal frequencies}
We first give a definition of \emph{stability} of the spectrum of a family of linear operators that is suitable for the setting of the present paper.
\begin{definition}[Stability of the spectrum for a family of restricted operators]
\label{def:stabspect}
Let $O =\bigcup_{\alpha\in I} O_{\alpha}\subseteq \C$, with $O_{\alpha}\subset \C$ open and $I$ an index set. Let $ \mathbf{H}$ be a Hilbert space and let $A: {\rm dom}(A)\subseteq \mathbf{H}\to \mathbf{H}$ be a densely defined linear operator. Let $ \mathbf{H}'_{\alpha}\subseteq \mathbf{H}_{\alpha}\subseteq \mathbf{H}$ be families of closed subspaces, such that the following restrictions of $A$ are densely defined operators: $A: \mathbf{H}'_{\alpha}\supseteq {\rm dom}(A_{\alpha})\to \mathbf{H}_{\alpha}$. Denote the corresponding restricted operators $A_{\alpha}$.

Let $B$ be a linear operator, such that its restrictions to $ \mathbf{H}'_{\alpha}$ results in operators $B_{\alpha}: \mathbf{H}'_{\alpha}\to \mathbf{H}_{\alpha}$ that are bounded with respect to the operator norms $\|\cdot\|_{\mathbf{H}_{\alpha}'\to \mathbf{H}_{\alpha}}$ for all $\alpha\in I$. 

Then we say that $A$ has a \emph{stable spectrum in $O$ in the direction $B$, with respect to the the families of subspaces $\{\mathbf{H}_{\alpha}, \mathbf{H}'_{\alpha}\}_{\alpha\in I}$}, if the following holds: for any compact domain $K\subset O$ and for all $\epsilon>0$, there exists a $\delta_0>0$ such that for all $\alpha\in I$ and for all $0<\delta\leq \delta_0$:
\begin{enumerate}
\item ${\rm Spect}(A_{\alpha}+\delta \cdot B_{\alpha})\cap O_{\alpha}\cap K \subseteq \bigcup_{s\in {\rm Spect}(A_{\alpha})\cap O_{\alpha}} D_{\epsilon}(s)$, with $D_{\epsilon}(s)$ the open disc of radius $\epsilon$ around $s\in \C$.
\item For all $s\in {\rm Spect}(A_{\alpha})\cap O_{\alpha}\cap K$,  
\begin{equation*}
{\rm Spect}(A_{\alpha}+\delta\cdot  B_{\alpha})\cap K\cap D_{\epsilon}(s)\neq \emptyset.
\end{equation*}
\end{enumerate}
\end{definition}

\begin{figure}[H]
	\begin{center}
\includegraphics[scale=0.75]{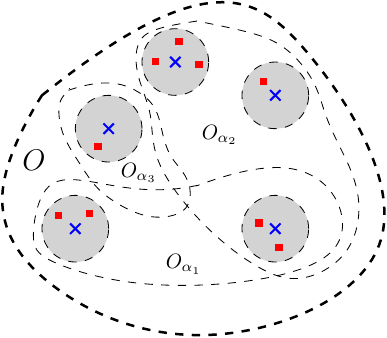}
\end{center}
\caption{A pictorial representation of stability of the spectrum in the case of a spectrum of isolated eigenvalues of $A_{\alpha}$ represented by crosses. The eigenvalues of $A_{\alpha}+\delta\cdot  B_{\alpha}$ are represented by squares and the shaded disks represent $D_{\epsilon}(s)$ with $s\in  {\rm Spect}(A_{\alpha})$.}
	\label{fig:stabQNF}
\end{figure}
Using the above definition, the following theorem gives a precise formulation of Theorem \ref{thm:rough1b}(i):
\begin{theorem}
\label{thm:stab}
Let $\mathcal{A}_{\delta \cdot Q}$ denote the infinitesimal generator of time translations corresponding to the following perturbed wave operator:
\begin{equation*}
\square_{g_{M,a}}+\delta \cdot \varrho^{-3}Q(\varrho \cdot)
\end{equation*}
with $Q$ the following second-order differential operator:
\begin{equation*}
Q=Q^{\tau \tau}(r,\vartheta,\varphi_*)\partial_{\tau}^2+Q^{\tau x}(r,\vartheta,\varphi_*)\partial_{\tau}\partial_x+Q^{\tau \vartheta}(r,\vartheta,\varphi_*)\partial_{\tau}\partial_{\vartheta}+Q^{\tau \varphi}(r,\vartheta,\varphi_*)\partial_{\tau}\partial_{\varphi_*}+Q^{\tau}(r,\vartheta,\varphi_*)\partial_{\tau}+\tilde{Q},
\end{equation*}
where $Q^{\tau}, Q^{\tau \tau},Q^{\tau x},Q^{\tau \vartheta}, Q^{\tau \varphi_*}$ are bounded and $\tilde{Q}$ is a differential operator of order $\leq 2$ acting functions on $\Sigma$, such that $\tilde{Q}_s=\tilde{Q}: H_{\sigma,R_{\infty},2}^{N_+}\to H_{\sigma,R_{\infty}}^{N_+}$ is bounded for all $|\sigma|<\sigma_0$, with $\sigma_0>0$ or $\sigma_0=\infty$. Let $B$ be the operator satisfying
\begin{equation*}
\mathcal{A}+\delta \cdot B=\mathcal{A}_{\delta\cdot Q}.
\end{equation*}
Let $s_0\in \C$ be such that $L_{s_0}^{-1}: H_{\sigma,R_{\infty}}^{N_+}\to H_{\sigma,R_{\infty},2}^{N_++1}$ is well-defined. Then the spectrum of $\mathcal{A}: {\rm dom}(\mathcal{A})\supset \mathbf{H}^1\to \mathbf{H}^1$ is stable in $\{\re s>-(\frac{1}{2}+N_+)\kappa_+\}\cap B_{\sigma_0^2}\cap\Omega=B_{\sigma_0^2}\cap\bigcup_{\sigma\in \R}\Omega_{\sigma,N_+} \subset \C$ with respect to the families of subspaces $\{\mathbf{H}_{\sigma}^{N_+}, D_{\sigma,R_{\infty}}^{N_+}(L_{s_0})\oplus H_{\sigma,R_{\infty},1}^{N_++1})\}_{\sigma\in (-\sigma_0,\sigma_0)}$ in the direction $B$.

In particular, if $Q$ vanishes on $\Sigma \cap\{r\geq R_{\infty}\}$ and has smooth coefficients, then the above stability result applies in the full $\Omega$.
\end{theorem}
\begin{proof}
By a slight modification of the proof of Proposition \ref{prop:relALs}, we obtain the following relation:
\begin{equation*}
\mathcal{A}_{\delta\cdot Q}-s\mathbf{1}=\begin{pmatrix}
0 & \mathbf{1}\\
\mathbf{1} & P_{\delta\cdot  Q}-s\mathbf{1}
\end{pmatrix}\begin{pmatrix}
\frac{1}{2M^2h(0)-Mxb_{\tau \tau}-\delta\cdot Q^{\tau \tau} } (L_s+\delta\cdot  \tilde{Q})& 0\\
0&\mathbf{1}
\end{pmatrix}\begin{pmatrix}
\mathbf{1}& 0\\
-s \mathbf{1} & \mathbf{1}
\end{pmatrix},
\end{equation*}
with $P_{\delta\cdot Q}=\frac{2M^2h(0)-Mxb_{\tau \tau} }{2M^2h(0)-Mxb_{\tau \tau}-\delta\cdot Q^{\tau \tau}} P+\frac{\delta }{2M^2h(0)-Mxb_{\tau \tau}-\delta\cdot Q^{\tau \tau}} \left(Q^{\tau x}\partial_x+Q^{\tau \vartheta}\partial_{\vartheta}+Q^{\tau \varphi_*}\partial_{\varphi_*}+Q^{\tau}\mathbf{1}\right)$.

Note finally, that for all $\delta>0$ $P_{\delta\cdot Q}f\in H_{\sigma,R_{\infty}}^{N_+}$ if $f\in  H_{\sigma,R_{\infty},1}^{N_++1}$, so we can restrict the domain of $\mathcal{A}_{\delta\cdot Q}$ to $D_{\sigma,R_{\infty}}^{N_+}(L_{s_0})\oplus H_{\sigma,R_{\infty},1}^{N_++1}$.

By the above expression for $\mathcal{A}_{\delta\cdot Q}-s\mathbf{1}$ and Corollary \ref{cor:relLA}, invertibility of $\mathcal{A}_{\delta \cdot Q}-s\mathbf{1}$ is equivalent to invertibility of $L_s+\delta \cdot \tilde{Q}:  H_{\sigma,R_{\infty}}^{N_+}\to  H_{\sigma,R_{\infty},2}^{N_++1}$, so we can apply Proposition \ref{prop:pertLs} below to obtain the desired stability statement.
\end{proof}

The following theorem is a precise version of Theorem \ref{thm:rough2b}(ii) (which leads to Theorem \ref{thm:rough1b}(ii)).
\begin{theorem}
\label{lm:kappaconvergence}
Let $\{\Lambda_n\}$ be a sequence of cosmological constants satisfying $M^2\Lambda_n\ll 1$ and $\Lambda_n\downarrow 0$. Suppose that there exists an $s_*\in \Omega$ with $\ker L_{s_*}^{-1}\neq \{0\}$ and consider $\sigma\in \R$ and $N_+\in \N$ such that $s_*\in \Omega_{\sigma, N_+}$.

Then there exists a sequence $s_n\in \Omega_{\sigma, N_+}$, such that $\ker L_{s_n,\Lambda_n}^{-1}\neq \{0\}$ and $s_n\to s_*$ as $n\to \infty$. Furthermore, for all $\hpsi\in \ker L_{s_*}^{-1}$, there exists a sequence $\{\hpsi_n\}$ in $H^{N_+}_{\sigma,R_{\infty}}$ with $\hpsi_n\in  \ker L_{s_n,\Lambda_n}^{-1}$ and
\begin{equation*}
	\|\hpsi-\hpsi_n\|_{H^{N_+}_{\sigma, R_{\infty}}}\to 0\quad (n\to \infty).
\end{equation*}
Conversely, let $\{s_n\}$ be a convergent sequence of elements in $\ker L_{s,\Lambda_n}$ and $\Lambda_n\to 0$ as $n\to \infty$, such that the limit $s_*$ satisfies: $s_*\in \{|{\rm arg} z|<\frac{2\pi}{3}\}$.  Then $s_*\in \ker L_{s,0}$.
\end{theorem}
\begin{proof}
	In the context of the $\kappa$-regularized operator $L_{s,\kappa}$, this proposition is proved in Lemma \ref{lm:kappaconvergence}. This is then generalized to the true Kerr--de Sitter operator $L_{s,\Lambda}$ in \S \ref{sec:KdS}.
\end{proof}

\subsection{Application 2: a characterization of scattering resonances}
In order to justify the use of the term ``regularity quasinormal frequencies'', we need to make a connection with other definitions of ``quasinormal frequencies'' that have previously appeared in the literature.

It is common to define ``quasinormal frequencies'' as scattering resonances associated to the meromorphic continuation of the cut-off resolvent operator corresponding to a ``standard'' time function whose level sets are asymptotically flat hypersurfaces. In the case of Kerr, this is typically take to be the Boyer--Lindquist time function $t$. One first considers the operator
\begin{equation*}
A(\omega)=e^{i\omega t}\square_{g_{M,a}}(e^{-i\omega t}\cdot (\cdot)),
\end{equation*}
which can be inverted on suitable spaces of functions on $\Sigma'=\{t=0\}$ with so-called outgoing boundary conditions and where $\omega\in \C$ with suitably large $\im \omega>0$. The inverse, denoted $R(\omega)$, is also called the \emph{resolvent operator}.

While $R(\omega)$ fails to make sense in a region intersecting both $\{\im z>0\}$ and $\{\im z<0\}\subset \C$, one can conjugate $R(\omega)$ with a multiplication operator $\chi \mathbf{1}$, where $\chi\in C_c^{\infty}(\Sigma')$ and try to meromorphically continue the \emph{cut-off resolvent}:
\begin{equation*}
\chi R(\omega)\chi
\end{equation*}
as an operator on suitable function spaces. The poles of the meromorphic continuation of $\chi R(\omega)\chi$ are called \emph{scattering resonances}.

The theorem below provides a more precise version of Theorem \ref{thm:rough1}(iii). This theorem demonstrates that scattering resonances form a \emph{subset} of the set of regularity quasinormal frequences, thereby justifying the nomenclature.

\begin{theorem}
\label{thm:cutoffresolvmain}
Let $N_+\in \N_0$. The cut-off resolvent operator:
\begin{equation*}
\chi R(\omega) \chi : H^{N_+}_{\snabla}(\Sigma')\to H^{N_++1}_{\snabla}(\Sigma')
\end{equation*}
is well-defined for any cut-off function $\chi\in C_c^{\infty}(\Sigma')$ and for all $\omega\in \C$ with $\im \omega>\omega_0>0$ suitably large and $\omega\mapsto \chi R(\omega) \chi $ can be meromorphically continued to $i \Omega\cap\{\im \omega>-\left(\frac{1}{2}+N_+\right)\kappa_+\}$. Let $i \mathscr{Q}_{\rm res} $ denote the union over the set of poles of $\omega\mapsto \chi R(\omega) \chi $ over $N_+$. Then
\begin{equation*}
\mathscr{Q}_{\rm res} \subseteq \mathscr{Q}_{\rm reg}.
\end{equation*}
Furthermore, for $\frac{|a|}{M}$ suitably small, or under the restriction to functions with a bounded azimuthal number $m$,
\begin{equation*}
\chi R(\omega) \chi : L^{2}(\Sigma')\to H^2(\Sigma')
\end{equation*}
is well-defined and can be meromorphically continued to $i \Omega$, with poles in the set $\mathscr{Q}_{\rm res} $.
\end{theorem}
\begin{proof}
Consider the extension of $\Sigma'$ to the bifurcation sphere, denoted $\overline{\Sigma'}$. By standard local-in-time energy estimates (see also the proof of Theorem \ref{thm:mainthmLs}), it follows that the time evolution map $\mathcal{S}(\tau): H^1(\overline{\Sigma'})\times L^2(\overline{\Sigma'})\to H^1(\overline{\Sigma'})\times L^2(\overline{\Sigma'})$ corresponding to initial data $(\phi,\partial_t\phi)$ on $\overline{\Sigma'}$ is a bounded linear operator for each $\tau$, with a norm that can be bounded by $e^{s_0\tau}$, for suitably large $s_0>0$. By standard results of functional analysis, see for example \cite{evans}[\S 7.4, Theorem 2] and \cite{hil48}[Theorem 11.6.1], the corresponding infinitesimal generator $\mathcal{A}: H^1(\overline{\Sigma'})\times L^2(\overline{\Sigma'})\supset {\rm dom}(\mathcal{A})\to H^1(\overline{\Sigma'})\times L^2(\overline{\Sigma'})$ is closed and densely defined with $(\mathcal{A}-s\mathbf{1})^{-1}: H^1(\overline{\Sigma'})\times L^2(\overline{\Sigma'})\to {\rm dom}(\mathcal{A})$ a well-defined bounded operator for $\re s>s_0$.

In this case, we can apply \eqref{eq:relAandLs}, with $\frac{L_s}{ (2Mh(0)-Mx b_{\tau\tau})}$ replaced with $\frac{r^2 A(\omega)}{\mathbbm{h}(2(r^2 +a^2)-\Delta \mathbbm{h})-a^2\sin^2\theta}$, where $\mathbbm{h}=\frac{r^2+a^2}{\Delta}$ with $\omega=i\cdot s$ and $P$ replaced with an appropriate first order operator to conclude via \eqref{eq:rangeexpr} that:
\begin{equation*}
R(\omega): \frac{\Delta}{r^2}L^2(\overline{\Sigma'})\to H^1(\overline{\Sigma'})
\end{equation*}
is well-defined for suitably large $\im \omega>0$ and we can consider $\chi R(\omega) \chi: L^2(\Sigma')\to H^1(\Sigma')$.

In Corollary \ref{cor:resonances} below, we prove the existence of the meromorphic continuation of $\chi R(\omega) \chi$, defined on appropriate subspaces, and its relation to $ \mathscr{Q}_{\rm reg}$.
\end{proof}

\begin{remark}
	The results of Theorem \ref{thm:cutoffresolvmain} in fact apply also when considering cut-off resolvents with respect to time functions whose level sets intersect $\mathcal{H}^+$ to the future of the bifurcation sphere and are asymptotically flat. We denote the level set corresponding to ``time zero'' in this case by $\widetilde{\Sigma}$. Now, the cut-off functions $\chi$ are allowed to be supported near $\mathcal{H}^+$ but the support remains compact. In the domains and codomains of the maps in Theorem \ref{thm:cutoffresolvmain}, $\Sigma'$ is then replaced with $\widetilde{\Sigma}$.
\end{remark}

\section{Guide to the rest of the paper}
\label{sec:guide}
We outline the structure of the rest of the paper and relate each section to the steps sketched in \S \ref{sec:ideastechniques}, which concern the operator $L_s$. We first focus on the $\Lambda=0$ setting.\\

\paragraph{Theorem \ref{thm:rough2}(i):}
\begin{itemize}
\item In \S \ref{sec:redshiftelliptic}, we derive red-shift estimates and elliptic estimates (featuring a finite number of derivatives) and couple them to obtain a single estimate in the bounded region $\varrho\leq \eta^{-1}R_{\infty}$.
\item In \S \ref{sec:gevrey}, we derive Gevrey-type estimates (featuring an infinite number of derivatives) in the region $\varrho\geq 2R_{\infty}$: estimates with higher-order radial derivatives are derived in \S \ref{sec:gevreyho} and estimates with lower-order radial derivatives are derived in \S \ref{sec:gevreylo}. 
\item The Gevrey estimates are coupled with the red-shift and elliptic estimates in \S \ref{sec:fullycoupledestmates} to obtain a single, global estimate.
\item In \S \ref{sec:invertingmathfrakL}, we establish the invertibility of $\mathfrak{L}_{s,\lambda,\kappa}$ with $\kappa>0$ and construct $\mathfrak{L}_{s,\lambda}^{-1}$ as a $\kappa \downarrow 0$ limit.
\item In \S \ref {sec:invtLs}, we establish the meromorphicity and Fredholm properties of $L_s^{-1}$.
\end{itemize}
Theorem \ref{thm:rough2}(ii) and Theorem \ref{thm:rough2b}(ii):
\begin{itemize}
\item In \S \ref{sec:invtLs}, we also establish the independence of the regularity quasinormal spectrum with respect to a change in Hilbert spaces, corresponding to different $N_+$, $R_{\infty}$ and $\sigma$. We also establish Theorem \ref{thm:rough2b}(ii) for the $\kappa$-regularized operator.
\end{itemize}

We subsequently prove Theorem \ref{thm:rough2b}(i):
\begin{itemize}
\item In \S \ref{sec:stabqnf}, we establish stability of the quasinormal spectrum under suitably small perturbations at the level of $L_s$.
\end{itemize}

Since the analysis in the $\Lambda>0$ case actually follows, with minor modifications, from the analysis that is already required for the $\Lambda=0$ case, we only take $\Lambda>0$ in \S \ref{sec:KdS}, where we focus primarily on the differences with the $\Lambda=0$ case. Theorem \ref{thm:rough2}(i) for $\Lambda>0$ and Theorem \ref{thm:rough2b}(ii) for the true Kerr--de Sitter operator follow from \S \ref{sec:KdS}.

Finally, we establish the relation between regularity quasinormal frequencies and scattering resonances:
\begin{itemize}
\item Theorem \ref{thm:rough2}(iii): In \S \ref{sec:resonances}, we construct the meromorphic continuation of the cut-off resolvent operator and relate scattering resonances to regularity quasinormal frequencies.
\end{itemize}

\section{Red-shift and elliptic estimates}
\label{sec:redshiftelliptic}
We will consider different types of estimates in the following two regions:
\begin{enumerate}[label=(\Alph*)]
\item $r_+ \leq r\leq \frac{1}{2}R_{\infty}$,
\item $ \frac{1}{4}R_{\infty} \leq  \varrho\leq \eta^{-1}R_{\infty}$,
\end{enumerate}
with $\R \ni \eta<2$ and where $R_{\infty}M^{-1}$ will be chosen suitably large. We will later also need to restrict to suitably small $\eta$.

The estimates in (A) are the spacelike analogues of the spacetime \emph{red-shift estimates} of Dafermos--Rodnianski  \cite{redshift} and they establish, in particular, control of the non-degenerate $H^1$ norm of $\hat{\psi}\in C^{\infty}(\widehat{\Sigma})$ in terms of the $L^2$ norm of $\mathfrak{L}_{s,\lambda}\hat{\psi}$. These estimates rely crucially on the positivity of the surface gravity $\kappa_+$ associated to the event horizon at $r=r_+$. In contrast with the spacetime red-shift estimates, which apply only to regions where $r$ is sufficiently close to $r_+$, they hold in regions of arbitrarily large $r$.

Note that for $M^{-1}R_{\infty}$ sufficiently large, the region (A) covers the ergoregion of Kerr,  where the vector field $T$ is not timelike, and therefore ``standard'' (degenerate) elliptic estimates do not hold.

The estimates in (B) are more similar to standard elliptic estimates, in the sense that they establish control of a (weighted) $H^2$ norm of $\hat{\psi}\in C^{\infty}(\widehat{\Sigma})$ in terms of a (weighted) $L^2$ norm of $\mathfrak{L}_{s,\lambda}\hat{\psi}$ and rely on the fact that the spacetime metric is sufficiently close to the Minkowski metric in appropriate coordinates, by the asymptotic flatness property of the Kerr metric. The main difference with standard elliptic estimates, is the addition of a degenerate weight function in the norms, which will play an important role in the coupling with the Gevrey estimates of \S \ref{sec:gevrey}.

\subsection{Red-shift estimate in (A)}
\label{sec:redshift}
We first consider the region with:
\begin{equation*}
r_+ \leq \varrho\leq \frac{1}{2}R_{\infty}.
\end{equation*}
Recall from \eqref{eq:operatorbounded} that in this region $\varrho=r$, $\vartheta=\theta$ and we can write:
\begin{equation*}
	r^{-2}\rho^{2}\mathfrak{L}_{s,\lambda,\kappa}=r\Delta \partial_r^2(r^{-1}\cdot)+r\frac{d\Delta}{dr}\partial_r(r^{-1}\cdot)+2ar\partial_r\partial_{\varphi_*}(r^{-1}\cdot)+\slashed{\Delta}_{\s^2}(\cdot)+sr(4 Mr-2(1-\mathbbm{h})\Delta)\partial_r(r^{-1}\cdot)-\lambda \pi_{\leq l}.
\end{equation*}
\begin{proposition}
\label{prop:redshift}
Let $\re s>-\frac{\kappa_+}{2}$ and $\delta=1+2\kappa_+^{-1} \min\{\re s,0\}$. Then there exists a constant $C_s=C_s(\kappa_+,s)>0$ and a numerical constant $C>0$ such that
\begin{multline}
\label{eq:redshift}
\int_{r_+}^{\frac{1}{3}R_{\infty}}\int_{\s^2} \frac{1}{2}(r-M)(1+2\kappa_+^{-1} \min\{\re s,0\})r^{\delta}|\partial_r(r^{-1}\hpsi)|^2\\
+\frac{1}{2}(1+2\kappa_+^{-1} \min\{\re s,0\})r^{-3+\delta}\left[|\snabla_{\s^2}\hpsi|^2+\lambda |\pi_{\leq l}\hpsi|^2\right]\,d\upsigma dr\\
+\int_{\s^2}  r^{\delta-2}\left[|\snabla_{\s^2}\hpsi|^2+\lambda|\pi_{\leq l}\hpsi|^2\right]\Big|_{r=r_+}\,d\upsigma\\
\leq \boxed{C\int_{\frac{1}{3}R_{\infty}}^{\frac{1}{2}R_{\infty}}\int_{\s^2}  r^{-3+\delta}\left[|\snabla_{\s^2}\hpsi|^2+\lambda|\pi_{\leq l}\hpsi|^2\right]\,d\upsigma dr}+C_s\int_{r_+}^{\frac{1}{2}R_{\infty}}\int_{\s^2}r^{-3+\delta}|\mathfrak{L}_{s,\lambda,\kappa}\hpsi|^2\,d\upsigma dr.
\end{multline}
\end{proposition}
\begin{proof}
Let $\hphi=r^{-1}\hpsi$. Without loss of generality, we will assume $\hphi=\pi_{\leq l}\hphi$ in the derivations below. The $\pi_{>l}\hphi$ part of general $\hphi$ can be estimated in exactly the same way by simply taking $\lambda=0$.

We multiply $r^{-2}\rho^{2}\mathfrak{L}_{s,\lambda}\hpsi$ with $\chi_{r\leq \frac{1}{3}R_{\infty}}r^{\delta-1}\overline{\partial_r\hphi}$, where $\delta>0$ will be fixed later, and where $\chi_{r\leq \frac{1}{3}R_{\infty}}$ is a smooth cut-off function that is equal to 1 for $r\leq \frac{1}{3}R_{\infty}$ and vanishes when $r\geq \frac{1}{2}R_{\infty}$. We then apply the Leibniz rule to obtain:
\begin{multline*}
\re(\chi_{r\leq \frac{1}{3}R_{\infty}}r^{\delta-1}\overline{\partial_r\hphi}\cdot r^{-2}\rho^{2}\mathfrak{L}_{s,\lambda,\kappa}\hpsi)=\Delta r^{\delta}\chi_{r\leq \frac{1}{3}R_{\infty}}\re(\overline{\partial_r\hphi} \partial_r^2\hphi)+\frac{d\Delta}{dr}\chi_{r\leq \frac{1}{3}R_{\infty}}r^{\delta}|\partial_r\hphi|^2\\
+2ar^{\delta}\chi_{r\leq \frac{1}{3}R_{\infty}}\re(\overline{\partial_r\hphi}\partial_r\partial_{\varphi_*}\hphi)+r^{\delta}\re(\overline{\partial_r\hphi}\slashed{\Delta}_{\s^2}\hphi)+(4Mr- 2(\mathbbm{h}-1)\Delta))r^{\delta}\chi_{r\leq \frac{1}{3}R_{\infty}} {\re s}|\partial_r\hphi|^2-\lambda r^{\delta}\re(\overline{\partial_r \hphi}\hphi)\\
=\frac{1}{2}\partial_r\left(\Delta r^{\delta}\chi_{r\leq \frac{1}{3}R_{\infty}}|\partial_r\hphi|^2-r^{\delta}\chi_{r\leq \frac{1}{3}R_{\infty}}|\snabla_{\s^2}\hphi|^2-\lambda\chi_{r\leq \frac{1}{3}R_{\infty}}|\hphi|^2\right)+\partial_{\varphi_*}(\ldots)+\textnormal{div}_{\s^2}(\ldots)\\
+\frac{1}{2}r^{\delta}\chi_{r\leq \frac{1}{3}R_{\infty}}\left[\frac{d \Delta}{dr}-\delta  r^{-1}\Delta +(8 M r-4(\mathbbm{h}-1)\Delta)) {\re s}\right]|\partial_r\hphi|^2+\frac{\delta}{2}r^{\delta-1}\chi_{r\leq \frac{1}{3}R_{\infty}}\left[|\snabla_{\s^2}\hphi|^2+\lambda |\hphi|^2\right]\\
-\frac{1}{2}\frac{d\chi_{r\leq \frac{1}{3}R_{\infty}}}{dr}\left(\Delta r^{\delta}\chi_{r\leq \frac{1}{3}R_{\infty}}|\partial_r\hphi|^2-r^{\delta}\chi_{r\leq \frac{1}{3}R_{\infty}}|\snabla_{\s^2}\hphi|^2-\lambda\chi_{r\leq \frac{1}{3}R_{\infty}}|\hphi|^2\right),
\end{multline*}
where we employed the schematic notation $\partial_{\varphi_*}(\ldots)+\textnormal{div}_{\s^2}(\ldots)$ for total derivatives and divergences on $\s^2$, as these terms vanish after integration over $\s^2$.

We first estimate via Young's inequality:
\begin{equation*}
\left|\re\left[r^{\delta-1}\overline{\partial_r\hphi}\cdot r^{-2}\rho^{2}\mathfrak{L}_{s,\lambda,\kappa}\hpsi\right]\right|\leq \frac{\delta}{2}r^{\delta}(r-M)|\partial_r\hphi|^2+\delta^{-1}r^{-2+\delta}(r-M)^{-1}|\mathfrak{L}_{s,\lambda,\kappa}\hpsi|^2.
\end{equation*}

Furthermore, using that $\mathbbm{h}\leq \frac{r^2+a^2}{\Delta}$, we obtain
\begin{equation*}
8 M r-4(\mathbbm{h}-1)\Delta\geq 8Mr-4(r^2+a^2)+4\Delta=0
\end{equation*}
so for $\re s\geq 0$, we have that $(8 M r-4(\mathbbm{h}-1)\Delta)) {\re s}\geq 0$.

For $\re s<0$, we use that $\mathbbm{h}\geq 1$ to obtain
\begin{equation*}
(8 M r-4(\mathbbm{h}-1)\Delta) \re s\geq 8Mr\re s
\end{equation*}
and we have that
\begin{equation}
\label{eq:redshiftfactor}
	\frac{d \Delta}{dr}-\delta  r^{-1}\Delta +8M r{\re s}-\frac{\delta}{2}(r-M)> 2(r-M)\left(1-\frac{3}{4}\delta +2 \frac{r_+-M}{r-M}\kappa_+^{-1} {\re s}\right),
\end{equation}
where we recall the following expression for the surface gravity $\kappa_+$:
\begin{equation*}
	\kappa_+=\frac{r_+-M}{2Mr_+}.
\end{equation*}

Hence, for ${\re s}>-\frac{\kappa_+}{2}$ and $0<\frac{3}{4}\delta<1+\min\{2\kappa_+^{-1}{\re s},0\}$, the right-hand side of \eqref{eq:redshiftfactor} is non-negative definite.

Since $\frac{d\chi_{r\leq \frac{1}{3}R_{\infty}}}{dr}>0$ and there exists a numerical constant $C>0$ such that $-\frac{d\chi_{r\leq \frac{1}{3}R_{\infty}}}{dr}\leq C R_{\infty}^{-1}$, we can estimate:
\begin{equation*}
-\frac{1}{2}\frac{d\chi_{r\leq \frac{1}{3}R_{\infty}}}{dr}\left(\Delta r^{\delta}\chi_{r\leq \frac{1}{3}R_{\infty}}|\partial_r\hphi|^2-r^{\delta}\chi_{r\leq \frac{1}{3}R_{\infty}}|\snabla_{\s^2}\hphi|^2-\lambda\chi_{r\leq \frac{1}{3}R_{\infty}}|\hphi|^2\right)\geq -CR_{\infty}^{-1}r^{\delta}\left(|\snabla_{\s^2}\hphi|^2+\lambda |\hphi|^2\right).
\end{equation*}

After setting $\delta=1+2\kappa_+^{-1} \min\{{\re s},0\}$ and integrating over $[r_+,\frac{1}{2}R_{\infty}]\times \s^2$, we obtain the following estimate: for ${\re s}>-\frac{\kappa_+}{2}$, there exists a $C_s=C(\kappa_+,{\re s},\delta)>0$, such that
\begin{multline*}
\int_{r_+}^{\frac{1}{3}R_{\infty}}\int_{\s^2} \frac{1}{2}(r-M)(1+2\kappa_+^{-1} \min\{{\re s},0\})r^{\delta}|\partial_r\hphi|^2\,d\upsigma dr\\
+\frac{1}{2}(1+2\kappa_+^{-1}  \min\{{\re s},0\})r^{-1+\delta}\left[|\snabla_{\s^2}\hphi|^2+\lambda |\hphi|^2\right]+\int_{\s^2}  r^{\delta}\left[|\snabla_{\s^2}\hphi|^2+\lambda|\hphi|^2\right]\Big|_{r=r_+}\,d\upsigma\\
\leq C\int_{\frac{1}{3}R_{\infty}}^{\frac{1}{2}R_{\infty}}\int_{\s^2}  r^{-1+\delta}\left[|\snabla_{\s^2}\hphi|^2+\lambda|\hphi|^2\right]\,d\upsigma dr+C_s\int_{r_+}^{\frac{1}{2}R_{\infty}}\int_{\s^2}r^{-3+\delta}|\mathfrak{L}_{s,\lambda,\kappa}\hpsi|^2\,d\upsigma dr.
\end{multline*}
\end{proof}

In order to improve the range of allowed values of ${\re s}$, we need to consider higher-order $r$-derivatives of $\hpsi$. In addition, we will consider higher angular derivatives.
\begin{lemma}
Let $k,m\in \N_0$ and $r_+ \leq \varrho\leq \frac{1}{2}R_{\infty}$. Then, for $\hphi=r^{-1}\hpsi$:
	\begin{multline}
	\label{eq:hoeqnearhor}
		\partial_r^k\slashed{\Delta}_{\s^2}^m(r^{-3}\rho^2 \mathfrak{L}_{s,\lambda,\kappa}\hpsi)=\Delta \partial_r^{2+k}\slashed{\Delta}_{\s^2}^m\hphi+(1+k)\frac{d\Delta}{dr}\partial_r^{k+1}\slashed{\Delta}_{\s^2}^m\hphi+2a\partial_{\varphi_*}\partial_r^{k+1}\slashed{\Delta}_{\s^2}^m\hphi+\slashed{\Delta}_{\s^2}^{1+m}\partial_r^k\hphi\\
		+4M rs \partial_r^{k+1}\slashed{\Delta}_{\s^2}^m\hphi-(\lambda \pi_{\leq l}-k(k+1)- 4 k M  s)\partial_r^k\slashed{\Delta}_{\s^2}^m\hphi-2\sum_{j=0}^{k}\frac{k!}{j!(k-j)!}\partial_r^{k-j}((\mathbbm{h}-1)\Delta)\partial_r^j\slashed{\Delta}_{\s^2}^m\hphi.
	\end{multline}
\end{lemma}
\begin{proof}
We will first prove \eqref{eq:hoeqnearhor} with $m=0$ via induction.  The $k=0$ case follows from \eqref{eq:operatorbounded}. Assume \eqref{eq:hoeqnearhor} holds for $k=K$. Then
\begin{multline*}
\partial_r^{K+1}(r^{-3}\rho^2 \mathfrak{L}_{s,\lambda,\kappa}\hpsi)=\partial_r\Big(\Delta \partial_r^{2+K}\hphi+(1+K)\frac{d\Delta}{dr}\partial_r^{K+1}\hphi+2a\partial_{\varphi_*}\partial_r^{K+1}\hphi+\slashed{\Delta}_{\s^2}\partial_r^K\hphi+4M r s \partial_r^{K+1}\hphi\\
-(\lambda \pi_{\leq l}-K(K+1)- 4 K M s)\partial_r^K\hphi-2\sum_{j=0}^{K}\frac{K!}{j!(K-j)!}\partial_r^{K-j}((\mathbbm{h}-1)\Delta)\partial_r^j\slashed{\Delta}_{\s^2}^m\hphi\Big)\\
=\partial_r^{3+K}\hphi+(K+2)\frac{d\Delta}{dr}\partial_r^{K+2}\hphi+2a\partial_{\varphi_*}\partial_r^{K+2}\hphi+\slashed{\Delta}_{\s^2}\partial_r^{K+1}\hphi+4M r s \partial_r^{K+2}\hphi\\
-\left(\lambda\pi_{\leq l}-K(K+1)-\frac{d^2\Delta}{dr^2}(K+1)- 4M (K+1) s\right)\partial_r^{K+1}\hphi\\
-2\sum_{j=0}^{K+1}\frac{(K+1)!}{j!(K+1-j)!}\partial_r^{K+1-j}((\mathbbm{h}-1)\Delta)\partial_r^j\slashed{\Delta}_{\s^2}^m\hphi.
\end{multline*}
Using that $K(K+1)+\frac{d^2\Delta}{dr^2}(K+1)=(K+2)(K+1)$, it therefore follows that \eqref{eq:hoeqnearhor} with $m=0$ holds for $k=K+1$.

The case $m>0$ follows immediately from the fact that $\slashed{\Delta}_{\s^2}$ and $\partial_r^{K+1}(r^{-3}\rho^2 \mathfrak{L}_{s,\lambda,\kappa}(\cdot))$ commute.
\end{proof}

We state now a higher-order version of Proposition \ref{prop:redshift}.
\begin{corollary}
\label{cor:redshiftho}
Let $N_+,m\in \N_0$, ${\re s}>-\kappa_+\left(\frac{1}{2}+N_+\right)$ and $\delta=1+2N_++2\kappa_+^{-1}\min\{{\re s},0\}$. Assume that
\begin{equation*}
\lambda-N_+(N_++1)\gg \frac{N_+^2M^2|s|^2}{1+2N_++2\kappa_+^{-1}\min\{{\re s},0\}}.
\end{equation*}

Then there exists a constant $C_s=C_s(\kappa_+,s)>0$ and $C=C(N_+)>0$ such that:
\begin{multline}
\label{eq:redshiftho}
\sum_{k=0}^{N_+}(1+2N_++2\kappa_+^{-1} \min\{{\re s},0\})R_{\infty}^{2k}\int_{r_+}^{\frac{1}{3}R_{\infty}}\int_{\s^2} (r-M)r^{\delta}|\snabla_{\s^2}^m\partial_r^{k+1}(r^{-1}\hpsi)|^2\\
+r^{-3+\delta}\left[|\snabla_{\s^2}^{m+1}\partial_r^k\hpsi|^2+(\lambda \pi_{\leq l}-k(k+1)) |\snabla_{\s^2}^m\partial_r^k\hpsi|^2\right]\,d\upsigma dr\\
\leq \boxed{C\sum_{k=0}^{N_+}R_{\infty}^{2k-2}\int_{\frac{1}{3}R_{\infty}}^{\frac{1}{2}R_{\infty}}\int_{\s^2}  r^{-1+\delta}\left[|\snabla_{\s^2}^{m+1}\partial_r^k\hpsi|^2+(\lambda \pi_{\leq l}-k(k+1))|\snabla_{\s^2}^{m}\partial_r^k\hpsi|^2\right]\,d\upsigma dr}\\
+C_sR_{\infty}^{2N_+}\int_{r_+}^{\frac{1}{2}R_{\infty}}\int_{\s^2}r^{-1+\delta}|\snabla_{\s^2}^m\partial_r^{N_+}(r^{-3}\rho^2\mathfrak{L}_{s,\lambda,\kappa}\hpsi)|^2\,d\upsigma dr.
\end{multline}
\end{corollary}
\begin{proof}
	We repeat the proof of Proposition \ref{prop:redshift}, but with \eqref{eq:operatorbounded} replaced by the higher-order equation \eqref{eq:hoeqnearhor} and, for $m>0$, we additionally integrate by parts in the $\s^2$ direction. This gives for $k\in \N_0$:
	\begin{multline}
\label{eq:redshiftk}
\int_{r_+}^{\frac{1}{2}R_{\infty}}\int_{\s^2} (r-M)(1+2k+2\kappa_+^{-1} \min\{{\re s},0\})r^{\delta}\chi_{r\leq \frac{1}{3}R_{\infty}}|\snabla_{\s^2}^m\partial_r^{k+1}\hphi|^2\\
+\frac{2}{3}(1+2k+2\kappa_+^{-1}\min\{{\re s},0\})r^{-1+\delta}\chi_{r\leq \frac{1}{3}R_{\infty}}\left[|\snabla_{\s^2}^{m+1}\partial_r^k\hphi|^2+(\lambda \pi_{\leq l}-k(k+1)) |\snabla_{\s^2}^m\partial_r^k\hphi|^2\right]\,d\upsigma dr\\
+\int_{\s^2}  r^{\delta}\left[|\snabla_{\s^2}^{m+1}\partial_r^k\hphi|^2+(\lambda \pi_{\leq l}-k(k+1))|\snabla_{\s^2}^m\partial_r^k\hphi|^2\right]\Big|_{r=r_+}\,d\upsigma\\
\leq C\int_{\frac{1}{3}R_{\infty}}^{\frac{1}{2}R_{\infty}}\int_{\s^2}  r^{-1+\delta}\left[|\snabla_{\s^2}^{m+1}\partial_r^k\hphi|^2+(\lambda \pi_{\leq l}-k(k+1))|\snabla_{\s^2}^{m}\partial_r^k\hphi|^2\right]\,d\upsigma dr\\
+C k^2M^2|s|^2\int_{r_+}^{\frac{1}{2}R_{\infty}}\int_{\s^2} r^{\delta}\chi_{r\leq \frac{1}{3}R_{\infty}} |\snabla_{\s^2}^m\partial_r^k\hphi|^2\,d\upsigma dr+C_kM^2|s|^2\sum_{j=0}^{k-1}\int_{r_+}^{\frac{1}{2}R_{\infty}}\int_{\s^2} r^{\delta}\chi_{r\leq \frac{1}{3}R_{\infty}} |\snabla_{\s^2}^m\partial_r^j\hphi|^2\,d\upsigma dr\\
+C_s\int_{r_+}^{\frac{1}{2}R_{\infty}}\int_{\s^2}r^{-1+\delta}|\snabla_{\s^2}^m\partial_r^k(r^{-3}\rho^2\mathfrak{L}_{s,\lambda,\kappa}\hpsi)|^2\,d\upsigma dr.
\end{multline}

Note that we can absorb the term with a factor $Ck^2M^2|s|^2$ into
\begin{equation*}
+\frac{2}{3}(1+2k+2\kappa_+^{-1}\min\{{\re s},0\})r^{-1+\delta}\chi_{r\leq \frac{1}{3}R_{\infty}}\left[|\snabla_{\s^2}^{m+1}\partial_r^k\hphi|^2+(\lambda \pi_{\leq l}-k(k+1)) |\snabla_{\s^2}^m\partial_r^k\hphi|^2\right]\,d\upsigma dr,
\end{equation*}
if we take $\lambda-k(k+1)\gg \frac{k^2|s|^2}{1+2k+2\kappa_+^{-1}\min\{{\re s},0\}}$. This gives the desired control for $k=N_+$ in the sum on the left-hand side of \eqref{eq:redshiftk}.

In order to control the $k<N_+$ terms in the sum on the left-hand side of \eqref{eq:redshiftk} as well as bound the $k<N+1$ terms on the right-hand side of \eqref{eq:redshiftk}, we will use that we can absorb these lower-order terms into the $k=N_+$ terms by applying a standard Hardy inequality $N_+-k$ times, see for example \cite{aagkerr}[Lemma 2.6, eq. (2.21)], and by taking $\lambda$ to be appropriately large. For example,
\begin{multline*}
\int_{r_+}^{\frac{1}{2}R_{\infty}}\int_{\s^2} \chi_{r\leq \frac{1}{3}R_{\infty}}r^{-1+\delta}|\snabla_{\s^2}\partial_r^{k-1}\hphi|^2\,d\upsigma dr\leq C \int_{r_+}^{\frac{1}{2}R_{\infty}}\int_{\s^2} \chi_{r\leq \frac{1}{3}R_{\infty}}r^{1+\delta}|\snabla_{\s^2}\partial_r^{k}\hphi|^2\,d\upsigma dr\\
+C\int_{\frac{1}{3}R_{\infty}}^{\frac{1}{2}R_{\infty}}\int_{\s^2}r^{-1+\delta}|\snabla_{\s^2}\partial_r^{k-1}\hphi|^2\,d\upsigma dr\\
\leq  C R_{\infty}^2\int_{r_+}^{\frac{1}{2}R_{\infty}}\int_{\s^2} \chi_{r\leq \frac{1}{3}R_{\infty}}r^{-1+\delta}|\snabla_{\s^2}\partial_r^{k}\hphi|^2\,d\upsigma dr+C\int_{\frac{1}{3}R_{\infty}}^{\frac{1}{2}R_{\infty}}\int_{\s^2}r^{-1+\delta}|\snabla_{\s^2}\partial_r^{k-1}\hphi|^2\,d\upsigma dr. \qedhere
\end{multline*}

\end{proof}

In the sections below, we will consider a coupled system of estimates consisting of \eqref{eq:redshiftho} with $k+l\leq N_+$ and $l\geq 1$, and appropriate elliptic estimates in the region $\{r\geq \frac{1}{4}R_{\infty}\}$. In this way, we will be able to absorb the integral over the region $\{\frac{1}{3}R_{\infty}\leq r \leq \frac{1}{2}R_{\infty}\}$ in the boxed term on the right-hand side of \eqref{eq:redshiftho}.

We can in fact improve the estimate in \eqref{cor:redshiftho} in the case of $|a|\ll M$ or under the restriction to fixed azimuthal modes by applying a standard elliptic estimate in $r\leq \frac{1}{2}R_{\infty}$ in order to include higher-order derivative terms on the right-hand side, with weights that degenerate at $r=r_+$. This will be relevant when considering the cut-off resolvent operator.
\begin{proposition}
\label{prop:redshifthopluselliptic}
Let $N_+,m\in \N_0$, ${\re s}>-\kappa_+\left(\frac{1}{2}+N_+\right)$ and $\delta=\frac{2}{3}(1+2N_++2\kappa_+^{-1}\min\{{\re s},0\})$. Assume that
\begin{equation*}
\lambda-N_+(N_++1)\gg \frac{N_+^2M^2|s|^2}{1+2N_++2\kappa_+^{-1}\min\{{\re s},0\}}.
\end{equation*}
Then for $|a|\ll M$ or $\hpsi=\sum_{m=-m_0}^{m_0}\hpsi_m$, with $m_0>0$ and $\Phi\psi_m=im \psi_m$, there exists a constant $C_s=C_s(a, m_0,\kappa_+,s)>0$ and $C=C(N_+,m_0)>0$ such that:
\begin{multline}
\label{eq:redshifthopluselliptic}
\sum_{k=0}^{N_+}(1+2N_++2\kappa_+^{-1} \min\{{\re s},0\})R_{\infty}^{2k}\int_{r_+}^{\frac{1}{3}R_{\infty}}\int_{\s^2} (r-M)r^{\delta}|\snabla_{\s^2}^m\partial_r^{k+1}(r^{-1}\hpsi)|^2+\Delta^2r^{\delta-1}|\snabla_{\s^2}^m\partial_r^{k+2}(r^{-1}\hpsi)|^2\\
+r^{-3+\delta}\left[|\snabla_{\s^2}^{m+2}\partial_r^k\hpsi|^2+\Delta |\snabla_{\s^2}^{m+1}\partial_r^{k+1}\hpsi|^2+(\lambda \pi_{\leq l}-k(k+1))^2 |\snabla_{\s^2}^m\partial_r^k\hpsi|^2\right]\,d\upsigma dr\\
\leq \boxed{C\sum_{k=0}^{N_+}R_{\infty}^{2k}\int_{\frac{1}{3}R_{\infty}}^{\frac{1}{2}R_{\infty}}\int_{\s^2}  r^{-3+\delta}\left[|\snabla_{\s^2}^{m+1}\partial_r^k\hpsi|^2+(\lambda \pi_{\leq l}-k(k+1))|\snabla_{\s^2}^{m}\partial_r^k\hpsi|^2\right]\,d\upsigma dr}\\
+C_sR_{\infty}^{2N_+}\int_{r_+}^{\frac{1}{2}R_{\infty}}\int_{\s^2}r^{-1+\delta}|\snabla_{\s^2}^m\partial_r^{N_+}(r^{-3}\rho^2\mathfrak{L}_{s,\lambda,\kappa}\hpsi)|^2\,d\upsigma dr.
\end{multline}
\end{proposition}
\begin{proof}
We complement the estimates in the proof of Proposition \ref{cor:redshiftho} with a standard elliptic estimate. In the $N_+=0$ case, we integrate
\begin{equation*}
\chi |r^{-2}\rho^{2}\mathfrak{L}_{s,\lambda,\kappa}\hphi+\lambda \pi_{\leq l}\hphi-sr(4 Mr-2(1-\mathbbm{h})\Delta)\partial_r\hphi-r\frac{d\Delta}{dr}\partial_r\hphi-2ar\partial_r\partial_{\varphi_*} \hphi|^2=\chi |r\Delta \partial_r^2\hphi+\slashed{\Delta}_{\s^2}\hphi|^2
\end{equation*}
and integrate by parts the mixed term arising from the square norm on the right-hand side.

We moreover use that the integral of $\chi r^2 a^2|\partial_r\partial_{\varphi_*} \hphi|^2$ can be absorbed into the remaining terms of $|a|$ sufficiently small or for $\hphi$ supported on bounded set of azimuthal modes.
\end{proof}

\subsection{Weighted elliptic estimates in (B)}
Here we will use that
\begin{equation*}
R_{\infty}\gg M,
\end{equation*}
so in the region $\varrho\geq \frac{1}{4}R_{\infty}$, the operator $\mathfrak{L}_{s,\lambda,\kappa}$ is suitably close to the analogous operator in Minkowski (i.e. the $M=0$ case).

To make the elliptic estimates compatible with the Gevrey estimates in \S \ref{sec:gevreylo}, we will need to add a sufficiently degenerate weight to the elliptic estimates: $(\lf+1)^{-w_p}$, where $w_p: [0,r_+^{-1}]\to \R$ was defined in \eqref{eq:defw}. Recall, 
\begin{equation*}
w_p(x)=\begin{cases}
 p  &x\leq  \frac{1 }{ 4R_{\infty}},\\
0  & x\geq \frac{1}{3R_{\infty}}.
\end{cases}
\end{equation*}

We make the following observations regarding $w$:
\begin{itemize}
\item We have that $w\geq 0$ and $w\equiv 0$ in the region $\varrho\leq 3R_{\infty}$ where $\mathfrak{L}_{s,\lambda,\kappa}$ does not agree with the corresponding operator in Minkowski. 
\item By employing standard cut-off functions, we can define $w$ such that
\begin{equation}
\label{eq:derw}
\left|\left|\frac{dw_p}{dx}\right|\right|_{L^{\infty}[0,r_+^{-1}]}\leq C p R_{\infty},
\end{equation}
for some numerical constant $C>0$.
\end{itemize}
From \eqref{eq:operatorgeneral}, it follows that we can split
\begin{align*}
	\mathfrak{L}_{s,\lambda,\kappa}\hpsi=&\:\partial_x(x(x+\kappa \chi_{\varrho \geq R_{\infty}})\partial_x \hpsi)+(\slashed{\Delta}_{\s^2,\vartheta}-\lambda \pi_{\leq l}-\chi_{\varrho \geq R_{\infty}}S_{\lambda} \pi_{\leq l}) \hpsi+\mathfrak{R}_{s,\lambda},\ \textnormal{with}\\
\mathfrak{R}_{s,\lambda}\hpsi:=&\:	2s\partial_x \hpsi+\chi_{\varrho\leq R_{\infty}}(\mathcal{K}_s^{\infty}\hpsi+M x \mathcal{B}_s^{\infty}\hpsi)-\chi_{\varrho\leq \frac{1}{2}R_{\infty}}\mathcal{K}_s^{+}\hpsi.
\end{align*}
We will treat $\mathfrak{R}_{s,\lambda}\hpsi$ as an inhomogeneous term that can be absorbed, because it will come with additional factors of $x$ compared to the terms on the left-hand side of the elliptic estimates and $x$ will be small in the regions under consideration. Note that $\mathfrak{R}_{s,\lambda}\hpsi-2s\partial_x$ contains all deviations from corresponding operator in Minkowski. This will follow from the following lemma:
\begin{lemma}
\label{lm:estresttermelliptic}
Let $k\in \N_0$. Then there exists a numerical constant $C=C(k)>0$, such that
\begin{multline*}
	\int_{\s^2}|\partial_x^k\mathfrak{R}_{s,\lambda}\hpsi|^2\,d\upsigma\leq C\int_{\s^2}\Big[|s|^2|\partial_x^{k+1}\hpsi|^2+M^2\chi_{\varrho\leq R_{\infty}}\Big(x^6|\partial_x^{k+2}\hpsi|^2+x^4|\snabla_{\s^2} \partial_x^{k+1}\hpsi|^2+x^2|\slashed{\Delta}_{\s^2}\partial_x^k\hpsi|^2\\
	+x^2|\partial_{x}^{k+1}\hpsi|^2+|s|^2|\snabla_{\s^2}\partial_x^k\hpsi|^2+(|s|^2+x^2)|\partial_{x}^{k}\hpsi|^2\Big)\Big]\,d\upsigma.
\end{multline*}
\end{lemma}
\begin{proof}
When $\varrho> R_{\infty}$, we simply have that
\begin{equation*}
	\mathfrak{R}_{s,\lambda}\hpsi=2s\partial_x\hpsi,
\end{equation*}
so the estimate is immediate. When $\varrho\leq R_{\infty}$, we observe that the top-order derivative terms in $x \mathcal{B}_s^{\infty}$ and $\mathcal{K}_s^{+}$ have additional powers of $x$ compared to the top-order derivative terms in the remaining terms of $\mathfrak{L}_{s,\lambda,\kappa}$.
\end{proof}
\begin{proposition}
\label{prop:ellipticaway}
Let $0<\delta\leq 1$ and assume the following relations:
\begin{align*}
0<\eta<& \: \frac{1}{4},\\
\kappa\ll &\: \eta R_{\infty}^{-1},\\
M^{-1}R_{\infty}\gg &\: 1,\\
\lambda\gg &\: 1+|s|^2\eta^{-2}R_{\infty}^{2}+p.
\end{align*}
Then, there exists a numerical constant $C>0$ such that:
\begin{multline}
\label{eq:ellipticaway}
\sum_{\ell\in \N_0}\int_{\eta R_{\infty}^{-1}}^{3R_{\infty}^{-1}} x^{1-\delta}(\lf+1)^{-w_p}\Bigg(|\partial_x(x(x+\kappa \chi_{\varrho \geq R_{\infty}})\partial_x\hpsi_{\ell})|^2+|(\slashed{\Delta}_{\s^2}-\lambda\pi_{\leq l}-\chi_{\varrho \geq R_{\infty}}S_{\lambda}\pi_{\leq l})\hpsi_{\ell}|^2\\
+2x(x+\kappa \chi_{\varrho \geq R_{\infty}})\left[|\snabla_{\s^2}\partial_x \hpsi_{\ell}|^2+\lambda \pi_{\leq l}|\partial_x\hpsi_{\ell}|^2+\chi_{\varrho \geq R_{\infty}}|\sqrt{S_{\lambda}} \pi_{\leq l}\partial_x\hpsi_{\ell}|^2\right]\Bigg)\,d\upsigma dx\\
\leq \boxed{C\int_{\frac{1}{4}R_{\infty}}^{\frac{1}{3}R_{\infty}}\int_{\s^2} r^{-2+\delta}\left |(\slashed{\Delta}_{\s^2}-\lambda  \pi_{\leq l})\hpsi\right|\cdot\left|\partial_{r}\hpsi\right|\,d\upsigma dr}\\
+\boxed{C(\lf+1)^{-p} \sum_{\ell\in \N_0}\int_{\s^2}x^{2-\delta}(x+\kappa) \lf(\lf+1)|\hpsi_{\ell}|\cdot|\partial_x\hpsi_{\ell}|\,d\upsigma\Big|_{x=\eta R_{\infty}^{-1}}}\\
+C\sum_{\ell\in \N_0}\int_{\eta R_{\infty}^{-1}}^{4 R_{\infty}^{-1}} x^{1-\delta} (\lf+1)^{-w_p}|(\mathfrak{L}_{s,\lambda,\kappa}\hpsi)_{\ell}|^2\,d\upsigma dx.
\end{multline}
\end{proposition}
\begin{proof}
We split:
\begin{equation*}
\partial_x(x(x+\kappa \chi_{\varrho\geq R_{\infty}})\partial_x\hpsi)+(\slashed{\Delta}_{\s^2}-\lambda\pi_{\leq l}-\chi_{\varrho \geq R_{\infty}}S_{\lambda}\pi_{\leq l})\hpsi=\mathfrak{L}_{s,\lambda,\kappa}\hpsi-\mathfrak{R}_{s,\lambda}\hpsi	
\end{equation*}
Then we consider the $\ell$-th spherical harmonics $\hpsi_{\ell}$, take the square norm of both sides, multiply both sides with the weight factor $x^q (\lf+1)^{-w_p}\chi_{\varrho \geq \frac{1}{3}R_{\infty}}$, with $q\in \R$ and $\chi_{\varrho \geq \frac{1}{3}R_{\infty}}$ a smooth cut-off function which equals 1 when $\varrho\geq \frac{1}{3}R_{\infty}$ and vanishes for $\varrho\leq \frac{1}{4}R_{\infty}$,  and we integrate over $\s^2$ to obtain:
\begin{multline*}
	\overbrace{\int_{\s^2}x^q (\lf+1)^{-w_p}\chi_{\varrho \geq \frac{1}{3}R_{\infty}}\left(|\partial_x(x(x+\kappa \chi_{\varrho \geq R_{\infty}})\partial_x\hpsi_{\ell})|^2+|(\slashed{\Delta}_{\s^2}-\lambda\pi_{\leq l}-\chi_{\varrho \geq R_{\infty}}S_{\lambda}\pi_{\leq l})\hpsi_{\ell}|^2\right)\,d\upsigma}^{=:J_0}\\
	+\int_{\s^2}2x^q (\lf+1)^{-w_p}\chi_{\varrho \geq \frac{1}{3}R_{\infty}}\re\left((\slashed{\Delta}_{\s^2}-\lambda \pi_{\leq l}-\chi_{\varrho \geq R_{\infty}}S_{\lambda} \pi_{\leq l})\overline{\hpsi}_{\ell}\cdot \partial_x(x(x+\kappa \chi_{\varrho \geq R_{\infty}})\partial_x\hpsi_{\ell})\right)\,d\upsigma\\
	=\int_{\s^2}x^q (\lf+1)^{-w_p}\chi_{\varrho \geq \frac{1}{3}R_{\infty}}|(\mathfrak{L}_{s,\lambda,\kappa}\hpsi)_{\ell}-(\mathfrak{R}_{s,\lambda}\hpsi	)_{\ell}|^2\,d\upsigma.
\end{multline*}
For the sake of convenience, we will restrict to $\ell\leq l$. The $\ell>l$ case proceeds entirely analogously, but without the $\lambda$ and $S_{\lambda}$ terms.

We integrate by parts on $\s^2$ to obtain:
\begin{multline*}
	\int_{\s^2}2x^q (\lf+1)^{-w_p} \chi_{\varrho\geq \frac{1}{3} R_{\infty}} \re\left((\slashed{\Delta}_{\s^2}-\lambda-\chi_{\varrho \geq R_{\infty}}S_{\lambda} )\overline{\hpsi}_{\ell}\cdot \partial_x(x(x+\kappa \chi_{\varrho \geq R_{\infty}})\partial_x\hpsi_{\ell}\right)\,d\upsigma\\
	=\overbrace{2\int_{\s^2}\partial_x\left[x^q(\lf+1)^{-w_p} \chi_{\varrho\geq \frac{1}{3} R_{\infty}} \re\left((\slashed{\Delta}_{\s^2}-\lambda-\chi_{\varrho \geq R_{\infty}}S_{\lambda})\overline{\hpsi}_{\ell}\cdot x(x+\kappa \chi_{\varrho \geq R_{\infty}})\partial_x\hpsi_{\ell}\right)\right]\,d\upsigma}^{=:J_1}\\
+\overbrace{2\int_{\s^2}x^q (\lf+1)^{-w_p}x(x+\kappa \chi_{\varrho \geq R_{\infty}}) \chi_{\varrho\geq \frac{1}{3} R_{\infty}} \left[|\snabla_{\s^2}\partial_x \hpsi_{\ell}|^2+\lambda|\partial_x\hpsi_{\ell}|^2+\chi_{\varrho \geq R_{\infty}}|\sqrt{S_{\lambda}}\partial_x\hpsi_{\ell}|^2\right]\,d\upsigma}^{=:J_2}\\
\overbrace{-2\int_{\s^2}\partial_x(x^q (\lf+1)^{-w_p}) x(x+\kappa \chi_{\varrho\geq R_{\infty}}) \chi_{\varrho\geq \frac{1}{3} R_{\infty}} \re\left((\slashed{\Delta}_{\s^2}-\lambda-\chi_{\varrho \geq R_{\infty}}S_{\lambda})\overline{\hpsi}_{\ell}\cdot\partial_x\hpsi_{\ell}\right)\,d\upsigma}^{=:J_3}\\
+\overbrace{2\int_{\s^2}x^q (\lf+1)^{-w_p} x(x+\kappa \chi_{\varrho\geq R_{\infty}})\partial_x(\chi_{\varrho\geq R_{\infty}}) \chi_{\varrho\geq \frac{1}{3} R_{\infty}} \re(S_{\lambda}\overline{\hpsi}_{\ell}\cdot\partial_x\hpsi_{\ell})\,d\upsigma}^{=:J_4}\\
-\overbrace{2\int_{\s^2}x^q (\lf+1)^{-w_p}x(x+\kappa \chi_{\varrho\geq R_{\infty}})\partial_x(\chi_{\varrho\geq \frac{1}{3} R_{\infty}})\re\left((\slashed{\Delta}_{\s^2}-\lambda-\chi_{\varrho \geq R_{\infty}}S_{\lambda})\overline{\hpsi}_{\ell}\cdot\partial_x\hpsi_{\ell}\right)\,d\upsigma}^{=:J_5}.
\end{multline*}
Note first that $J_2$ is non-negative definite. We will show that we can absorb $J_3$ and $J_4$ into $J_0$ and $J_2$ by taking $\lambda$ suitably large.

We apply Young's inequality together with \eqref{eq:derw} to obtain for $x\in [\eta R_{\infty}^{-1},4 R^{-1}_{\infty}]$:
\begin{multline*}
|J_3|\leq\\
 \int_{\s^2} \left(|q|x^{-1}+  \left|\frac{dw_p}{dx}\right|(\lf+1)^{-1}\right)x^{q+1}(x+\kappa\chi_{\varrho \geq R_{\infty}} ) (\lf+1)^{-w_p}(\ell(\ell+1)+\lambda+2l\ell\chi_{\varrho \geq R_{\infty}} )|\hpsi_{\ell}\|\partial_x\hpsi_{\ell}|\,d\upsigma\\
\leq \left(\frac{|q| }{ \lf x}+C pR_{\infty} (\lf+1)^{-1}\lf^{-1}\right) \lf  \int_{\s^2}x^{q+1}(x+\kappa\chi_{\varrho \geq R_{\infty}} )    (\lf+1)^{-w_p} (\ell(\ell+1)+\lambda+2l\ell\chi_{\varrho \geq R_{\infty}} )|\hpsi_{\ell}\|\partial_x\hpsi_{\ell}|\,d\upsigma\\
\leq \frac{1}{4} \int_{\s^2} x^q (\lf+1)^{-w_p}|(\slashed{\Delta}_{\s^2}-\lambda-\chi_{\varrho \geq R_{\infty}}S_{\lambda})\hpsi_{\ell}|^2\,d\upsigma\\
+  C\left(\frac{|q| }{ \lf x}+C pR_{\infty} (\lf+1)^{-1}\lf^{-1}\right)^2x(x+\kappa \chi_{\varrho \geq R_{\infty}}) \int_{\s^2}x^{q}(\lf+1)^{-w_p}x(x+\kappa \chi_{\varrho \geq R_{\infty}})   \lf^2 |\partial_x\hpsi_{\ell}|^2\,d\upsigma.
\end{multline*}
The first term on the very  right-hand side above can be absorbed into $J_0$. The second term on the right-hand side can be absorbed into $J_2$ by taking $ \kappa\ll \eta R_{\infty}^{-1}$ and $\lambda\gg q^2+p$ to obtain:
\begin{equation*}
 C\left(\frac{|q| }{ \lf x}+C pR_{\infty} \lf^{-2}\right)^2x(x+\kappa \chi_{\varrho \geq R_{\infty}})=C\left(|q|  \lf^{-1}+C px R_{\infty} \lf^{-2}\right)^2(1+\kappa x^{-1}\chi_{\varrho \geq R_{\infty}})\leq \frac{1}{4}.
\end{equation*}
and noting that
\begin{equation*}
\lf(\lf+1)=\ell(\ell+1)+l(l+1)+2l\ell\leq 2(\ell(\ell+1)+2l(l+1)).
\end{equation*}
We conclude that $|J_3|\leq \frac{1}{4}J_0+\frac{1}{4}J_2$.

To estimate $|J_4|$, we first note that there exists a constant $C>0$ such that $\left|\frac{d\chi_{\varrho \geq R_{\infty}}}{d\varrho}\right|\leq C R_{\infty}^{-1}$, so that
\begin{equation*}
\left|\frac{d\chi_{\varrho \geq R_{\infty}}}{dx}\right|\leq C R_{\infty}.
\end{equation*}
By repeating the estimates for $|J_3|$ and taking $\lambda\gg M^{-1}R_{\infty}$, we conclude that  $|J_4|\leq \frac{1}{4}J_0+\frac{1}{4}J_2$. We conclude that
\begin{equation*}
\sum_{j=0}^4J_j\geq \frac{1}{2}(J_0+J_2).
\end{equation*}

We can estimate $J_5$ by observing that $\partial_x\chi_{\varrho\geq \frac{1}{3}R_{\infty}}$ is supported in $\frac{1}{4}R_{\infty}\leq \varrho\leq \frac{1}{3}R_{\infty}$, where $\varrho=r$, $w(x)=0$ and there exists a numerical constant $C>0$ such that $|\partial_x\chi_{\varrho\geq \frac{1}{3}R_{\infty}}|\leq CR_{\infty}$. We obtain:
\begin{equation*}
|J_5|\leq \int_{\s^2} r^{-q+1}|(\slashed{\Delta}_{\s^2}-\lambda) \hpsi_{\ell}|\cdot|\partial_r\hpsi_{\ell}|\,d\upsigma
\end{equation*}
and $\supp J_5\subset \{\frac{1}{4}R_{\infty}\leq r\leq \frac{1}{3}R_{\infty}\}$.

We apply Young's inequality to estimate:
\begin{multline*}
\int_{\s^2}x^q (\lf+1)^{-w_p}|(\mathfrak{L}_{s,\lambda,\kappa}\hpsi)_{\ell}-(\mathfrak{R}_{s,\lambda}\hpsi)_{\ell}|^2\,d\upsigma\\
\leq 2\int_{\s^2}x^q (\lf+1)^{-w_p}|(\mathfrak{L}_{s,\lambda,\kappa}\hpsi)_{\ell}|^2\,d\upsigma+2\int_{\s^2}x^q(\lf+1)^{-w_p}|(\mathfrak{R}_{s,\lambda}\hpsi)_{\ell}|^2\,d\upsigma.
\end{multline*}

Note that $(\mathfrak{R}_{s,\lambda}\hpsi)_{\ell}$ does not need to agree with $(\mathfrak{R}_{s,\lambda}(\hpsi_{\ell}))$. However, since $w\equiv 0$ in $\{\varrho\leq  2R_{\infty}\}$ and $\mathfrak{R}_{s,\lambda}\hpsi=2s\partial_x \hpsi$ in $r\geq R_{\infty}$, we can conclude that:
\begin{equation*}
\sum_{\ell\in \N_0}\int_{\s^2}(\lf+1)^{-w_p}|(\mathfrak{R}_{s,\lambda}\hpsi)_{\ell}|^2\,d\upsigma=\begin{cases}\int_{\s^2}|\mathfrak{R}_{s,\lambda}\hpsi|^2\,d\upsigma\quad &\textnormal{if}\quad x\geq \frac{1}{2R_{\infty}},\\
\sum_{\ell\in \N_0}\int_{\s^2}(\lf+1)^{-w_p}|\mathfrak{R}_{s,\lambda}(\hpsi_{\ell})|^2\,d\upsigma\quad &\textnormal{if}\quad x<\frac{1}{2R_{\infty}}.
\end{cases}
\end{equation*}
By applying Lemma \ref{lm:estresttermelliptic}, we can take $MR_{\infty}^{-1}\ll 1$ and $\lambda\gg |s|^2\eta^{-2}R_{\infty}^2$ to obtain:
\begin{equation*}
2\sum_{\ell\in \N_0}\int_{\s^2}x^q (\lf+1)^{-w_p}|(\mathfrak{R}_{s,\lambda}\hpsi)_{\ell}|^2\,d\upsigma\leq \frac{1}{4}\sum_{\ell\in \N_0}(J_0+J_2).
\end{equation*}
We now integrate over $x$ to obtain:
\begin{multline*}
\frac{1}{4}\sum_{\ell\in \N_0}\int_{x(\eta^{-1})}^{x(\frac{1}{4}R_{\infty})} x^q (\lf+1)^{-w_p}\Bigg(|\partial_x(x(x+\kappa \chi_{\varrho \geq R_{\infty}})\partial_x\hpsi_{\ell})|^2+|(\slashed{\Delta}_{\s^2}-\lambda-\chi_{\varrho \geq R_{\infty}}S_{\lambda})\hpsi_{\ell}|^2\\
+2x(x+\kappa \chi_{\varrho \geq R_{\infty}})\left[|\snabla_{\s^2}\partial_x \hpsi_{\ell}|^2+\lambda|\partial_x\hpsi_{\ell}|^2+\chi_{\varrho \geq R_{\infty}}|\sqrt{S_{\lambda}}\partial_x\hpsi_{\ell}|^2\right]\Bigg)\,d\upsigma dx\\
\leq C\int_{\frac{1}{4}R_{\infty}}^{\frac{1}{3}R_{\infty}}\int_{\s^2} r^{-q-1}\left |(\slashed{\Delta}_{\s^2}-\lambda)\hpsi\right|\cdot\left|\partial_{r}\hpsi\right|\,d\upsigma dr\\
+2(\lf+1)^{-p} \sum_{\ell\in \N_0}\int_{\s^2}x^{q+1}(x+\kappa) \lf(\lf+1)|\hpsi_{\ell}|\cdot|\partial_x\hpsi_{\ell}|\,d\upsigma\Big|_{x=\eta R_{\infty}^{-1}}\\
+2\sum_{\ell\in \N_0}\int_{x(\eta^{-1})}^{x(\frac{1}{4}R_{\infty})} x^q (\lf+1)^{-w_p}|(\mathfrak{L}_{s,\lambda,\kappa}\hpsi)_{\ell}|^2\,d\upsigma dx,
\end{multline*}
and we conclude that \eqref{eq:ellipticaway} holds after taking $q=1-\delta$.
\end{proof}

We will also consider higher-order analogues of \eqref{eq:ellipticaway}. For this it is useful to define:
\begin{equation*}
\hpsi^{(n)}:=\partial_x^n\hpsi.
\end{equation*}
Then it can be easily verified that:
\begin{multline*}
\partial_x^n(\partial_x(x(x+\kappa \chi_{\varrho \geq R_{\infty}})\partial_x\hpsi)+(\slashed{\Delta}_{\s^2}-\lambda  \pi_{\leq l}-\chi_{\varrho \geq R_{\infty}}S_{\lambda}  \pi_{\leq l})\hpsi)=\partial_x(x(x+\kappa  \chi_{\varrho \geq R_{\infty}})\partial_x\hpsi^{(n)}\\
+(\slashed{\Delta}_{\s^2}-\lambda  \pi_{\leq l}-\chi_{\varrho \geq R_{\infty}}S_{\lambda}  \pi_{\leq l})\hpsi^{(n)}+2(k+1)\left(x+\frac{1}{2}\kappa  \chi_{\varrho \geq R_{\infty}}\right)\partial_x\hpsi^{(n)} \\
+\kappa \sum_{k=0}^{n+1}{n+1\choose k}\partial_x^{n+1-k}(x\chi_{\varrho \geq R_{\infty}})\partial_x\hpsi^{(k)}-\sum_{k=0}^{n}{n\choose k}\partial_x^{n-k}(\chi_{\varrho \geq R_{\infty}})S_{\lambda}\hpsi^{(k)}.
\end{multline*}
Define
\begin{multline*}
\mathfrak{R}^{(n)}_{s,\lambda}:=\partial_x^n	\mathfrak{R}_{s,\lambda}-2(k+1)\left(x+\frac{1}{2}\kappa  \chi_{\varrho \geq R_{\infty}}\right)\partial_x\hpsi^{(n)} -\kappa \sum_{k=0}^{n+1}{n+1\choose k}\partial_x^{n+1-k}(x\chi_{\varrho \geq R_{\infty}})\partial_x\hpsi^{(k)}\\
+\sum_{k=0}^{n}{n\choose k}\partial_x^{n-k}(\chi_{\varrho \geq R_{\infty}})\pi_{\leq l}S_{\lambda}\hpsi^{(k)}.
\end{multline*}
Then
\begin{equation*}
\partial_x^n(\mathfrak{L}_{s,\lambda,\kappa}	\hpsi)=\partial_x(x(x+\kappa  \chi_{\varrho \geq R_{\infty}})\partial_x\hpsi^{(n)})+(\slashed{\Delta}_{\s^2}-\lambda\pi_{\leq l}-\chi_{\varrho \geq R_{\infty}}S_{\lambda}\pi_{\leq l})\hpsi^{(n)}+\mathfrak{R}^{(n)}_{s,\lambda}.
\end{equation*}

We obtain the following analogue of Proposition \ref{eq:ellipticaway}, including additional commutation with vector fields $(x^2\partial_x)^k$.
\begin{corollary}
\label{cor:ellipticawayho}
	Let $N\in \N_0$, $0<\delta<1$ and $p\geq 0$ and assume the following smallness relations:
\begin{align*}
0<\eta<& \: \frac{1}{4},\\
\kappa\ll &\: \eta R_{\infty}^{-1},\\
M^{-1}R_{\infty}\gg &\: 1,\\
\lambda-N(N+1)\gg &\: 1+|s|^2\eta^{-2}R_{\infty}^{2}+p.
\end{align*}
Then, there exists a numerical constant $C=C(N)>0$ such that:
\begin{multline}
\label{eq:ellipticawayho}
\sum_{k=0}^N\sum_{\ell\in \N_0}\int_{x(\eta^{-1})}^{3R_{\infty}^{-1}} x^{1-\delta} (\lf+1)^{-w_p}\Bigg(|\partial_x(x(x+\kappa \chi_{\varrho \geq R_{\infty}})\partial_x(x^2\partial_x)^k\hpsi_{\ell})|^2+|(\slashed{\Delta}_{\s^2}-\lambda \pi_{\leq l}-\chi_{\varrho \geq R_{\infty}}S_{\lambda}\pi_{\leq l})(x^2\partial_x)^k\hpsi_{\ell}|^2\\
+2x(x+\kappa \chi_{\varrho \geq R_{\infty}})\left[|\snabla_{\s^2}\partial_x (x^2\partial_x)^k\hpsi_{\ell}|^2+\lambda \pi_{\leq l}|\partial_x(x^2\partial_x)^k\hpsi_{\ell}|^2+\chi_{\varrho \geq R_{\infty}}|\sqrt{S_{\lambda} \pi_{\leq l}}\partial_x(x^2\partial_x)^k\hpsi_{\ell}|^2\right]\Bigg)\,d\upsigma dx\\
\leq\sum_{k=0}^N \boxed{C\int_{\frac{1}{4}R_{\infty}}^{\frac{1}{3}R_{\infty}}\int_{\s^2} r^{-2+\delta}\left |(\slashed{\Delta}_{\s^2}-\lambda \pi_{\leq l})\partial_r^k\hpsi\right|\cdot\left|\partial_{r}^{k+1}\hpsi\right|\,d\upsigma dr}\\
+\boxed{C(\lf+1)^{-p} \sum_{\ell\in \N_0}\int_{\s^2}x^{2+4k-\delta}(x+\kappa) \lf(\lf+1)|\hpsi_{\ell}^{(k)}|\cdot|\partial_x\hpsi_{\ell}^{(k)}|\,d\upsigma\Big|_{x=\eta R_{\infty}^{-1}}}\\
+C\sum_{\ell\in \N_0}\int_{x(\eta^{-1})}^{4R_{\infty}^{-1}} x^{1-\delta} (\lf+1)^{-w_p}|(x^2\partial_x)^k(\mathfrak{L}_{s,\lambda,\kappa}\hpsi)_{\ell}|^2\,d\upsigma dx.
\end{multline}
\end{corollary}
\begin{proof}
	We repeat the proof of Proposition \ref{prop:ellipticaway} with $\hpsi^{(k)}$ replacing $\hpsi$ and with $q=1-\delta+4k$, sum over $k$, and use that all the terms coming from $\mathfrak{R}^{(n)}_{s,\lambda}$ can be absorbed in the terms $J_0$ and $J_2$.
\end{proof}

\subsection{Combining the estimates in (A) and (B)}
In this section we combine the estimates from Corollary \ref{cor:redshiftho} and Corollary \ref{cor:ellipticawayho} by taking $\lambda$ suitably large.

\begin{corollary}
\label{cor:fullellipticest}
	Let $N_+\in \N_0$, $0<\delta<1$, $p\geq0$, and ${\re s}>-\kappa_+\left(\frac{1}{2}+N_+\right)$, and assume the following smallness relations:
\begin{align*}
0<\eta<& \: \frac{1}{4},\\
\kappa\ll &\: \eta R_{\infty}^{-1},\\
M^{-1}R_{\infty}\gg &\: 1,\\
\lambda-N_+(N_++1)\gg &\: p+N_+|s|^2(1+2N_++2\kappa_+^{-1}{\min\{\re s,0\}})^{-1}+(1+2N_++2\kappa_+^{-1}{\min\{\re s,0\}})^{-2}\\
&+|s|^2\eta^{-2}R_{\infty}^{2}+(M^{-1}R_{\infty})^{2(1+N_+)}.
\end{align*}
Then, there exists a constant $C=C(N_+)>0$ such that:
\begin{multline}
\label{eq:fullellipticest}
\sum_{k=0}^{N_+} \int_{r_+}^{\frac{1}{2}R_{\infty}} |\partial_r^{1+k}\snabla_{\s^2}\hpsi|^2+\lambda|\partial_r^{1+k}\hpsi|^2+|(\slashed{\Delta}_{\s^2}-\lambda  \pi_{\leq l})\hpsi|^2+\chi_{r\geq \frac{1}{3}R_{\infty}}|\partial_r^{2+k}\psi|^2\,d\upsigma dr\\
+\sum_{k=0}^{N_+}\sum_{\ell\in \N_0}\int_{\eta R_{\infty}^{-1}}^{x(\frac{1}{2}R_{\infty})} x^{1-\delta} (\lf+1)^{-w_p}\Bigg(|\partial_x(x(x+\kappa \chi_{\varrho \geq R_{\infty}})\partial_x(x^2\partial_x)^k\hpsi_{\ell})|^2+|(\slashed{\Delta}_{\s^2}-\lambda  \pi_{\leq l}-\chi_{\varrho \geq R_{\infty}}S_{\lambda} \pi_{\leq l})(x^2\partial_x)^k\hpsi_{\ell}|^2\\
+2x(x+\kappa \chi_{\varrho \geq R_{\infty}})\left[|\snabla_{\s^2}\partial_x (x^2\partial_x)^k\hpsi_{\ell}|^2+\lambda  \pi_{\leq l}|\partial_x(x^2\partial_x)^k\hpsi_{\ell}|^2+\chi_{\varrho \geq R_{\infty}}|\sqrt{S_{\lambda}  \pi_{\leq l}}\partial_x(x^2\partial_x)^k\hpsi_{\ell}|^2\right]\Bigg)\,d\upsigma dx\\
\leq \boxed{C(\lf+1)^{-p}\sum_{k=0}^{N_+} \sum_{\ell\in \N_0}\int_{\s^2}x^{2+4k-\delta}(x+\kappa) \lf(\lf+1)|\hpsi_{\ell}^{(k)}|\cdot|\partial_x\hpsi_{\ell}^{(k)}|\,d\upsigma\Big|_{x=\eta R_{\infty}^{-1}}}\\
+C_s \sum_{k=0}^{N_+} \int_{r_+}^{\frac{1}{2}R_{\infty}}  |\partial_r^k\mathfrak{L}_{s,\lambda,\kappa}\hpsi|^2+\chi_{r\leq \frac{1}{3}R_{\infty}}(|\snabla_{\s^2}\partial_r^k\mathfrak{L}_{s,\lambda,\kappa}\hpsi|^2+\lambda  \pi_{\leq l}|\partial_r^k\mathfrak{L}_{s,\lambda,\kappa}\hpsi|^2)\,d\upsigma dr\\
+C_s\sum_{k=0}^{N_+} \sum_{\ell\in \N_0}\int_{\eta R_{\infty}^{-1}}^{x(\frac{1}{2}R_{\infty})} x^{1-\delta}(\lf+1)^{-w_p}|(x^2\partial_x)^k(\mathfrak{L}_{s,\lambda,\kappa}\hpsi)_{\ell}|^2\,d\upsigma dx.
\end{multline}
Furthermore, for suitably small $\frac{|a|}{M}$ or for $\hphi=\sum_{m=-m_0}^{m_0}\hphi_m$, with $m_0>0$, we can omit the $|\snabla_{\s^2}\partial_r^k\mathfrak{L}_{s,\lambda,\kappa}\hpsi|^2$ term on the right-hand side.
\end{corollary}
\begin{proof}
We apply Corollary \ref{cor:redshiftho} with arbitrary $N_+\in \N_0$ and $m\in \{0,1\}$ and we also apply Corollary \ref{cor:ellipticawayho} with the same value of $N_+$. We estimate the first term on the right-hand side of \eqref{eq:ellipticawayho} as follows: let $\epsilon>0$ be arbitrarily small, then we apply Young's inequality and use that $\hpsi=r\phi$ in the region under consideration to obtain:
\begin{equation*}
C\sum_{k=0}^{N_+}\int_{\s^2}	 r^{-2+\delta}\left |(\slashed{\Delta}_{\s^2}-\lambda  \pi_{\leq l})\partial_r^k\hpsi\right|\cdot\left|\partial_{r}^{k+1}\hpsi\right|\,d\upsigma \leq \sum_{k=0}^{N_+} \int_{\s^2}	 \epsilon r^{-1+\delta}\left |(\slashed{\Delta}_{\s^2}-\lambda  \pi_{\leq l})\partial_r^k\hphi\right|^2+\frac{C^2}{4\epsilon \lambda}\lambda r^{1+\delta}\left|\partial_{r}^{k+1}\hphi\right|^2\,d\upsigma.
\end{equation*}
We take $\lambda$ suitably large so that $\frac{ C^2}{4\epsilon \lambda}<\epsilon$ and $\epsilon\ll (1+2N_++2\kappa_+^{-1}{\re s})$. Then we can absorb the right-hand side above into the left-hand side of \eqref{eq:redshiftho} with $m=0$, and $m=1$, after multiplying the $m=0$ estimates by a factor $\lambda$.

By taking $\epsilon$ suitably small compared to the constant $C R_{\infty}^{2N_+}$ appearing on the right-hand side of \eqref{eq:redshiftho}, we can also absorb the first integral on the right-hand side of \eqref{eq:redshiftho} into the left-hand side of \eqref{eq:ellipticawayho}.
\end{proof}

\begin{remark}
	Note that the estimate \eqref{eq:fullellipticest} still has $\hpsi$ and its derivatives appearing in the boxed boundary term at $x=\eta R_{\infty}^{-1}$. In the next section we will show that we can absorb this term into the left-hand side of appropriate Gevrey-type estimates.
\end{remark}

\section{Gevrey estimates}
\label{sec:gevrey}
In this section, we consider estimates in the region $x\leq x_0$, where we define:
\begin{equation*}
x_0:=\frac{1}{2R_{\infty}}.
\end{equation*}
 By \eqref{eq:operatorfaraway}, we can express in the region $x\leq x_{0}$:
\begin{equation*}
(\mathfrak{L}_{s,\lambda,\kappa}\hpsi)_{\ell}=\partial_{x}(x(\kappa+x)\partial_{x}\hpsi_{\ell})+2 s\partial_{x}\hpsi_{\ell}-\mathfrak{l}(\mathfrak{l}+1)\hpsi_{\ell}.
\end{equation*}
Hence, \textbf{all the estimates in this section are estimates for the rescaled Laplace-transformed wave operator in Minkowski}, restricted to fixed spherical harmonic modes with angular frequencies $\geq \mathfrak{l}$. We will moreover require these estimates to be uniform in $\lf$.

Denote $\hpsi^{(n)}=\partial_x^n\hpsi$. Then it can be easily verified that:
\begin{equation}
\label{eq:commeq}
(\partial_x^n\mathfrak{L}_{s,\lambda,\kappa}\hpsi)_{\ell}=x(\kappa+x)\partial_{x}\hpsi^{(n+1)}_{\ell}+2(n+1) \left(x+\frac{\kappa}{2}\right)\partial_x \hpsi^{(n)}+2 s\partial_{x}\hpsi^{(n)}_{\ell}+(n(n+1)-\mathfrak{l}(\mathfrak{l}+1))\hpsi_{\ell}^{(n)}.
\end{equation}
We will estimate higher-order energy norms of $\hpsi$ in terms of similar energy norms of $\mathfrak{L}_{s,\lambda,\kappa}\hpsi$ in two steps.
\begin{itemize}
\item[\bf Step 1:]	We consider $\hpsi^{(n)}_{\ell}=\partial_x^n\hpsi_{\ell}$ and derive $\dot{H}^1$ estimates for $\hpsi^{(n)}_{\ell}$ with $n\geq \lf$. We then sum the estimates for different values of $n$ to obtain an $L^2$-based Gevrey-type estimate. \emph{Here it is important that we will not see any terms $\hpsi^{(n)}_{\ell}$ with $n<\lf$.}
\item[\bf Step 2:] We estimate weighted $\dot{H}^1$ norms of $\hpsi^{(n)}_{\ell}$ with $n<\lf$ in terms of a weighted $\dot{H}^1$ norm of $\hpsi^{(n)}_{\ell}$. To be able to sum these lower-order estimates, we apply elliptic-type estimates in the region $x\leq x_0$.
\end{itemize}

\subsection{Step 1: higher-order derivatives}
\label{sec:gevreyho}
We will first derive weighted $L^2$ estimates for $\hpsi^{(n)}$ with $n\geq \lf$. We obtain these estimates by simply rearranging the terms in \eqref{eq:commeq} to obtain:
\begin{equation}
\label{eq:mainidentityforgevrey}
x(\kappa+x)\partial_{x}\hpsi^{(n+1)}_{\ell}+2(n+1) \left(x+\frac{\kappa}{2}\right)\partial_x \hpsi_{\ell}^{(n)}+2 s\partial_{x}\hpsi^{(n)}_{\ell}=(\mathfrak{l}(\mathfrak{l}+1)-n(n+1))\hpsi_{\ell}^{(n)}+(\partial_x^n\mathfrak{L}_{s,\lambda,\kappa}\hpsi)_{\ell}.
\end{equation}
and by taking the square norm on both sides, treating the right-hand side as an inhomogeneous term. By considering appropriate weighted sums in $n$ of these $L^2$ estimates, we  will then be able to absorb the terms arising from $(\mathfrak{l}(\mathfrak{l}+1)-n(n+1))\hpsi_{\ell}^{(n)}$ on the right-hand side of \eqref{eq:mainidentityforgevrey}.
\begin{proposition}
\label{prop:pregev}
Let $\kappa\geq 0$, $M=1$, $\ell,l\in \N_0$, $N_{\infty}\in \N_0$, with $N_{\infty}>\lf$, and let $\alpha,\beta,\nu,\tilde{\sigma},\epsilon >0$, with $\mu\in (0,1)$. Then there exists a constant $C>0$ such that:
\begin{multline}
\label{eq:pregev}
\sum_{n=\lf}^{N_{\infty}}\frac{|2s|^2\tilde{\sigma}^{2n}}{n!^2(n+1)!^2}\int_0^{x_0}(1-\mu \nu^{-1}) x^2(x+\kappa)^2 |\hat{\psi}_{\ell}^{(n+2)}|^2\,dx\\
+\sum_{n=\lf}^{N_{\infty}}\frac{|2s|^2\tilde{\sigma}^{2n}}{n!^2(n+1)!^2}\int_0^{x_0} \left[4-\beta-(1-\mu)(\alpha+\tilde{\sigma}^{-2}\alpha^{-1})-6(n+1)^{-1}\right]\\
\times(n+1)^2 \left(x+\frac{1}{2}\kappa\right)^2|\hat{\psi}_{\ell}^{(n+1)}|^2\,dx\\
+\sum_{n=\lf}^{N_{\infty}}\frac{|2s|^2\tilde{\sigma}^{2n}}{n!^2(n+1)!^2}\int_0^{x_0} \left(1-\mu\nu-4\beta^{-1}\left(\frac{\max\{-\re s,0\}}{|s|}\right)^2-(1+\epsilon)\tilde{\sigma}^2\left(1+2(n-1)^{-1}\right)^2\right)\\
\times |2s|^2|\hat{\psi}_{\ell}^{(n+1)}|^2\,dx\\
\boxed{-\frac{|2s|^{2(N_{\infty}+1)}\tilde{\sigma}^{2(N_{\infty}+1)}}{(N_{\infty}+1)!^2(N_{\infty}+2)!^2}\int_0^{x_0}(1-\mu)\alpha^{-1}\tilde{\sigma}^{-2} (N_{\infty}+2)^2\left(x+\frac{1}{x^{-1}+2\kappa^{-1}}\right)^2|\hat{\psi}_{\ell}^{(N_{\infty}+2)}|^2\,dx}\\
+\sum_{n=\lf}^{N_{\infty}}\frac{|2s|^2\tilde{\sigma}^{2n}}{n!^2(n+1)!^2}x_0^3(n+1) |\hpsi^{(n+1)}|^2\big|_{x=x_0}\\
\leq C(1+\epsilon^{-1}) \sum_{n=\lf}^{N_{\infty}}\frac{|2s|^2\tilde{\sigma}^{2n}}{n!^2(n+1)!^2}\int_0^{x_0}|\partial_x^n(\mathfrak{L}_{s,\lambda,\kappa}\hpsi)_{\ell}|^2\,dx.
\end{multline}
\end{proposition}
\begin{proof}
We will omit the subscript $\ell$ in $\hpsi_{\ell}$ for the sake of convenience. Rearranging \eqref{eq:commeq} and taking the square norm on both sides, we obtain:
\begin{multline*}
x^2(x+\kappa)^2|\hat{\psi}^{(n+2)}|^2+4(n+1)^2\left(x+\frac{1}{2}\kappa\right)^2|\hat{\psi}^{(n+1)}|^2+4|s|^2|\hat{\psi}^{(n+1)}|^2\\
+T_1+T_2+T_3=\left|[\lf(\lf+1)-n(n+1)]\hat{\psi}^{(n)}+\partial_x^n (\mathfrak{L}_{s,\lambda,\kappa}\hpsi)_{\ell}\right|^2,
\end{multline*}
with
\begin{align*}
T_1:=&\:2(n+1) x\left(x+\frac{1}{2}\kappa\right)(x+\kappa)\partial_x(|\hat{\psi}^{(n+1)}|^2),\\
T_2:=&\:4 \re(2s)(n+1)\left(x+\frac{1}{2}\kappa\right)|\hat{\psi}^{(n+1)}|^2,\\
T_3:=&\:2x(x+\kappa)\re\left(\hat{\psi}^{(n+2)}2\overline{s}\overline{\hat{\psi}^{(n+1)}} \right).\\
\end{align*}

We integrate $T_1$ by parts to obtain:
\begin{multline*}
\int_0^{x_0}T_1\,dx\geq 2(n+1) x_0\left(x_0+\frac{1}{2}\kappa\right)(x_0+\kappa)|\hat{\psi}^{(n+1)}|^2|_{x=x_0}\\
 -\int_{0}^{x_0}(n+1)(6x^2+6\kappa x+\kappa^2)|\hat{\psi}^{(n+1)}|^2\,dx\\
\geq 2(n+1)x_0^3 |\hat{\psi}^{(n+1)}|^2|_{x=x_0}-\int_{0}^{x_0}6(n+1)\left(x+\frac{1}{2}\kappa\right)^2|\hat{\psi}^{(n+1)}|^2\,dx.
\end{multline*}
If ${\re s}\geq0$, then $T_2\geq 0$. If ${\re s}<0$, then we introduce a constant $\beta>0$ and apply Young's inequality to split:
\begin{equation*}
|T_2|\leq \beta (n+1)^2\left(x+\frac{1}{2}\kappa\right)^2 |\hat{\psi}^{(n+1)}|^2+ 4\beta^{-1}\left(\frac{|\re s|}{|s|}\right)^2|2s|^2 |\hat{\psi}^{(n+1)}|^2.
\end{equation*}
We introduce additionally $0\leq \mu\leq 1$ and split
\begin{equation*}
\begin{split}
|T_3|\leq 2x(x+\kappa)|2s| |\hat{\psi}^{(n+1)}\|\hat{\psi}^{(n+2)}|=&\:\mu \cdot 2x(x+\kappa)|2s| |\hat{\psi}^{(n+1)}\|\hat{\psi}^{(n+2)}|\\
&+(1-\mu)\cdot 2x(x+\kappa)|2s| |\hat{\psi}^{(n+1)}\|\hat{\psi}^{(n+2)}|.
\end{split}
\end{equation*}
We further estimate via Young's inequality with an additional weight $\nu>0$:
\begin{equation*}
\mu \cdot 2x(x+\kappa)|2s| |\hat{\psi}^{(n+1)}\|\hat{\psi}^{(n+2)}|\leq \mu \nu |2s|^2 |\hat{\psi}^{(n+1)}|^2+\mu\nu^{-1} x^2(x+\kappa)^2|\hat{\psi}^{(n+2)}|^2,
\end{equation*}
and we estimate the remaining term with a slightly different Young's inequality, involving the weight $\alpha$ and an additional factor $ \frac{|2s|^2\tilde{\sigma}^2}{(n+1)^2(n+2)^2}$ that appears in the summation below:
\begin{equation*}
\begin{split}
(1-\mu)\cdot &2x(x+\kappa)|2s\|\hat{\psi}^{(n+1)}\|\hat{\psi}^{(n+2)}|\leq  \alpha (1-\mu) (n+1)^2 \left(x+\frac{1}{2}\kappa\right)^2|\hat{\psi}^{(n+1)}|^2\\
&+(1-\mu)\alpha^{-1}\tilde{\sigma}^{-2}\cdot \frac{|2s|^2\tilde{\sigma}^2}{(n+1)^2(n+2)^2} (n+1+1)^2 \frac{x^2(x+\kappa )^2}{\left(x+\frac{1}{2}\kappa\right)^2}|\hat{\psi}^{(n+1+1)}|^2.
\end{split}
\end{equation*}
We moreover have that $\frac{x^2(x+\kappa )^2}{\left(x+\frac{1}{2}\kappa\right)^2}=(x+\frac{1}{x^{-1}+2\kappa^{-1}})^2\leq (x+\frac{1}{2}\kappa)^2$
Finally, we estimate
\begin{multline*}
|[\lf(\lf+1)-n(n+1)]\hat{\psi}^{(n)}+\partial_x^n((\mathfrak{L}_{s,\lambda,\kappa}\hpsi)_{\ell})|^2\leq (1+\epsilon)(n(n+1)-\lf(\lf+1))^2|\hat{\psi}^{(n+1-1)}|^2\\
+(1+\epsilon^{-1})|(\mathfrak{L}_{s,\lambda,\kappa}\hpsi)_{\ell}|^2.
\end{multline*}

Now we integrate the above estimates over the interval $[0,x_0]$ to obtain:
\begin{multline*}
\int_0^{x_0}(1-\mu \nu^{-1}) x^2(x+\kappa)^2  |\hat{\psi}^{(n+2)}|^2\,dx+2(n+1)x_0^3|\hat{\psi}^{(n+1)}|^2|_{x=x_0}\\
+\int_0^{x_0} \left[4-\beta-\alpha(1-\mu)-6(n+1)^{-1}\right](n+1)^2 \left(x+\frac{1}{2}\kappa\right)^2|\hat{\psi}^{(n+1)}|^2\,dx\\
\boxed{-\int_0^{x_0}(1-\mu)\alpha^{-1}\tilde{\sigma}^{-2}\cdot \frac{|2s|^2\tilde{\sigma}^2}{(n+1)^2(n+2)^2} (n+1+1)^2 \left(x+\frac{1}{x^{-1}+2\kappa^{-1}}\right)^2|\hat{\psi}^{(n+1+1)}|^2\,dx}\\
+\int_0^{x_0} \left(1-\mu\nu-4\beta^{-1}\left(\frac{\max\{-\re s,0\}}{|s|}\right)^2\right) |2s|^2 |\hat{\psi}^{(n+1)}|^2\,dx\\
\boxed{-\int_0^{x_0} (1+\epsilon)|2s|^2\tilde{\sigma}^2\left(\frac{n(n+1)-\lf(\lf+1)}{n(n-1)}\right)^2\frac{n^2(n-1)^2}{|2s|^2\tilde{\sigma}^2} |\hat{\psi}^{(n+1-1)}|^2\,dx}\\
\leq (1+\epsilon^{-1})\int_0^{x_0}|\partial_x^n(\mathfrak{L}_{s,\lambda,\kappa}\hpsi)_{\ell}|^2\,dx.
\end{multline*}
We can sum the above estimate over $n$ with the following weights:
\begin{equation*}
\sum_{n=\lf}^{N_{\infty}} \frac{|2s|^{2n}\tilde{\sigma}^{2n}}{n!^2(n+1)!^2}\left[\cdot\right],
\end{equation*}
and absorbing all but the top order term with a bad sign into the terms with good signs, we obtain \eqref{eq:pregev}.
\end{proof}
In order for the unboxed terms on the left-hand side of \eqref{eq:pregev} to be non-negative definite, the coefficients $\tilde{\sigma},\alpha,\beta,\mu,\nu$ need to satisfy the following compatibility conditions: when ${\re s}<0$
\begin{align}
\label{eq:compatibility1}
0<\mu<&\:\nu,\\
\label{eq:compatibility2}
	4-\beta-(1-\mu)(\alpha+\tilde{\sigma}^{-2}\alpha^{-1})>&\: 0,\\
	\label{eq:compatibility3}
	1-\mu\nu-4\beta^{-1}\left(\frac{|{\re s}|}{|s|}\right)^2-\tilde{\sigma}^2>&\: 0.
\end{align}
For ${\re s}\geq 0$, we can take $\mu=0$, $\beta=1$ and $\alpha=2$, so we instead need 
\begin{align}
\label{eq:compatibility4}
\frac{1}{4}<\tilde{\sigma}^2<1.
\end{align}
Indeed, if \eqref{eq:compatibility1}--\eqref{eq:compatibility4} are satisfied, then we can take $\lambda$ sufficiently large and $\epsilon$ sufficiently small, to guarantee non-negative definiteness of the unboxed terms on the left-hand side of \eqref{eq:pregev}.

\begin{lemma}
	\label{lm:compatibility}
There exist $\tilde{\sigma}$, $\alpha,\beta,\mu,\nu$ satisfying \eqref{eq:compatibility1}--\eqref{eq:compatibility4} if and only if:
	\begin{enumerate}[label=\emph{(\roman*)}]
		\item  $-{\re s}\leq \frac{1}{2}|s|$,
		\item For fixed $2|s|\tilde{\sigma}=\sigma$, with $\sigma\in \R$, we have $s\in \Omega_{\sigma}$, where
	\begin{multline*}
	\Omega_{\sigma}:=\left\{x+iy\in \C\,|\, x\in \R_-,\,y\in \R,\, x^2+y^2<\sigma^2,\,3y^2-5x^2>\sigma^2\right\}\\
	\cup \left\{	x+iy\in \C\,|\, x\in[0,\infty),\, \frac{\sigma^2}{4}<x^2+y^2<\sigma^2\right\}.
	\end{multline*}
	\item For $s\in \{-\frac{1}{2}|s|\leq \re s<0\}$.
\begin{equation*}
\frac{1}{4}<\tilde{\sigma}^2< \frac{3}{4}-2\left(\frac{\re s}{|s|}\right)^2.	
\end{equation*}
Furthermore, 
\begin{equation*}
	\Omega:=\bigcup_{\sigma\in \R}\Omega_{\sigma}=\left\{z\in \C\,,\, |\arg(z)|<\frac{2\pi}{3}\right\}.
\end{equation*}
	\end{enumerate}
\end{lemma}
\begin{proof}
	Appendix \ref{app:pfoptparam}.
\end{proof}
By combining Proposition \ref{prop:pregev} with Lemma \ref{lm:compatibility}, we immediately obtain:
\begin{corollary}
\label{cor:gevcompat}
Let $\kappa\geq 0$, $M=1$, $\ell,l\in \N_0$, $N_{\infty}\in \N_0$, with $N_{\infty}>\lf$. Let $\sigma\in \R$ and let $s\in \Omega_{\sigma}$. Then, there exists a constant $C_{s,\sigma}=C(s,\sigma)>0$, such that:
\begin{multline}
\label{eq:gevcompat}
\sum_{n=\lf}^{N_{\infty}}\frac{\sigma^{2n}}{n!^2(n+1)!^2}\int_0^{x_0} x^2(x+\kappa)^2  |\hat{\psi}_{\ell}^{(n+2)}|^2\,dx+\sum_{n=\lf}^{N_{\infty}}\frac{\sigma^{2n}}{n!^2(n+1)!^2}(n+1)x_0^3 |\hpsi^{(n+1)}|^2\big|_{x=x_0}\\
+\sum_{n=\lf}^{N_{\infty}}\frac{\sigma^{2n}}{n!^2(n+1)!^2}\int_0^{x_0} \left[|2s|^2+(n+1)^2 \left(x+\frac{1}{2}\kappa\right)^2\right]|\hat{\psi}_{\ell}^{(n+1)}|^2\,dx\\
\boxed{-C_{s,\sigma}\frac{|2s|^{2(N_{\infty}+1)}\tilde{\sigma}^{2(N_{\infty}+1)}}{(N_{\infty}+1)!^2(N_{\infty}+2)!^2}\int_0^{x_0} (N_{\infty}+2)^2 \left(x+\frac{1}{x^{-1}+2\kappa^{-1}}\right)^2|\hat{\psi}_{\ell}^{(N_{\infty}+2)}|^2\,dx}\\
\leq C_{s,\sigma}\sum_{n=\lf}^{N_{\infty}}\frac{\sigma^{2n}}{n!^2(n+1)!^2}\int_0^{x_0}|\partial_x^n(\mathfrak{L}_{s,\lambda,\kappa}\hpsi)_{\ell}|^2\,dx.
\end{multline}
\end{corollary}

We now restrict to the case $\kappa>0$ to absorb the boxed boundary term in \eqref{eq:gevcompat} and arrive at a closed, uniform estimate in $\kappa$:
\begin{corollary}
\label{cor:maingev}
Let $\kappa> 0$, $M=1$, $\ell,l\in \N_0$, $N_{\infty}\in \N_0$, with $N_{\infty}>\lf$. Let $\sigma\in \R$ and let $s\in \Omega_{\sigma}$. Then there exists a constant $C_{s,\sigma}=C(s,\sigma)>0$ and a $N_{\kappa}>\lf$, such that for all $N_{\infty}\geq N_{\kappa}$:
\begin{multline}
\label{eq:gevcompat}
\sum_{n=\lf}^{N_{\infty}}\frac{\sigma^{2n}}{n!^2(n+1)!^2}\int_0^{x_0} x^2(x+\kappa)^2  |\hat{\psi}_{\ell}^{(n+2)}|^2\,dx+\sum_{n=\lf}^{N_{\infty}}\frac{\sigma^{2n}}{n!^2(n+1)!^2}(n+1)x_0^3 |\hpsi^{(n+1)}|^2\big|_{x=x_0}\\
+\sum_{n=\lf}^{N_{\infty}}\frac{\sigma^{2n}}{n!^2(n+1)!^2}\int_0^{x_0} \left[|2s|^2+(n+1)^2 \left(x+\frac{1}{2}\kappa\right)^2\right]|\hat{\psi}_{\ell}^{(n+1)}|^2\,dx\\
\leq C_{s,\sigma}\sum_{n=\lf}^{N_{\infty}}\frac{\sigma^{2n}}{n!^2(n+1)!^2}\int_0^{x_0}|\partial_x^n(\mathfrak{L}_{s,\lambda,\kappa}\hpsi)_{\ell}|^2\,dx.
\end{multline}
\end{corollary}
\begin{proof}
Note that
\begin{equation*}
	C_{s,\sigma}\frac{|2s|^{2}\tilde{\sigma}^{2}}{(N_{\infty}+1)^2(N_{\infty}+2)^2}(N_{\infty}+2)^2 \left(x+\frac{1}{x^{-1}+2\kappa^{-1}}\right)^2< 4\frac{C_{s,\sigma}|2s|^{2}\tilde{\sigma}^{2} \kappa^{-2}}{(N_{\infty}+1)^2}\kappa^2x^2,
\end{equation*}
so, for $N_{\infty}+1> 4 \sqrt{C_{s,\sigma}}|2s\|\tilde{\sigma}|\kappa^{-1}=:N_{\kappa}+1$, we can absorb the boxed term in \eqref{eq:gevcompat} into the first term on the right-hand side of \eqref{eq:gevcompat}.
\end{proof}

\begin{remark}
Note that the estimate \eqref{eq:gevcompat} is \emph{closed}, i.e.\ there are only terms involving $(\mathfrak{L}_{s,\lambda,\kappa}\hpsi)_{\ell}$ on the right-hand side. This means that we can control weighted $L^2$-norms of $\hpsi^{(n)}_{\ell}$ with $n\geq \lf+1$ in $x\leq x_0$ \underline{independently} from $\hpsi^{(n)}_{\ell}$ with $n\leq \lf$. 
\end{remark}

\subsection{Step 2: lower-order derivatives}
\label{sec:gevreylo}
In this section, we will establish estimates for appropriately weighted $\dot{H}^1$ estimates in the region $\{0<x\leq x_0\}$ for $\hpsi^{(n)}$ with $n\leq \lf$ in terms of a $\dot{H^1}$ norm for $\hpsi^{(\lf+1)}$ and lower-order boundary terms.

The estimates in this section rely on a different rearranging of \eqref{eq:commeq}, compared to Step 1:
\begin{equation*}
x(\kappa+x)\partial_{x}\hpsi^{(n+1)}_{\ell}+2(n+1) \left(x+\frac{\kappa}{2}\right)\partial_x \hpsi_{\ell}^{(n)}-(\mathfrak{l}(\mathfrak{l}+1)-n(n+1))\hpsi_{\ell}^{(n)}=-2 s\partial_{x}\hpsi^{(n)}_{\ell}+(\partial_x^n\mathfrak{L}_{s,\lambda,\kappa}\hpsi)_{\ell}.
\end{equation*}
The main estimates are obtained by considering the square norm of both sides of the equation and summing over $n$ with a appropriately weighted sum.

\begin{proposition}
\label{prop:mainestloder}
Let $\kappa\leq x_0$ and $\gamma>1$. Then there exists a constant $C=C(\gamma)>0$ such that
\begin{multline}
\label{eq:mainestloder}
\sum_{n=0}^{\lf-1}\frac{\gamma^{2n} |2s|^{2n}(\lf-(n+1))!^2}{(\lf+n+1)!^2}\int_0^{x_0}\int_{\s^2}(x^4+\kappa^2x^2)|\hat{\psi}^{(n+2)}|^2+(x^2(\lf+1)^2+\kappa^2n(n+1)+|2s|^2)|\hat{\psi}^{(n+1)}|^2\\
+(\lf-n)^2(\lf+n+1)^2|\hat{\psi}^{(n)}|^2\,d\upsigma dx\\
\leq C\frac{|2s|^{2(\lf-1)}\gamma^{2(\lf-1)}}{(2\lf!)^2} \int_0^{x_0}\int_{\s^2} |\hat{\psi}^{(\lf+1)}|^2\,dx+C\frac{|2s|^{2(\lf-1)}\gamma^{2(\lf-1)}}{(2\lf!)^2}\int_{\s^2}|\hat{\psi}^{(\lf)}|^2|_{x=x_0}\,d\upsigma\\
+C\sum_{n=0}^{\lf-1}\frac{\gamma^{2n} |2s|^{2n}(\lf-(n+1))!^2}{(\lf+n+1)!^2}(\lf-n)(\lf+n+1)(\lf+1)^2 x_0^{-1}\int_{\s^2}|\hpsi^{(n)}|^2|_{x=x_0}d\upsigma\\
+C\sum_{n=0}^{\lf-1}\frac{\gamma^{2n} |2s|^{2n}(\lf-(n+1))!^2}{(\lf+n+1)!^2}\int_0^{x_0}\int_{\s^2}|\partial_x^n\mathfrak{L}_{s,\lambda,\kappa}\hpsi|^2\,d\upsigma dx.
\end{multline}

Furthermore, there exists a constant $C=C(\gamma,|s|)>0$ such that
\begin{multline}
\label{eq:boundestetay0}
(\lf+1)^{-(4N_++1)}  \sum_{n=0}^{N_++1} \int_{\s^2}x |\hat{\psi}^{(n)}|^2|_{x=\eta x_0}d\upsigma\leq C \frac{|2s|^{2(\lf-1)}\gamma^{2(\lf-1)}}{(2\lf!)^2} \int_0^{x_0}\int_{\s^2} |\hat{\psi}^{(\lf+1)}|^2\,dx\\
+C\frac{|2s|^{2(\lf-1)}\gamma^{2(\lf-1)}}{(2\lf!)^2}\int_{\s^2}|\hat{\psi}^{(\lf)}|^2|_{x=x_0}\,d\upsigma\\
+C\sum_{n=0}^{\lf-1}\frac{\gamma^{2n} |2s|^{2n}(\lf-(n+1))!^2}{(\lf+n+1)!^2}(\lf-n)(\lf+n+1)(\lf+1)^2 x_0^{-1}\int_{\s^2}|\hpsi^{(n)}|^2|_{x=x_0}d\upsigma\\
+C\sum_{n=0}^{\lf-1}\frac{\gamma^{2n} |2s|^{2n}(\lf-(n+1))!^2}{(\lf+n+1)!^2}\int_0^{x_0}\int_{\s^2}|\partial_x^n\mathfrak{L}_{s,\lambda,\kappa}\hpsi|^2\,d\upsigma dx.
\end{multline}
\end{proposition}
\begin{proof}
Consider \eqref{eq:commeq}, suppress the $\ell$ subscript and rearrange the terms to obtain:
\begin{multline}
\label{eq:sqnormloworder}
|(\partial_x^n\mathfrak{L}_{s,\lambda,\kappa}\hpsi)-2s\hpsi^{(n+1)}|^2=|x(\kappa+x)\partial_{x}^2\hpsi^{(n)}+2(n+1) \left(x+\frac{\kappa}{2}\right)\partial_x\hpsi^{(n)}-(\mathfrak{l}(\mathfrak{l}+1)-n(n+1))\hpsi^{(n)}|^2\\
=x^2(x+\kappa)^2|\hat{\psi}^{(n+2)}|^2+4(n+1)^2\left(x+\frac{1}{2}\kappa\right)^2|\hat{\psi}^{(n+1)}|^2+(\lf(\lf+1)-n(n+1))^2|\hat{\psi}^{(n)}|^2\\
+T_1+T_2+T_3,
\end{multline}
with
\begin{align*}
T_1:=&\:2(n+1) x\left(x+\frac{1}{2}\kappa\right)(x+\kappa)\partial_x(|\hat{\psi}^{(n+1)}|^2),\\
T_2:=&\:-2(\mathfrak{l}(\mathfrak{l}+1)-n(n+1))x(\kappa+x)\re(\partial_{x}^2\hpsi^{(n)}\cdot \overline{\hpsi^{(n)}}),\\
T_3:=&\: -2(n+1)(\lf (\mathfrak{l}+1)-n(n+1))\left(x+\frac{\kappa}{2}\right)\partial_x(|\hpsi^{(n)}|^2).
\end{align*}
We now integrate both sides of \eqref{eq:sqnormloworder} over $[0,x_0]\times \s^2$, but we suppress the integration over $\s^2$ in the notation. 

We first estimate the integral of $T_1$ by integrating by parts:
\begin{multline*}
\int_0^{x_0}T_1\,dx=-\int_0^{x_0}(n+1)(6x^2+6\kappa x+\kappa^2)|\hpsi^{(n+1)}|^2\,dx\\
+2(n+1) x_0\left(x_0+\frac{1}{2}\kappa\right)(x_0+\kappa)|\hat{\psi}^{(n+1)}|^2|_{x=x_0}\\
\geq -\int_0^{x_0}6(n+1)(x(x+\kappa)+\frac{1}{6}\kappa^2) |\hat{\psi}^{(n+1)}|^2\,dx+2(n+1) x_0\left(x_0+\frac{1}{2}\kappa\right)(x_0+\kappa)|\hat{\psi}^{(n+1)}|^2|_{x=x_0}
\end{multline*}

To estimate $T_2$, we integrate by parts twice:
\begin{multline*}
\int_0^{x_0}T_2\,dx=\int_0^{x_0}2(\mathfrak{l}(\mathfrak{l}+1)-n(n+1))x(\kappa+x) |\hpsi^{(n+1)}|^2\,dx-\int_0^{x_0}(\mathfrak{l}(\mathfrak{l}+1)-n(n+1))(\kappa+2x)\partial_x(|\hpsi^{(n)}|^2)\,dx\\
-2(\mathfrak{l}(\mathfrak{l}+1)-n(n+1))x_0(\kappa+x_0)\re(\hpsi^{(n+1)}\cdot  \overline{\hpsi^{(n)}})|_{x=x_0}\\
\geq \int_0^{x_0}2(\mathfrak{l}(\mathfrak{l}+1)-n(n+1))x(\kappa+x) |\hpsi^{(n+1)}|^2\,dx+2\int_0^{x_0}(\mathfrak{l}(\mathfrak{l}+1)-n(n+1))|\hpsi^{(n)}|^2\,dx\\
-2(\mathfrak{l}(\mathfrak{l}+1)-n(n+1))x_0(\kappa+x_0)\re(\hpsi^{(n+1)}\cdot  \overline{\hpsi^{(n)}})|_{x=x_0}-(\mathfrak{l}(\mathfrak{l}+1)-n(n+1))(\kappa+2x_0)|\hpsi^{(n)}|^2|_{x=x_0}
\end{multline*}

We also integrate $T_3$ by parts to obtain:
\begin{equation*}
\int_0^{x_0}T_3\,dx\geq \int_0^{x_0}2(n+1)(\lf (\mathfrak{l}+1)-n(n+1))|\hpsi^{(n)}|^2\,dx-2(n+1)(\lf (\mathfrak{l}+1)-n(n+1))\left(x_0+\frac{\kappa}{2}\right)|\hpsi^{(n)}|^2|_{x=x_0}.
\end{equation*}

Finally, we apply Young's inequality to estimate for $\epsilon>0$:
\begin{equation*}
|(\partial_x^n\mathfrak{L}_{s,\lambda,\kappa}\hpsi)-2s\hpsi^{(n+1)}|^2\leq (1+\epsilon)|2s|^2|\hpsi^{(n+1)}|^2+(1+\epsilon^{-1})|\partial_x^n\mathfrak{L}_{s,\lambda,\kappa}\hpsi|^2.
\end{equation*}

We combine the above estimates to obtain for $\gamma>1$:
\begin{multline*}
\int_0^{x_0}x^2(x+\kappa)^2|\hat{\psi}^{(n+2)}|^2\,dx\\
+\int_0^{x_0}\left[(2(n+1)(n-1)+2\mathfrak{l}(\mathfrak{l}+1))x(\kappa+x)+(n+1)^2\kappa^2 \right]|\hat{\psi}^{(n+1)}|^2\,dx\\
+\int_0^{x_0}(\lf(\lf+1)-n(n+1))(\lf(\lf+1)-n(n+1)+2(n+2))|\hat{\psi}^{(n)}|^2\,dx\\
\boxed{-\int_0^{x_0}\frac{|2s|^2\gamma^{2}}{(\lf(\lf+1)-(n+1)(n+2))^2} \gamma^{-2}(1+\epsilon)(\lf(\lf+1)-(n+1)(n+2))^2|\hat{\psi}^{(n+1)}|^2\,dx}\\
\leq C(\lf(\lf+1)-n(n+1))(\lf+1)^2 x_0^{-1}|\hpsi^{(n)}|^2|_{x=x_0}+(1+\epsilon^{-1})\int_0^{x_0}|\partial_x^n\mathfrak{L}_{s,\lambda,\kappa}\hpsi|^2\,dx.
\end{multline*}

In order to absorb the boxed term above, we will sum over $n$ with the following weights:
\begin{equation*}
\sum_{n=0}^{\lf-1}\frac{\gamma^{2n} |2s|^{2n}}{(\lf(\lf+1))^2(\lf(\lf+1)-1(1+1))^2\ldots (\lf(\lf+1)-n(n+1))^2}.
\end{equation*}

Using that $\lf(\lf+1)-n(n+1)=(\lf-n)(\lf+n+1)$. We can simplify for $n\leq \lf-1$:
\begin{equation*}
(\lf(\lf+1))^2(\lf(\lf+1)-1(1+1))^2\ldots (\lf(\lf+1)-n(n+1))^2=\frac{(\lf+n+1)!^2}{(\lf-(n+1))!^2},
\end{equation*}
so that the summation becomes
\begin{equation*}
\sum_{n=0}^{\lf-1}\frac{\gamma^{2n} |2s|^{2n}(\lf-(n+1))!^2}{(\lf+n+1)!^2}.
\end{equation*}
We obtain:
\begin{multline*}
\sum_{n=0}^{\lf-1}\frac{\gamma^{2n} |2s|^{2n}(\lf-(n+1))!^2}{(\lf+n+1)!^2}\int_0^{x_0}x^4|\hat{\psi}^{(n+2)}|^2\,dx\\
+\int_0^{x_0}\left[((n+1)^2+\mathfrak{l}(\mathfrak{l}+1))x^2 +|2s|^2\right]|\hat{\psi}^{(n+1)}|^2\,dx\\
+\int_0^{x_0}(1-(1+\epsilon)\gamma^{-2})(\lf(\lf+1)-n(n+1))^2|\hat{\psi}^{(n)}|^2\,dx\\
\leq C\frac{|2s|^{2(\lf-1)}\gamma^{2(\lf-1)}}{(2\lf!)^2} \int_0^{x_0} |\hat{\psi}^{(\lf)}|^2\,dx\\
+C\sum_{n=0}^{\lf-1}\frac{\gamma^{2n} |2s|^{2n}(\lf-(n+1))!^2}{(\lf+n+1)!^2}(\lf-n)(\lf+n+1)(\lf+1)^2 x_0^{-1}|\hpsi^{(n)}|^2|_{x=x_0}\\
+C(1+\epsilon^{-1})\sum_{n=0}^{\lf-1}\int_0^{x_0}\frac{\gamma^{2n} |2s|^{2n}(\lf-(n+1))!^2}{(\lf+n+1)!^2}|\partial_x^n\mathfrak{L}_{s,\lambda,\kappa}\hpsi|^2\,dx.
\end{multline*}

By applying a Hardy inequality we can further estimate:
\begin{equation*}
\frac{|2s|^{2(\lf-1)}\gamma^{2(\lf-1)}}{(2\lf!)^2} \int_0^{x_0} |\hat{\psi}^{(\lf)}|^2\,dx\leq 4x_0^2\frac{|2s|^{2(\lf-1)}\gamma^{2(\lf-1)}}{(2\lf!)^2} \int_0^{x_0} |\hat{\psi}^{(\lf+1)}|^2\,dx+2\frac{|2s|^{2(\lf-1)}\gamma^{2(\lf-1)}}{(2\lf!)^2} x_0 |\hpsi^{(\lf)}|^2|_{x=x_0}.
\end{equation*}
We conclude \eqref{eq:mainestloder}.

By integrating by parts $T_3$ in the region $[\eta x_0,x_0]\times \s^2$ in the argument above and using the integral estimates above, we can moreover estimate boundary terms at $x=\eta x_0$. We obtain in particular also:
\begin{multline*}
(\lf+1)^{-4(N_++1)} \sum_{n=0}^{N_++1}x |\hat{\psi}^{(n)}|^2|_{x=\eta x_0}\leq C\frac{|2s|^{2(\lf-1)}\gamma^{2(\lf-1)}}{(2\lf!)^2} \int_0^{x_0} |\hat{\psi}^{(\lf)}|^2\,dx\\
+C\sum_{n=0}^{\lf-1}\frac{\gamma^{2n} |2s|^{2n}(\lf-(n+1))!^2}{(\lf+n+1)!^2}(\lf-n)(\lf+n+1)(\lf+1)^2 x_0^{-1}|\hpsi^{(n)}|^2|_{x=x_0}\\
+C(1+\epsilon^{-1})\sum_{n=0}^{\lf-1}\int_0^{x_0}\frac{\gamma^{2n} |2s|^{2n}(\lf-(n+1))!^2}{(\lf+n+1)!^2}|\partial_x^n\mathfrak{L}_{s,\lambda,\kappa}\hpsi|^2\,dx.
\end{multline*}
\end{proof}

We now estimate the boundary terms at $x=x_0$ that appear on the right-hand sides of the estimates in Proposition \ref{prop:mainestloder}.

\begin{lemma}
\label{lm:boundterms}
Let $\kappa\leq x_0$ and $\lf^2\gg |s|^2x_0^{-2}$. Then there exists constants $C>0$ suitably large and $d_0>0$ suitably small, with $d_0$ independent of $|s|$, $\lf$ and $x_0$, such that for all $n\in \N_0$:
\begin{multline}
\label{eq:sumboundary}
\sum_{n=0}^{\infty}\frac{(d_0x_0)^{2n} }{(\max\{\lf+1,n+1\})^{2n}}|\hat{\psi}^{(n+2)}_{\ell}|^2\big|_{x=x_0}\leq Cx_0^{-4} \left[(\lf+1)^4|\hat{\psi}_{\ell}|^2+(x_0^2+|2s|^2)|\hat{\psi}^{(1)}_{\ell}|^2\right]\big|_{x=x_0}\\
+Cx_0^{-4}\sum_{n=0}^{\infty}\frac{(d_0x_0)^{2n}}{(\max\{\lf+1,n+1\})^{2n}}|\partial_x^{n}(\mathfrak{L}_{s,\lambda,\kappa}\hpsi)_{\ell}|^2\big|_{x=x_0}.
\end{multline}
Furthermore, for any $\gamma>0$, with $\lf\gg d_0^{-1}x_0^{-1}\gamma |s|$, we can estimate
\begin{multline}
\label{eq:sumboundarygamma}
\sum_{n=0}^{\lf-1}\frac{\gamma^{2n}|2s|^{2n} (\lf-(n+1))!^2}{(\lf+n+1)!^2}|\hat{\psi}^{(n+2)}_{\ell}|^2\big|_{x=x_0}\leq Cx_0^{-4}\left[|\hat{\psi}_{\ell}|^2+(x_0^2+|2s|^2)(\lf+1)^{-4}|\hat{\psi}^{(1)}_{\ell}|^2\right]\big|_{x=x_0}\\
+C(\lf+1)^{-4}x_0^{-4}\sum_{n=0}^{\lf-1}\frac{(d_0x_0)^{2n}}{(\max\{\lf+1,n+1\})^{2n}}|\partial_x^{n}(\mathfrak{L}_{s,\lambda,\kappa}\hpsi)_{\ell}|^2\big|_{x=x_0}.
\end{multline}
\end{lemma}
\begin{proof}
We will first establish the following estimate inductively: for $n\in \N_0$
\begin{multline}
\label{eq:singletermboundary}
\frac{(d_0'x_0)^{2n}}{(\max\{\lf+1,n+1\})^{2n}}|\hat{\psi}^{(n+2)}_{\ell}|^2\big|_{x=x_0}\leq Cx_0^{-4} \left[(\lf+1)^4|\hat{\psi}_{\ell}|^2+(x_0^2+|2s|^2)|\hat{\psi}^{(1)}_{\ell}|^2\right]\\
+Cx_0^{-4}\sum_{k=0}^{n}\frac{(d_0'x_0)^{2n}}{(\max\{\lf+1,k+1\})^{2k}}|\partial_x^{k}(\mathfrak{L}_{s,\lambda,\kappa}\hpsi)_{\ell}|^2\big|_{x=x_0}.
\end{multline}
with $d_0<d_0'$. Note that the $n=0$ case follows immediately from \eqref{eq:commeq} with $n=0$. Now suppose \eqref{eq:singletermboundary} holds for $n\leq N$. By applying \eqref{eq:commeq} with $n=N+1$, we obtain
\begin{multline*}
\frac{(d_0'x_0)^{2(N+1)} }{(\max\{\lf+1,N+2\})^{2(N+1)}}|\hat{\psi}^{(N+3)}|^2\big|_{x=x_0}\leq  C\frac{(d_0'x_0)^{2(N+1)}}{(\max\{\lf+1,N+1\})^{2N}}\\
\times \Bigg[ ((N+2)^2x_0^{-2}+|2s|^2)|\hat{\psi}^{(N+2)}|^2\big|_{x=x_0}+x_0^{-4}(\lf(\lf+1)-(N+1)(N+2))^2|\hat{\psi}^{(N+1)}|^2\big|_{x=x_0}\\
+x_0^{-4} |\partial_x^{N+1}(\mathfrak{L}_{s,\lambda,\kappa}\hpsi)_{\ell}|^2\big|_{x=x_0}\Bigg]\\
\leq C\frac{(d_0'x_0)^{2N}}{(\max\{\lf+1,N+1\})^{2N}}\frac{d_0'^2((N+2)^2+|2s|^2x_0^{-2})}{(\max\{\lf+1,N+1\})^{2}}|\hat{\psi}^{(N+2)}|^2\big|_{x=x_0}\\
+C\frac{(d_0'x_0)^{2(N-1)}}{(\max\{\lf+1,N\})^{2(N-1)}}\frac{d_0'^4(\lf(\lf+1)-(N+1)(N+2))^2}{(\max\{\lf+1,N+1\})^{2}(\max\{\lf+1,N+2\})^{2}}|\hat{\psi}^{(N+1)}|^2\big|_{x=x_0}\\
+C\frac{(d_0'x_0)^{2(N+1)} }{(\max\{\lf+1,N+2\})^{2(N+1)}}x_0^{-4} |\partial_x^{N+1}(\mathfrak{L}_{s,\lambda,\kappa}\hpsi)_{\ell}|^2\\
\leq \frac{C}{2}\frac{(d_0'x_0)^{2N}}{(\max\{\lf+1,N+1\})^{2N}}|\hat{\psi}^{(N+2)}|^2\big|_{x=x_0}+\frac{C}{2}\frac{(d_0'x_0)^{2(N-1)} }{(\max\{\lf+1,N\})^{2(N-1)}}|\hat{\psi}^{(N+1)}|^2\big|_{x=x_0}\\
+C\frac{(d_0'x_0)^{2(N+1)} }{(\max\{\lf+1,N+2\})^{2(N+1)}}x_0^{-4} |\partial_x^{N+1}(\mathfrak{L}_{s,\lambda,\kappa}\hpsi)_{\ell}|^2,
\end{multline*}
where we arrived at the final inequality by taking $d_0'$ suitably small,  $|2s|^2x_0^{-2}\leq \lf^2$. We now conclude \eqref{eq:singletermboundary} with $n=N+1$ by applying \eqref{eq:singletermboundary} with $n=N$ and $n=N-1$.

To obtain \eqref{eq:sumboundary}, we let $0<{d_0}<{d_0}'$ and we sum both sides of \eqref{eq:singletermboundary} to obtain:
\begin{multline*}
\sum_{n=0}^{\infty}\frac{(d_0x_0)^{2n} }{(\max\{\lf+1,n+1\})^{2n}}|\hat{\psi}^{(n+2)}_{\ell}|^2\big|_{x=x_0}\leq C\sum_{n=0}^{\infty}\left(\frac{d_0}{d_0'}\right)^nx_0^{-4} \left[|\hat{\psi}_{\ell}|^2+(x_0^2+|2s|^2)(\lf(\lf+1))^{-2}|\hat{\psi}^{(1)}_{\ell}|^2\right]\\
+Cx_0^{-4}\sum_{n=0}^{\infty}\left(\frac{d_0}{d_0'}\right)^{2(n-k)}\sum_{k=0}^{\infty}\frac{(d_0x_0)^{2k} }{(\max\{\lf+1,k+1\})^{2k}}|\partial_x^{k}(\mathfrak{L}_{s,\lambda,\kappa}\hpsi)_{\ell}|^2\big|_{x=x_0}\\
\leq Cx_0^{-4} \left[|\hat{\psi}_{\ell}|^2+(x_0^2+|2s|^2)(\lf(\lf+1))^{-2}|\hat{\psi}^{(1)}_{\ell}|^2\right]\\
+Cx_0^{-4}\sum_{n=0}^{\infty}\frac{(d_0x_0)^{2n} }{(\max\{\lf+1,n+1\})^{2n}}|\partial_x^{n}(\mathfrak{L}_{s,\lambda,\kappa}\hpsi)_{\ell}|^2,
\end{multline*}
where we used that $\sum_{j=0}^{\infty}\left(\frac{d_0}{d_0'}\right)^{2j}<\infty$. We now conclude \eqref{eq:sumboundary}.

By applying Stirling's approximation in the following form: for any $k\in \N$
\begin{equation*}
k!\sim k^{k+\frac{1}{2}}e^{-k},
\end{equation*}
we can estimate for $n\leq \lf-1$:
\begin{multline*}
\frac{\gamma^{2n}|2s|^{2n}(\lf-(n+1))!^2}{(\lf+(n+1))!^2}\leq C \gamma^{2n}|2s|^{2n}\frac{(\lf-(n+1))^{2\lf-2(n+1)+1}}{(\lf+(n+1))^{2\lf+2(n+1)+1}}e^{4(n+1)}\\
=C  \gamma^{2n}|2s|^{2n}e^{4(n+1)}\left(1-2\frac{n+1}{\lf+(n+1)}\right)^{4(n+1)}(\lf+(n+1))^{-4(n+1)}\\
\leq  C\left(e^{-2}d_0x_0|2s|^{-1}\gamma^{-1}(\lf+1)\right)^{-2(n+1)} (d_0x_0)^{2n}(\lf+1)^{-2n-4}\leq (\lf+1)^{-4}(d_0x_0)^{2n}(\lf+1)^{-2n}
\end{multline*}
for $\lf\gg d_0^{-1}x_0^{-1}\gamma |2s|$. Hence, \eqref{eq:sumboundarygamma} follows.
\end{proof}

\subsection{Combining all derivatives}
We combine Corollary \ref{cor:gevcompat} with Proposition \ref{prop:mainestloder} and Lemma \ref{lm:boundterms} to obtain an estimate in the region $\{x\leq x_0\}$, with on the right-hand side only the terms $\partial_x^{n}(\mathfrak{L}_{s,\lambda,\kappa}\hpsi)_{\ell}$ and boundary terms at $x=x_0$ involving only the lowest two derivatives $\hpsi$ and $\hpsi^{(1)}$:
\begin{corollary}
\label{cor:fullgevrey}
	Let $M=1$, $N_+\in \N_0$, $\lf>|s|x_0^{-1}$ and $0<\kappa<x_0$ and let $d_0$ be the constant from Lemma \ref{lm:boundterms}. Let $s\in \Omega_{\sigma}$. Then, for $\lambda$ suitably large, there exists a constant $C_{s,\sigma}=C(s,\sigma)>0$, an numerical constant $C>0$, such that for all $\N_0\ni N_{\infty}\geq N_{\kappa}$:
	\begin{multline}
	\label{eq:fullgevrey}
\sum_{n=0}^{\lf-1}\frac{(4\sigma)^{2n}(\lf-(n+1))!^2}{(\lf+n+1)!^2}\int_0^{x_0}\int_{\s^2}(x^4+\kappa^2x^2)|\hat{\psi}^{(n+2)}|^2+(x^2(\lf+1)^2+(n+1)^2\kappa^2+|2s|^2)|\hat{\psi}^{(n+1)}|^2\\
+(\lf-n)^2(\lf+n+1)^2|\hat{\psi}^{(n)}|^2\,d\upsigma dx\\
+(\lf+1)^3\sum_{n=\lf}^{N_{\infty}}\frac{\sigma^{2n}}{n!^2(n+1)!^2}\int_0^{x_0}\int_{\s^2}x^4|s|^{-2}|\hpsi^{(n+2)}|^2+ \left[1+(n+1)^2x^2|s|^{-2}\right]|\hat{\psi}_{\ell}^{(n+1)}|^2\,d\upsigma dx\\
+\sum_{n=0}^{\infty}\frac{(d_0x_0)^{2n} }{(\max\{\lf+1,n+1\})^{2n}}\int_{\s^2}|\hpsi^{(n+2)}|^2|_{x=x_0}d\upsigma+ (\lf+1)^{-4(N_++1)}\sum_{n=0}^{N_++1} \int_{\s^2}x|\hat{\psi}^{(n)}|^2|_{x=\eta x_0}d\upsigma\\
\leq \boxed{C\int_{\s^2}|\hat{\psi}_{\ell}|^2+(\lf+1)^{-4}|\hat{\psi}^{(1)}_{\ell}|^2\big|_{x=x_0}\,d\upsigma}\\
+C\sum_{n=0}^{\lf-1}\frac{(4\sigma)^{2n}(\lf-(n+1))!^2}{(\lf+n+1)!^2}\int_0^{x_0}\int_{\s^2}|\partial_x^{n}(\mathfrak{L}_{s,\lambda,\kappa}\hpsi)_{\ell}|^2\,d\upsigma dx\\
+C(\lf+1)^3\sum_{n=\lf}^{N_{\infty}}\frac{\sigma^{2n}}{n!^2(n+1)!^2}\int_0^{x_0}\int_{\s^2}|\partial_x^{n}(\mathfrak{L}_{s,\lambda,\kappa}\hpsi)_{\ell}|^2\,d\upsigma dx\\
+C\sum_{n=0}^{\infty}\frac{(d_0x_0)^{2n}}{(\max\{\lf+1,n+1\})^{2n}}\int_{\s^2}|\partial_x^{n}(\mathfrak{L}_{s,\lambda,\kappa}\hpsi)_{\ell}|^2|_{x=x_0}d\upsigma.
	\end{multline}
\begin{proof}
Note first that the boundary terms at $x=x_0$ on the right-hand sides of \eqref{eq:mainestloder} and \eqref{eq:boundestetay0} can be estimated directly via \eqref{eq:sumboundarygamma}.

Furthermore, the $x$-integral of $\hpsi^{(\lf+1)}$ on the right-hand sides of \eqref{eq:mainestloder} and \eqref{eq:boundestetay0} can be controlled by the $x$-integral of $\hpsi^{(\lf+1)}$ on the left-hand side of \eqref{eq:gevcompat}, provided we multiply both sides of \eqref{eq:gevcompat} with the factor
\begin{equation*}
|s|^{-2}\frac{\gamma^{2\lf}(\lf+1)!^2\lf!^2}{\tilde{\sigma}^{2\lf}(2\lf)!^2}=|s|^{-2}(\lf+1)^2\frac{\gamma^{2\lf}}{\tilde{\sigma}^{2\lf}}\frac{\lf!^4}{(2\lf)!^2}
\end{equation*}
By applying Stirling's approximation, we have that
\begin{equation*}
\frac{\lf!^4}{(2\lf)!^2}\sim \frac{\lf^{4\lf +2} e^{-4\lf}}{(2\lf)^{4\lf+1}e^{-4\lf}}=2^{-4\lf}\lf.
\end{equation*}
By taking $\gamma=4\tilde{\sigma}$, we therefore obtain:
\begin{equation*}
|s|^2\frac{\gamma^{2\lf}(\lf+1)!^2\lf!^2}{\tilde{\sigma}^{2\lf}(2\lf)!^2}\sim |s|^{-2}(\lf+1)^3. \qedhere
\end{equation*}
\end{proof}

\end{corollary}
\section{A global estimate for $\mathfrak{L}_{s,\lambda,\kappa}$}
\label{sec:fullycoupledestmates}
We combine Corollaries \ref{cor:fullellipticest} and \ref{cor:fullgevrey}, to obtain a global estimate for an $\hpsi$ in terms of $\mathfrak{L}_{s,\lambda,\kappa}\hpsi$. Recall the Hilbert spaces $H^{N_+}_{\sigma,\delta}$ and $H^{N_++1}_{\sigma, R_{\infty},2}$ introduced in \S \ref{sec:hilbertspaces}.
\begin{proposition}
\label{prop:fullmainest}
Let $M=1$, $N_+\in \N_0$, $0<\kappa\ll \eta x_0$. Take $s\in \Omega_{\sigma}\cap\{{\re s}>(\frac{1}{2}+N_+)\kappa_+\}$ and $p=4(N_++1)+2$. Then, for $M^{-1}R_{\infty}$ and $\lambda$ suitably large, there exists a constant $C_{s,\lambda,\delta, R_{\infty}}>0$ and a constant $N_{\kappa}>\lf$, such that for $\mathfrak{L}_{s,\lambda,\kappa}\hpsi\in H^{N_+}_{\sigma,\delta}$:
\begin{equation}
\label{eq:fullmainestcpt}
	\|\hpsi\|_{H^{N_++1}_{\sigma,R_{\infty},2}}^2+\kappa^2(\|x\partial_x^2\hpsi\|_{G_{\sigma,R_{\infty}}}^2+\|\slashed{\Delta}_{\s^2}\hpsi\|_{G_{\sigma,R_{\infty}}}^2) \leq C_{s,\sigma}\|\mathfrak{L}_{s,\lambda,\kappa}\hpsi\|_{H^{N_+}_{\sigma,R_{\infty}}}^2.
\end{equation}
\end{proposition}
\begin{proof}
We need to show that:
\begin{multline}
\label{eq:fullmainest}
\sum_{k=0}^1\|\snabla_{\s^2}^{k }\hpsi\|_{H^{N_++1}(\Sigma\cap \{\varrho\leq \frac{1}{3}R_{\infty}\})}^2+\|\hpsi\|_{H^{N_++2}(\Sigma\cap \{\frac{1}{3}R_{\infty}\leq \varrho \leq 3 R_{\infty}\})}^2\\
+\left|\left\|(-\slashed{\Delta}_{\s^2}+\mathbf{1})^{-\frac{w_p}{4}}\hpsi\right|\right|_{H^{N_++2}(\Sigma\cap\{3R_{\infty}\leq \varrho\leq 4R_{\infty}\})}^2+\|\hpsi\|_{B_{R_{\infty},2}}^2+\|\hpsi\|_{G_{\sigma,R_{\infty}, 2}}^2	+\kappa^2\|x\partial_x^2\hpsi\|_{G_{\sigma,R_{\infty}}}^2\\
\leq C_{s,\lambda, R_{\infty}}\Biggl[\|\mathfrak{L}_{s,\lambda,\kappa}\hpsi\|_{H^{N_+}(\Sigma\cap\{\varrho\leq 3R_{\infty}\})}^2+\|\snabla_{\s^2}\mathfrak{L}_{s,\lambda,\kappa}\hpsi\|_{H^{N_+}(\Sigma\cap\{\varrho\leq \frac{1}{3}R_{\infty}\})}^2\\
+\|(-\slashed{\Delta}_{\s^2}+\mathbf{1})^{-\frac{w_p}{4}}(\mathfrak{L}_{s,\lambda,\kappa}\hpsi)\|_{H^{N_+}(\Sigma\cap\{3R_{\infty}\leq\varrho\leq 4R_{\infty}\})}^2+\|\mathfrak{L}_{s,\lambda,\kappa}\hpsi\|_{B_{R_{\infty}}}^2+\|\mathfrak{L}_{s,\lambda,\kappa}\hpsi\|_{G_{\sigma,R_{\infty}}}^2\Biggr].\end{multline}

We consider \eqref{eq:fullellipticest} and we commute additionally with $\slashed{\Delta}_{\s^2}$ to control additional angular derivatives on the left-hand side and obtain control of the desired $H^{N_+}$, $H^{N_++1}$ and $H^{N_++2}$ norms. 

By using that $M^{-1}R_{\infty}\ll \sqrt{\lambda}$, we can absorb the boxed term on the right-hand side of \eqref{eq:fullgevrey} into the left-hand side of \eqref{eq:fullellipticest}. In order to subsequently absorb the boxed terms that appear on the right-hand side of \eqref{eq:fullellipticest} into the left-hand side of  \eqref{eq:fullgevrey}, we take $p= 4(N_++1)+2$. Finally, we note that
\begin{equation*}
\sum_{\ell=\in \N_0}\int_{\s^2}(\lf+1)^{-w_p}|f_{\ell}|^2\,d\upsigma \leq  \sum_{\ell\in \N_0}\int_{\s^2}2^{\frac{w_p}{2}}(\ell(\ell+1)+\lambda)^{-\frac{p}{2}}|f_{\ell}|^2\,d\upsigma\leq 2^{\frac{w_p}{2}}(1+\lambda^{-\frac{w_p}{2}})\int_{\s^2}|(-\slashed{\Delta}_{\s^2}+\mathbf{1})^{-\frac{w_p}{4}}f|^2\,d\upsigma.
\end{equation*}
We also have that
\begin{multline*}
\sum_{\ell=\in \N_0}\int_{\s^2}(\lf+1)^{-w_p}|f_{\ell}|^2\,d\upsigma = \sum_{\ell=0}^{l}\int_{\s^2}(\ell+l+1)^{-w_p}|f_{\ell}|^2\,d\upsigma+ \sum_{\ell=l+1}^{\infty }(\ell+1)^{-w_p}|f_{\ell}|^2\,d\upsigma\\
\geq ((2l+1)^{w_p}+1) \sum_{\ell=0}^{\infty}(1+\ell(\ell+1))^{\frac{-w_p}{2}}|f_{\ell}|^2\,d\upsigma\geq ((2l+1)^{w_p}+1)\int_{\s^2}|(-\slashed{\Delta}_{\s^2}+\mathbf{1})^{-\frac{w_p}{4}}f|^2\,d\upsigma.
\end{multline*}
Finally, we observe that in the case $\ell>l$, we have that $\lf=\ell$, whereas in the case $\ell\leq l$, we use boundedness of $\ell$ to ensure that the $\lf$-dependent norms in \eqref{eq:fullellipticest} are comparable to the norms appearing in \eqref{eq:fullmainest}.
\end{proof}

\section{Inverting $\mathfrak{L}_{s,\lambda,\kappa}$}
\label{sec:invertingmathfrakL}
In this section, we will construct $\mathfrak{L}_{s,\lambda,\kappa}^{-1}$ for $\kappa\geq 0$ with an appropriate Hilbert space domain and with a suitably large $\lambda$. 

We first make use of the regularizing effect of $\kappa>0$ to invert $\mathfrak{L}_{s,\lambda,\kappa}$ as an operator on suitable Sobolev spaces. We then appeal to the uniform-in-$\kappa$ estimate from Proposition \ref{prop:fullmainest} to construct $\mathfrak{L}_{s,\lambda,\kappa=0}^{-1}$ as a limit of a sequence of operators $\mathfrak{L}_{s,\lambda,\kappa_n}^{-1}$, with $\kappa_n>0$, $\kappa_n\to 0$ as $n\to \infty$.

Define $H^k(\mathfrak{L}_{s,\lambda,\kappa})$ as the closure of $C^{\infty}([0,r_+^{-1}]_x\times \s^2)$ with respect to the graph norm:
\begin{equation*}
	\|\cdot\|_{H^k([0,r_+^{-1}]_x\times \s^2)}+ \|\mathfrak{L}_{s,\lambda,\kappa}(\cdot)\|_{H^k([0,r_+^{-1}]_x\times \s^2)}.
\end{equation*}
Then the map:
\begin{equation*}
	\mathfrak{L}_{s,\lambda,\kappa}: H^k(\mathfrak{L}_{s,\lambda,\kappa})\to H^k([0,r_+^{-1}]_x\times \s^2)
\end{equation*}
is a well-defined bounded linear operator for all $k\in \N_0$.

\begin{theorem}
\label{tm:invertkappapos}
	Let $N\geq N_+\in \N_0$ and ${\re s}>-(N+\frac{1}{2})\min\{\kappa_+,\kappa\}$. Then, there exists a $\lambda>0$ suitably large, such that the inverse map:
	\begin{multline*}
		\mathfrak{L}_{s,\lambda,\kappa}^{-1}: H^{N}([0,r_+^{-1}]_x\times \s^2)\cap\left\{f\in L^2([0,r_+^{-1}]_x\times \s^2)\,|\, \exists m\in \Z\,:\,\partial_{\varphi}f=imf\right\}\\
		\to H^N(\mathfrak{L}_{s,\lambda,\kappa})\subseteq H^{N+1}([0,r_+^{-1}]_x\times \s^2)
	\end{multline*}
	is a bounded linear operator.
\end{theorem}
\begin{proof}
Since we are considering a domain restricted to finite azimuthal modes (with $\varphi$-dependence $e^{im\varphi}$), we can apply a straightforward modification of \cite{warn15}[Theorem 5.7] to obtain the existence of the desired inverse.
\end{proof}
\begin{remark}
The ability to restrict to fixed azimuthal modes on Kerr is not crucial. One can in principle derive an analogue of Theorem \ref{tm:invertkappapos} without making use of this assumption and instead make use of red-shift estimates depending on $\kappa$ in the large-$r$-region, instead of Gevrey-type estimates. See also the related results in \cite{vasy1}.
\end{remark}

By applying Proposition \ref{prop:fullmainest}, which is uniform in $\kappa$ and $m$, we can obtain the existence of $\mathfrak{L}_{s,\lambda,\kappa}^{-1}$ with domain $H_{ \sigma,R_{\infty}}^{N_+}$.
\begin{corollary}
\label{cor:unifboundinvkappapos}
	Let $\sigma\in \R$ and take $s\in \Omega_{\sigma}\cap\{{\re s}>(\frac{1}{2}+N_+)\kappa_+\}$. Then, there exists a suitably small $\kappa>0$ and a suitably large $R_{\infty}$ and $\lambda \in \N_0$, such that:
	\begin{equation*}
		\mathfrak{L}_{s,\lambda,\kappa}^{-1}: H_{ \sigma,R_{\infty}}^{N_+}\to  H_{\sigma, R_{\infty},2}^{N_++1}
			\end{equation*}
	is a bounded linear operator and there exists a constant $C_{s,\sigma, R_{\infty}}>0$ such that:
	\begin{equation}
	\label{eq:opboundkappapos}
	\|\mathfrak{L}_{s,\lambda,\kappa}^{-1}\|_{ H_{ \sigma, R_{\infty}}^{N_+}\to  H_{ \sigma, R_{\infty},2}^{N_++1}}\leq C_{s,\sigma, R_{\infty}}.
	\end{equation}
	Furthermore, for all $f\in H_{ \sigma, R_{\infty}}^{N_+}$,
	\begin{equation}
	\label{eq:addestposkappa}
	\kappa \|x\partial_x^2\mathfrak{L}_{s,\lambda,\kappa}^{-1}(f)\|_{G_{\sigma,R_{\infty}}}+\kappa^2 \|\slashed{\Delta}_{\s^2 }\mathfrak{L}_{s,\lambda,\kappa}^{-1}(f)\|_{G_{\sigma,R_{\infty}}}\leq C_{s,\sigma}\|f\|_{H_{ \sigma, R_{\infty}}^{N_+}}.
	\end{equation}
\end{corollary}
\begin{proof}
	Let $\kappa_+>\kappa>0$ and $s\in \Omega_{\sigma}\cap\{{\re z}>(\frac{1}{2}+N_+)\kappa_+\}$. Then there exists a $N_{\kappa}>N_+$ such that ${\re s}>-(N_{\kappa}+\frac{1}{2})\kappa$. Let $f_m\in C^{\infty}([0,r_+^{-1}]_x\times \s^2))\cap H_{ \sigma,R_{\infty}}^{N_+}$, such that $\partial_{\varphi}f_m=im f_m$ for some $m\in \Z$. Then we can apply Theorem \ref{tm:invertkappapos} to conclude that $\mathfrak{L}_{s,\lambda,\kappa}^{-1}(f_m)\in  C^{\infty}([0,r_+^{-1}]_x\times \s^2)$ is well-defined for suitably large $\lambda$.
	
	By Proposition \ref{prop:fullmainest}, it then follows that for $\lambda$ suitably large and $\kappa$ suitably small, we can estimate
	\begin{equation*}
		\|\mathfrak{L}_{s,\lambda,\kappa}^{-1}(f_m)\|_{H^{N_++1}_{\sigma, R_{\infty},2}}^2\leq C_{s,\sigma,\lambda}\|f_m\|_{H^{N_+}_{ \sigma,R_{\infty}}}^2.
	\end{equation*}
	Let  $f\in C^{\infty}([0,r_+^{-1}]_x\times \s^2))\cap H_{ \sigma,R_{\infty}}^{N_+}$. Then we can decompose $f=\sum_{m\in \Z}f_m$, with $\partial_{\varphi}f_m=imf_m$.
	Since the constant in the above estimate does not depend on $m$, we can sum over $m$ to conclude that for general $f\in C^{\infty}([0,r_+^{-1}]_x\times \s^2))\cap H_{ \sigma,R_{\infty}}^{N_+}$, $\mathfrak{L}_{s,\lambda,\kappa}^{-1}(f)=\sum_{m\in \Z}\mathfrak{L}_{s,\lambda,\kappa}^{-1}(f_m)$ is well defined.
	
	Hence, $\mathfrak{L}_{s,\lambda,\kappa}^{-1}(C^{\infty}([0,r_+^{-1}]_x\times \s^2)\cap H_{ \sigma,R_{\infty}}^{N_+})\subseteq H_{ \sigma,R_{\infty},2}^{N_++1}$. Since $C^{\infty}([0,r_+^{-1}]_x\times \s^2)\cap H_{ \sigma,R_{\infty}}^{N_+}$ is dense in  $H_{ \sigma,R_{\infty}}^{N_+}$, we can extend $\mathfrak{L}_{s,\lambda,\kappa}^{-1}$ to act on $H_{ \sigma,R_{\infty}}^{N_+}$ and we must have that $\mathfrak{L}_{s,\lambda,\kappa}^{-1}( H_{ \sigma,R_{\infty}}^{N_+})\subseteq H_{ \sigma,R_{\infty},2}^{N_++1}$.
\end{proof}

\begin{proposition}
\label{prop:invcalL}
	Let $\sigma\in \R$ and take $s\in \Omega_{\sigma}\cap\{{\re z}>-(\frac{1}{2}+N_+)\kappa_+\}$. Then, there exists a suitably large $\lambda \in \N_0$, such that:
	\begin{equation}
	\label{eq:resolventwelldef}
		\mathfrak{L}_{s,\lambda}^{-1}=\mathfrak{L}_{s,\lambda,0}^{-1}: H_{ \sigma,R_{\infty}}^{N_+}\to  H_{ \sigma,R_{\infty},2}^{N_++1}
			\end{equation}
	is a bounded linear operator and there exists a constant $ C_{s,\sigma}>0$ such that:
	\begin{equation*}
	\|\mathfrak{L}_{s,\lambda}^{-1}\|_{ H_{ \sigma,R_{\infty}}^{N_+}\to  H_{ \sigma,R_{\infty},2}^{N_++1}}\leq  C_{s,\sigma}.
	\end{equation*}
Furthermore, for any sequence $\{\kappa_n\}$, with $\kappa_n\to 0$ as $n\to \infty$
\begin{equation*}
\mathfrak{L}_{s,\lambda,\kappa_n}^{-1}\to \mathfrak{L}_{s,\lambda}^{-1}
\end{equation*}
as $n\to \infty$, with respect to the operator norm on $H_{ \sigma,R_{\infty}}^{N_+}$.
\end{proposition}
\begin{proof}
Note that there exist $\sigma_2^2>\sigma_1^2>\sigma^2$, such that $s\in \Omega_{\sigma_1}\cap\Omega_{\sigma_2}$. Let $f\in C^{\infty}([0,r_+^{-1}]_x\times \s^2)\cap H_{ \sigma_2, R_{\infty}}^{N_++1}$.

Let $\kappa_n$ be a monotonically non-increasing sequence with $\kappa_n>0$ and $\kappa_n\to 0$ as $n\to \infty$. By Corollary \ref{cor:unifboundinvkappapos}, $\hpsi_n:=\mathfrak{L}_{s,\lambda,\kappa_n}^{-1}(f)$ is a well-defined, uniformly bounded sequence in $H_{\sigma_2, R_{\infty},2}^{N_++1}$. By \eqref{eq:compact2}, $H_{ \sigma_2,R_{\infty},2}^{N_++1}\ssubset H_{\sigma_1, R_{\infty},2}^{N_++1}$, so $\hpsi_n$ admits a convergent subsequence $\{\hpsi_{n_j}\}$ with respect to $H_{\sigma_1, R_{\infty},2}^{N_++1}$, with limit $\hpsi\in  H_{\sigma_1, R_{\infty},2}^{N_++1}$. In fact, since $\{\hpsi_{n_j}\}$ is a bounded sequence in $H_{\sigma_2,R_{\infty},2}^{N_++1}$, it admits a further subsequence with a weak limit in $H_{\sigma_2,R_{\infty},2}^{N_++1}$. By the uniqueness of limits, this limit must also be $\hpsi$, so in fact $\hpsi\in  H_{\sigma_2, R_{\infty},2}^{N_++1}$.

We will now show that $\mathfrak{L}_{s,\lambda,0}(\hpsi)=f$. We have that
\begin{equation}
\label{eq:splittingkappalimit}
	\mathfrak{L}_{s,\lambda,0}(\hpsi)=f-(\mathfrak{L}_{s,\lambda,\kappa_{n_j}}-\mathfrak{L}_{s,\lambda,0})(\hpsi)+\mathfrak{L}_{s,\lambda,\kappa_n}(\hpsi-\hpsi_{n_j}).
\end{equation}
Note that:
\begin{equation*}
	(\mathfrak{L}_{s,\lambda,\kappa_{n_j}}-\mathfrak{L}_{s,\lambda,0})(\hpsi)=\kappa_{n_j}\chi_{\varrho \geq R_{\infty}}( x \partial_x^2\hpsi+\partial_x\hpsi).
\end{equation*}
Hence, by Proposition \ref{prop:producttotalnorm} 
\begin{multline*}
	\|(\mathfrak{L}_{s,\lambda,\kappa_{n_j}}-\mathfrak{L}_{s,\lambda,0})(\hpsi)\|_{H^{N_+}_{ \sigma,R_{\infty}}}\leq 2\kappa_{n_j} (\|\chi_{\varrho \geq R_{\infty}}\hpsi^{(1)}\|_{H^{N_+}_{ \sigma,R_{\infty}}}+\|x \chi_{\varrho \geq R_{\infty}}\hpsi^{(2)}\|_{H^{N_+}_{ \sigma,R_{\infty}}})\\
	\leq C\kappa_{n_j} (\|\hpsi\|_{H^{N_++2}_{\sigma,R_{\infty},2}}+\|x\hpsi^{(2)}\|_{G_{\sigma,R_{\infty}}}).
\end{multline*}

To bound $\|x\hpsi^{(2)}\|_{G_{\sigma, R_{\infty}}}$, we first apply Lemma \ref{lm:productGnorm}, from which it follows that
\begin{equation*}
\|x\hpsi^{(2)}\|_{G_{\sigma,R_{\infty}}}\leq C\|\hpsi^{(2)}\|_{G_{\sigma,R_{\infty}}}\leq C\|\hpsi^{(1)}\|_{G_{\sigma,R_{\infty},2}}.
\end{equation*}
Then we observe that we can use that $\sigma_1>\sigma$ to further estimate:
\begin{equation*}
\|\hpsi^{(1)}\|_{G_{\sigma,R_{\infty},2}}\leq C \|\hpsi \|_{G_{\sigma_1,R_{\infty},2}}.
\end{equation*}

We can similarly bound:
\begin{equation*}
	\|\mathfrak{L}_{s,\lambda,\kappa_{n_j}}(\hpsi-\hpsi_{n_j})\|_{H^{N_+}_{ \sigma,R_{\infty}}}\leq C \|\hpsi-\hpsi_{n_j}\|_{H^{N_++2}_{\sigma_1,R_{\infty},2}}\to 0\quad \textnormal{as $j\to \infty$}.
\end{equation*}

Taking the limit $j\to \infty$ in \eqref{eq:splittingkappalimit} therefore gives $\mathfrak{L}_{s,\lambda,0}(\hpsi)=f$. Hence, $\mathfrak{L}_{s,\lambda,0}^{-1}: C^{\infty}([0,r_+^{-1}]_x\times \s^2))\cap H_{\sigma_2, R_{\infty}}^{N_+}\to H_{\sigma_2, R_{\infty},2}^{N_++1}$ is well-defined. Using density of $ C^{\infty}([0,r_+^{-1}]_x\times \s^2))\cap H_{\sigma_2, \delta_2}^{N_+}$ in $H_{\sigma_2, \delta_2}^{N_+}$, we conclude  \eqref{eq:resolventwelldef}.

To conclude uniform operator convergence, we first write:
\begin{equation*}
\mathfrak{L}_{s,\lambda,\kappa_{m}}^{-1}-\mathfrak{L}_{s,\lambda,\kappa_{n}}^{-1}=\mathfrak{L}_{s,\lambda,\kappa_{m}}^{-1}\circ (\mathfrak{L}_{s,\lambda,\kappa_{n}}-\mathfrak{L}_{s,\lambda,\kappa_{m}})\circ \mathfrak{L}_{s,\lambda,\kappa_{n}}^{-1}
\end{equation*}
Let $f\in H_{\sigma,R_{\infty}}^{N_+}$, then $\hpsi:=\mathfrak{L}_{s,\lambda,\kappa_{n}}^{-1}(f)\in H_{\sigma,R_{\infty},2}^{N_++1}$. Furthermore,
\begin{equation*}
(\mathfrak{L}_{s,\lambda,\kappa_{n}}-\mathfrak{L}_{s,\lambda,\kappa_{m}})(\hpsi)=(\kappa_n-\kappa_m)\partial_x(\chi_{\varrho \geq R_{\infty}} x \partial_x\hpsi),
\end{equation*}
so $(\mathfrak{L}_{s,\lambda,\kappa_{n}}-\mathfrak{L}_{s,\lambda,\kappa_{m}})(\hpsi)$ vanishes in $\varrho\leq \frac{1}{2}R_{\infty}$. Define
\begin{equation*}
\hpsi':=\chi_{\varrho \geq R_{\infty}}  \partial_x\hpsi\in H_{\sigma,R_{\infty}}^{N_+},
\end{equation*}
then $\|\mathfrak{L}_{s,\lambda,\kappa_{m}}^{-1}(\hpsi')\|_{H_{\sigma, R_{\infty},2}^{N_++1}}\leq C_{s,\sigma}\|f\|_{ H_{\sigma,R_{\infty}}^{N_+}}$ by \eqref{eq:opboundkappapos}. We now express:
\begin{multline}
\label{eq:keyeqdiffkappa}
(\mathfrak{L}_{s,\lambda,\kappa_{n}}-\mathfrak{L}_{s,\lambda,\kappa_{m}})(\hpsi)=(\kappa_n-\kappa_m) \partial_x(x\hpsi')=(\kappa_n-\kappa_m) \partial_x(x\mathfrak{L}_{s,\lambda,\kappa_{m}}\mathfrak{L}_{s,\lambda,\kappa_{m}}^{-1}(\hpsi'))\\
=(\kappa_n-\kappa_m)\mathfrak{L}_{s,\lambda,\kappa_{m}}( \partial_x\circ  x\cdot )\mathfrak{L}_{s,\lambda,\kappa_{m}}^{-1}(\hpsi')+(\kappa_n-\kappa_m)[\partial_x\circ x,\mathfrak{L}_{s,\lambda,\kappa_{m}}\cdot] \mathfrak{L}_{s,\lambda,\kappa_{m}}^{-1}(\hpsi')
\end{multline}
and hence
\begin{equation*}
(\mathfrak{L}_{s,\lambda,\kappa_{m}}^{-1}-\mathfrak{L}_{s,\lambda,\kappa_{n}}^{-1})(f)=(\kappa_n-\kappa_m) \left(\partial_x(x\mathfrak{L}_{s,\lambda,\kappa_{m}}^{-1}(\hpsi'))+\mathfrak{L}_{s,\lambda,\kappa_{m}}^{-1} [\partial_x\circ x\cdot,\mathfrak{L}_{s,\lambda,\kappa_{m}}] \mathfrak{L}_{s,\lambda,\kappa_{m}}^{-1}(\hpsi')\right).
\end{equation*}
For $x\leq x_0$, we have that
\begin{equation*}
[\partial_x\circ x\cdot,\mathfrak{L}_{s,\lambda,\kappa_{m}}] =-\kappa_m x\partial_x^2-(2s+\kappa_m)\partial_x.
\end{equation*}
Using that $\|\mathfrak{L}_{s,\lambda,\kappa_{m}}^{-1}(\hpsi')\|_{H_{\sigma, R_{\infty},2}^{N_++1}}\leq C_{s,\sigma}\|f\|_{ H_{\sigma,R_{\infty}}^{N_+}}$ and additionally that $\kappa_m \|x\partial_x^2\mathfrak{L}_{s,\lambda,\kappa_{m}}^{-1}(\hpsi')\|_{G_{\sigma,\delta}}\leq C_{s,\sigma}\|f\|_{ H_{\sigma,R_{\infty}}^{N_+}}$ by \eqref{eq:addestposkappa}, we conclude that
\begin{equation*}
\|(\mathfrak{L}_{s,\lambda,\kappa_{m}}^{-1}-\mathfrak{L}_{s,\lambda,\kappa_{n}}^{-1})(f)\|_{H_{\sigma,R_{\infty}}^{N_+}}\leq C\cdot C_{s,\sigma}|\kappa_n-\kappa_m\||f\|_{ H_{\sigma,R_{\infty}}^{N_+}}.
\end{equation*}
Since $\{\kappa_n\}$ is converges to 0, it is Cauchy, and by the above estimate, $\{\mathfrak{L}_{s,\lambda,\kappa_{n}}^{-1}\}$ must be Cauchy in the Banach space of bounded linear operators $B(H_{\sigma,R_{\infty}}^{N_+},H_{\sigma,R_{\infty}}^{N_+})$ with respect to the operator norm, so it converges. By uniqueness of limits, the limit must be equal to the operator $\mathfrak{L}_{s,\lambda,0}^{-1}$ defined above.
\end{proof}
\section{Inverting $L_s$: boundedness, compactness and Fredholm theory}
\label{sec:invtLs}
In this section, we will establish the meromorphicity of the maps
\begin{equation*}
s\mapsto L_{s,\kappa}^{-1}
\end{equation*}
with $\kappa\geq 0$.

By \eqref{eq:Lsaspert} and \eqref{eq:defLskappa}, we have the following relation between $L_{s,\kappa}$ and $\mathfrak{L}_{s,\lambda,\kappa}$.
\begin{equation*}
L_{s,\kappa}=\mathfrak{L}_{s,\lambda,\kappa}+\epsilon_0 \cdot B_{s}+K_{s,\lambda},
\end{equation*}
where $\epsilon_0>0$ can be taken arbitrarily small.

\begin{proposition}
\label{prop:boundandcompact}
Let $\lambda\in \R$. Assume that $\mathbbm{h}\equiv 1$ in $\{\varrho\leq \frac{1}{2}R_{\infty}\}$. Then the linear operators
\begin{align*}
B_s: H_{\sigma, R_{\infty},2}^{N_++1}\to H_{\sigma,R_{\infty}}^{N_+},\\
K_{s,\lambda}: H_{\sigma, R_{\infty},2}^{N_++1}\to   H_{\sigma, R_{\infty},2}^{N_++1}
\end{align*}
are bounded independently of $\epsilon_0$. Let $\mathcal{R}:H_{\sigma,R_{\infty}}^{N_+} \to H_{\sigma, R_{\infty},2}^{N_++1}$ be bounded. Then $K_{s,\lambda}\circ \mathcal{R}: H_{\sigma,R_{\infty}}^{N_+}\to H_{\sigma, R_{\infty},2}^{N_++1}\subset H_{\sigma,R_{\infty}}^{N_+}$ is compact.
\end{proposition}
\begin{proof}
Note first that the compactness of $K_{s,\lambda}\circ \mathcal{R}$ follows directly from the boundedness of $K_{s,\lambda}$ by the fact that $H_{\sigma, R_{\infty},2}^{N_++1}\ssubset H_{\sigma,R_{\infty}}^{N_+}$ by \eqref{eq:compact3}.

Recall that
\begin{align*}
B_s=&\:M\epsilon_0^{-1} x\chi_{\varrho\geq R_{\infty}} {B}^{\infty}_s,\\
K_{s,\lambda}=&\: \lambda\pi_{\leq l}+\chi_{\varrho\geq R_{\infty}}S_{\lambda} \pi_{\leq l}+{K}^{\infty}_s+\chi_{\varrho\leq \frac{1}{2}R_{\infty}}{K}_{s}^+.
\end{align*}
We first consider $B_s$. Take $M^{-1}R_{\infty}$ to be suitably large, so that $\epsilon_0^{-1} x\leq 1$ in $\supp \chi_{\varrho\geq R_{\infty}}$. We need to establish boundedness of $\chi_{\varrho\geq R_{\infty}} {B}^{\infty}_s$.

Recall moreover that
\begin{align*}
B^{\infty}_sf=&\:b_{xx}\cdot x^2\partial_x^2f+s b_{x\tau }\cdot \partial_x  f+s^2b_{\tau \tau}f+a^2( xb_{x \vartheta}\cdot x^2\partial_x+s x b_{\tau \vartheta})  \sin\vartheta \partial_{\vartheta}f \\
+&\:a^4x^{4}\cdot b_{\vartheta \vartheta}(\sin\vartheta\partial_{\vartheta})^2f+ b_{x} \cdot x \partial_xf+ a^2x^2b_{\vartheta}\cdot \sin \vartheta\partial_{\vartheta} f+sb_{\tau}f+ b\cdot  f.
\end{align*}

We apply Proposition \ref{prop:producttotalnorm} to estimate $\mathcal{B}^{\infty}_sf$:
\begin{multline*}
\|{B}_{s}^{\infty}f\|_{H^{N_+}_{\sigma,\delta}}\leq C_s \left[\sum_{\star\in \{x,\tau,\vartheta\}} \|b_{\star \star}\|_{H^{N_+}_{ \sigma,R_{\infty};\textnormal{coeff}}}+\|b_{\star}\|_{H^{N_+}_{ \sigma,R_{\infty};\textnormal{coeff}}}+\|b\|_{H^{N_+}_{ \sigma,R_{\infty};\textnormal{coeff}}}\right]\\
\times \left(\|f\|_{H^{N_+}_{\sigma,R_{\infty}}}+\|\partial_xf\|_{H^{N_+}_{\sigma,R_{\infty}}}+a^2\|x^2\sin \vartheta \partial_{\vartheta}f\|_{H^{N_+}_{\sigma,R_{\infty}}}+\| \chi_{\varrho\geq R_{\infty}} x^2\partial_x^2f\|_{H^{N_+}_{\sigma,R_{\infty}}}\right)\\
+a^2\| \chi_{\varrho\geq R_{\infty}} x^4(\sin \vartheta \partial_{\vartheta})^2f\|_{H^{N_+}_{\sigma,R_{\infty}}})\\
\leq C_s \left[\sum_{\star\in \{x,\tau,\vartheta\}} \|b_{\star \star}\|_{H^{N_+}_{ \sigma,R_{\infty};\textnormal{coeff}}}+\|b_{\star}\|_{H^{N_+}_{ \sigma,R_{\infty};\textnormal{coeff}}}+\|b\|_{H^{N_+}_{ \sigma,R_{\infty};\textnormal{coeff}}}\right] \|f\|_{H^{N_++1}_{\sigma, R_{\infty},2}}\\
\leq C_{{\rm coeff},s}\cdot \|f\|_{H^{N_++1}_{\sigma, R_{\infty},2}},
\end{multline*}
where the final inequality follows from the analyticity properties of $b, b_{\star}$ and $b_{\star \star}$ that were established in Lemma \ref{lm:operatorcoeffb}. Note moreover that we used here that the $H^{N_++1}_{\sigma, R_{\infty},2}$ norm controls two more derivatives (with degenerate factors) than the $H^{N_+}_{\sigma,R_{\infty}}$ norm in the region $\varrho\geq \frac{1}{2}R_{\infty}$ and one more derivative globally. Boundedness of $\|B_s\|$ follows by combining the above estimates.

We now turn to $K_{s,\lambda}$. Recall that
\begin{align*}
K^{\infty}_sf=&\:-2M^2h(0) s^2 f,\\
r^{-2}\rho^{2}{K}_{s}^+f=&\:\left[2r-\frac{d}{dr}(\Delta \mathbbm{h})\right]s (r^{-1}f)+[\mathbbm{h}^2\Delta-2\mathbbm{h}(r^2+a^2)+a^2\sin^2\theta] s^2 f\\
&-2a(\mathbbm{h}-1) s\partial_{\varphi_*} f.
\end{align*}
In order for all terms in $K_s$ to be zeroth-order terms, we use the assumption $\mathbbm{h}\equiv 1$ in $\varrho\leq R_{\infty}$.

Hence, applying Proposition \ref{prop:producttotalnorm}, it is straightforward to estimate
\begin{equation*}
\|\chi_{\varrho\leq \frac{1}{2}R_{\infty}}{K}_{s}^+f\|_{H^{N_++1}_{\sigma, R_{\infty},2}}\leq C_{{\rm coeff},s}\cdot \|f\|_{H^{N_++1}_{\sigma, R_{\infty},2}}.
\end{equation*}

Boundedness of $K^{\infty}_s, \chi_{\varrho\leq \frac{1}{2}R_{\infty}}{K}_{s}^+: H_{s,R_{\infty},2}^{N_++1}\to  H_{s,R_{\infty},2}^{N_++1}$ is therefore immediate. Consider $\lambda\pi_{\leq l}$ and $\chi_{\varrho\geq R_{\infty}}S_{\lambda} \pi_{\leq l}$. The operator $\lambda\pi_{\leq l}:  H_{s,R_{\infty},2}^{N_++1}\to  H_{s,R_{\infty},2}^{N_++1}$ can easily be seen to be bounded. Finally, since $\int_{\s^2}|S_{\lambda} \pi_{\leq l}f|^2\,d\upsigma\leq \lambda^2 \int_{\s^2} |\snabla_{\s^2}f_{\leq l}|^2\leq \lambda^4\int_{\s^2}|f_{\leq l}|^2\,d\upsigma$, we have that
\begin{equation*}
\chi_{\varrho\geq R_{\infty}}S_{\lambda} \pi_{\leq l}:  H_{\sigma, R_{\infty},2}^{N_++1}\to  H_{\sigma, R_{\infty},2}^{N_++1}.
\end{equation*}
is bounded.
\end{proof}

Recall the definitions of a meromorphic family of operators in Definition \ref{def:meromfamily} and the definition of the cokernel of an operator in Definition \ref{eq:comode}. In the proposition below, we establish meromorphicity of the inverses $L_{s,\kappa}^{-1}$ with $\kappa\geq0$, and we relate the poles to the cokernel of $L_{s,\kappa}$.

\begin{proposition}
\label{prop:mainprop}
Let $\sigma\in \R$ and $N_+\in \N_0$.  For $\kappa\geq 0$ sufficiently small and $M^{-1}R_{\infty}$ sufficiently large:
\begin{enumerate}[label=\emph{(\roman*)}]
\item The family of operators $\{L_{s,\kappa}^{-1}: H_{\sigma,R_{\infty}}^{N_+}\to H_{\sigma, R_{\infty},2}^{N_++1}\}$ with $s\in \Omega_{\sigma}\cap\{{\re s}>-(\frac{1}{2}+N_+)\kappa_+\}$ is meromorphic. We denote with $\mathscr{Q}_{\sigma,R_{\infty},N_+,\kappa}$ the set of poles of $\{L_{s,\kappa}^{-1}\}$. The set $\mathscr{Q}_{\sigma,R_{\infty},N_+,\kappa}$ is independent of the precise choice of $\mathbbm{h}$.
\item Let $s\in \Omega_{\sigma}\cap\{{\re s}>-(\frac{1}{2}+N_+)\kappa_+\}$. Then the space $ \ker L_{s}$ is finite dimensional with $\dim \ker L_{s,\kappa}=\dim \textnormal{coker}\,  L_{s,\kappa}$ and $\textnormal{ran}\,L_{s,\kappa}^*=(\ker L_{s,\kappa})^{\perp}\subset H_{\sigma, R_{\infty},2}^{N_++1}$, with $L_{s,\kappa}^*$ the Hermitian adjoint and $\textnormal{coker}\, L_{s,\kappa}^*$ is isomorphic to $\ker L_{s,\kappa}$. 
	\end{enumerate}
\end{proposition}
\begin{proof}
Without loss of generality, we consider the $\kappa=0$ case. The argument proceeds in exactly the same way for $\kappa>0$. We first suppose that $\mathbbm{h}=1$ in $\{\varrho\leq \frac{1}{2}R_{\infty}\}$. By Proposition \ref{prop:invcalL}, $\mathfrak{L}_{s,\lambda}^{-1}: H_{\sigma,R_{\infty}}^{N_+}\to H_{\sigma, R_{\infty},2}^{N_++1}$ is well-defined for suitably large $\lambda$. We can express:
\begin{equation*}
\mathfrak{L}_{s,\lambda}+\epsilon_0 B_s=(\mathbf{1}+\epsilon_0 B_s\mathfrak{L}_{s,\lambda}^{-1})\mathfrak{L}_{s,\lambda}.
\end{equation*}
By Proposition \ref{prop:boundandcompact}, $B_s: H_{\sigma, R_{\infty},2}^{N_++1}\to H_{\sigma,R_{\infty}}^{N_+}$ is bounded, so if we take $\epsilon_0>\frac{1}{\|B_s\|\cdot  \|\mathfrak{L}_{s,\lambda}^{-1}\|}$, we have that $\epsilon_0 B_s\mathfrak{L}_{s,\lambda}^{-1}:  H_{\sigma,R_{\infty}}^{N_+}\to  H_{\sigma,R_{\infty}}^{N_+}$ is bounded with operator norm less than 1. Hence, $(\mathbf{1}+\epsilon_0 B_s\mathfrak{L}_{s,\lambda}^{-1})^{-1}: H_{\sigma,R_{\infty}}^{N_+}\to  H_{\sigma,R_{\infty}}^{N_+}$ is a bounded operator that we can express as follows in terms of a Neumann series:
\begin{equation*}
(\mathbf{1}+\epsilon_0 B_s\mathfrak{L}_{s,\lambda}^{-1})^{-1}=\sum_{k=0}^{\infty}(-1)^k(\epsilon_0 B_s\mathfrak{L}_{s,\lambda}^{-1})^k.
\end{equation*}
We now conclude that
\begin{equation*}
-T_s:=(\mathfrak{L}_{s,\lambda}+\epsilon_0 B_s)^{-1}=\mathfrak{L}_{s,\lambda}^{-1}(\mathbf{1}+\epsilon_0 B_s\mathfrak{L}_{s,\lambda}^{-1})^{-1}: H_{\sigma,R_{\infty}}^{N_+}\to H_{\sigma, R_{\infty},2}^{N_++1}
\end{equation*}
is a well-defined bounded linear operator and we can express:
\begin{equation*}
L_s=-(\mathbf{1}-K_{s,\lambda} T_s)T_s^{-1}
\end{equation*}
on $H_{\sigma, R_{\infty},2}^{N_++1}$.

Formally, $L_s^{-1}=-T_s(\mathbf{1}-K_{s,\lambda} T_s)^{-1}$. By the analytic Fredholm theorem (see for example \cite{rero06}[Theorem 8.92]), $(\mathbf{1}-K_{s,\lambda} T_s)^{-1}: H_{\sigma,R_{\infty}}^{N_+}\to H_{\sigma,R_{\infty}}^{N_+}$ defines a meromorphic family of bounded operators if $\{K_{s,\lambda}T_s\}$ is holomorphic and compact as an operator from $H_{\sigma,R_{\infty}}^{N_+}$ to $H_{\sigma,R_{\infty}}^{N_+}$. Compactness follows directly from Proposition \ref{prop:boundandcompact}. Consider the following difference quotient: let 
\begin{multline*}
\frac{K_{s,\lambda}T_s-K_{s_0,\lambda}T_{s_0}}{s-s_0}(f)=\frac{K_{s,\lambda}-K_{s_0,\lambda}}{s-s_0}T_s+K_{s_0,\lambda}\frac{T_s-T_{s_0}}{s-s_0}\\
=\frac{K_{s,\lambda}-K_{s_0,\lambda}}{s-s_0}T_s-K_{s_0,\lambda}T_s\frac{\mathfrak{L}_{s,\lambda}-\mathfrak{L}_{s_0,\lambda}+\epsilon_0 (B_s-B_{s_0})}{s-s_0}T_{s_0}\\
=\frac{K_{s,\lambda}-K_{s_0,\lambda}}{s-s_0}T_s-K_{s_0,\lambda}T_s\frac{L_{s}-{L}_{s_0}-(K_{s,\lambda}-K_{s_0,\lambda})}{s-s_0}T_{s_0}
\end{multline*}
Note that
\begin{align*}
\frac{L_{s}-{L}_{s_0}}{s-s_0}(f)=&\:2 \partial_{x}f+M x\left[b_{x\tau }\partial_{x}f+(s+s_0)b_{\tau \tau}f+x b_{\tau \vartheta})  \sin\vartheta \partial_{\vartheta}f+b_{\tau}f\right]-2M^2(s+s_0)h(0)f,\\
\frac{K_{s,\lambda}-K_{s_0,\lambda}}{s-s_0}=&\:-2M^2(s+s_0)h(0)f
\end{align*}
By the above expressions, it is straightforward to see that $\lim_{s\to s_0}\frac{K_{s,\lambda}T_s-K_{s_0,\lambda}T_{s_0}}{s-s_0}$ is a well-defined operator from $H_{\sigma,R_{\infty}}^{N_+}$ to $H_{\sigma,R_{\infty}}^{N_+}$.

We conclude that $L_s^{-1}: H_{\sigma,R_{\infty}}^{N_+}\to H_{\sigma,\delta,2}^{N_++2}$ constitute a meromorphic family of operators, so (i) holds when $\mathbbm{h}=1$ in $\{\varrho\leq \frac{1}{2}R_{\infty}\}$.

Since $K_{s,\lambda}T_{s}$ is compact, we have by the Fredholm alternative, that
\begin{equation*}
\dim \ker (T_sL_s)^*=\dim \ker (\mathbf{1}-(K_{s,\lambda}T_s)^*)=\dim  \ker (\mathbf{1}-(K_{s,\lambda}T_s)^*)=\dim \ker (T_sL_s)<\infty
\end{equation*}
and
\begin{equation*}
\textnormal{ran}((T_sL_s)^*)=\textnormal{ran}\,(\mathbf{1}-(K_{s,\lambda}T_{s})^*)=[\ker (\mathbf{1}-K_{s,\lambda}T_s)]^{\perp}=(\ker (T_sL_s))^{\perp},
\end{equation*}
with $T_s^*$ and $L_s^*$ the Hilbert space adjoints of $T_s$ and $L_s$, respectively. By the boundedness of $T_s$ and $T_s L_s$, we have that $(T_sL_s)^*=L_s^*T_s^*: H_{\sigma, R_{\infty},2}^{N_++1}\to H_{\sigma, R_{\infty},2}^{N_++1}$.

Suppose $\hpsi\in \ker L_s$. Then $L_s\hpsi=0$, so $T_sL_s\hpsi=0$, so $\hpsi\in \ker (T_sL_s)$. Suppose now that $\hpsi\in \ker (T_sL_s)$, then $T_sL_s\hpsi=0$. Since $\ker T_s=\{0\}$, we must have that $L_s\hpsi=0$, so $\ker L_s= \ker (T_sL_s)$. Similarly, since $\ker T_s=\{0\}$, we have that $\textnormal{ran}\,T_s^*=H_{\sigma,R_{\infty}}^{N_+}$, so $\textnormal{ran}\, L_s^*=\textnormal{ran}\, (L_s^* T_s^*)$.

Finally, observe that 
\begin{equation*}
\textnormal{coker}\,(L_s^*)=H_{\sigma,R_{\infty}}^{N_+}/\textnormal{ran}\,(L_s^*)\cong(\textnormal{ran}\,(L_s^*))^{\perp}=(\textnormal{ran}\,((T_sL_s)^*))^{\perp}=\ker (T_sL_s)=\ker L_s
\end{equation*}
and we conclude (ii) when $\mathbbm{h}=1$ in $\{\varrho\leq \frac{1}{2}R_{\infty}\}$.

Now, consider the case of more general $\mathbbm{h}$, which do not need to satisfy $\mathbbm{h}\equiv 1$ in $\{\varrho\leq \frac{R_{\infty}}{2}\}$. Denote with $\Sigma^{\mathbbm{h}}$ the hypersurface $\Sigma$ corresponding to the more general function $\mathbbm{h}$. Similarly, denote with $L_{s,\kappa}^{\mathbbm{h}}$ and $\tau^{\mathbbm{h}}$ the operator $L_s$ and the time function $\tau$ with respect to $\mathbbm{h}$, respectively. Let $\mathbbm{h}_0$ satisfy $\mathbbm{h}_0=1$ in $\{\varrho\leq \frac{1}{2}R_{\infty}\}$. 

By assumption on the hyperboloidal time functions, we must have that $\mathbbm{h}_0-\mathbbm{h}$ is integrable in $r$, so
\begin{equation*}
\tau^{\mathbbm{h}}-\tau^{\mathbbm{h}_0}=\int_{\varrho}^{\infty} (\mathbbm{h}_0-\mathbbm{h})(\varrho')\,d\varrho'
\end{equation*}
is well-defined.

We will show that invertibility of $L_{s,\kappa}^{\mathbbm{h}_0}$ is equivalent to invertibility of $L_{s,\kappa}^{\mathbbm{h}}$.

Let $Q_s: L^2(\Sigma^{\mathbbm{h}})\to L^2(\Sigma^{\mathbbm{h}_0})$ be defined as follows: let $\overline{f}: D^+(\Sigma^{\mathbbm{h}})\to \C$ denote the trivial extension of $f\in L^2(\Sigma^{\mathbbm{h}})$ to $D^+(\Sigma^{\mathbbm{h}})$, i.e.\ $\overline{f}(\tau^{\mathbbm{h}},\varrho,\vartheta,\varphi_*)=\overline{f}(\varrho,\vartheta,\varphi_*)$. Then we define
\begin{equation*}
Q_s(f)(\varrho,\vartheta,\varphi_*)=(e^{s\tau^{\mathbbm{h}}}\overline{f})|_{\Sigma^{\tau^{\mathbbm{h}_0}}}(\varrho,\vartheta,\varphi_*)=e^{s \int_{\varrho}^{\infty} (\mathbbm{h}_0-\mathbbm{h})(\varrho')\,d\varrho'}f(\varrho,\vartheta,\varphi_*)
\end{equation*}
to obtain
\begin{equation*}
L_s^{\mathbbm{h}}=Q_s^{-1}\circ L_s^{\mathbbm{h}_0}\circ Q_s,
\end{equation*}
with $Q_s^{-1}(f)=e^{-s \int_{\varrho}^{R_{\infty}} (\mathbbm{h}_0-\mathbbm{h})(\varrho')\,d\varrho'}f$. Since the factor $e^{s \int_{\varrho}^{R_{\infty}} (\mathbbm{h}_0-\mathbbm{h})(\varrho')\,d\varrho'}$ is analytic in $x=\frac{1}{\varrho}$ near $x=0$ by assumption on $\mathbbm{h}$, we can apply Proposition \ref{prop:producttotalnorm} to conclude that there exists a constant $B>0$ such that
\begin{align*}
\|Q_s(f)\|_{H^{N_+}_{\sigma,R_{\infty}}}\leq&\: B \|f\|_{H^{N_+}_{\sigma,R_{\infty}}}\quad \textnormal{and}\quad \|Q_s(f)\|_{H^{N_++1}_{\sigma,R_{\infty},2}}\leq B \|f\|_{H^{N_++1}_{\sigma,R_{\infty},2}},\\
\|Q_s^{-1}(f)\|_{H^{N_+}_{\sigma,R_{\infty}}}\leq &\:B \|f\|_{H^{N_+}_{\sigma,R_{\infty}}}\quad \textnormal{and}\quad \|Q_s^{-1}(f)\|_{H^{N_++1}_{\sigma,R_{\infty},2}}\leq B \|f\|_{H^{N_++1}_{\sigma,R_{\infty},2}}.
\end{align*}
The restrictions $Q_s,Q_s^{-1}: H^{N_+}_{\sigma,R_{\infty}}\to H^{N_+}_{\sigma,R_{\infty}}$ and  $Q_s,Q_s^{-1}: H^{N_++1}_{\sigma,R_{\infty},2}\to H^{N_++1}_{\sigma,R_{\infty},2}$ are therefore bounded operators, so
\begin{equation*}
(L_s^{\mathbbm{h}})^{-1}=Q_s^{-1}\circ (L_s^{\mathbbm{h}_0})^{-1}\circ Q_s: H^{N_+}_{\sigma,R_{\infty}}\to H^{N_++1}_{\sigma,R_{\infty},2}
\end{equation*}
is well-defined and bounded if and only if $(L_s^{\mathbbm{h}_0})^{-1}: H^{N_+}_{\sigma,R_{\infty}}\to H^{N_++1}_{\sigma,R_{\infty},2}$ is well-defined and bounded. This moreover implies that the poles of $(L_s^{\mathbbm{h}})^{-1}$are independent of the choice of $\mathbbm{h}$, provided that $\mathbbm{h}$ satisfies the assumptions stated in \S \ref{sec:asymphypfol}.
\end{proof}

We denote $\mathscr{Q}_{\sigma,R_{\infty},N_+}:=\mathscr{Q}_{\sigma,R_{\infty},N_+,0}$. In the next proposition, we will show that sets $\mathscr{Q}_{\sigma,R_{\infty},N_+}$ do not depend on the choice of $R_{\infty}$, $N_+$ and $\sigma$, provided $R_{\infty}$ is taken to be suitably large and $\sigma$ is restricted to an appropriate interval. We will also show that the union of $\mathscr{Q}_{\rm reg}:=\bigcup_{\sigma\in \R, N_+\in \N_0}\mathscr{Q}_{\sigma,R_{\infty},N_+}$ has no accumulation points in $\Omega =\left\{z\in \C\setminus\{0\}\,,\,|{\rm arg}(z)|<\frac{2}{3}\pi\right\}$.

We first need to establish suitably convergence properties of the elements in $\mathscr{Q}_{\sigma,R_{\infty},N_+,\kappa}$ as $\kappa\downarrow 0$ and with respect to small perturbations $L_{s,\kappa}+\delta \widetilde{Q}$:
\begin{lemma}
\label{lm:kappaconvergence}
Let $\{\kappa_n\}$ and $\{\delta_n\}$ be a sequences in $[0,\infty)$ such that $\kappa_n\to 0$ and $\delta_n\to 0$ as $n\to \infty$. Define $T_{s,n}:=-(\mathfrak{L}_{s,\lambda,\kappa_n}+\delta_n\cdot \widetilde{Q}+\epsilon_0 B_s)^{-1}$, with $\widetilde{Q}\in B(H_{\sigma, R_{\infty},2}^{N_++1},H_{\sigma,R_{\infty}}^{N_+})$ and let $T_{s}:=-(\mathfrak{L}_{s,\lambda}+\epsilon_0 B_s)^{-1}$. 
\begin{enumerate}[label=\emph{(\roman*)}]
\item Suppose that there exists a $s_*\in  \Omega_{\sigma}\cap\{\re(z)>-(\frac{1}{2}+N_+)\kappa_+\}$ such that $\ker (\mathbf{1}-K_{s_*,\lambda}T_{s_*})\neq \{0\}$. Then there exist a sequence $\{s_n\}$ in $ \Omega_{\sigma}\cap\{\re(z)>-(\frac{1}{2}+N_+)\kappa_+\}$ such that $s_n\to s_*$ as $n\to \infty$ and
\begin{equation*}
\ker (\mathbf{1}-K_{s_n,\lambda}T_{s_n,n})\neq \{0\}.
\end{equation*}
Furthermore, for all $\hpsi_{s_*}\in \ker L_{s_*}$, there exists a sequence $\{\hpsi_{s_n}\}$ in $\ker (\mathbf{1}-K_{s_n,\lambda}T_{s_n,n})$, such that
\begin{equation*}
\|\hpsi_{s_n}-\hpsi_{s_*}\|_{H_{\sigma,R_{\infty}}^{N_+}}\to 0
\end{equation*}
as $n\to \infty$.
\item Conversely, suppose that there exists a sequence $\{s_n\}$ in $\Omega_{\sigma}\cap\{\re(z)>-(\frac{1}{2}+N_+)\kappa_+\}$, $\hpsi_{s_n}$ with $||\hpsi_{s_n}||_{H^{N_+}_{\sigma,R_{\infty}}}=1$ and $\{\kappa_n\}$, such that $s_n\to s_*\in \Omega_{\sigma}\cap\{\re(z)>-(\frac{1}{2}+N_+)\kappa_+\}$, $\kappa_n\downarrow 0$ and $\hpsi_n\in \ker (\mathbf{1}-K_{s_n,\lambda}T_{s_n,n})$. Then there exists a $0\neq \hpsi\in \ker L_{s_*}$.
\end{enumerate}
\end{lemma}
\begin{proof}
 The key ingredients of the proof of (i) are the convergence property $\mathfrak{L}_{s,\lambda,\kappa_n}^{-1}\to \mathfrak{L}_{s,\lambda}^{-1}$ with respect to the operator norm, which we established in \eqref{prop:invcalL}, together with the fact that  $(\mathfrak{L}_{s,\lambda,\kappa}+\delta_n\cdot \widetilde{Q})^{-1}$ is a well-defined bounded operator for $\delta_n< \frac{1}{\|\widetilde{Q}\|\cdot \|\mathfrak{L}_{s,\lambda,\kappa}^{-1}\|}$, which follows from the proof of Proposition \ref{prop:mainprop}, and $(\mathfrak{L}_{s,\lambda,\kappa}+\delta_n\cdot \widetilde{Q})^{-1}$ converges to $\mathfrak{L}_{s,\lambda,\kappa}^{-1}$ with respect to the operator norm. We also use that  $\ker L_s=\ker (\mathbf{1}-K_{s_*,\lambda}T_{s_*})$. Then Hurwitz's theorem is applied. We refer the reader to \cite{gajwar19a}[Proposition 11.6 and 11.7] for the details, where $K_{s_n,\lambda}T_{s_n,n}$ replaces the operator $A_{\kappa_n,s}$ from \cite{gajwar19a}.
 
We prove (ii) by first proceeding as in the proof of Proposition \ref{prop:invcalL}. Let $\sigma'<\sigma$. Then $\hpsi_{s_n}\in H^{N_++1}_{\sigma,R_{\infty},2} \ssubset H^{N_++2}_{\sigma',R_{\infty},2}$, so $\{\hpsi_n\}$ admits a convergent subsequence with non-trivial limit $\hpsi_{s_*} \in H^{N_++1}_{\sigma',R_{\infty},2}$. Since $\{\hpsi_n\}$ is also weakly convergent in $H^{N_++1}_{\sigma,R_{\infty},2}$ and there is uniqueness of limits, we must have that $\hpsi_{s_*} \in H^{N_++1}_{\sigma,R_{\infty},2}$.

We will now show that $\hpsi_{s_*}\in \ker L_{s_*}$. By splitting
\begin{equation*}
	\mathbf{1}-K_{s_*,\lambda}T_{s_*}=(\mathbf{1}-K_{s_n,\lambda}T_{s_n,n})+(K_{s_n,\lambda}T_{s_n,n}-K_{s_*,\lambda}T_{s_n}),
\end{equation*}
we obtain
\begin{equation*}
	(\mathbf{1}-K_{s_*,\lambda}T_{s_*})\hpsi_{s_*}=(\mathbf{1}-K_{s_*,\lambda}T_{s_*})(\hpsi_{s_*}-\hpsi_{s_n})+(K_{s_n,\lambda}T_{s_n,n}-K_{s_*,\lambda}T_{s_n})\hpsi_{s_*}.
\end{equation*}
The terms on the right-hand side go to zero in $H^{N_++1}_{\sigma',R_{\infty},2}$ as $n\to \infty$, so we conclude that $\hpsi_{s_*}\in \ker L_{s_*}$.
\end{proof}

\begin{proposition}
\label{prop:accumulation}
Let $N+\in \N_0$ and let $R_{\infty}>0$ be suitably large, such that Proposition \ref{prop:mainprop} applies. 
\begin{enumerate}[label=\emph{(\roman*)}]
\item Let $\tilde{N}_+\geq N_+$. Then $\mathscr{Q}_{\sigma,R_{\infty},\tilde{N}_+}=\mathscr{Q}_{\sigma, R_{\infty},N_+}$, $\ker L_s\subset C^{\infty}(\hat{\Sigma})$ and $\ker (\mathfrak{L}_{s,\lambda,\kappa}+\epsilon_0 B_s+K_{s,\lambda})\subset C^{\infty}(\hat{\Sigma})$ for $s\in \Omega_{\sigma}\cap\{\re(z)>-(\frac{1}{2}+N_+)\kappa_+\}$ with $\sigma\in \R$.
\item Let $\widetilde{R}_{\infty}\geq R_{\infty}$. Then $\mathscr{Q}_{\sigma,\widetilde{R}_{\infty},N_+}=\mathscr{Q}_{\sigma, R_{\infty},N_+}$.
\item Let $|\sigma'|>|\sigma|$ and $\Omega_{\sigma'}\cap \Omega_{\sigma}\neq \emptyset$. Then
\begin{equation*}
\mathscr{Q}_{\sigma,R_{\infty},N_+}\cap (\Omega_{\tilde{\sigma}}\cap \Omega_{\sigma})=\mathscr{Q}_{\sigma',R_{\infty},N_+}\cap (\Omega_{\sigma'}\cap \Omega_{\sigma}).
\end{equation*}
In particular, $\mathscr{Q}_{\rm reg}=\bigcup_{\sigma\in \R, N_+\in \N_0}\mathscr{Q}_{\sigma,R_{\infty},N_+}$ has no accumulation points in\\ $\Omega =\left\{z\in \C\setminus\{0\}\,,\,|{\rm arg}(z)|<\frac{2}{3}\pi\right\}$.
\end{enumerate}
\end{proposition}
\begin{proof}
We will first prove (i). Let $\hpsi_s\in \ker L_s\subset H^{N_++1}_{\sigma,R_{\infty},2}$. Then
\begin{equation*}
(\mathfrak{L}_{s,\lambda}+\epsilon_0 B_s)\hpsi_s=-K_{s,\lambda}\hpsi_s.
\end{equation*}
In the  proof of Proposition \ref{prop:mainprop}, we showed that for $\epsilon_0$ suitably small, $(\mathfrak{L}_{s,\lambda}+\epsilon_0 B_s)^{-1}: H_{\sigma, R_{\infty}}^{N_+}\to H_{\sigma, R_{\infty},2}^{N_++1}$ is a well-defined bounded operator and hence, there exists a constant $C>0$, such that
\begin{equation}
\label{eq:estkerLs}
\|\hpsi_s\|_{H^{N_++1}_{\sigma,R_{\infty},2}}=\|(\mathfrak{L}_{s,\lambda}+\epsilon_0 B_s)^{-1}K_{s,\lambda}\hpsi_s\|_{H^{N_++1}_{\sigma,R_{\infty},2}}\leq C \|K_{s,\lambda}\hpsi_s\|_{H^{N_+}_{\sigma,R_{\infty}}}.
\end{equation}
By Proposition \ref{prop:boundandcompact}, we have in fact that $\|K_{s,\lambda}\hpsi_s\|_{H^{N_++1}_{\sigma,R_{\infty},2}}<\infty$.

Now let $\tilde{N}_+=N_++1$, then \eqref{eq:estkerLs} applies with $\tilde{N}_+$ replacing $N_+$:
\begin{equation*}
\|\hpsi_s\|_{H^{\tilde{N}_++1}_{\sigma,R_{\infty},2}}\leq C \|K_{s,\lambda}\psi_s\|_{H^{\tilde{N}_+}_{\sigma,R_{\infty}}}\leq C \|K_{s,\lambda}\hpsi_s\|_{H^{N_++1}_{\sigma,R_{\infty},2}}<\infty.
\end{equation*}
Therefore $\hpsi_s\in H^{N_++2}_{\sigma,R_{\infty},2}$.

Conversely, suppose that $\hpsi\in \ker L_s\subset H^{\tilde{N}_++1}_{\sigma,R_{\infty},2}$ with $\tilde{N}_+=N_++1$. By \eqref{eq:estkerLs}, we have that
\begin{equation*}
\|\hpsi_s\|_{H^{N_++1}_{\sigma,R_{\infty},2}}\leq  \|K_{s,\lambda}\hpsi_s\|_{H^{N_+}_{\sigma,R_{\infty}}}\leq \|K_{s,\lambda}\hpsi_s\|_{H^{\tilde{N}_++1}_{\sigma,R_{\infty},2}}<\infty,
\end{equation*}
where we used in particular that for $w_{\tilde{p}}$ defined with $\tilde{p}= 4(\tilde{N}_++1)+2=4(N_++1)+2+4$ and $w_p$ defined with $p= 4(N_++1)+2$:
\begin{equation*}
\|(\mathbf{1}-\slashed{\Delta}_{\s^2})^{-\frac{w_{p}}{4}}K_{s,\lambda}\hpsi_s\|_{H^{N_+}_{\sigma,R_{\infty}}}\leq \|\snabla_{\s^2}(\mathbf{1}-\slashed{\Delta}_{\s^2})^{-\frac{w_{\tilde{p}}}{4}}K_{s,\lambda}\hpsi_s\|_{H^{\tilde{N}_+}_{\sigma,R_{\infty}}} \leq \|(\mathbf{1}-\slashed{\Delta}_{\s^2})^{-\frac{w_{\tilde{p}}}{4}}\hpsi_s\|_{H^{\tilde{N}_++1}_{\sigma,R_{\infty},2}}.
\end{equation*}
Hence,  $\hpsi\in \ker L_s\subset H^{{N}_++1}_{\sigma,R_{\infty},2}$.

We can repeat this procedure inductively to conclude that the kernel of $L_s$ with respect to $H^{N_++1}_{\sigma,R_{\infty},2}$ agrees with the kernel of $L_s$ with respect to $H^{\tilde{N}_++1}_{\sigma,R_{\infty},2}$ for arbitrary $\tilde{N}_+\geq N_+$. Therefore, we must have that $\mathscr{Q}_{\sigma,R_{\infty},\tilde{N}_+}=\mathscr{Q}_{\sigma, R_{\infty},N_+}$.

Note that $(\mathfrak{L}_{s,\lambda}+\epsilon_0 B_s)(\snabla_{\s^2}\hpsi_s)=-\snabla_{\s^2}K_{s,\lambda}\hpsi_s+[\snabla_{\s^2},\mathfrak{L}_{s,\lambda}+\epsilon_0 B_s]\hpsi_s$.

It is straightforward to show that
\begin{equation*}
\|[\snabla_{\s^2},\mathfrak{L}_{s,\lambda}+\epsilon_0 B_s]\hpsi_s\|_{H^{N_+}_{\sigma,R_{\infty}}}\leq C\|\hpsi_s\|_{H^{N_++1}_{\sigma,R_{\infty},2}}.
\end{equation*}
Hence, we can proceed as above to conclude that
\begin{equation*}
\|\snabla_{\s^2}\hpsi_s\|_{H^{N_++1}_{\sigma,R_{\infty},2}}\leq \|\snabla_{\s^2}K_{s,\lambda}\psi_s\|_{H^{N_+}_{\sigma,R_{\infty}}}+\|[\snabla_{\s^2},\mathfrak{L}_{s,\lambda}+\epsilon_0 B_s]\psi_s\|_{H^{N_+}_{\sigma,R_{\infty}}}\leq C\|\psi_s\|_{H^{N_++1}_{\sigma,R_{\infty},2}}<\infty.
\end{equation*}
Hence, $\snabla_{\s^2}\hpsi_s\in H^{N_++1}_{\sigma,R_{\infty},2}$. By repeating this procedure for $\snabla_{\s^2}^k\hpsi_s$, we conclude in particular that 
\begin{equation}
\label{eq:estHN}
\|\hpsi_s\|_{H^{N_++1}(\hat{\Sigma})}+\|\psi_s\|_{H^{N_++1}_{\sigma,R_{\infty},2}}\leq C\|\psi_s\|_{H^{N_++1}_{\sigma,R_{\infty},2}}
\end{equation}
so $\hpsi_s\in H^{N_++1}_{\sigma,R_{\infty},2}\cap H^{N_++1}(\hat{\Sigma})$, and in fact  $\hpsi_s\in H^{\tilde{N}_++1}_{\sigma,R_{\infty},2}\cap H^{\tilde{N}_++1}(\hat{\Sigma})$ for all $\tilde{N}_+\geq N_+$, so $\hpsi_s\in C^{\infty}(\hat{\Sigma})\cap H^{N_++1}_{\sigma,R_{\infty},2}$. To conclude that $\ker (\mathfrak{L}_{s,\lambda,\kappa}+\epsilon_0 B_s+K_{s,\lambda})\subset C^{\infty}(\hat{\Sigma})$, we simply repeat the above arguments with $\mathfrak{L}_{s,\lambda,\kappa}$ replacing $\mathfrak{L}_{s,\lambda}$.

We now turn to (ii). Since $\hpsi_s\in C^{\infty}(\hat{\Sigma})$, it follows immediately that for $\hpsi_s\in  \ker L_s\subseteq H^{N_++1}_{\sigma,R_{\infty},2}$ with $\tilde{R}_{\infty}\geq R_{\infty}$, $\hpsi_s\in H^{N_++1}_{\sigma,\tilde{R}_{\infty},2}$. It remains to show that $\hpsi_s\in  \ker L_s\subseteq H^{N_++1}_{\sigma,\tilde{R}_{\infty},2}$ implies that $\hpsi_s\in H^{N_++1}_{\sigma,R_{\infty},2}$. For this, we will apply Lemma \ref{lm:kappaconvergence} with $Q\equiv 0$, which says that there exists a sequence $\{\kappa_n\}$ in $(0,\infty)$ and $s_n\in \Omega_{\sigma}\cap\{\re(z)>-(\frac{1}{2}+N_+)\kappa_+\}$ with $\kappa_n\downarrow 0$ as $n\to \infty$, and there exists a sequence $\psi_{s_n}\in \ker (\mathfrak{L}_{s,\lambda,\kappa_n}+\epsilon B_s+K_{s,\lambda})$ such that
\begin{equation*}
\|\hpsi_{s_n}-\hpsi_{s}\|_{H_{\sigma,\tilde{R}_{\infty}}^{N_+}}\to 0
\end{equation*}
as $n\to \infty$.

Observe that
\begin{equation*}
(\mathfrak{L}_{s,\lambda,\kappa_n}+\epsilon_0 B_s)\hpsi_{s_n}=-K_{s,\lambda}\hpsi_{s_n}.
\end{equation*}

Since any estimates involving $\mathfrak{L}_{s,\lambda,\kappa_n}$ apply with $H^{N_++1}_{\sigma,R_{\infty}}$ and $H^{N_++1}_{\sigma,R_{\infty},2}$ norms where the infinite sums $\sum_{n=\ldots}^{\infty}$ are replaced with sums $\sum_{n=\ldots}^{N_{\infty}}$ with $N_{\infty}<\infty$ suitably large depending on $\kappa_n$, we can first apply the estimates with finite $N_{\infty}$, using that $\hpsi_{s_n},\hpsi_{s_m}\in C^{\infty}(\hat{\Sigma})$, and subsequently take the limit $N_{\infty}\to \infty$ to obtain: there exists an $N>0$ (independent of $n$) such that
\begin{equation*}
\|\hpsi_{s_n}\|_{H^{N_++1}_{\sigma,R_{\infty},2}}\leq C\|\hpsi_{s_n}\|_{H^{N}(\hat{\Sigma})},
\end{equation*}
where we moreover used that all sufficiently higher-order derivatives can appearing on the right-hand side of the $N_{\infty}<\infty$ estimates can be absorbed to the left-hand side, leaving only lower-order derivatives on the right-hand side. For $N_+$ suitably large depending on $N$, we obtain from \eqref{eq:estHN} applied to the $\kappa>0$ setting with $\tilde{R}_{\infty}$ replacing $R_{\infty}$:
\begin{equation*}
\|\hpsi_{s_n}\|_{H^{N_++1}_{\sigma,R_{\infty},2}}+\|\hpsi_{s_n}\|_{H^{N_++1}(\hat{\Sigma})}\leq C\|\hpsi_{s_n}\|_{H^{N_++1}_{\sigma,\tilde{R}_{\infty},2}}.
\end{equation*}
In particular, $\{\hpsi_{s_n}\}$ is a bounded sequence in $H^{N_++1}_{\sigma,R_{\infty},2}\cap H^{N_++1}(\hat{\Sigma})$, with respect to the norm $\|\cdot \|_{H^{N_++1}_{\sigma,R_{\infty},2}}+\|\cdot \|_{H^{N_++1}(\hat{\Sigma})}$ so it admits a weakly convergent subsequence. Since $H^{N_++1}_{\sigma,R_{\infty},2}\cap H^{N_++1}(\hat{\Sigma})\subset H^{N_++1}_{\sigma,\tilde{R}_{\infty},2}$ and $\hpsi_{s_n}$  converges strongly to $\psi_s$ with respect to $\|\cdot\|_{H^{N_++1}_{\sigma,\tilde{R}_{\infty},2}}$, we can conclude by uniqueness of limits that $\hpsi_{s}\in H^{N_++1}_{\sigma,\tilde{R}_{\infty},2}$.

We therefore conclude that  $\mathscr{Q}_{\sigma,\widetilde{R}_{\infty},N_+}=\mathscr{Q}_{\sigma, R_{\infty},N_+}$.

Finally, consider (iii). First of all, we have that $H^{N_++1}_{\tilde{\sigma},R_{\infty},2}\subset H^{N_++1}_{\sigma,R_{\infty},2}$, so if $\psi_s\in \ker L_s$ with respect to $H^{N_+}_{\tilde{\sigma},R_{\infty}}$, then $\psi_s\in \ker L_s$ with respect to $H^{N_+}_{\sigma,R_{\infty}}$.

Conversely, suppose that $\psi_s\in \ker L_s$ with respect to $H^{N_+}_{\sigma,R_{\infty}}$. Suppose also that $s\in \Omega_{\sigma'}$. As above, we can find a sequence $\{\kappa_n\}$ in $(0,\infty)$, $s_n\in \Omega_{\sigma}\cap  \Omega_{\sigma'}\cap\{\re(z)>-(\frac{1}{2}+N_+)\kappa_+\}$ with $\kappa_n\downarrow 0$ as $n\to \infty$ and $\psi_{s_n}\in \ker (\mathfrak{L}_{s,\lambda,\kappa_n}+\epsilon B_s+K_{s,\lambda})$ with
\begin{equation*}
\|\hpsi_{s_n}-\hpsi_{s}\|_{H_{\sigma,R_{\infty}}^{N_+}}\to 0.
\end{equation*}
We use the same strategy as before: we apply the $N_{\infty}<\infty$ estimates with respect to $H^{N_++1}_{\sigma',R_{\infty}}$ to $(\mathfrak{L}_{s,\lambda,\kappa_n}+\epsilon_0 B_s)\hpsi_{s_n}$ and use that any higher-order derivatives can be absorbed into the left-hand side, leaving
\begin{equation*}
\|\hpsi_{s_n}\|_{H^{N_++1}_{\sigma,R_{\infty},2}}\leq C\|\hpsi_{s_n}\|_{H^{N}(\hat{\Sigma}\cap \{x\leq x_0\})}\leq \|\hpsi_{s_n}\|_{H^{N_++1}_{\sigma,R_{\infty},2}}.
\end{equation*}
Hence, $\{\hpsi_{s_n}\}$ admits a weakly convergent subsequence with limit in $H^{N_++1}_{\sigma',R_{\infty},2}$. Since $H^{N_++1}_{\sigma',R_{\infty},2}\subset H^{N_++1}_{\sigma,R_{\infty},2}$ and uniqueness of limits holds, we conclude that $\psi_s\in H^{N_++1}_{\sigma',R_{\infty},2}$. Hence, $s$ is an element of $\ker L_s$ with respect to $H^{N_++1}_{\sigma',R_{\infty},2}$, from which (iii) follows.
\end{proof}

\section{Stability of the quasinormal spectrum}
\label{sec:stabqnf}
Recall the definition of the stability of the spectrum of a family of operators in Definition \ref{def:stabspect}. The following proposition forms the main ingredient for the stability results in Theorem \ref{thm:stab}.

\begin{proposition}
\label{prop:pertLs}
Consider the operator $L_s + \delta \cdot \widetilde{Q}: H^{N_+}_{\sigma,R_{\infty}}(L_s)\to H^{N_+}_{\sigma,R_{\infty}}$, with $\widetilde{Q}: H^{N_+}_{\sigma,R_{\infty},2}\to H^{N_+}_{\sigma,R_{\infty}}$ bounded for all $|\sigma|<\sigma_0$, with $\sigma_0>0$ or $\sigma_0=\infty$. Let $\epsilon>0$ and $D\subset \Omega_{\sigma}$. Then for all $\epsilon>0$ and $K\subset \Omega_{\sigma}$ compact, there exists a $\delta_0>0$ such that for all $0<\delta\leq \delta_0$:
\begin{enumerate}[label=\emph{(\roman*)}]
\item For $|\sigma|<\sigma_0$,  the family $s\mapsto (L_s+\delta \cdot \widetilde{Q})^{-1}:  H^{N_+}_{\sigma,R_{\infty}}\to  H^{N_+}_{\sigma,R_{\infty},2}$ is meromorphic on $K$ and holomorphic on $K\setminus \bigcup_{z\in \mathscr{Q}_{\rm reg}}B_{\delta}(z)\cap K$.
\item For all $z\in \mathscr{Q}_{\rm reg}\cap B_{\sigma_0^2}$, there exists an $s\in B_{\delta}(z)$, such that $s$ is a pole of $s\mapsto (L_s+\delta \cdot \widetilde{Q})^{-1}$.
\end{enumerate}
\end{proposition}
\begin{proof}
Part (i) follows from the fact that for $s\in K\setminus ( \bigcup_{z\in \mathscr{Q}_{\rm reg}}B_{\delta}(z)\cap K)$, we can write.
\begin{equation*}
L_s+\delta \cdot \widetilde{Q}=(1+\delta \cdot \widetilde{Q} L_s^{-1})L_s
\end{equation*}
and for suitably small $\delta$, $\|\widetilde{Q} L_s^{-1}\|_{H^{N_+}_{\sigma,R_{\infty}}\to H^{N_+}_{\sigma,R_{\infty}}}<1$ for all $s\in K\setminus (\bigcup_{z\in \mathscr{Q}_{\rm reg}}B_{\delta}(z)\cap K)$, provided $|\sigma|<\sigma_0$, so
\begin{equation*}
(L_s+\delta \cdot \widetilde{Q})^{-1}=L_s^{-1}(1+\delta \cdot Q L_s^{-1})^{-1}
\end{equation*}
is well-defined.

Part (ii) follows directly from Lemma \ref{lm:kappaconvergence} with $\kappa_n=0$ and the observation that for $|\sigma|<\sigma_0$ we have that
\begin{equation*}
	\bigcup_{\sigma\in (-\sigma_0,\sigma_0)}\Omega_{\sigma}=\Omega\cap B_{\sigma_0^2}. \qedhere
\end{equation*}
\end{proof}

\section{Kerr--de Sitter quasinormal modes}
\label{sec:KdS}
Let $(\mathcal{M}_{M,a,\Lambda},g_{M,a,\Lambda})$ be Kerr de Sitter spacetimes with mass $M$, angular momentum/mass $a$ and cosmological constant $\Lambda$, such that $\kappa_+$ remains uniformly bounded away from zero as $\Lambda\downarrow 0$. We denote by $r_+$ the radius of the event horizon, by $r_c$ the radius of the cosmological horizon and by $\kappa_c$ the surface gravity of the cosmological horizon. Then it can easily be verified that:
\begin{equation*}
	\kappa_c\sim \sqrt{\Lambda}\sim r_c^{-1}.
\end{equation*}

The wave operator $\square_{g_{M,a,\Lambda}}$ takes the following form with respect to Boyer--Linquist coordinates $(t,r,\theta,\varphi)$:
\begin{multline*}
	\rho^2 \square_{g_{M,a,\Lambda}}\phi=\partial_r(\Delta_r \partial_r\phi) -\Xi^2\left[\frac{(r^2+a^2)^2}{\Delta_r}-\frac{a^2\sin^2\theta}{\Delta_{\theta}}\right]\partial_t^2\phi-2a\frac{\Xi^2}{\Delta_r}\left[(r^2+a^2)-\frac{\Delta_r}{\Delta_{\theta}}\right]\partial_t\partial_{\varphi}\phi\\
	+\frac{1}{\sin\theta}\partial_{\theta}(\Delta_{\theta}\sin\theta \partial_{\theta}\phi)+\Xi^2\left(\frac{1}{\Delta_{\theta} \sin^2\theta}-\frac{a^2}{\Delta_r}\right)\partial_{\varphi}^2\phi,
\end{multline*}
with
\begin{align*}
	\rho^2=&\:r^2+a^2\cos^2\theta,\\
	\Delta_r=&\: (r^2+a^2)\left(1-\frac{\Lambda}{3}r^2\right)-2Mr,\\
	\Delta_{\theta}=&\:1+\frac{\Lambda a^2}{3}\cos^2\theta,\\
	\Xi=&\: 1+\frac{\Lambda a^2}{3}.
\end{align*}
Furthermore, $\kappa_c=-\frac{\Xi}{2(r_c^2+a^2)}\frac{d}{dr}\Delta_r|_{r=r_c}$. By using that $\Delta(r_c)=0$, we moreover obtain:
\begin{equation}
\label{eq:estLambdarc}
	0\leq 1-\frac{\Lambda}{3r_c^2}=\frac{2M}{1+a^2r_c^{-2}}r_c^{-1}\leq 2M r_c^{-1}.
\end{equation}

With respect to the vector fields
\begin{align*}
	T=&\:\Xi \partial_t,\\
	\Phi=&\:\Xi \partial_{\varphi},\\
	Y=&\partial_r-\frac{r^2+a^2}{\Delta_r}\Xi\partial_t-\frac{a}{\Delta_r}\Xi\partial_{\varphi},
\end{align*}
we can express:
\begin{multline*}
	\rho^2 \square_{g_{M,a,\Lambda}}\phi=Y(\Delta_r Y\phi) +2a  Y\Phi \phi+\frac{1}{\sin\theta}\partial_{\theta}(\Delta_{\theta}\sin\theta \partial_{\theta}\phi)+\frac{1}{\Delta_{\theta} \sin^2\theta}\Phi^2\phi\\
	+2(r^2+a^2) YT\phi +2r T\phi+\frac{1}{\Delta_{\theta}}a^2\sin^2\theta T^2\phi+2a\frac{1}{\Delta_{\theta}}T\Phi\phi.
\end{multline*}
Define $v=\Xi^{-1} t+r_*$ with $\frac{dr_*}{dr}(r)=\frac{r^2+a^2}{\Delta_r}$ and $\varphi_*=\Xi^{-1} \varphi+\int_{r_0}^{r}\frac{a}{\Delta_r}\mod 2\pi$. With respect to $(v,r,\theta,\varphi_*)$ coordinates, we have that $T=\partial_v$, $\Phi=\partial_{\varphi_*}$ and $Y=\partial_r$.

In a region where $r-r_c$ is suitably small, we will view $\square_{g_{M,a,\Lambda}}\phi$ as a perturbation of $\square_{g_{0,a,\Lambda_c}}\phi$, which is the de Sitter wave operator with cosmological constant $\Lambda_c=\frac{3}{r_c^2}$. We denote $\mathring{\Delta}_{\theta}=\frac{\Lambda_c a^2}{3}\cos^2\theta$ and $\mathring{\Delta}_{r}= (r^2+a^2)\left(1-\frac{\Lambda_c}{3}r^2\right)$, which correspond to the expressions for $\Delta_r$ and $\Delta_{\theta}$, respectively, in the case $M=0$ and $\Lambda=\Lambda_c$. 

Note that we can estimate
\begin{equation*}
	\left|\Delta_{\theta}-\mathring{\Delta}_{\theta}\right|=a^2\left|\frac{\Lambda}{3}-r_c^{-2}\right|\cos^2\theta\leq C a^2\kappa_c^2.
\end{equation*}

Then
\begin{multline*}
	\rho^2 \square_{g_{M,a,\Lambda}}\phi=\rho^2 \square_{g_{0,a,\Lambda_c}}\phi+ Y((\Delta_r-\mathring{\Delta}_{r})Y\phi)+\frac{1}{\sin\theta}\partial_{\theta}((\Delta_{\theta}-\mathring{\Delta}_{\theta})\sin\theta \partial_{\theta}\phi)\\
	+(\Delta_{\theta}^{-1}-\mathring{\Delta}_{\theta}^{-1})\left(\frac{1}{ \sin^2\theta}\Phi^2\phi+a^2\sin^2\theta T^2\phi+2aT\Phi\phi \right).
\end{multline*}

Consider the following coordinate transformations in the region $\{r\geq \frac{R_{\infty}}{\sqrt{2}}\}$ with $r_+<R_{\infty}<r_c$:
\begin{align*}
	\varrho(r,\theta)=&\:\sqrt{\frac{1}{1+a^2r_c^{-2}}(r^2 \mathring{\Delta}_{\theta}+a^2\sin^2\theta)},\\
\tan \vartheta=&\: \sqrt{\frac{1+a^2r^{-2}}{1+a^2r_c^{-2}}}\tan \theta,\\
\tilde{\varphi}=&\:\varphi-\frac{a}{r_c^2}t.
\end{align*}
Note that
\begin{align}
	\label{eq:varrhods}
	\varrho(r,\theta)=\:&r+\frac{a^2}{r_c^2+a^2}\sin^2\theta(r_c-r)+r^{-1} O_1((r_c-r)^2),\\
	\label{eq:varthetads}
	\vartheta(r,\theta)=\:&\theta+\frac{a^2}{r^2}O_1(r_c-r).
\end{align}

The coordinates $(t,\varrho,\vartheta,\tilde{\varphi})$ correspond to standard static coordinates on de Sitter with cosmological constant $\Lambda_c$ (see for example \cite{givo17}[\S V]). That is to say, we can write:
\begin{equation*}
	g_{0,a,\Lambda_c}=-\left(1-\Lambda_c \varrho^2\right)dt^2+\left(1-\Lambda_c \varrho^2\right)^{-1}d\varrho^2+\varrho^2(d\vartheta^2+\sin^2\vartheta d\tilde{\varphi}^2).
\end{equation*}

Let $\tilde{u}(t,\varrho)=t-\int_{R_{\infty}}^{\varrho}\frac{r_c^2}{r_c^2-\rho'^2}\,d\varrho'$. Let $T=\partial_{\tilde{u}}$ and $X=\partial_{\varrho}$. Then we can express:
\begin{equation*}
	\square_{g_{0,a,\Lambda_c}}\phi=r_c^{-2}\varrho^{-2}X(\varrho^2(r_c^2-\varrho^2) X\phi)+\varrho^{-2}\slashed{\Delta}_{\s^2,\vartheta}\phi-2\varrho^{-1}XT(\varrho \phi).
\end{equation*}
Let $\psi=\varrho \phi$. Then
\begin{equation*}
	\varrho^3\left(\square_{g_{0,a,\Lambda_c}}\phi-\frac{2}{r_c^2}\phi\right)=\varrho^2 X((1-r_c^{-2}\varrho^2) X\psi)+\slashed{\Delta}_{\s^2,\vartheta}\psi-2\varrho^2XT\psi.
\end{equation*}

Let $x=\varrho^{-1}-r_c^{-1}$, so that $\partial_x=-\rho^2X$.

We then obtain:
\begin{equation*}
	\varrho^3\left(\square_{g_{0,a,\Lambda_c}}\phi-\frac{2}{r_c^2}\phi\right)=\partial_x(x(x+2r_c^{-1})\partial_x\psi)+\slashed{\Delta}_{\s^2,\vartheta}\psi+2\partial_xT\psi.
\end{equation*}

We now extend the functions $\varrho$, $\vartheta$, $\tilde{\varphi}$ to $\{r\geq r_+\}$, such that $\varrho=r$, $\vartheta=\theta$ and $\tilde{\varphi}=\varphi$ for $r\leq \frac{R_{\infty}}{2}$. We consider the following time function:
\begin{equation*}
	\tau(t,\varrho)=\Xi^{-1} t+\int_{r_0}^r\frac{(\varrho^2+a^2)(1+\mathbbm{h}(\varrho'))}{\Delta_r (\varrho')}\,d\varrho'.
\end{equation*}
We consider $\mathbbm{h}$ such that $\tau$ are smooth,  spacelike hypersurfaces that intersect $\mathcal{H}^+$ and $\mathcal{C}^+$ to the future of the bifurcation spheres. We moreover demand that the $\tau$-level sets are analytic in a neighbourhood of $\mathcal{C}^+$ and that $\mathbbm{h}$ and its derivatives can be uniformly bounded as $\Lambda\downarrow 0$.

We then consider the coordinate chart $(\tau,x=\varrho^{-1}-r_c^{-1},\vartheta,\tilde{\varphi})$.

\begin{proposition}
The conformal wave operator on Kerr--de Sitter can be decomposed as follows:
\begin{equation*}
	\varrho^3\left(\square_{g_{M,a,\Lambda}}\phi-\frac{2}{3}\Lambda \phi\right)= \partial_x(x(2\kappa_c+x) \partial_x\psi)+\slashed{\Delta}_{\s^2,\vartheta}\psi-2\partial_x T\psi+\epsilon_0 \mathcal{B}_{\Lambda}^{\infty}+\mathcal{K}_{\Lambda}^{\infty}+ \kappa_c^2  \mathcal{B}_{\Lambda}^{\mathcal{C}^+},
\end{equation*}
where the coefficients of the differential operators $\mathcal{B}_{\Lambda}^{\infty}$, $\mathcal{K}_{\Lambda}^{\infty}$ and $\mathcal{B}_{\Lambda}^{\mathcal{C}^+}$ stay bounded as $\Lambda\downarrow 0$. Furthermore, $\mathcal{B}_{\Lambda}^{\infty}$ and $\mathcal{K}_{\Lambda}^{\infty}$ have the same form as the operators $\mathcal{B}^{\infty}$ and $\mathcal{K}^{\infty}$ from Lemma \ref{lm:operatorcoeffb} and $\mathcal{B}_{\Lambda}^{\mathcal{C}^+}$ can contain up to second-order angular derivatives and $\tau$-derivatives, up to first order $x$-derivatives and second-order $x$-derivatives of the form $\kappa_c x\partial_x^2$. The constant $\epsilon_0>0$ can be chosen arbitrarily small for $M^{-1}R_{\infty}$ suitably large.
\end{proposition}
\begin{proof}
The decomposition follows by applying Lemma \ref{lm:operatorcoeffb}	(the $\Lambda=0$) case, combined with the observation that additional terms that vanish when $\Lambda=0$ contain a factor of $\Lambda$ and hence can be written as follows $\kappa_c^2  \mathcal{B}_{\Lambda}^{\mathcal{C}^+}$, where $\mathcal{B}_{\Lambda}^{\mathcal{C}^+}$ is a second-order operator with coefficients that stay bounded as $\Lambda\downarrow 0$.

We will not derive the precise form of $\mathcal{B}_{\Lambda}^{\infty}$, $\mathcal{K}_{\Lambda}^{\infty}$ and $\mathcal{B}_{\Lambda}^{\mathcal{C}^+}$ in detail, but will will include a schematic version of the key relevant computation.

By applying \eqref{eq:varrhods} and \eqref{eq:varthetads}, we can write:
\begin{multline*}
	\varrho^3\rho^{-2}Y((\Delta_r-\mathring{\Delta}_r)Y\phi)=\varrho^3\rho^{-2}Y((\Delta_r-\mathring{\Delta}_r)Y(\varrho^{-1}\psi))\\
	=\varrho Y(\varrho^{-1}(\Delta_r-\mathring{\Delta}_r)Y\psi)-\varrho Y(\varrho^{-2}(\Delta_r-\mathring{\Delta}_r)\psi)\\
	=\frac{d}{dr}(\Delta_r-\mathring{\Delta}_r)\left(X\psi-\varrho^{-1}\psi\right)+\ldots\\
	=-2(\kappa_c-r_c^{-1}) r_c^2X\psi+\frac{d}{dr}\left(\frac{\Lambda-\Lambda_c}{3}r^2(r^2+a^2)+2M r\right)\varrho^{-1}\psi+\ldots\\
	=+2(\kappa_c-r_c^{-1}) \partial_x\psi+\frac{4}{3}(\Lambda-3r_c^{-2})r_c^2\psi+\ldots,
\end{multline*}
where the terms in $\ldots$ are part of $\epsilon_0 \mathcal{B}_{\Lambda}^{\infty}$, $\mathcal{K}_{\Lambda}^{\infty}$ and $\kappa_c^2 \mathcal{B}_{\Lambda}^{\mathcal{C}^+}$, for $M^2\Lambda$ suitably small depending on $\epsilon_0$.

Note moreover that by \eqref{eq:estLambdarc}:
\begin{align*}
	\frac{4}{3}\left(\Lambda-3r_c^{-2}\right)=r_c^{-2}(\Lambda r_c^{-2}-3)=M O(r_c^{-3}),\\
	\frac{2}{3}\Lambda_c r_c^2\psi=\frac{2}{3}\Lambda r_c^2\psi+M O(r_c^{-1})\psi.
\end{align*}
\end{proof}

Denote $L_{s,\Lambda}:=\varrho^3\left(\square_{g_{M,a,\Lambda}}-\frac{2}{3}\Lambda \right)(\varrho^{-1}\cdot)$, $B_{s,\Lambda}^{\infty}= e^{-s\tau}\mathcal{B}_{\Lambda}^{\infty}(e^{s\tau}(\cdot))$, $K_{s,\Lambda}^{\infty}= e^{-s\tau}\mathcal{K}_{\Lambda}^{\infty}(e^{s\tau}(\cdot))$ and $B_{s,\Lambda}^{\mathcal{C}^+}= e^{-s\tau}\mathcal{B}_{\Lambda}^{\mathcal{C}^+}(e^{s\tau}(\cdot))$. Then
\begin{equation*}
	L_{s,\Lambda}=\partial_x(x(x+2\kappa_c) \partial_x\psi)+\slashed{\Delta}_{\s^2,\vartheta}\psi-2 s\partial_x \psi+M x B_{s,\Lambda}^{\infty}+K_{s,\Lambda}^{\infty}+ \kappa_c^2  B_{s,\Lambda}^{\mathcal{C}^+}.
\end{equation*}

\begin{proposition}
We can decompose:
\begin{equation*}
	L_{s,\Lambda}=\mathfrak{L}_{s,\lambda,2\kappa_c}+\epsilon_0\cdot B_{s,\Lambda}+\kappa_c^2 B_{s,\Lambda}^{\mathcal{C}^+}+K_{s,\lambda,\Lambda},
\end{equation*}
where 
\begin{align}
\label{eq:operatorboundedLambda}
	r^{-2}\rho^{2}\mathfrak{L}_{s,\lambda,\kappa_c}=&\:r\partial_r(\Delta_r \partial_r(r^{-1}\cdot))+2ar\partial_r\partial_{\varphi_*}(r^{-1}\cdot)+\slashed{\Delta}_{\s^2}(\cdot)\\ \nonumber
	&\:+sr(4 Mr-2(1-\mathbbm{h})\Delta_r)\partial_r(r^{-1}\cdot)-\lambda\pi_{\leq l}\quad \textnormal{for}\quad  r=\varrho\leq \frac{R_{\infty}}{2},\\
	\label{eq:operatorfarawayLambda}
	(\mathfrak{L}_{s,\lambda,2\kappa_c}f)_{\ell}=&\:\partial_{x}(x(2\kappa_c+x)\partial_{x}f_{\ell})+2 s\partial_{x}f_{\ell}-\mathfrak{l}(\mathfrak{l}+1)f_{\ell}\quad \textnormal{for}\quad\varrho\geq R_{\infty}.
	\end{align}
Furthermore:
\begin{enumerate}[label=\emph{(\roman*)}]
\item There exists a constant $C_{s,\sigma, R_{\infty}}>0$ such that:
	\begin{equation*}
	\|\mathfrak{L}_{s,\lambda,2\kappa_c}^{-1}\|_{ H_{ \sigma, R_{\infty}}^{N_+}\to  H_{ \sigma, R_{\infty},2}^{N_++1}}\leq C_{s,\sigma, R_{\infty}}.
	\end{equation*}
	For all $f\in H_{ \sigma, R_{\infty}}^{N_+}$ we have additionally that
	\begin{equation*}
	\kappa_c \|x\partial_x^2\mathfrak{L}_{s,\lambda,2\kappa_c}^{-1}(f)\|_{G_{\sigma,R_{\infty}}}+\kappa_c \|\slashed{\Delta}_{\s^2 }\mathfrak{L}_{s,\lambda,\kappa}^{-1}(f)\|_{G_{\sigma,R_{\infty}}}\leq C_{s,\sigma}\|f\|_{H_{ \sigma, R_{\infty}}^{N_+}}.
	\end{equation*}
\item For $M^2\Lambda$ suitably small,
$B_{s,\Lambda}\circ \mathfrak{L}_{s,2\kappa_c,\lambda}^{-1}: H^{N_+}_{\sigma,R_{\infty}}\to H^{N_+}_{\sigma,R_{\infty}}$ and $\kappa_c B_s^{\mathcal{C}^+}\circ \mathfrak{L}_{s,2\kappa_c,\lambda}^{-1}: H^{N_+}_{\sigma,R_{\infty}}\to H^{N_+}_{\sigma,R_{\infty}}$ are bounded operators and $K_{s,\lambda,\Lambda}\circ \mathfrak{L}_{s,2\kappa_c,\lambda}^{-1}: H^{N_+}_{\sigma,R_{\infty}}\to H^{N_+}_{\sigma,R_{\infty}}$ is compact.
\end{enumerate}
\end{proposition}
\begin{proof}
	The statement in (i) follows from an analogue of Corollary \ref{cor:unifboundinvkappapos}. Due to the presence of the factor $\Delta_r$ instead of $\Delta$ in \eqref{eq:operatorboundedLambda}, we cannot directly apply Corollary \ref{cor:unifboundinvkappapos}. Instead we need to prove an analogue of Corollary \ref{cor:redshiftho}	for $\Delta$ replaced with $\Delta_r$. Since the surface gravity $\kappa_+$ of $\mathcal{H}^+$ is still positive, by assumption, and away from $r=r_+$ the difference $\Delta_r-\Delta$ can be made arbitrarily small in $\{r\leq \frac{1}{2}R_{\infty}\}$ for $M^2\Lambda$ suitably small, the red-shift estimates in Corollary \ref{cor:redshiftho} still hold.
	
	The statement in (ii) then follows straightforwardly by repeating the proof of Proposition \ref{prop:boundandcompact}. The only difference is the presence of the operator $\kappa_c B_s^{\mathcal{C}^+}$, whose boundedness follows from the estimates in (i) that feature $\kappa_c$-dependent terms on the left-hand side.
	\end{proof}
	
	We then directly obtain the following corollary:
	
	\begin{corollary}
	\label{cor:Kdsmainresult}
		The conclusions of Propositions \ref{prop:mainprop}--\ref{prop:accumulation} hold also for $L_{s,2\kappa_c}$ replacing $L_{s,\kappa}$.
	\end{corollary}

\section{Scattering resonances}
\label{sec:resonances}
In this section, we will establish meromorphicity of cut-off resolvents associated to the Boyer--Lindquist time coordinate $t$.

Let $R_{\infty}>0$ be arbitrarily large. Let $\chi_{\infty}\in C^{\infty}(\Sigma)$, such that $\textnormal{supp}\, \chi_{\infty}\subseteq [r_+,r_{\rm max}]$ for some arbitrary $r_{\rm max }<\frac{1}{2} R_{\infty}$.

We first show that we can simplify the domain of $L_s^{-1}$ by considering $\chi_{\infty}L_s^{-1} \chi_{\infty}$.
\begin{corollary}
\label{cor:infcutoffresest}
The cut-off operators $L_s^{-1} \chi_{\infty}$ and $\chi_{\infty}L_s^{-1} \chi_{\infty}$ satisfy the following properties:
\begin{enumerate}[label=\emph{(\roman*)}]
\item Let $N_+\in \N_+$ and $\delta_0>0$. Then there exists a $\sigma_0(\delta_0)\in \R$ such that $L_s^{-1} \chi_{\infty}: H^{N_+}_{\snabla}(\Sigma)\to H^{N_++1}_{\sigma_0, R_{\infty},2}$ defines a meromorphic family of operators in $\{|\arg z|<\frac{2}{3}\pi\}\cap \{\re z>-(\frac{1}{2}+N_+)\kappa_+\}\cap \{|z|\geq \delta_0\}$, with poles that lie in $\mathscr{Q}_{\rm reg}$.
\item $\chi_{\infty}L_s^{-1} \chi_{\infty}: H^{N_+}_{\snabla}(\Sigma)\to H^{N_++1}_{\snabla}(\Sigma)$ defines a meromorphic family of operators for $s\in \{|\arg z|<\frac{2}{3}\pi\}\cap\{\re s>-(\frac{1}{2}+N_+)\kappa_+\}$, with poles that lie in $\mathscr{Q}_{\rm reg}$.
	\item $\chi_{\infty} L_s^{-1} \chi_{\infty}: C^{\infty}(\Sigma)\to C^{\infty}(\Sigma)$ is well-defined for all $s\in \Omega\setminus \mathscr{Q}_{\rm reg}$.
	 \end{enumerate}
\end{corollary}
\begin{proof}
Let $s\in \Omega\setminus \mathscr{Q}_{\rm reg}$. By Proposition \ref{prop:mainprop}, there exists an $N_+$ such that $\re s>-(\frac{1}{2}+N_+)\kappa_+$, so $L_s^{-1}: H_{\sigma,R_{\infty}}^{N}\to H_{\sigma,R_{\infty},2}^{N+1}$ is well-defined for all $N\geq N_+$ and for an appropriate choice of $\sigma$. Let $f\in H^{N_+}_{\snabla}(\Sigma)$. Then $\chi_{\infty} f\in H_{\sigma,R_{\infty}}^{N_+}$ for all $\sigma\in \R$ and $L_s^{-1} (\chi_{\infty}f)\in H_{\sigma,R_{\infty},2}^{N+1}$ for all $\sigma$ such that $s\in \Omega_{\sigma}$. Since $H_{\sigma,R_{\infty},2}^{N+1}\subset H_{\sigma',R_{\infty},2}^{N+1}$, for all $\sigma\geq \sigma'$, we obtain (i) by taking $\sigma_0$ to be the minimum of $\sigma$, such that $\Omega^2_{\sigma}$ cover the region $\{|\arg z|<\frac{2}{3}\pi\}\cap \{\re z>-(\frac{1}{2}+N_+)\kappa_+\}\cap \{|z|\geq \delta_0\}$. The above implies moreover that $\chi_{\infty}L_s^{-1} (\chi_{\infty}f)\in H^{N_++1}_{\snabla}(\Sigma)$, so (ii) also follows.

In particular, if $f\in C^{\infty}(\Sigma)$, then we can take $N_+$ arbitrarily large in (ii) to conclude (iii).
\end{proof}

Now let $\chi_{\infty,+}\in C^{\infty}(\Sigma)$, such that $\textnormal{supp}\, \chi_{\infty,+}\subseteq [r_{\rm min },r_{\rm max}]$, for some arbitrary $r_{\min}>r_+$. Then can moreover obtain further meromorphicity results in the full set $\Omega$, in the case $\frac{|a|}{M}\ll 1$ or for bounded azimuthal frequencies $m$.

\begin{proposition}
\label{prop:mainpropcutoff}
The operator $\chi_{\infty} L_s^{-1}\chi_{\infty,+} :\pi_{|m|\leq m_0}(L^2(\Sigma))\to \pi_{|m|\leq m_0}(H^2(\Sigma))$ is well-defined for all $s\in \Omega\setminus \mathscr{Q}_{\rm reg}$ and defines a meromorphic family of operators on $\Omega$ with poles that lie in $\mathscr{Q}_{\rm reg}$. Furthermore, for $\frac{|a|}{|M|}$ suitably small, $\chi_{\infty}    L_s^{-1}\chi_{\infty,+} :L^2(\Sigma)\to H^2(\Sigma)$ is well-defined and defines a meromorphic family of operators on $\Omega$ with poles that lie in $\mathscr{Q}_{\rm reg}$. In both cases, there exists a constant $C>0$, with $C=C(a,r_{\rm min}, r_{\rm max}, R_{\infty}, s)>0$ or $C=C(m_0,r_{\rm min}, r_{\rm max},R_{\infty}, s)>0$ such that
\begin{equation*}
\|\chi_{\infty}L_s^{-1}(\chi_{\infty,+} f)\|_{H^2(\Sigma)}\leq C_s \|f\|_{L^2(\Sigma)}.
\end{equation*}
\end{proposition}
\begin{proof}
For $\frac{|a|}{|M|}$ suitably small or when restricting to the span of azimuthal modes with $|m|\leq m_0$, we can omit the angular derivatives $|\snabla_{\s^2}\partial_r^k\mathfrak{L}_{s,\lambda,\kappa}\hpsi|^2$ on the right-hand side of \eqref{eq:fullellipticest} by applying Proposition \ref{prop:redshifthopluselliptic} to obtain the following improved version of Corollary \ref{cor:infcutoffresest}(i): for $s\notin \mathscr{Q}_{\rm reg}$
\begin{equation}
\label{eq:auxcutoffest1}
\|\chi_{\infty}L_s^{-1}(\chi_{\infty,+} f)\|_{H^{N+1}(\Sigma)}+\left|\left|\left(1-\frac{r_+}{r}\right) \chi_{\infty}L_s^{-1}(\chi_{\infty,+} f)\right|\right|_{H^{N+2}(\Sigma)}\leq C_s \|f\|_{H^N(\Sigma)}.
\end{equation}

It is straightforward to improve \eqref{eq:auxcutoffest1} in the region $r\leq r_{\rm min}$ by applying the estimates in the proof of Corollary \ref{cor:redshiftho}, but restricted to $r\in [r_+,r_{\rm min}]$. It then follows that
\begin{equation}
\label{eq:auxcutoffest2}
\|\chi_{\infty}L_s^{-1}(\chi_{\infty,+} f)\|_{H^{N_++2}(\Sigma)}\leq C \|f\|_{H^N_+(\Sigma)}+C\|\partial_r^{N_++1}(\chi_{\infty,+} f)\|_{L^2(\Sigma\cap\{r\leq r_{\rm min}\})}\leq C \| f\|_{H^{N_+}(\Sigma)},
\end{equation}
where the final inequality follows from the fact that $\chi_{\infty,+}f=0$ in $\{r\leq r_{\rm min} \}$.

Suppose that $s\in \Omega$ satisfies $\re s>-\frac{1}{2}\kappa_+$, then the proposition holds. Suppose now instead that $-\frac{3}{2}\kappa_+< \re s\leq -\frac{1}{2}\kappa_+$ and let $s'\in \Omega$ be such that $\re s'>-\frac{1}{2}\kappa_+$. 

Let $f\in C^{\infty}(\Sigma)$ and define $\hpsi=L_s^{-1}(\chi_{\infty,+} f)$. We can express:
\begin{equation*}
L_{s}^{-1}=L_{s}^{-1}L_{s'}L_{s'}^{-1}=L_{s}^{-1}(L_{s}+(L_{s'}-L_s))L_{s'}^{-1}=L_{s'}^{-1}+L_{s}^{-1}(L_{s'}-L_s)L_{s'}^{-1}.
\end{equation*}
We can now estimate
\begin{multline*}
\|\chi_{\infty} \hpsi\|_{H^2(\Sigma)}=\|\chi_{\infty}L_{s}^{-1}(\chi_{\infty,+} f)\|_{H^2(\Sigma)}\\
\leq \|\chi_{\infty}L_{s'}^{-1}(\chi_{\infty,+} f)\|_{H^2(\Sigma)}+ \|\chi_{\infty}L_{s}^{-1}(L_{s'}-L_s)L_{s'}^{-1}(\chi_{\infty,+} f)\|_{H^2(\Sigma)}\\
\leq C_{s',r_{\rm min}}\|f\|_{L^2(\Sigma)}+C_s\|\chi_{\infty}(L_{s'}-L_s)L_{s'}^{-1}(\chi_{\infty,+} f)\||_{H^1(\Sigma)}\\
\leq C_{s',r_{\rm min}}\|f\|_{L^2(\Sigma)}+C_s|s-s'\||\chi_{\infty}L_{s'}^{-1}(\chi_{\infty,+} f)\|_{H^2(\Sigma)}\\
\leq C_{s',r_{\rm min}}(1+C_s|s-s'|)\|f\|_{L^2(\Sigma)},
\end{multline*}
where we applied \eqref{eq:auxcutoffest2} with $N=0$ and Corollary \ref{cor:infcutoffresest}(i) with $N_+=1$ to arrive at the second line, and then we applied \eqref{eq:auxcutoffest2} with $N_+=0$ once again to arrive at the last line. By a density argument, we therefore obtain $\|\chi_{\infty}\hpsi\|_{H^2(\Sigma)}\leq C \|f \|_{L^2(\Sigma)}$ for all $f\in L^2(\Sigma)$.

We can repeat this procedure to conclude that the desired inequality holds for all $s\in \Omega$.
\end{proof}

We consider now the differential operator
\begin{equation*}
A(\omega):=e^{i\omega t}\square_{g_{M,a}}(e^{-i\omega  t} (\cdot)): C^{\infty}(\R_r\times \s^2_{(\theta,\varphi)})\to C^{\infty}(\R_r\times \s^2_{(\theta,\varphi)}).
\end{equation*}
We will construct the corresponding resolvent operator $R(\omega)=A(\omega)^{-1}$ and the associated cut-off resolvent operator $ \chi_{\infty,+}R(\omega) \chi_{\infty,+}$ by relating $A(\omega)$ to $L_s$. Denote $\Sigma':=\{t=0\}$.

Let $Q_s: L^2(\Sigma')\to L^2(\Sigma\setminus\{r=r_+\})$ be defined as follows: let $\overline{f}: D^+(\Sigma')\to \C$ denote the trivial extension of $f\in L^2(\Sigma')$ to $D^+(\Sigma')$, i.e.\ $\overline{f}(t,r,\theta,\varphi_*)=\overline{f}(r,\theta,\varphi_*)$. Then
\begin{equation*}
Q_s(f)(r,\theta,\varphi_*)=(e^{st}\overline{f})|_{\Sigma}(r,\theta,\varphi^*).
\end{equation*}
Furthermore, $Q_s\chi: L^2(\Sigma')\to L^2(\Sigma)$ is well-defined and bounded. We consider also the inverse $Q_s^{-1}: L^2(\Sigma) \to L^2(\Sigma')$, to obtain the bounded operator:
\begin{equation*}
\chi_{\infty,+}Q_s^{-1} :  H^2(\Sigma)\to  H^2(\Sigma').
\end{equation*}

Let $s=-i\omega$, then we can express
\begin{equation*}
A(\omega)=Q_s^{-1}\circ (\varrho^{-3}L_s(\varrho \cdot))\circ Q_s.
\end{equation*}
so that
\begin{equation*}
\chi_{\infty,+} R(\omega) \chi_{\infty,+}=(\chi_{\infty,+} Q_s^{-1})\circ (\varrho^{-1}L_s(\varrho^{3} \cdot))^{-1}\circ (Q_s\chi_{\infty,+}).
\end{equation*}
\begin{corollary}
\label{cor:resonances}
The cut-off resolvent operator:
\begin{equation*}
\chi_{\infty,+} R(\omega) \chi_{\infty,+}: H^{N_+}_{\snabla}(\Sigma')\to H^{N_++1}_{\snabla}(\Sigma')
\end{equation*}
can be meromorphically continued to $\omega\in (i\Omega\setminus  i\mathscr{Q}_{\rm reg})\cap \{\im \omega>-(\frac{1}{2}+N_+)\kappa_+\}$ and $\omega\mapsto \chi_{\infty,+} R(\omega) \chi_{\infty,+} $ is meromorphic on $i\Omega\cap \{\im \omega>-(\frac{1}{2}+N_+)\kappa_+\}$. Let $i \mathscr{Q}_{\rm res}$ denote the set of poles of $\omega\mapsto \chi_{\infty,+} R(\omega)\chi_{\infty,+}$ in $i\Omega$. Then
\begin{equation*}
\mathscr{Q}_{\rm res} \subseteq \mathscr{Q}_{\rm reg}.
\end{equation*}

Furthermore, the above results hold also with the following domain changes:
\begin{align*}
\chi_{\infty,+}R(\omega) \chi_{\infty,+}:&\: \pi_{|m|\leq m_0}(L^2(\Sigma'))\to \pi_{|m|\leq m_0}(H^2(\Sigma'))\quad \textnormal{for any $m_0\in [0,\infty)$},\\
\chi_{\infty,+}R(\omega) \chi_{\infty,+}:&\: L^2(\Sigma')\to H^2(\Sigma')\quad \textnormal{for suitably small $\frac{|a|}{M}$}.
\end{align*}
\end{corollary}
\begin{proof}
By taking $s=-i\omega$ and and using that $r=\varrho$ in $\{r\leq \frac{1}{2}R_{\infty}\}$, we obtain
\begin{equation*}
\chi_{\infty,+} R(\omega) \chi_{\infty,+}= (\chi_{\infty,+}  Q_s^{-1}) \cdot  r^{-1}L_s^{-1}\cdot  (r^{3}  Q_s \chi_{\infty,+}).
\end{equation*}

We conclude by applying Corollary \ref{cor:infcutoffresest} and Proposition \ref{prop:mainpropcutoff} and using boundedness of $\chi_{\infty,+}  Q_s^{-1}$ and $Q_s \chi_{\infty,+}$.
\end{proof}

\appendix

\section{Proof of Lemma \ref{lm:compatibility}}
\label{app:pfoptparam}
We may assume that $\re s<0$, since the $\re s\geq 0$ follows immediately from \eqref{eq:compatibility4}. 

In the case $\re s<0$, the statements in Lemma \ref{lm:compatibility} follow from the proof of \cite{gajwar19a}[Lemma 9.7]. For the sake of completeness, we carry out the proof below.

Recall \eqref{eq:compatibility1}, \eqref{eq:compatibility2} and \eqref{eq:compatibility3}:
\begin{align*}
0<\mu<&\:\nu,\\
	4-\beta-(1-\mu)(\alpha+\tilde{\sigma}^{-2}\alpha^{-1})>&\: 0,\\
	1-\mu\nu-4\beta^{-1}\left(\frac{|{\re s}|}{|s|}\right)^2-\tilde{\sigma}^2>&\: 0.
\end{align*}
We will first determine for what values of $\frac{|\Re s|}{|s|}$ \eqref{eq:compatibility1}--\eqref{eq:compatibility3} can hold simultaneously. First of all, by \eqref{eq:compatibility1}, it is sufficient to show that \eqref{eq:compatibility2} and \eqref{eq:compatibility3} hold with $\nu$ replaced $\mu$, which together imply that:
\begin{equation}
\label{eq:mainineqparam}
\frac{1}{\frac{4-\beta}{1-\mu}\alpha-\alpha^2}	<\tilde{\sigma}^2<1-\mu^2-4\beta^{-1}\left(\frac{|{\re s}|}{|s|}\right)^2
\end{equation}
We can minimize the term on the very left-hand side by taking $\alpha=\frac{1}{2}\frac{4-\beta}{1-\mu}$, so that \eqref{eq:mainineqparam} becomes:
\begin{equation*}
	\frac{4(1-\mu)^2}{(4-\beta)^2}	<\tilde{\sigma}^2<1-\mu^2-4\beta^{-1}\left(\frac{|{\re s}|}{|s|}\right)^2.
\end{equation*}
Hence, we can find a $\tilde{\sigma}$ satisfying the above inequality if and only if
\begin{equation*}
	\left(\frac{|{\re s}|}{|s|}\right)^2<\frac{\beta}{4}\left(1-\mu^2-\frac{4 (1-\mu)^2}{(4-\beta)^2}\right).
\end{equation*}
We will first pick $\mu\in [0,1]$ so that the above expression is maximized. Note the right-hand side vanishes for $\mu=1$ and is non-positive for $\mu=0$. Furthermore, a maximum is attained in $(0,1)$ if
\begin{equation*}
	\frac{d}{d\mu}\left(1-\mu^2-\frac{4 (1-\mu)^2}{(4-\beta)^2}\right)=-2\mu+\frac{8}{(4-\beta)^2}(1-\mu)
\end{equation*}
vanishes, which is the case for
\begin{equation*}
	\mu=\frac{4}{(4-\beta)^2}\frac{1}{\frac{4}{(4-\beta)^2}+1}=\frac{1}{1+\frac{(4-\beta)^2}{4}}=\frac{1}{1+b^2},
\end{equation*}
where $b=2-\frac{\beta}{2}$. Hence,
\begin{equation*}
	\left(\frac{|{\re s}|}{|s|}\right)^2<\left(1-\frac{b}{2}\right)\frac{b^2}{1+b^2}.
\end{equation*}
Finally, since the right-hand side vanishes for $b=0$ and $b=2$, the maximum will be obtained if
\begin{equation*}
	0=\frac{d}{db}\left(\left(1-\frac{b}{2}\right)\frac{b^2}{1+b^2}\right)=-\frac{b(b-1)(b^2+b+4)}{2(1+b^2)^2},
\end{equation*}
which is the case for $b=1$. We conclude that \eqref{eq:compatibility1}--\eqref{eq:compatibility3} can only hold if 
\begin{equation*}
	\left(\frac{|{\re s}|}{|s|}\right)^2<\frac{1}{4}.
\end{equation*}
or $\re s>-\frac{1}{2}|s|$. Note that this corresponds to the sector $\frac{\pi}{2}<\arg z<\theta_0$, with $\theta_0\in (\frac{\pi}{2},\pi)$ satisfying: $\cos \theta_0=-\frac{1}{2}$. That means that $\theta_0=\frac{2}{3}\pi$.

For $\beta=2$, $\mu=\frac{1}{2}$ and $\alpha=2$, \eqref{eq:mainineqparam} reduces to:
\begin{equation*}
\frac{1}{4}	<\tilde{\sigma}^2<\frac{3}{4}-2\left(\frac{|{\re s}|}{|s|}\right)^2.
\end{equation*}
If we fix $\sigma^2=|2s|^2\tilde{\sigma}^2|$, we therefore have that $|s|^2<\sigma^2$ and
\begin{equation*}
	\sigma^2<3|s|^2-8|\re s|^2=3 |\im s|^2-5 |\re s|^2.
\end{equation*}
This finishes the proof of Lemma \ref{lm:compatibility}.

\section{Compact embeddings}
The following Hilbert spaces are compactly embedded:
\begin{proposition}
Let $R_{\infty}\geq 3 r_+$ and $N_+\in \N_0$. Then
\begin{align}
\label{eq:compact1}
H_{\sigma_2,R_{\infty}}^{N_++1} \ssubset &\: H_{\sigma_1,R_{\infty}}^{N_+}\quad \textnormal{for all $0>|\sigma_2|>|\sigma_1|$},\\
\label{eq:compact2}
H_{\sigma_2,R_{\infty},2}^{N_++2} \ssubset &\: H_{\sigma_1,R_{\infty},2}^{N_++1}\quad \textnormal{for all $0>|\sigma_2|>|\sigma_1|$},\\
\label{eq:compact3}
H_{\sigma, R_{\infty},1}^{N_++1} \ssubset &\: H_{\sigma,R_{\infty}}^{N_+}\quad \textnormal{for all $\sigma\in \R$.}
\end{align}
\end{proposition}
\begin{proof}
We first consider \eqref{eq:compact3}. Let $\{f_k\}$ be a sequence in $H_{\sigma, R_{\infty},1}^{N_++1}$ satisfying $\|f_k\|_{H_{\sigma, R_{\infty},1}^{N_++1}}=1$. To conclude \eqref{eq:compact3}, we need to show that it admits a Cauchy subsequence with respect to the $H_{\sigma, R_{\infty}}^{N_+}$-norm.

We consider first the $B_{R_{\infty}}$-part of the norm:
\begin{multline*}
\|f_k\|_{B_{R_{\infty}}}^2=\sum_{\ell\in \N_0}\sum_{n=0}^{\infty}\frac{(d_0x_0)^{2n}}{(\max\{\ell+1,n+1\})^{2n}}\||\partial_x^{n}(f_k)_{\ell}\|^2_{L^2(\s^2)}(x_0)\\
=\sum_{\ell\in \N_0}\|(f_k)_{\ell}\|^2_{L^2(\s^2)}(x_0)+\sum_{\ell\in \N_0}\sum_{n=0}^{\infty}\frac{1}{(\max\{\ell+1,n+1\})^2}\frac{(d_0x_0)^{2n}}{(\max\{\ell+1,n+1\})^{2n}}\||\partial_x^{n+1}(f_k)_{\ell}\|_{L^2(\s^2)}^2(x_0)\\
=\sum_{\ell\in \N_0}\sum_{|m|\leq \ell}|(f_k)_{\ell m}|^2(x_0)\\
+\sum_{\ell\in \N_0}\sum_{|m|\leq \ell}\sum_{n=0}^{\infty}\frac{1}{(\max\{\ell+1,n+1\})^2}\frac{(d_0x_0)^{2n}}{(\max\{\ell+1,n+1\})^{2n}}|\partial_x^{n+1}(f_k)_{\ell m}|^2(x_0)
\end{multline*}
Since $\|f_k\|_{B_{R_{\infty}}}^2\leq 1$, we have by Bolzano--Weierstrass that there exists a subsequence $f_{k_l}$ such that for fixed $L>0$
\begin{align*}
\sum_{\ell=0}^L\sum_{|m|\leq \ell}|(f_{k_l})_{\ell m}|^2(x_0)\to &\: a_{L,N},\\
\sum_{\ell=0}^L\sum_{|m|\leq \ell}\sum_{n=0}^{N}\frac{1}{(\max\{\ell+1,n+1\})^2}\frac{(d_0x_0)^{2n}}{(\max\{\ell+1,n+1\})^{2n}}|\partial_x^{n+1}(f_{k_l})_{\ell m}|^2(x_0)\to &\: b_{L,N}
\end{align*}
for fixed $L, N$. We moreover have that the ``high-frequency part'' of the norm of the difference $f_k-f_r$ can be estimated as follows:
\begin{align*}
\sum_{\ell\geq L}\|(f_k-f_r)_{\ell}\|^2_{L^2(\s^2)}(x_0)\leq  \frac{1}{L(L+1)}&\sum_{\ell\geq L}\|(\snabla_{\s^2}(f_k-f_r))_{\ell}\|^2_{L^2(\s^2)}(x_0)\leq \frac{C}{L(L+1)},\\
\sum_{\ell\in \N_0, n\in \N_0, \textnormal{$\ell\geq L$ or $n\geq N$}}\frac{1}{(\max\{\ell+1,n+1\})^2}&\frac{(d_0x_0)^{2n}}{(\max\{\ell+1,n+1\})^{2n}}\||\partial_x^{n+1}(f_k-f_r)_{\ell}\|_{L^2(\s^2)}^2(x_0)\\
\leq  \frac{1}{(\min\{N,L\}+1)^2}(\|\snabla_{\s^2}(f_k-f_r)\|_{B_{R_{\infty}}}^2&+\|\partial_x(f_k-f_r)\|_{B_{R_{\infty}}}^2)\leq\frac{C}{(\min\{N,L\}+1)^2}
\end{align*}
Hence, $\{f_{k_l}\}$ is Cauchy with respect to $\|\cdot\|_{B_{R_{\infty}}}$.

Consider now the $G_{\sigma,R_{\infty}}$-part of the norm:
\begin{multline*}
\|f_k\|_{G_{\sigma,R_{\infty}}}^2=\sum_{\ell\in \N_0}\sum_{n=0}^{\max\{\ell-1,0\}}\frac{(4\sigma)^{2n}(\ell-(n+1))!^2}{(\ell+n+1)!^2} \int_0^{x_0}\| \partial_x^{n}f_{\ell}\|_{L^2(\s^2)}(x)\, dx\\
	+(\ell+1)^3\sum_{n=\ell}^{\infty}\frac{\sigma^{2n}}{n!^2(n+1)!^2} \int_0^{x_0}\|\partial_x^{n}f_{\ell}\|_{L^2(\s^2)}(x)\, dx\\
	=\sum_{\ell\in \N_0}\sum_{|m|\leq \ell}\sum_{n=0}^{\max\{\ell-1,0\}}\frac{(4\sigma)^{2n}(\ell-(n+1))!^2}{(\ell+n+1)!^2} \int_0^{x_0}| \partial_x^{n}f_{m \ell}|^2(x)\, dx\\
	+(\ell+1)^3\sum_{n=\ell}^{\infty}\frac{\sigma^{2n}}{n!^2(n+1)!^2} \int_0^{x_0}|\partial_x^{n}f_{m\ell}|^2(x)\, dx.
\end{multline*}
Now we can proceed as above: we first split the sum into a part where $\ell\leq L$ and $n\leq N$ and, rather than applying Bolzano--Weierstrass, we apply Arzel\`a--Ascoli for functions on $[0,x_0]$ to obtain a convergent subsequence. Here we moreover use that we can consider the $C^N([0,x_0])$ norm rather than the $H^N([0,x_0])$ norm via a Sobolev inequality on $[0,x_0]$, since the norm $\|f_k\|_{G_{\sigma,R_{\infty}}}$ bounds arbitrarily many $x$-derivatives.

Subsequently, we bound the remaining high-frequency part part of $\|f_k-f_r\|_{G_{\sigma,R_{\infty}}}^2$, where $n>N$ or $\ell\geq L$,  with a constant $\frac{1}{(\min\{N,L\}+1)^2}$ by using the boundedness of $\|f_k-f_r\|_{G_{\sigma, R_{\infty}}}+\|\snabla_{\s^2}(f_k-f_r)\|_{G_{\sigma,R_{\infty}}}+\|\partial_x(f_k-f_r)\|_{G_{\sigma,R_{\infty}}}$.

To conclude the proof of \eqref{eq:compact3}, we consider $\|f\|_{H^{N_+}(\Sigma\cap\{\varrho\leq 2R_{\infty}\})}^2+\|\snabla_{\s^2}f\|_{H^{N_+}(\Sigma\cap\{\varrho \leq \frac{1}{3}R_{\infty}\})}^2+\sum_{\ell\in \N_0}|(\ell(\ell+1)+1)^{-\frac{w_p}{4}}f_{\ell}\|^2_{H^{N_+}(\Sigma\cap\{\frac{1}{3}R_{\infty}\leq \varrho \leq 3R_{\infty}\})}$. We proceed as before, but rather than applying Arzel\`a--Ascoli to estimate the bounded angular frequencies $\ell\leq L$, we apply Rellich--Kondrachov on $[r_+^{-1},\frac{1}{R_{\infty}}]$.

We can prove \eqref{eq:compact1} and \eqref{eq:compact2} in an analogous manner, where we use obtain an smallness factor in the high-frequency estimates by using that $(4\sigma_1)^{2n}=(\frac{\sigma_2}{\sigma_1})^{-2n}(4\sigma_2)^{2n}$.
\end{proof}

\section{Norms of pointwise products}
In this section, we consider norms of products of functions: $\|g\cdot f\|_{H^{N_+}_{ \sigma,R_{\infty}}}$, where $g$ will take the role of a coefficient in a differential operator and $f$ will take the role of a function in the domain of that operator. Since the $H^{N_+}_{ \sigma,R_{\infty}}$-norm features an infinite number of derivatives, one cannot immediately estimate $\|g\cdot f\|_{H^{N_+}_{ \sigma,R_{\infty}}}$ by a product of $\|f\|_{H^{N_+}_{ \sigma,R_{\infty}}}$ with an appropriate norm for $g$.

In this section, we will prove:
\begin{proposition}
\label{prop:producttotalnorm}
Let $\beta>2$. Then there exists a constant $C>0$ such that for all appropriately regular functions $f,g; [0,x_0]\times \s^2\to \C$, we can estimate:   
\begin{equation*}
\|g\cdot f\|_{H^{N_+}_{ \sigma,R_{\infty}}}\leq C\|g\|_{H^{N_+}_{ \sigma,R_{\infty};\textnormal{coeff}}}\cdot \|f\|_{H^{N_+}_{\sigma,R_{\infty}}},
\end{equation*}
with
\begin{multline*}
\|g\|^2_{H^{N_+}_{ \sigma,R_{\infty};\textnormal{coeff}}}:=\|g\|_{H^{N_++2}(\Sigma\cap\{\varrho\leq 3R_{\infty}\})}^2+\|\snabla_{\s^2}g\|_{H^{N_++2}(\Sigma\cap\{\varrho\leq \frac{1}{3}R_{\infty}\})}^2\\
+\sum_{\ell\in \N_0}\|(1+\slashed{\Delta}_{\s^2})^{\frac{w_p+5}{4}} g_{\ell}\|_{H^{N_++2}(\Sigma\cap\{3R_{\infty}\leq \varrho\leq 4R_{\infty}\})}^2+\sum_{j=0}^{\infty}\sum_{\ell\in \N_0}e^{2\ell}(\ell+1)^5\frac{(d_0x_0)^{2j}(j+1)^2}{j!^2 }\|\partial_x^jg_{\ell}|_{x=x_0}\|^2_{L^2(\s^2)}\\
+\sum_{j\in \N_0}\sum_{\ell\in \N_0}\frac{e^{5(\ell+j)}}{j!^4}\|\partial_x^jg_{\ell}\|^2_{L^{\infty}([0,x_0]\times \s^2)}.
\end{multline*}
\end{proposition}
The proof follows from combining Lemmas \ref{lm:productBnorm} and \ref{lm:productGnorm} below, which each deal with a different part of the norm $\|g\cdot f\|_{H^{N_+}_{ \sigma,R_{\infty}}}$.
\begin{lemma}
\label{lm:productBnorm}
Let $\beta>2$. Then there exists a constant $C>0$ such that for all appropriately regular functions $f,g; [0,x_0]\times \s^2\to \C$, we can estimate:
	\begin{equation}
	\label{eq:spherenormproduct}
	\|gf\|_{B_{R_{\infty}}}^2\leq C\|f\|_{B_{R_{\infty}}}^2\left[\sum_{j=0}^{\infty}\sum_{\ell\in \N_0}e^{2\ell}(\ell+1)^5\frac{(d_0x_0)^{2j}(j+1)^2}{j!^2 }\|\partial_x^jg_{\ell}|_{x=x_0}\|^2_{L^2(\s^2)}\right].
\end{equation}
Furthermore,
\begin{multline}
\label{eq:spherenormproductw}
\|gf\|_{H^{N_+}(\Sigma\cap\{r\leq 3R_{\infty}\})}^2+\|\snabla_{\s^2}(gf)\|_{H^{N_+}(\Sigma\cap\{r\leq \frac{1}{3}R_{\infty}\})}^2+\|(-\slashed{\Delta}_{\s^2}-\mathbf{1})^{-\frac{w_p}{4}}(gf)_{\ell}\|_{H^{N_+}(\Sigma\cap\{3R_{\infty}\leq r\leq 4R_{\infty}\})}^2\\
\leq \|f\|_{H^{N_+}(\Sigma\cap\{r\leq 3R_{\infty}\})}^2\|g\|_{H^{N_++2}(\Sigma\cap\{r\leq 3R_{\infty}\})}^2+\|\snabla_{\s^2}f\|_{H^{N_+}(\Sigma\cap\{r\leq \frac{1}{3}R_{\infty}\})}^2\|\snabla_{\s^2}g\|_{H^{N_++2}(\Sigma\cap\{r\leq \frac{1}{3}R_{\infty}\})}^2\\
+\|(\mathbf{1}-\slashed{\Delta}_{\s^2})^{-\frac{w_p}{4}} f\|_{H^{N_+}(\Sigma\cap\{3R_{\infty}\leq r\leq 4R_{\infty}\})}^2 \|(\mathbf{1}-\slashed{\Delta}_{\s^2})^{\frac{w_p+5}{4}}  g\|_{H^{N_+}(\Sigma\cap\{3R_{\infty}\leq r\leq 4R_{\infty}\})}^2.
\end{multline}
\end{lemma}
\begin{proof}
Consider the functions $g,f: [r_+,\infty)_{\varrho}\times \s^2_{\vartheta,\varphi*} \to \C$. We denote with $g^{(n)}=\partial_x^n g$ and $f^{(n)}=\partial_x^nf$. We can expand the product $g f$ in spherical harmonics:
\begin{equation*}
(gf)_{\ell m}(x)=\sum_{\tilde{\ell}\in \N_0}\sum_{\tilde{m}=-\tilde{\ell}}^{\tilde{\ell}}\sum_{\ell'\in \N_0}\sum_{m'=-\ell'}^{m=\ell'}g_{\ell' m'}(x)f_{\tilde{\ell}\tilde{m}}(x)\la Y_{\ell' m'}Y_{\tilde{\ell} \tilde{m}},Y_{\ell m}\ra_{L^2(\s^2)}.
\end{equation*}
Since $\la Y_{\ell' m'}Y_{\tilde{\ell} \tilde{m}},Y_{\ell m}\ra_{L^2(\s^2)}=0$ if $m\neq m'+\tilde{m}$, we can write
\begin{equation*}
	\int_{\s^2}|(gf)_{\ell}|^2\,d\upsigma=\sum_{\tilde{\ell}\in \N_0}\sum_{\tilde{m}=-\tilde{\ell}}^{\tilde{\ell}}\sum_{\ell'\in \N_0}\sum_{m'=-\ell'}^{m=\ell'}|g_{\ell' m'}(x)|^2|f_{\tilde{\ell}\tilde{m}}(x)|^2|\la Y_{\ell' m'}Y_{\tilde{\ell} \tilde{m}},Y_{\ell (m'+\tilde{m})}\ra_{L^2(\s^2)}|^2.
\end{equation*}
By applying Cauchy--Schwarz and using that spherical harmonics have unit norm on $\s^2$, we can further estimate:
\begin{multline*}
	|\la Y_{\ell' m'}Y_{\tilde{\ell} \tilde{m}},Y_{\ell (m'+\tilde{m})}\ra_{L^2(\s^2)}|^2\leq \|Y_{\ell' m'}Y_{\tilde{\ell}\tilde{m}}\|_{L^2(\s^2)}^2\|Y_{\ell (m'+\tilde{m})}\|_{L^2(\s^2)}^2\\
	\leq  \|Y_{\ell' m'}\|_{L^{\infty}(\s^2)}^2\|Y_{\tilde{\ell}\tilde{m}}\|_{L^2(\s^2)}^2\|Y_{\ell (m'+\tilde{m})}\|_{L^2(\s^2)}^2\leq \|Y_{\ell' m'}\|_{L^{\infty}(\s^2)}^2.
\end{multline*}
By a standard Sobolev inequality on $\s^2$, we can further estimate:
\begin{equation*}
	\|Y_{\ell' m'}\|_{L^{\infty}(\s^2)}^2\leq C \|Y_{\ell' m'}\|_{H^2(\s^2)}^2\leq C(\ell'+1)^4.
\end{equation*}
Furthermore, from the vanishing properties of Clebsch--Gordan coefficients (see for example \cite{sak94}[\S3.7]) $\la Y_{\ell' m'}Y_{\tilde{\ell} \tilde{m}},Y_{\ell (m'+\tilde{m})}\ra_{L^2(\s^2)}=0$ if $\ell'+\tilde{\ell}<\ell$ or $|\tilde{\ell}-\ell'|>\ell$.
Finally, observe that by Cauchy--Schwarz and $\sum_{j=1}j^{-2}=\frac{\pi^2}{6}$:
\begin{multline*}
	|(gf)_{\ell}^{(n)}|^2=\left|\sum_{k=0}^n \binom{n}{k}(g^{(n-k)}f^{(k)})_{\ell}\right|^2\leq\sum_{k=0}^n(n-k+1)^{-2}  \sum_{k=0}^n\binom{n}{k}^2 (n-k+1)^2 |(g^{(n-k)}f^{(k)})_{\ell}|^2\\
	\leq \frac{\pi^2}{6}  \sum_{k=0}^n\binom{n}{k}^2 (n-k+1)^2 |(g^{(n-k)}f^{(k)})_{\ell}|^2.
\end{multline*}
Combining the above observations, we obtain:
\begin{multline*}
	\sum_{\ell\in \N_0}  \sum_{n=0}^{\infty}\frac{(d_0x_0)^{2n}}{(\max\{\ell+1,n+1\})^{2n}}\int_{\s^2}|(gf)_{\ell}^{(n)}|^2\,d\upsigma\\
	\leq C\sum_{\ell\in \N_0} \sum_{\{\ell',\tilde{\ell}\in \N_0\,:\,|\tilde{\ell}-\ell'|\leq \ell\leq \ell+\ell'\}}  \sum_{n=0}^{\infty}\frac{(d_0x_0)^{2n}}{(\max\{\ell+1,n+1\})^{2n}}\\\
	\times\sum_{k=0}^n\binom{n}{k}^2 |f^{(k)}_{\tilde{\ell}\tilde{m}}|^2(\ell'+1)^4(n-k+1)^2|g^{(n-k)}_{\ell'm'}|^2.
\end{multline*}

We introduce the summation index $j=n-k$, we rearrange the summation and separate the cases $\tilde{\ell}\leq \ell$ and $\tilde{\ell}>\ell$ to obtain:
\begin{multline*}
	\sum_{n=0}^{\infty}\sum_{\ell\in \N_0} \frac{(d_0x_0)^{2n}}{(\max\{\ell+1,n+1\})^{2n}}\int_{\s^2}|(gf)_{\ell}^{(n)}|^2\,d\upsigma\\
	\leq C\sum_{k=0}^{\infty}\sum_{\tilde{\ell}\in \N_0}\sum_{\tilde{m}=-\tilde{\ell}}^{\tilde{\ell}}\frac{(d_0x_0)^{2k}}{(\max\{\tilde{\ell}+1,k+1\})^{2k}}|f_{\tilde{\ell}\tilde{m}}^{(k)}|^2\\
	\times\Biggl[\sum_{\ell'\in \N_0}\sum_{m'=-\ell'}^{\ell'}\sum_{\ell=|\tilde{\ell}-\ell'|}^{\tilde{\ell}+\ell'}\sum_{j=0}^{\infty}(d_0x_0)^{2j}\frac{(\max\{\tilde{\ell}+1,k+1\})^{2k}}{(\max\{\ell+1,k+j+1\})^{2(k+j)}}\binom{k+j}{k}^2\\
	\times(j+1)^2 (\ell'+1)^4|g_{\ell'm'}^{(j)}|^2\Biggr]
\end{multline*}
We will now bound the terms inside the square brackets uniformly in $k$ and $\tilde{\ell}$.

Note first of all that
\begin{equation*}
	(k+j+1)^{-2j}\binom{k+j}{k}^2=\frac{1}{j!^2}(k+j+1)^{-2j}(k+j)^2(k+j-1)^2\ldots (k+1)^2\leq \frac{1}{j!^2}.
\end{equation*}

Now suppose $\tilde{\ell}+1\leq k+1$. Then 
\begin{equation*}
	\left(\frac{\max\{\tilde{\ell}+1,k+1\}}{\max\{\ell+1,k+j+1\}}\right)^{2k}\leq \left(\frac{k+1}{k+j+1}\right)^{2k}\leq 1.
\end{equation*}
Suppose $\tilde{\ell}+1>k+1$. Then
\begin{equation*}
	\left(\frac{\max\{\tilde{\ell}+1,k+1\}}{\max\{\ell+1,k+j+1\}}\right)^{2k}\leq \left(\frac{\tilde{\ell}+1}{\max\{\ell+1,k+1\}}\right)^{2k}=e^{2k\log\left(\frac{\tilde{\ell}+1}{\max\{\ell+1,k+1\}}\right)}.
\end{equation*}
If $\tilde{\ell}\leq \ell$, then the right-hand side can be bounded by 1. Suppose therefore that $\tilde{\ell}>\ell$. We have that $\tilde{\ell}\leq \ell'+\ell$, so we can bound further:
\begin{equation*}
	e^{2k\log\left(\frac{\tilde{\ell}+1}{\max\{\ell+1,k+1\}}\right)}\leq e^{2k\log\left(1+\frac{\ell'}{k+1}\right)}\leq e^{2\ell'\frac{k+1}{\ell'}\log\left(1+\frac{\ell'}{k+1}\right)}\leq e^{2\ell'},
\end{equation*}
where we obtained the final inequality by using that $x^{-1}\log (1+x)\leq 1$ for $x\geq 0$.

We can therefore infer that:
\begin{multline*}
	\sum_{n=0}^{\infty}\sum_{\ell\in \N_0}\frac{(d_0x_0)^{2n}}{(\max\{\ell+1,n+1\})^{2n}}\int_{\s^2}|(gf)_{\ell}^{(n)}|^2\,d\upsigma\\
	\leq C\sum_{k=0}^{\infty}\sum_{\tilde{\ell}\in \N_0}\sum_{\tilde{m}=-\tilde{\ell}}^{\tilde{\ell}}\frac{(d_0x_0)^{2k}}{(\max\{\tilde{\ell}+1,k+1\})^{2k}}|f_{\tilde{\ell}\tilde{m}}^{(k)}|^2\\
	\times\Biggl[\sum_{\ell'\in \N_0}\sum_{m'=-\ell'}^{m=\ell'}\sum_{\ell=|\tilde{\ell}-\ell'|}^{\tilde{\ell}+\ell'}\sum_{j=0}^{\infty}(\ell'+1)^4e^{2 \ell'}\frac{(j+1)^2 (d_0x_0)^{2j}}{j!^2 }|g_{\ell'm'}^{(j)}|^2\Biggr]
\end{multline*}
Then \eqref{eq:spherenormproduct} follows from the observation that $\sum_{\ell=|\tilde{\ell}-\ell'|}^{\tilde{\ell}+\ell'}1=\tilde{\ell}+\ell'-|\tilde{\ell}-\ell'|\leq 2\ell'$.

The estimate \eqref{eq:spherenormproductw} follows in an analogous manner, where the above computations become simplified since the relevant norms depend only a finite number of derivatives, and we use that for $\tilde{\ell}\leq \ell$, $(\ell(\ell+1)+1)^{-w}\leq (\tilde{\ell}(\tilde{\ell}+1)+1)^{-w}$ and for $\tilde{\ell}>\ell$ (in which case $\tilde{\ell}<\ell+\ell'$):
\begin{multline}
(\ell(\ell+1)+1)^{-w}= (\tilde{\ell}(\tilde{\ell}+1)+1)^{-w}\left(\frac{\tilde{\ell}(\tilde{\ell}+1)+1}{\ell(\ell+1)+1}\right)^w\leq (\tilde{\ell}(\tilde{\ell}+1)+1)^{-w}\left(\frac{(\ell+\ell')(\ell+\ell'+1)}{\ell(\ell+1)+1}\right)^w\\
=(\tilde{\ell}(\tilde{\ell}+1)+1)^{-w}\left(1+\frac{\ell'(\ell'+1)}{\ell(\ell+1)+1}+\frac{2\ell \ell'}{\ell(\ell+1)+1}\right)^w\\
\leq (\tilde{\ell}(\tilde{\ell}+1)+1)^{-w}\left(1+\ell'(\ell'+1)\right)^w.
\end{multline}

\end{proof}

\begin{lemma}
\label{lm:productGnorm}
Ther exists a constants $C,\beta>0$ such that for all appropriately regular functions $f,g; [0,x_0]\times \s^2\to \C$, we can estimate:
\begin{equation}
\label{eq:productGnorm}
\|gf\|^2_{G_{\sigma,R_{\infty}}}\leq C\|f\|^2_{G_{\sigma,R_{\infty}}}\left[\sum_{j\in \N_0}\sum_{\ell'\in \N_0}\frac{e^{5(\ell'+j)}}{j!^4}\|\partial_x^jg_{\ell'}\|^2_{L^{\infty}([0,x_0]\times \s^2)}\right].
\end{equation}
\end{lemma}
\begin{proof}
We can express:
\begin{multline*}
\|gf\|^2_{G_{\sigma,R_{\infty}}}=\sum_{\ell\in \N_0}\sum_{n=0}^{\ell+l-1}\frac{(4\sigma)^{2n}(\ell-(n+1))!^2}{(\ell+n+1)!^2}\int_0^{x_0}\int_{\s^2}|\partial_x^{n}(gf)_{\ell}|^2\,d\upsigma dx\\
	+(\lf+1)^3\sum_{n=\ell+l}^{\infty}\frac{\sigma^{2n}}{n!^2(n+1)!^2}\int_0^{x_0}\int_{\s^2}|\partial_x^{n}(gf)_{\ell}|^2\,d\upsigma dx.
\end{multline*}

We then proceed as in the proof of Lemma \ref{lm:productBnorm} by decomposing $g$ and $f$ into spherical harmonics: $g=\sum_{\ell'\in \N_0} \sum_{m=-\ell'}^{\ell'} g_{\ell'm'}Y_{\ell' m'}$, $f=\sum_{\tilde{\ell}\in \N_0} \sum_{\tilde{m}=-\tilde{\ell}}^{\tilde{\ell}} f_{\tilde{\ell} \tilde{m}}Y_{\tilde{\ell} \tilde{m}}$ and using that
\begin{equation*}
|(gf)_{\ell}^{(n)}|^2\leq \frac{\pi^2}{6}  \sum_{k=0}^n\binom{n}{k}^2 (n-k+1)^2 |(g^{(n-k)}f^{(k)})_{\ell}|^2.
\end{equation*}

We then obtain:
\begin{multline*}
\|gf\|^2_{G_{\sigma,R_{\infty}}}\leq C \sum_{\ell\in \N_0}\sum_{\{\ell',\tilde{\ell}\in \N_0\,|\,|\tilde{\ell}-\ell'|\leq \ell\leq \ell+\ell'\}} \sum_{\tilde{m}=-\tilde{\ell}}^{\tilde{\ell}}  \sum_{m=-\ell'}^{\ell'} \sum_{n=0}^{\ell-1}\frac{(4\sigma)^{2n}(\ell-(n+1))!^2}{(\ell+n+1)!^2} \\
\times\int_0^{x_0}|f_{\tilde{\ell}\tilde{m}}^{(k)}|^2\binom{n}{k}^2(n-k+1)^2\|Y_{\ell' m'}\|_{L^{\infty}(\s^2)}^2|g_{\ell'm'}^{(n-k)}|^2\,d\upsigma dx\\
	+\left(\frac{4\sigma }{d_0x_0}\right)^{-2\ell}(\ell+1)^3\sum_{n=\ell}^{\infty}\frac{\sigma^{2n}}{n!^2(n+1)!^2}\int_0^{x_0}\int_{\s^2}|f_{\tilde{\ell}\tilde{m}}^{(k)}|^2\binom{n}{k}^2(n-k+1)^2\|Y_{\ell' m'}\|_{L^{\infty}(\s^2)}^2|g_{\ell'm'}^{(n-k)}|^2\,d\upsigma dx.
\end{multline*}

By introducing the summation index $j=n-k$ and taking $\sup_{x\in[0,x_0]}|g^{(n-k)}_{\ell'm'}|^2(x)$ outside of the $x$-integral, we can rewrite the right-hand side above to obtain:
\begin{multline*}
\|gf\|^2_{G_{\sigma,R_{\infty}}}\leq C \sum_{\tilde{\ell}\in \N_0}\sum_{\tilde{m}=-\tilde{\ell}}^{\tilde{\ell}}\sum_{k=0}^{\tilde{\ell}-1}\frac{(4\sigma)^{2k}(\tilde{\ell}-(k+1))!^2}{(\tilde{\ell}+k+1)!^2} \\
\times  \int_0^{x_0}|f^{(k)}_{\tilde{\ell}\tilde{m}}|^2\,dx   \\
\times\Biggl[  \sum_{\ell'\in \N_0}\sum_{\ell=|\tilde{\ell}-\ell'|}^{\tilde{\ell}+\ell'}\sum_{m=-\ell'}^{\ell'} \Biggl(\sum_{j=0}^{\ell-1-k}(4\sigma)^{2j} \underbrace{\frac{(\tilde{\ell}+k+1)!^2({\ell}-(k+j+1))!^2}{({\ell}+k+j+1)!^2(\tilde{\ell}-(k+1))!^2}\frac{(k+j)!^2}{k!^2j!^2}}_{=:J_1}(j+1)^2\|g^{(j)}_{\ell' m'}Y_{\ell'm'}\|_{L^{\infty}}^2\\
+\sum_{j=\ell-k}^{\infty}4^{-2k}(4\sigma)^{2j}(\ell+1)^3 \underbrace{\frac{(\tilde{\ell}+k+1)!^2}{(\tilde{\ell}-(k+1))!^2(k+j+1)!^2(k+j)!^2}\frac{(k+j)!^2}{k!^2j!^2}}_{=:J_2}(j+1)^2\|g^{(j)}_{\ell' m'}Y_{\ell'm'}\|_{L^{\infty}}^2\Biggr)\Biggr]\\
+\sum_{k=\tilde{\ell}+l}^{\infty}\frac{\sigma^{2k}(\tilde{\ell}+1)^3}{(k+1)!^2k!^2}\int_0^{x_0}|f^{(k)}_{\tilde{\ell}\tilde{m}}|^2\,dx\Biggl[  \sum_{\ell'\in \N_0}\sum_{\ell=|\tilde{\ell}-\ell'|}^{\tilde{\ell}+\ell'}\sum_{m=-\ell'}^{\ell'} \Biggl(\sum_{j=0}^{\min\{\ell-1-k,0\}}4^{2k}(4\sigma)^{2j}(\tilde{\ell}+1)^{-3}   \\
\times \underbrace{\frac{(k+1)!^2k!^2({\ell}-(k+j+1))!^2}{({\ell}+k+j+1)!^2}\frac{(k+j)!^2}{k!^2j!^2}}_{=:J_3}(j+1)^2\|g^{(j)}_{\ell' m'}Y_{\ell'm'}\|_{L^{\infty}}^2\\
+\sum_{j=\max\{\ell+l-k,0\}}^{\infty}\sigma^{j}\left(\frac{\ell+1}{\tilde{\ell}+1}\right)^3 \underbrace{\frac{(k+1)!^2k!^2}{(k+j+1)!^2(k+j)!^2}\frac{(k+j)!^2}{k!^2j!^2}}_{=:J_4}(j+1)^2\|g^{(j)}_{\ell' m'}Y_{\ell'm'}\|_{L^{\infty}}^2\Biggr)\Biggr].
\end{multline*}
We will simplify the terms inside square brackets.
We will first estimate $J_1$. We have that:
\begin{equation*}
J_1=\frac{(\tilde{\ell}+k+1)^2(\tilde{\ell}+k)^2\ldots (\tilde{\ell}-k)^2}{(\ell+k+j+1)^2\ldots (\ell+k+j)^2\ldots (\ell-k-j)^2}\frac{(k+j)!^2}{k!^2j!^2}
\end{equation*}
Suppose that $\tilde{\ell}\leq \ell$. Then, using that $\ell\geq k+j+1$, we estimate
\begin{multline*}
J_1\leq \frac{1}{(\ell+k+j+1)^2\ldots (\ell+k+1)^2(\ell-k-1)^2 \ldots (\ell-k-j)^2}\frac{(k+j)!^2}{k!^2j!^2}\\
\leq \frac{(k+j)^2(k+j-1)^2\ldots (k+1)^2}{(2k+j+2+j)^2\ldots (2k+j+2)^2}\frac{1}{j!^4}\leq \frac{1}{j!^4}.
\end{multline*}
Now suppose that $\tilde{\ell}\geq \ell+1$. Then we estimate using that $k+j\leq \ell-1$
\begin{multline*}
J_1\leq \frac{(\tilde{\ell}+k+1)^2\ldots (\tilde{\ell}-k)^2}{(\ell+k+1)^2\ldots (\ell-k)^2}\\
\times \frac{1}{(\ell+k+j+1)^2\ldots (\ell+k+2)^2(\ell-k-1)^2\ldots (\ell-k-j)^2}\frac{(k+j)!^2}{k!^2j!^2}\\
\leq  \frac{(\tilde{\ell}+k+1)^2\ldots(\ell+k+2)^2}{(\tilde{\ell}-k-1)^2\ldots (\ell-k)^2}\frac{1}{j!^4}\\
\leq\frac{(2\ell+1+(\tilde{\ell}-\ell-1))^2\ldots(2\ell+1)^2}{(\tilde{\ell}-\ell)!^2}\frac{1}{j!^4}\\
\leq Ce^{2(\tilde{\ell}-\ell)}(\tilde{\ell}-\ell)^{-2(\tilde{\ell}-\ell-1)-3}(2\ell+l+1+(\tilde{\ell}-\ell-1))^{2(\tilde{\ell}-\ell-1)}\frac{1}{j!^4}\\
\leq Ce^{2(\tilde{\ell}-\ell)}(\tilde{\ell}-\ell)^{-3}\left(1+\frac{2\ell+1}{\tilde{\ell}-\ell}\right)^{2(\tilde{\ell}-\ell-1)}\frac{1}{j!^4}
\end{multline*}
Suppose now that $\tilde{\ell}-\ell\geq \frac{1}{2}(2\ell+1)$. Then we can further estimate:
\begin{equation*}
J_1\leq Ce^{2(1+\log\left(\frac{3}{2}\right))(\tilde{\ell}-\ell)}(\tilde{\ell}-\ell)^{-3}\frac{1}{j!^4}\leq Ce^{2(1+\log\left(\frac{3}{2}\right))\ell'}\ell^{-3}\frac{1}{j!^4}.
\end{equation*}
If instead $\tilde{\ell}-\ell< \frac{1}{2}(2\ell+1)$, then $\tilde{\ell}\leq 2\ell+\frac{1}{2}$ and we estimate:
\begin{equation*}
J_1\leq  \frac{(\tilde{\ell}+k+1)^2\ldots (\tilde{\ell}-k)^2}{(\ell+k+1)^2\ldots (\ell-k)^2} \frac{1}{j!^4}\leq 4^{2k+1} \frac{1}{j!^4}.
\end{equation*}
We now turn to $J_2$.  

Note first that $k\geq \ell-j\geq \tilde{\ell}-(\ell'+j)$, so $\tilde{\ell}\leq k+(\ell'+j)$. Hence,
\begin{equation*}
\frac{(\tilde{\ell}+k+1)!^2}{(\tilde{\ell}-(k+1))!^2(k+1)!^4}\leq \frac{(\ell'+j-1+2(k+1))!^2}{(\ell'+j-1)!^2(k+1)!^4}
\end{equation*}
Via Stirling's approximation, we can further estimate
\begin{multline*}
 \frac{(\ell'+j-1+2(k+1))!^2}{(\ell'+j-1)!^2(k+1)!^4}\leq C \left(1+\frac{2(k+1)}{\ell'+j-1}\right)^{2(\ell'+j-1)+1}\left(2+\frac{\ell'+j-1}{k+1}\right)^{4(k+1)}\\
 \leq Ce^{2(k+1)\frac{\ell'+j-1}{2(k+1)}\log\left(1+\frac{2(k+1)}{\ell'+j-1}\right)}e^{4(\ell'+j-1) \frac{k+1}{\ell'+j-1} \log\left(2+\frac{\ell'+j-1}{k+1}\right)}\\
 \leq Ce^{2(k+1)\frac{\ell'+j-1}{2(k+1)}\log\left(1+\frac{2(k+1)}{\ell'+j-1}\right)}e^{4\log\left(1+\ell'+j\right)},
\end{multline*}
where we obtained the final inequality by using that $x^{-1}\log(2+x)$ is strictly decreasing. Suppose that $2(k+1)\geq \ell'+j-1$. Then, by the decreasing property of $x^{-1}\log(1+x)$, we have that
\begin{equation*}
e^{2(k+1)\frac{\ell'+j-1}{2(k+1)}\log\left(1+\frac{2(k+1)}{\ell'+j-1}\right)}\leq e^{2\log(2)(k+1)}\leq 2^{2(k+1)}.
\end{equation*}
Suppose now that $2(k+1)< \ell'+j-1$. Then we can estimate instead:
\begin{equation*}
e^{2(k+1)\frac{\ell'+j-1}{2(k+1)}\log\left(1+\frac{2(k+1)}{\ell'+j-1}\right)}= \left(1+\frac{2(k+1)}{\ell'+j-1}\right)^{2(\ell'+j-1)+1}\leq 2^{2(\ell'+j-1)+1}.
\end{equation*}
Hence, there exists constants $C,\beta>0$ such that
\begin{equation*}
\frac{(\tilde{\ell}+k+1)!^2}{(\tilde{\ell}-(k+1))!^2(k+1)!^4}\leq C2^{2\ell'+2j}2^{2k}.
\end{equation*}

From the above, we can therefore estimate $J_2$ as follows: 
\begin{equation*}
4^{-2k}J_2\leq 2^{-2k}2^{2\ell'+2j}\frac{(k+1)^2(k+1)!^2}{(k+j+1)!^2j!^2}\leq C 2^{-2k}k^2\frac{1}{j!^4}
\end{equation*}
Note also that by $\ell\leq k+j$
\begin{equation*}
(\ell+1)^3\leq (k+j+1)^3\leq (k+1)^3 \left(1+\frac{j}{k+1}\right)^3\leq (k+1)^3 (j+1)^3.
\end{equation*}

Consider now $J_3$. Note that by the Stirling approximation and the fact that $\ell\geq k+j$
\begin{equation*}
\frac{(\ell-(k+j+1))!^2(k+j+1)!^4}{(\ell+(k+j+1))!^2}\leq \frac{(k+j+1)!^4}{(2(k+j+1))!^2}\leq C 2^{-4(k+j+1)}(k+j+1),
\end{equation*}
so
\begin{equation*}
J_3\leq C 4^{-2(k+j)}\frac{(k+1)!^2}{(k+j+1)!^2(k+j+1)}\frac{1}{j!^2}\leq C 4^{-2(k+j)} (k+j+1)^{-1}\frac{1}{j!^4}.
\end{equation*}

We consider $J_4$. Note that
\begin{equation*}
J_4=\frac{(k+1)!^2}{(k+j+1)!^2}\frac{1}{j!^2}=\frac{1}{(k+j)^2\ldots (k+2)!^2}\frac{1}{j!^2}\leq \frac{1}{j!^4}.
\end{equation*}
Furthermore,
\begin{equation*}
\left(\frac{\ell+1}{\tilde{\ell}+1}\right)^3=\left(1+ \frac{\ell-\tilde{\ell}}{\tilde{\ell}+1}\right)^3\leq C(1+\ell')^3.
\end{equation*}
The inequality \eqref{eq:productGnorm} follows by combining the above estimates.
\end{proof}

\bibliographystyle{alpha}

\newcommand{\etalchar}[1]{$^{#1}$}


\begin{thebibliography}{CDM{\etalchar{+}}22}

\bibitem[AAG18a]{aag18}
Y.~Angelopoulos, S.~Aretakis, and D.~Gajic.
\newblock Horizon hair of extremal black holes and measurements at null
  infinity.
\newblock {\em Physical Review Letters}, 121(13):131102, 2018.

\bibitem[AAG18b]{paper2}
Y.~Angelopoulos, S.~Aretakis, and D.~Gajic.
\newblock Late-time asymptotics for the wave equation on spherically symmetric,
  stationary backgrounds.
\newblock {\em Advances in Mathematics}, 323:529--621, 2018.

\bibitem[AAG20]{paper4}
Y.~Angelopoulos, S.~Aretakis, and D.~Gajic.
\newblock Late-time asymptotics for the wave equation on extremal
  {R}eissner--{N}ordstr\"{o}m backgrounds.
\newblock {\em Advances in Mathematics}, 375:107363, 2020.

\bibitem[AAG23a]{aagkerr}
Y.~Angelopoulos, S.~Aretakis, and D.~Gajic.
\newblock Late-time tails and mode coupling of linear waves on {K}err
  spacetimes.
\newblock {\em Advances in Mathematics}, 417:108939, 2023.

\bibitem[AAG23b]{aagprice}
Y.~Angelopoulos, S.~Aretakis, and D.~Gajic.
\newblock {Price's law and precise late-time asymptotics for subextremal
  Reissner--Nordstr{\"o}m black holes}.
\newblock {\em Annales Henri Poincar{\'e}}, 24(9):3215--3287, 2023.

\bibitem[Abb21]{ligo21}
Virgo) Abbott, R. \emph{et al} (LIGO~Scientific.
\newblock Tests of general relativity with binary black holes from the second
  ligo-virgo gravitational-wave transient catalog.
\newblock {\em Physical review D}, 103(12):122002, 2021.

\bibitem[AC71]{agco71}
J.~Aguilar and J.-M. Combes.
\newblock A class of analytic perturbations for one-body schr{\"o}dinger
  hamiltonians.
\newblock {\em Communications in Mathematical Physics}, 22:269--279, 1971.

\bibitem[AKS23]{aags23}
S.~Aretakis, G.~Khanna, and S.~Sabharwal.
\newblock An observational signature for extremal black holes.
\newblock {\em arXiv:2307.03963}, 2023.

\bibitem[BC71]{baco71}
E.~Balslev and J.-M. Combes.
\newblock Spectral properties of many-body schr{\"o}dinger operators with
  dilatation-analytic interactions.
\newblock {\em Communications in Mathematical Physics}, 22:280--294, 1971.

\bibitem[BCS09]{bcs09}
E.~Berti, V.~Cardoso, and A.~O. Starinets.
\newblock Quasinormal modes of black holes and black branes.
\newblock {\em Classical and Quantum Gravity}, 26(163001), 2009.

\bibitem[BH35]{baha35}
W.G. Baber and H.R. Hass{\'e}.
\newblock The two centre problem in wave mechanics.
\newblock {\em Mathematical Proceedings of the Cambridge Philosophical
  Society}, 31(4):564--581, 1935.

\bibitem[BH08]{bonyh}
J.-F. Bony and D.~Hafner.
\newblock {Decay and non-decay of the local energy for the wave equation in the
  {D}e {S}itter - {S}chwarzschild metric}.
\newblock {\em Communications in Mathematical Physics}, 282:697--719, 2008.

\bibitem[BKS23]{bgs23}
L.~M. Burko, G.~Khanna, and S.~Sabharwal.
\newblock Aretakis hair for extreme kerr black holes with axisymmetric scalar
  perturbations.
\newblock {\em Physical Review D}, 107(12):124023, 2023.

\bibitem[BMB93]{ba}
A.~Bachelot and A.~Motet-Bachelot.
\newblock Les r\'{e}sonances d' un trou noir de {S}chwarzschild.
\newblock {\em Ann. Inst. H. Poincar\'{e}}, 59:368, 1993.

\bibitem[BZ97]{bazw97}
A.~S. Barreto and M.~Zworski.
\newblock Distribution of resonances for spherical black holes.
\newblock {\em Mathematical Research Letters}, 4(1):103--121, 1997.

\bibitem[Car68]{car68}
B.~Carter.
\newblock {Hamilton-Jacobi and Schrodinger separable solutions of Einstein’s
  equations}.
\newblock {\em Communications in Mathematical Physics}, 10:280--310, 1968.

\bibitem[CD75]{chadet75}
S.~Chandrasekhar and S.~Detweiler.
\newblock {The quasi-normal modes of the Schwarzschild black hole}.
\newblock {\em Proceedings of the Royal Society of London. A. Mathematical and
  Physical Sciences}, 344(1639):441--452, 1975.

\bibitem[CDM{\etalchar{+}}22]{cheetal22}
M.~H.-Y. Cheung, K~Destounis, R.~P. Macedo, E.~Berti, and V.~Cardoso.
\newblock Destabilizing the fundamental mode of black holes: the elephant and
  the flea.
\newblock {\em Physical Review Letters}, 128(11):111103, 2022.

\bibitem[CT22]{cate22}
M.~Casals and R.~{Teixeira da Costa}.
\newblock Hidden spectral symmetries and mode stability of subextremal kerr
  (-de sitter) black holes.
\newblock {\em Communications in Mathematical Physics}, 394(2):797--832, 2022.

\bibitem[Det77]{det77}
S.~Detweiler.
\newblock On resonant oscillations of a rapidly rotating black hole.
\newblock {\em Proceedings of the Royal Society of London. A. Mathematical and
  Physical Sciences}, 352(1670):381--395, 1977.

\bibitem[DGM20]{dgm20}
R.~G. Daghigh, M.~D. Green, and J.~C. Morey.
\newblock Significance of black hole quasinormal modes: A closer look.
\newblock {\em Physical Review D}, 101(10):104009, 2020.

\bibitem[DR09]{redshift}
M.~Dafermos and I.~Rodnianski.
\newblock The redshift effect and radiation decay on black hole spacetimes.
\newblock {\em Comm. Pure Appl. Math.}, 62:859--919, ar{X}iv:0512.119, 2009.

\bibitem[DRSR16]{part3}
M.~Dafermos, I.~Rodnianski, and Y.~Shlapentokh-Rothman.
\newblock Decay for solutions of the wave equation on {K}err exterior
  spacetimes {III: The full subextremal case} $|a| < m$.
\newblock {\em Annals of Math}, 183:787--913, 2016.

\bibitem[Dya11]{semyon1}
S.~Dyatlov.
\newblock Quasi-normal modes and exponential energy decay for the {K}err--de
  {S}itter black hole.
\newblock {\em Comm. Math. Phys.}, 306:119--163, 2011.

\bibitem[Dya12]{dya12}
S.~Dyatlov.
\newblock {Asymptotic distribution of quasi-normal modes for Kerr--de Sitter
  black holes}.
\newblock {\em Annales Henri Poincar{\'e}}, 13(5):1101--1166, 2012.

\bibitem[DZ19]{dyzw19}
S.~Dyatlov and M.~Zworski.
\newblock {\em Mathematical theory of scattering resonances}, volume 200.
\newblock American Mathematical Soc., 2019.

\bibitem[Eva98]{evans}
L.~C. Evans.
\newblock {\em Partial Differential Equations}.
\newblock Graduate Studies in Mathematics, 1998.

\bibitem[Gaj23]{gaj22b}
D.~Gajic.
\newblock {Azimuthal instabilities on extremal Kerr}.
\newblock {\em arXiv:2302.06636}, 2023.

\bibitem[Gan14]{gan14}
O.~Gannot.
\newblock Quasinormal modes for schwarzschild--ads black holes: exponential
  convergence to the real axis.
\newblock {\em Communications in Mathematical Physics}, 330:771--799, 2014.

\bibitem[GK22]{gake22}
D.~Gajic and Leonhard M.~A. Kehrberger.
\newblock On the relation between asymptotic charges, the failure of peeling
  and late-time tails.
\newblock {\em Classical and Quantum Gravity}, 39(19):195006, 2022.

\bibitem[GV17]{givo17}
G.~W. Gibbons and M.~S. Volkov.
\newblock {Zero mass limit of Kerr spacetime is a wormhole}.
\newblock {\em Physical Review D}, 96(2):024053, 2017.

\bibitem[GW20]{gajwar19b}
D.~Gajic and C.~Warnick.
\newblock A model problem for quasinormal ringdown of asymptotically flat or
  extremal black holes.
\newblock {\em Journal of Mathematical Physics}, 61(10), 2020.

\bibitem[GW21]{gajwar19a}
D.~Gajic and C.~Warnick.
\newblock {Quasinormal modes in extremal Reissner--Nordstr\"om spacetimes}.
\newblock {\em Communications in Mathematical Physics}, 385(3):1395--1498,
  2021.

\bibitem[GZ21]{gazw21}
J.~Galkowski and M.~Zworski.
\newblock {Outgoing solutions via Gevrey-2 properties}.
\newblock {\em Annals of PDE}, 7:1--13, 2021.

\bibitem[Hil48]{hil48}
E.~Hille.
\newblock {\em {Functional Analysis and Semi-Groups}}, volume XXXI of
  Colloquium Publications.
\newblock American Mathematical Society, 1948.

\bibitem[Hin22]{hintzprice}
P.~Hintz.
\newblock A sharp version of {P}rice's law for wave decay on asymptotically
  flat spacetimes.
\newblock {\em Communications in Mathematical Physics}, 389(1):491--542, 2022.

\bibitem[Hin24]{hin24}
P.~Hintz.
\newblock {Mode stability and shallow quasinormal modes of Kerr--de Sitter
  black holes away from extremality}.
\newblock {\em Journal of the European Mathematical Society}, online first,
  2024.

\bibitem[HS14]{gusmu2}
G.~Holzegel and J.~Smulevici.
\newblock {Quasimodes and a lower bound on the uniform energy decay rate for
  Kerr--AdS spacetimes}.
\newblock {\em Analysis \& PDE}, 7(5):1057--1090, 2014.

\bibitem[HX21]{hixi21}
P.~Hintz and YQ. Xie.
\newblock {Quasinormal modes and dual resonant states on de Sitter space}.
\newblock {\em Physical Review D}, 104(6):064037, 2021.

\bibitem[HX22]{hixi22}
P.~Hintz and YQ~Xie.
\newblock {Quasinormal modes of small Schwarzschild--de Sitter black holes}.
\newblock {\em Journal of Mathematical Physics}, 63(1), 2022.

\bibitem[HZ24]{hizw24}
M.~Hitrik and M.~Zworski.
\newblock {Overdamped QNM for Schwarzschild black holes}.
\newblock {\em arXiv:2406.15924}, 2024.

\bibitem[Jaf34]{jaf34}
G.~Jaff{\'e}.
\newblock {Zur theorie des Wasserstoffmolek{\"u}lions}.
\newblock {\em {Zeitschrift f{\"u}r Physik}}, 87:535--544, 1934.

\bibitem[JMAS21]{jamash21}
J.~L. Jaramillo, R.~P. Macedo, and L.~Al~Sheikh.
\newblock Pseudospectrum and black hole quasinormal mode instability.
\newblock {\em Physical Review X}, 11(3):031003, 2021.

\bibitem[Joy22]{joy22}
J.~Joykutty.
\newblock Existence of zero-damped quasinormal frequencies for nearly extremal
  black holes.
\newblock {\em Annales Henri Poincar{\'e}}, 23(12):4343--4390, 2022.

\bibitem[Ker63]{kerr63}
R.~P. Kerr.
\newblock {Gravitational field of a spinning mass as an example of
  algebraically special metrics}.
\newblock {\em Physical Review Letters}, 11(5):237, 1963.

\bibitem[KS99]{kosc99}
K.~D. Kokkotas and B.~G. Schmidt.
\newblock {Quasi-normal modes of stars and black holes}.
\newblock {\em Living Reviews in Relativity}, 2:1--72, 1999.

\bibitem[Lea85]{leav85}
E.~W. Leaver.
\newblock An analytic representation for the quasi-normal modes of kerr black
  holes.
\newblock {\em Proceedings of the Royal Society of London. A. Mathematical and
  Physical Sciences}, 402(1823):285--298, 1985.

\bibitem[LO24]{lukoh24}
J.~Luk and S.-J. Oh.
\newblock Late time tail of waves on dynamic asymptotically flat spacetimes of
  odd space dimensions.
\newblock {\em arXiv:2404.02220}, 2024.

\bibitem[LP68]{laxphillips}
P.~D. Lax and R.~S. Phillips.
\newblock {\em Scattering theory}.
\newblock Academic press, 1968.

\bibitem[Mos17]{mos17}
G.~Moschidis.
\newblock {Superradiant instabilities for short-range non-negative potentials
  on Kerr spacetimes and applications}.
\newblock {\em Journal of Functional Analysis}, 273(8):2719--2813, 2017.

\bibitem[MS54]{ms54}
J.~Meixner and F.~W. Sch{\"a}fke.
\newblock {\em {Mathieusche Funktionen und Sph{\"a}roidfunktionen: Mit
  Anwendungen auf Physikalische und Technische Probleme}}.
\newblock Springer Berlin Heidelberg, 1954.

\bibitem[Nol96]{nol96}
H.-P. Nollert.
\newblock About the significance of quasinormal modes of black holes.
\newblock {\em Physical Review D}, 53(8):4397, 1996.

\bibitem[NP68]{np2}
E.~T. Newman and R.~Penrose.
\newblock New conservation laws for zero rest mass fields in asymptotically
  flat space-time.
\newblock {\em Proc. R. Soc. A}, 305:175204, 1968.

\bibitem[NP99]{npr99}
H.-P. Nollert and R.~H. Price.
\newblock Quantifying excitations of quasinormal mode systems.
\newblock {\em Journal of Mathematical Physics}, 40(2):980--1010, 1999.

\bibitem[PI98]{preis98}
F.~Pretorius and W.~Israel.
\newblock {Quasi-spherical light cones of the Kerr geometry}.
\newblock {\em Classical and quantum Gravity}, 15(8):2289, 1998.

\bibitem[Pra99]{pra99}
D.~W. Pravica.
\newblock Top resonances of a black hole.
\newblock {\em Proceedings of the Royal Society of London. Series A:
  Mathematical, Physical and Engineering Sciences}, 455(1988):3003--3018, 1999.

\bibitem[Pre71]{press71}
W.~H. Press.
\newblock {Long wave trains of gravitational waves from a vibrating black
  hole}.
\newblock {\em Astrophysical Journal}, 170:L105, 1971.

\bibitem[Pri72]{Price1972}
R.~Price.
\newblock Non-spherical perturbations of relativistic gravitational collapse.
  {I}. {S}calar and gravitational perturbations.
\newblock {\em Phys. Rev. D}, 3:2419--2438, 1972.

\bibitem[PV24]{peva21}
O.~Petersen and A.~Vasy.
\newblock {Wave equations in the Kerr-de Sitter spacetime: the full subextremal
  range}.
\newblock {\em Journal of the European Mathematical Society}, online first,
  2024.

\bibitem[RR06]{rero06}
M.~Renardy and R.~C. Rogers.
\newblock {\em An introduction to partial differential equations}, volume~13.
\newblock Springer Science \& Business Media, 2006.

\bibitem[Sak94]{sak94}
J.~J. Sakurai.
\newblock {\em {Modern Quantum Mechanics, Revised Edition}}.
\newblock Addison Wesley Longman, 1994.

\bibitem[Sim73]{si73}
B.~Simon.
\newblock Resonances in n-body quantum systems with dilatation analytic
  potentials and the foundations of time-dependent perturbation theory.
\newblock {\em Annals of Mathematics}, 97(2):247--274, 1973.

\bibitem[SR15]{sr15}
Y.~Shlapentokh-Rothman.
\newblock {Quantitative mode stability for the wave equation on the Kerr
  spacetime}.
\newblock {\em Annales Henri Poincar{\'e}}, 16(1):289--345, 2015.

\bibitem[Stu24]{stuc24}
T.~Stucker.
\newblock {Quasinormal modes for the Kerr black hole}.
\newblock {\em to appear}, 2024.

\bibitem[SZ91]{sjzw91}
J.~Sjostrand and M.~Zworski.
\newblock Complex scaling and the distribution of scattering poles.
\newblock {\em Journal of the American Mathematical Society}, 4(4):729--769,
  1991.

\bibitem[Tat13]{tataru3}
D.~Tataru.
\newblock Local decay of waves on asymptotically flat stationary space-times.
\newblock {\em American Journal of Mathematics}, 135:361--401, 2013.

\bibitem[TdC20]{costa20}
R.~Teixeira~da Costa.
\newblock {Mode stability for the Teukolsky equation on extremal and
  subextremal Kerr spacetimes}.
\newblock {\em Communications in Mathematical Physics}, 378(1):705--781, 2020.

\bibitem[Vas13]{vasy1}
A.~Vasy.
\newblock Microlocal analysis of asymptotically hyperbolic and {K}err-de
  {S}itter spaces. {W}ith an appendix by {S}. {D}yatlov.
\newblock {\em Inventiones Math.}, 194:381--513, 2013.

\bibitem[Vis70]{vish70}
C.V. Vishveshwara.
\newblock {Scattering of Gravitational Radiation by a Schwarzschild
  Black-hole}.
\newblock {\em Nature}, 227(5261):936--938, 1970.

\bibitem[War15]{warn15}
C.~M. Warnick.
\newblock {On Quasinormal Modes of Asymptotically Anti-de Sitter Black Holes}.
\newblock {\em Communications in Mathematical Physics}, 333(2):959--1035, 2015.

\bibitem[War24]{war24}
C.~M. Warnick.
\newblock {(In)stabilities of de Sitter quasinormal spectra}.
\newblock {\em to appear}, 2024.

\bibitem[Whi89]{whiting}
B.~Whiting.
\newblock Mode stability of the {K}err black hole.
\newblock {\em J. Math. Phys.}, 30:1301, 1989.

\bibitem[Zer70]{zer70}
F.~J. Zerilli.
\newblock {Gravitational field of a particle falling in a Schwarzschild
  geometry analyzed in tensor harmonics}.
\newblock {\em Physical Review D}, 2(10):2141, 1970.

\end{thebibliography}
\end{document}